\documentclass[12pt,letterpaper]{article}
\usepackage{fullpage,indentfirst, rotating}
\usepackage{float}
\usepackage{color}

\usepackage{epsfig}
\usepackage{rotating}
\usepackage{subfigure}
\usepackage{amsfonts}
\usepackage{latexsym}
\usepackage{amsmath}
\usepackage{amssymb}
\usepackage{amsthm}
\usepackage{amstext}
\usepackage[utf8]{inputenc}
\usepackage[english]{babel}
\usepackage{mathbbol}
\usepackage{bigints}
\usepackage{longtable,array}
\usepackage{longtable}
\usepackage{mathtools}

\usepackage{bbm}
\usepackage{bm}
\usepackage{booktabs}
\usepackage{setspace}
\usepackage{xcolor}
\usepackage{blindtext}
\usepackage{sectsty}
\usepackage{todonotes}
\usepackage{bigints}
\usepackage{multicol}
\usepackage{multirow}
\usepackage{comment}
\usepackage[flushleft]{threeparttable}
\usepackage{enumitem}
\usepackage{relsize}
\sectionfont{\centering}
\subsectionfont{\centering}
\usepackage{dsfont}

\usepackage[ruled,vlined,longend]{algorithm2e}

\DontPrintSemicolon
\SetKwInput{KwData}{Data}
\SetKwInput{KwResult}{Result}

\usepackage{tikz}
\usetikzlibrary{positioning,shapes.geometric}

\usepackage[colorlinks=true, linkcolor=blue, citecolor=blue]{hyperref}
\usepackage{soul}
\usepackage[authoryear]{natbib}
\makeatletter
\newif\if@suppcites
\newwrite\supp@aux
\def\supp@routecites#1{%
  \@for\supp@key:=#1\do{%
    \if@suppcites
      \@ifundefined{supp@main@\supp@key}%
        {\immediate\write\supp@aux{\string\citation{\supp@key}}}%
        {\immediate\write\@auxout{\string\citation{\supp@key}}}%
    \else
      \expandafter\gdef\csname supp@main@\supp@key\endcsname{}%
      \immediate\write\@auxout{\string\citation{\supp@key}}%
    \fi}}
\renewcommand\NAT@sort@cites[1]{%
  \let\NAT@cite@list\@empty
  \@for\@citeb:=#1\do{\expandafter\NAT@star@cite\@citeb\@@}%
  \if@filesw\expandafter\supp@routecites\expandafter{\NAT@cite@list}\fi
  \@ifnum{\NAT@sort>\z@}{%
    \expandafter\NAT@sort@cites@\expandafter{\NAT@cite@list}%
  }{}%
}
\newcommand\StartSupplementCitations{%
  \if@filesw
    \immediate\openout\supp@aux=\jobname-supp.aux
    \immediate\write\supp@aux{\string\relax}%
    \immediate\write\supp@aux{\string\bibstyle{chicago}}%
    \immediate\write\supp@aux{\string\bibdata{references}}%
  \fi
  \global\@suppcitestrue}
\newcommand\SupplementReferences{%
  \if@filesw\immediate\closeout\supp@aux\fi
  }
\makeatother
\usepackage[top=1 in, bottom=1 in, left=1 in , right=1 in]{geometry}

\allowdisplaybreaks
\newcommand{\ben}{\begin{enumerate}}
\newcommand{\een}{\end{enumerate}}
\newcommand{\bit}{\begin{itemize}}
\newcommand{\eit}{\end{itemize}}

\newcommand{\indep}[0]{\ensuremath{\perp \! \! \! \perp}}

\newif\ifshowrevisions
\showrevisionstrue
\ifshowrevisions
  \newcommand{\rev}[1]{{\color{red}#1}}
  \newenvironment{revblock}{\color{red}}{}
  
\else
  \newcommand{\rev}[1]{#1}
  \newenvironment{revblock}{}{}
  
\fi

\DeclareMathOperator*{\cov}{\text{Cov}}
\DeclareMathOperator*{\var}{\text{Var}}

\providecommand{\keywords}[1]{\textit{Keywords:} #1}

\newtheorem{assumption}{Assumption}
\newtheorem{cor}{Corollary}
\newtheorem{proposition}{Proposition}
\newtheorem{theorem}{Theorem}
\newtheorem{lemma}{Lemma}
\newtheorem{defn}{Definition}
\newtheorem{exmp}{Example}
\theoremstyle{definition}
\newtheorem{remark}{Remark}
\newcounter{repbox}

\title{  Causal Identification under Interference:\\ 
The Role of Treatment Assignment Independence}
\author{
Julius Owusu\thanks{Department of Economics, Concordia University. Email: julius.owusu@concordia.ca. }
\and
Monika Avila M\'arquez\thanks{University of Geneva, Geneva Finance Research Institute. Email: monika.avila@unige.ch.}
}

\begin{document}
\onehalfspacing
\maketitle
\addtocontents{toc}{\protect\setcounter{tocdepth}{-2}}
\begin{abstract}
\thispagestyle{empty} 
Applied researchers often use identification formulas derived under individualistic treatment response (ITR), which rules out interference, even when interference is plausible but unobserved. We characterize what these identification formulas recover under arbitrary interference. For selection on observables, instrumental variables, sharp regression discontinuity, and difference-in-differences, each ITR-based formula equals an average direct effect (ADE)-type estimand plus a bias. Every term of the bias is the product of an interference component and a component measuring how the law of others' assignments differs across the two comparison groups. The bias therefore vanishes under a restriction on the assignment mechanism, generalized conditional assignment independence (CAI), whose design-specific forms include conditions used in the literature. Because generalized CAI cannot generally be tested from the single realization of the assignment vector observed in data, we develop a sensitivity analysis that relaxes it and derive identified sets, together with the corresponding breakdown frontiers and breakdown points, under arbitrary interference. When interference is anonymous and operates through a treated share, the structure yields tighter bounds indexed by a measure of assignment dependence that can be benchmarked with observed group-level assignment patterns. A calibrated simulation and a reanalysis of a temporary-agency employment study illustrate the paper's theoretical results and proposed sensitivity-analysis procedures.
\vspace{0.3cm}
\\

\end{abstract}

\setlength{\parindent}{3ex}
\newpage

\setcounter{page}{1}
\doublespacing
\section{Introduction}

Much of the causal empirical evidence in applied research using observational data centers on four research designs: selection on observables, instrumental variables (IV), regression discontinuity designs (RDD), and difference-in-differences (DiD). Their causal identification and interpretation rest on the no-interference, or \emph{individualistic treatment response} (ITR), assumption \citep{rubin1980randomization,manski2013identification}. This assumption, however, is difficult to sustain in many empirical settings. Training programs displace untreated job seekers \citep{crepon2013labor,gautier2018estimating}, health interventions affect untreated neighbors \citep{miguel2004worms,angelucci2009indirect}, and large market programs move equilibrium prices \citep{heckman1999human}.
In such settings, unit $i$'s potential outcome depends on both its own treatment $d_i$ and the treatments of other units, $\mathbf d_{-i}$, so the potential outcome is $Y_i(d_i,\mathbf d_{-i})$, and ITR fails. Because data on who interacts with whom are rarely available, empirical applications often acknowledge this dependence but use estimators whose interpretation and targets effectively assume it away. Examples span selection on observables \citep{ichino2008from,aizer2016long}, IV \citep{angrist1990lifetime,finkelstein2012oregon}, RDD \citep{ludwig2007does,meyersson2014islamic}, and DiD \citep{card1994minimum,bailey2015war}; Table~\ref{tab:interference-literature} of the Supplementary Material lists further studies.
A leading example is \citet{ichino2008from}, who study temporary-agency jobs in Italy using a matching estimator. Agency workers and nonparticipants plausibly competed for the same local vacancies, but the data contain no information on how they interacted.

This paper provides a unifying proposition characterizing what the standard ITR-based identification formulas recover under arbitrary interference. In particular, we show that, in a finite population with arbitrary interference, each ITR-based identification formula equals an average direct effect (ADE) plus a bias driven by the interaction between interference and dependence in assignment across units. The bias vanishes under an independence restriction on the assignment mechanism whose specific form depends on the design and is unified by a \textit{generalized conditional assignment independence (generalized CAI)} condition. Under this condition, the ITR-based estimand retains an ADE-type causal interpretation. We assess sensitivity to violations of the condition by deriving \textit{identified sets} and corresponding \textit{breakdown frontiers} under alternative forms of interference, thereby quantifying how far a study's conclusions survive as generalized CAI fails. 

The sensitivity analysis is the paper's main contribution. We quantify nonparametric departures from the generalized CAI condition and bound the resulting bias. We use these bias bounds to derive identified sets that are sharp given the observed outcome means for treated and control units, and then characterize the corresponding breakdown frontiers and breakdown points. We first allow arbitrary interference and then impose restricted forms of interference that yield tighter bounds indexed by dependence measures that can be benchmarked using groupings in observed data. While design-specific decompositions of the ITR-based formulas exist under different assumptions, the general decomposition we provide in this paper yields two main insights beyond those existing results. First, we show that generalized CAI is the condition under which the ITR-based formulas recover ADE-type estimands under interference, and that it specializes to each design's own assignment-independence condition. Second, we show that, for the pooled ITR-based formulas that practitioners compute, the generalized CAI condition implies unit-level assignment independence together with equal propensities among units sharing the same conditioning value, the no-hidden-bias condition of \citet{rosenbaum2002observational}. Under no interference, no hidden bias is all the assignment mechanism must satisfy for the formula to identify a causal effect; under interference it remains necessary but is no longer sufficient, because outcomes depend on others' assignments, which must therefore also be independent of own assignment.

Every ITR-based formula we study is a functional $\Phi_{\mathrm{ITR}}(P)$ of the observed-data law $P$ that contrasts mean outcomes across two comparison groups (cells), conditional on covariates or averaged over their distribution. These cells are treated and untreated units under selection on observables, the two instrument values under IV, the two sides of the cutoff in RDD, and the treated and comparison groups in DiD. Let $\Delta_i(\mathbf d_{-i}):=Y_i(1,\mathbf d_{-i})-Y_i(0,\mathbf d_{-i})$\label{intro:ade} denote unit $i$'s direct effect with the treatment vectors of others fixed at $\mathbf d_{-i}$. With each design's maintained identifying assumptions restated under interference, Proposition~\ref{unif:thm-master} in Section~\ref{sec:mapping} asserts that 
\begin{equation*}
\underbrace{\Phi_{\mathrm{ITR}}(P)}_{\text{ITR formula}}
\;=\;
\underbrace{\mathbbm{E}\big[\Delta_i(\mathbf D_{-i})\big]}_{\text{ADE-type causal estimand}}
\;+\;
\underbrace{\mathcal B}_{\text{bias}},
\end{equation*}
where $\mathbbm{E}$ averages the direct effect over a reference law of others' assignments and over the design's target population  (all units, compliers, units at the cutoff, or treated units), yielding the average direct effect (ADE) or a local analog. 

The bias $\mathcal B$ is a sum of product terms, each pairing an interference component with an assignment-dependence component that captures differences in the law of others' assignments across comparison cells. We show that the bias vanishes under generalized CAI, which requires comparison cells to face the same law of others' assignments, regardless of the magnitude of interference.
For selection on observables, the generalized CAI condition reduces to $\mathbf D_{-i}\perp\!\!\!\perp D_i\mid W_i$ for a randomly drawn unit, which also suffices for DiD. For IV, the condition concerns others' instruments rather than treatments, whereas for sharp RDD it requires continuity in the law of others' treatments at the cutoff.

The practical implication of the general decomposition is that interpreting standard ITR-based estimates under interference need not require observing a network, group structure, or exposure mapping. Generalized CAI depends only on the assignment mechanism, so its plausibility can be assessed using institutional knowledge of how treatment is assigned. It may hold even under strong interference, but fails when the design or environment induces dependence in assignment, as under fixed treatment counts, capacity constraints, common shocks to assignment, or peer effects in participation. 

Generalized CAI restricts the joint assignment law, of which the data reveal a single realization, and is generally untestable. Section~\ref{robustness} therefore develops a sensitivity analysis for departures from this condition. For each design, we therefore bound the resulting bias $\mathcal B$ by the product of an interference parameter and an assignment-dependence parameter measured by the \textit{total-variation distance} between the laws of others' assignments across the relevant comparison cells. These bounds yield identified sets, sharp given the cell means, for the corresponding ADE estimands and associated breakdown frontiers and points. Because the total-variation distance can overstate the dependence that matters under constrained assignment mechanisms, and need not track its magnitude, we further consider alternative forms of interference that yield identified sets and breakdown frontiers indexed by more accurate measures of dependence. The proposed sensitivity analyses are implementable in a companion \texttt{R} package \texttt{caisensitivity}, available upon request.

\paragraph{Relation to the literature.}\label{intro:known}
The ADE-bias decomposition results exist for each design under different forms of interference. For experiments, the difference in means targets a contrast over the conditional law of others' treatments whose gap from an average direct effect vanishes under Bernoulli randomization but not, in general, under fixed-count randomization \citep{hudgens2008toward,vanderweele2011effect,savje2021average,hu2022average,vazquezbare2023identification}. For selection on observables, the result is available under a known exposure mapping \citep{forastiere2021identification} and, under unrestricted neighbourhood interference \citep{leung2025identifying}. For IV, the Wald estimand need not equal the local average treatment effect (LATE) \citep{sobel2006randomized}, but it equals a complier direct effect under independent instruments \citep{kang2016peer,imai2021causal,hoshino2024causal}, with known biases when instruments are dependent by design \citep{vazquez2023causal,kang2018spillover}. For the RDD, results rely on local randomization \citep{aronow2017regression}, a known exposure mapping \citep{daltorrione2025regression}, continuity at the cutoff of the law of the other treatments to which a unit is exposed \citep{borusyak2024regression}, or linear-in-means spillovers among units with nearby running variables \citep{auerbach2024regression}, and results also exist for DiD \citep{clarke2017estimating,xu2023difference,butts2021difference,mealli2025difference}. Section~\ref{supp:relwork} of the Supplementary Material documents the closest antecedent for each design. A complementary literature models interference explicitly and recovers direct and spillover effects when network or group data exist \citep{halloran1995causal,hudgens2008toward,tchetgen2012causal,baird2018optimal,aronow2017estimating,forastiere2021identification,huber2021framework,aronow2021spillover}; we instead ask what standard ITR-based formulas recover when such data do not exist, making it natural to allow for arbitrary interference in deriving the ADE--bias decomposition.

The central gap in the literature is a unifying framework for how these design-specific results fit together and what the ITR-based formulas identify under partial failures of generalized CAI across alternative forms of interference.   To the best of our knowledge, no existing sensitivity analysis under interference takes the dependence between own and others' assignments as its sensitivity parameter. Existing analyses either treat assignment dependence as a known feature of the design \citep{savje2021average} or focus on other assumptions. For example, \citet{vanderweele2014interference} relax the absence of unmeasured confounding of own and others' treatments; \citet{egami2021spillover} relax the assumption that network spillovers are observed; and \citet{ortyashov2025sensitivity} relax, jointly, the absence of interference, the absence of unmeasured confounding, and the transportability of the study sample to the target population.

For brevity, we defer discussion of additional related papers to Section~\ref{supp:relwork} of the Supplementary Material.

The remainder of the paper is organized as follows. Section~\ref{sec:example} uses the temporary-agency study of \cite{ichino2008from} to motivate and illustrate the framework and summarizes Monte Carlo evidence showing that, in the presence of interference, an ITR-based estimator is biased when treatment (take-up) is dependent.  Section~\ref{sec:mapping} develops the general theory: a finite-population framework with arbitrary interference, a representation that covers the formulas of all four designs, and the general decomposition (Proposition~\ref{unif:thm-master}). Section~\ref{sec:mapping} also summarizes the design-specific results as special cases and collects their ingredients and identifying assumptions in Table~\ref{unif:tab-mapping}. Section~\ref{robustness} develops the sensitivity analysis and applies it to the temporary-agency data, and Section~\ref{sec: conclusion} concludes. Proofs of the main results are in the Supplementary Material, which also contains the design-specific identifying assumptions, ADE--bias decompositions and their proofs; a discussion of the related literature; primitive conditions, benchmark dependence structures and confidence intervals that complement the sensitivity analysis of Section~\ref{robustness}; an implementation guide; the details of the Monte Carlo design and results; the data, replication, and full results of the temporary-agency illustration; and a second application, to the minimum-wage study of \citet{card1994minimum}.

\section{An Illustrative Example: Temporary-Agency Jobs}\label{sec:example}
This section uses the temporary-agency study of \citet{ichino2008from} to motivate the paper’s framework and question. Interference is plausible but unobserved, ruling out methods that model it explicitly and leaving an ITR-based selection-on-observables identification formula as the object of interest. Its causal interpretation depends on whether a unit's take-up is associated with comparable workers' take-up, a feature not generally testable with the observed data. Using a Monte Carlo design calibrated to the data, we show that the resulting bias depends on the strength of the assignment dependence and becomes substantial under strong common shocks.

\begin{exmp}[Temporary-Agency Employment as a Pathway to Permanent Employment]\label{eg::lalonde}
\citet{ichino2008from} ask whether an assignment through a temporary work agency (TWA) in the first half of 2001 raised the probability of holding a permanent contract in late 2002 in \emph{Tuscany} and \emph{Sicily}. They compare TWA workers in nine provinces with a random sample of comparable residents under a selection-on-observables design and estimate, by propensity-score matching, an average effect on the treated (ATT) of $0.19$ ($0.06$) in Tuscany and $0.10$ ($0.05$) in Sicily. Interference is plausible because TWA workers compete with nonparticipants for the same local vacancies, agencies transmit information and referrals, and placements affect firms' subsequent hiring. 
\end{exmp}

The study's data contain demographics, work history, distance to the nearest agency, and ZIP code, but not who competed with whom or how information about agencies diffused across workers. This is a common empirical setting rather than an exceptional one \citep[see, e.g.,][]{colpitts2002targeting}. Without information on interactions, the researcher cannot construct an exposure mapping, condition on neighbors' treatments, or estimate spillover effects. The relevant question is therefore what the ITR-based formula actually used identifies.

Addressing this question requires understanding the source of geographic concentration in take-up. If take-up among comparable workers is idiosyncratic, own take-up is uninformative about others' take-up conditional on covariates. Assignment is then conditionally independent, and the matching formula recovers an ADE regardless of interference strength (according to Proposition \ref{unif:thm-master}). If instead take-up is shaped by networks that transmit job information or by agency recruitment drives, capacity limits and peer effects in participation that move take-up among nearby workers, own take-up becomes informative about others' assignments, so treated workers and their matched controls face different treatment environments.  The data show substantial geographic concentration in realized take-up. Roughly two-thirds of TWA workers live in the decile of ZIP codes with the highest take-up, and the treated share declines with distance to an agency. These patterns show that treatment probabilities vary systematically across locations, but not why. They are consistent both with conditionally independent take-up driven by location-specific observables and with common local forces, such as networks or demand shocks, that induce dependence in assignment. The data therefore do not reveal whether one worker’s take-up is informative about others’ take-up after conditioning on covariates (see Section~\ref{ex:nsw-supp} of the Supplementary Material). Section~\ref{sen appl} asks how much dependence in assignments an empirical conclusion can tolerate.

The consequences of such dependence can be substantial. In a Monte Carlo study calibrated to the Tuscan subsample of the study (details in Section~\ref{sec:montecarlo} of the Supplementary Material), ITR fails while unconfoundedness and overlap hold by construction. Using an inverse probability weighted estimator at the true propensity score, we compute the average bias across \(4{,}000\) Monte Carlo replications for a direct effect of \(0.18\) on the probability of holding a permanent contract in late 2002. It is zero when take-up is independent given covariates (i.e., assignments are conditionally independent); $0.012$ when entire ZIP codes take up together; $0.009$ under a regional common shock of the same magnitude; and $0.158$ under a strong regional shock. The results are summarized in column two of Table \ref{tab:mc-sens} in Section \ref{sen appl}.

These simulation results provide three main lessons. First, interference alone is not the obstacle. When CAI holds in the temporary-work setting, the standard formula still identifies a well-defined causal quantity, the effect of a worker's own agency assignment averaged over the realized assignment mechanism for others. Second, assignment dependence is the obstacle. Treated and untreated units with the same covariate values may face different distributions of others' treatments, a discrepancy the Monte Carlo study quantifies and the decomposition in Section~\ref{sec:mapping} characterizes. Third, the data cannot test this dependence without additional restrictions. The two plausible take-up mechanisms described above differ in the joint law of take-up across workers, whereas the data reveal only one realization of assignment from that law. Section~\ref{sen appl} therefore quantifies how much dependence in take-up is sufficient to overturn a conclusion drawn from the estimate.

\section{The Unifying Framework and Decomposition}\label{sec:mapping}
This section presents the paper's unifying result. 
Consider a finite population of $N$ units indexed by $i\in [N]:=\{1,\dots, N\}$.\label{unif:sec-setup} For each $i\in[N]$, let $D_i\in\{0,1\}$ denote the treatment assignment of unit $i$, and let $\mathbf{D}=(D_1, \dots, D_N)\in\{0,1\}^N$ denote the assignment vector of the population, with realizations $\mathbf d$. Potential outcomes are indexed by the full assignment vector, allowing arbitrary interference. Unit-level randomness is captured by a latent variable $\epsilon_i$ common across treatment states, so that $Y_i(\mathbf d)=f(\mathbf d, \epsilon_i)$ for a function $f$. Remark~\ref{supp:rem-design} of the supplement discusses the antecedents to this latent stochastic potential-outcome formulation.

\begin{assumption}[Arbitrary Interference]
\label{ass:arbitrary-interference}
For each unit $i\in[N]$ and each treatment assignment vector
$\mathbf d \in \{0,1\}^N$, there exists a potential outcome $Y_i(\mathbf d)$.
\end{assumption}

\begin{assumption}[Consistency]\label{ass:consistency} For all $i\in[N]$ and $\mathbf{d}\in \{0,1\}^{N}$,
\vspace{-0.5cm}
 \[Y_i = Y_i(\mathbf{d})\quad  \text{if}\quad \mathbf{D}=\mathbf{d},  \vspace{-0.5cm}\]
 where $Y_i\in\mathcal{Y}\subseteq \mathbbm{R}$ denotes the observed outcome.
\end{assumption}

Assumption\rev{~}\ref{ass:arbitrary-interference} requires neither partial interference, neighborhood interference, nor a known exposure mapping as in \citet{aronow2017estimating}. Which units' treatments affect unit \(i\)'s outcome, and how, is left unrestricted, and ITR is the special case in which  outcome depends on the unit's own treatment alone. Assumption~\ref{ass:consistency} requires outcomes to depend on the realized assignment vector, not on the mechanism that generates it.

For each $i\in[N]$, let $\mathbf{D}_{-i}\in\{0,1\}^{N-1}$ denote the vector $\mathbf D$ with its $i$th element deleted (the \emph{treatment of others}), with realizations $\mathbf{d}_{-i}$, and let $Y_i(d,\mathbf{d}_{-i}):=Y_i(\mathbf{d})$ when $d_i=d$. Define the \emph{individual direct effect (IDE)}\footnote{The IDE parameter was first discussed in \cite{halloran1995causal}; \cite{savje2021average} refers to it as the \textit{assignment-conditional unit-level treatment effect}.} at $\mathbf{d}_{-i}$ as $\Delta_i(\mathbf{d}_{-i}):= Y_i(1,\mathbf{d}_{-i}) - Y_i(0,\mathbf{d}_{-i})$; marginalizing over the joint law of $\mathbf D_{-i}$ and $\epsilon_i$ gives the ADE  of unit $i$, $\mathbbm{E}[\Delta_i(\mathbf{D}_{-i})]$. The estimands of interest in this paper average such unit-level quantities over the population, as defined below.  Section~\ref{supp:relwork} of the supplement discusses the relationship between the ADE and existing target estimands in the literature.

Let $W_i\in\mathcal W\subseteq \mathbbm{R}^p$ denote pretreatment variables (covariates), and let $\boldsymbol\epsilon:=(\epsilon_1,\dots,\epsilon_N)$ and $\mathbf W:=(W_1,\dots, W_N)$ denote the population vector of the latent and pretreatment variables.  Because the \(N\) units need not share a common law, we let \(I\) be drawn uniformly from \([N]\), independently of \((\boldsymbol\epsilon,\mathbf W,\mathbf D)\), and let \(P\) denote the induced joint law of \((I,\boldsymbol\epsilon,\mathbf W,\mathbf D)\). Under \(P\), the marginal law of the observed vector \((Y_I,W_I,D_I)\) is the average
\(N^{-1}\sum_{i=1}^N P_i\) of the laws \(P_i\) of \((Y_i,W_i,D_i)\).
Note that the random index \(I\) is a device for averaging over the population, and not a sampling step. The population ADE is then $\mathbbm{E}[\Delta_I(\mathbf{D}_{-I})]=N^{-1}\sum_{i=1}^N\mathbbm{E}[\Delta_i(\mathbf{D}_{-i})]$, the average of the unit-level ADEs, which equals the ADE of every unit when the joint laws of $(\epsilon_i, W_i, D_i,\mathbf D_{-i})$ coincide across units after relabelling the other units. Others' assignments are $\mathbf D_{-I}=(D_j)_{j\neq I}$, with coordinates labelled by the units $j$. Hereafter, with a slight abuse of notation, we suppress the random index \(I\). Whenever \(i\) is not explicitly quantified, it denotes the drawn unit \(I\). Thus, \((Y_i, W_i, D_i)\), \(\mathbf D_{-i}\), and \(\Delta_i\) denote \((Y_I, W_I, D_I)\), \(\mathbf D_{-I}\), and \(\Delta_I\), and all expectations, conditional expectations, and probability laws are taken under \(P\). By contrast, \(i\) retains its unit-specific meaning, with law \(P_i\), when it is explicitly quantified, as in ``for all \(i\in[N]\),'' ``for each unit \(i\),'' ``fix a unit \(i\),'' or sums over \(i\), and in the Monte Carlo study and applications, where \(i\) indexes observed units. Pointwise restrictions that hold for every unit, such as consistency, exclusion, and monotonicity, therefore also hold for \(I\). Statements involving laws, means, or independence pass from the unit-specific laws \(P_i\) to \(P\) only under the equal-composition condition~\eqref{unif:eq-ec} below, which we invoke explicitly whenever needed. Remark~\ref{supp:rem-design} of the Supplementary Material discusses this formulation and its antecedents. The identification formulas for the common observational designs are of the general form
\begin{equation}
\Phi_{\mathrm{ITR}}(P)
:=
\mathbbm{E}\!\left[\psi(Y_i,W_i,D_i;P)\right],
\label{eq:sutva-functional}
\end{equation}
where $\psi(\cdot;P)$ is a design-specific function of the observed-data distribution.

Under ITR and consistency, $Y_i=Y_i(1)D_i+Y_i(0)(1-D_i)$, and the joint assignment of treatments is irrelevant for identification. In particular, comparisons of expected outcomes across values of $D_i$ have a causal interpretation under $P$ regardless of the dependence between $D_i$ and $\mathbf D_{-i}$, provided that the unit-level conditional or unconditional assignment probability, $\Pr(D_i=1\mid W_i)$ or $\Pr(D_i=1)$, is common across units $i\in[N]$. Under arbitrary interference and consistency, by contrast,
\begin{equation}
Y_i=\sum_{\mathbf d_{-i}}\big(Y_i(0,\mathbf d_{-i})+\Delta_i(\mathbf d_{-i})D_i\big)\,\mathbbm{1}(\mathbf D_{-i}=\mathbf d_{-i}).
\label{eq:obs-rep}
\end{equation}
A comparison of expected outcomes across values of $D_i$ averages over the treatment assignments of other units $\mathbf D_{-i}$, whose distribution may differ between the two comparison cells. When those distributions coincide, the contrast averages IDEs over a common law. The causal interpretation of an ITR-based formula under arbitrary interference therefore depends on the assignment mechanism, a relationship we formalize below.

The identification formulas of the designs we consider compare mean observed outcomes between two comparable cells of units. With a slight abuse of notation hereafter, we use $(Y_i, D_i, W_i)$ generically for the outcome being compared, the binary variable defining the two cells, and any variables on which the comparison is conditioned. When no conditioning variable is required, omit $W_i$, as summarized in Panel~A of Table~\ref{unif:tab-mapping}. Under selection on observables, the notation is unchanged. Under DiD, $Y_i$ denotes $Y_{i1}-Y_{i0}$, the change in the outcome $Y_{it}$ between periods $t=0$ and $t=1$. Under IV, $Z_i$ and $\mathbf Z_{-i}$ replace $D_i$ and $\mathbf D_{-i}$, with potential treatments (take-ups) $D_i(z)$ and $\mathbf D_{-i}(\mathbf z_{-i})$ and potential outcome $Y_i(D_i(z),\mathbf D_{-i}(\mathbf z_{-i}))$. Under sharp RDD, $W_i$ is the running variable, i.e., $D_i=\mathbbm 1\{W_i\ge w_0\}$, and the cell means are the one-sided limits at $w_0$.\footnote{When treating a design separately, as in Section~\ref{supp:designs} of the Supplementary Material, we use design-specific notation.} With the foregoing convention, every formula is a functional of the cell-mean contrast $\mathbbm E[Y_i\mid D_i=1, W_i]-\mathbbm E[Y_i\mid D_i=0, W_i]$ and a normalizer $\kappa$ (defined under \eqref{unif:eq-template} below). 
Taking expectations of both sides of \eqref{eq:obs-rep} within a  comparison cell gives, for  $d\in\{0,1\}$
\begin{equation}
\mathbbm{E}[Y_i\mid D_i=d,W_i]=\sum_{\mathbf d_{-i}}\mathbbm{E}\big[Y_i(d,\mathbf d_{-i})\,\mathbbm{1}(\mathbf D_{-i}=\mathbf d_{-i})\mid D_i=d,W_i\big],
\label{unif:eq-cellmean}
\end{equation}
so that the contrast between the two comparison cells is
\begin{equation}
\begin{split}
\sum_{\mathbf d_{-i}}\Big\{&\mathbbm{E}\big[Y_i(1,\mathbf d_{-i})\mathbbm{1}(\mathbf D_{-i}=\mathbf d_{-i})\mid D_i=1,W_i\big]\\[-2pt]
&-\mathbbm{E}\big[Y_i(0,\mathbf d_{-i})\mathbbm{1}(\mathbf D_{-i}=\mathbf d_{-i})\mid D_i=0,W_i\big]\Big\}.
\end{split}
\label{unif:eq-contrast}
\vspace{-0.5cm}
\end{equation}
Each term in \eqref{unif:eq-contrast} jointly involves the potential outcome under a given configuration of others' treatments and the event that this configuration is realized within the corresponding comparison cell. Each expectation need not factor into a mean potential outcome and an assignment probability, and without such a factorization (separation) the cell contrast need not have a causal interpretation. Extending each design's baseline identifying assumption to accommodate interference yields the required separation, which we state in a common form across designs.

\begin{assumption}[Separation]\label{unif:ass-sep}
For all $d\in\{0,1\}$, all $\mathbf d_{-i}\in\{0,1\}^{N-1}$, and almost all $w$,
\begin{equation}
\begin{split}
&\mathbbm{E}\big[Y_i(d,\mathbf d_{-i})\mathbbm{1}(\mathbf D_{-i}=\mathbf d_{-i})\mid D_i=d,W_i=w\big]\\[2pt]
&\qquad=\underbrace{\mathbbm{E}\big[Y_i(d,\mathbf d_{-i})\mid D_i=d,W_i=w\big]}_{=:\ \mu_d(\mathbf d_{-i};w)}\;
\underbrace{\Pr(\mathbf D_{-i}=\mathbf d_{-i}\mid D_i=d,W_i=w)}_{=:\ \pi_d(\mathbf d_{-i};w)} .
\end{split}
\label{unif:eq-sep}
\end{equation}
\end{assumption}

Under $P$, a configuration $\mathbf d_{-i}$ of others' assignments also identifies the drawn unit, namely the unit $k$ it omits, and $\mu_d(\mathbf d_{-i};w)$ denotes the mean of $Y_k(d,\mathbf d_{-i})$ given $D_k=d$, $W_k=w$ and $I=k$; sums over $\mathbf d_{-i}$ accordingly run over the configurations of the others of every unit. Assumption~\ref{unif:ass-sep} states the condition implied by extending each design's baseline identifying assumption under ITR to settings with interference. It requires the potential outcome at a fixed assignment vector to be \textit{mean-independent} of others' treatments, given own treatment and the conditioning variable, which breaks each expectation of \eqref{unif:eq-contrast} into an outcome part $\mu_d$ and an assignment part $\pi_d$. Under selection on observables, \eqref{unif:eq-sep} is implied by unconfoundedness stated at every $\mathbf d_{-i}$, which also removes the conditioning on $D_i$ from $\mu_d$, so that $\mu_d(\mathbf d_{-i};w)=\mathbbm{E}[Y_i(d,\mathbf d_{-i})\mid W_i=w]$ \citep{hudgens2008toward, vanderweele2011effect, forastiere2021identification}. Under IV, instrument independence delivers \eqref{unif:eq-sep} with $Z_i$, $\mathbf Z_{-i}$ substituting for $D_i$, $\mathbf D_{-i}$ respectively \citep{sobel2006randomized, kang2016peer, imai2021causal,hoshino2024causal}. Under the sharp RDD, conditional independence of the untreated outcome from others' treatments given the running variable, together with continuity at the cutoff, delivers \eqref{unif:eq-sep} through one-sided limits on either side of the cutoff, with $\pi_1(\mathbf d_{-i};w)$ and $\pi_0(\mathbf d_{-i};w)$ replaced by the one-sided limits $\pi_+(\mathbf d_{-i})$ and $\pi_-(\mathbf d_{-i})$ of $\Pr(\mathbf D_{-i}=\mathbf d_{-i}\mid W_i=w)$ as $w\downarrow w_0$ and $w\uparrow w_0$, and, when they coincide, $\pi_{w_0}(\mathbf d_{-i})$ denotes their common value, the law of others' treatments at the cutoff \citep{aronow2017regression,daltorrione2025regression}. Finally, under DiD, mean independence of the outcome trends from others' assignments within treatment groups, together with parallel trends, delivers \eqref{unif:eq-sep} with the outcome change $Y_{i1}-Y_{i0}$ in place of $Y_i$ \citep{clarke2017estimating,xu2023difference,butts2021difference}. Table~\ref{unif:tab-mapping} gives the full correspondence, and we state the conditions in detail for each design in Section~\ref{supp:designs} of the Supplementary Material. 

To state the general result, let $d\in\{0,1\}$ denote the index of the two compared cells, $w$ the value of the conditioning variable, and $\mathbf d_{-i}$ the assignments of others, so that
$\mu_d(\mathbf d_{-i};w)$ and $\pi_d(\mathbf d_{-i};w)$ are the two factors in \eqref{unif:eq-sep} and $G$ is the law used to average the contrast. For selection on observables with ADE as target, these are own treatment, covariates, and others' treatments, with $G=F_W$, the law of covariates. Every statement below holds with each design's own counterparts, as Table~\ref{unif:tab-mapping} outlines. 

Substituting \eqref{unif:eq-sep} into each term of the contrast in \eqref{unif:eq-contrast}, averaging over $w$ with respect to $G$, and dividing by a normalizer $\kappa>0$  yields\label{unif:eq-slice}
\begin{equation}
\Phi_{\mathrm{ITR}}(P)=\frac{1}{\kappa}\int_{\mathcal W}\sum_{\mathbf d_{-i}}\big\{\mu_1(\mathbf d_{-i};w)\cdot\pi_1(\mathbf d_{-i};w)-\mu_0(\mathbf d_{-i};w)\cdot\pi_0(\mathbf d_{-i};w)\big\}\,G(dw).
\label{unif:eq-template}
\end{equation}
The normalizer $\kappa$ equals one in every design we consider except IV, where it is the population share of compliers, $\Pr(D_i(1)>D_i(0))=\mathbbm E[D_i(1)-D_i(0)]$, which equals the Wald denominator $\mathbbm E[D_i\mid Z_i=1]-\mathbbm E[D_i\mid Z_i=0]$ under $P$ when the instrument propensity is common across units, so that the two instrument cells have the same unit composition.
Throughout, we suppress the arguments of functions of $(\mathbf d_{-i};w)$ and of $w$ whenever no confusion arises.

The two comparison cells may therefore differ in two ways, in the means $\mu_d$ and in the laws $\pi_d$, and only the first difference is causal. Define the \emph{slice contrast}, i.e., the normalized contrast at $(\mathbf d_{-i},w)$, as
\begin{equation}
\theta(\mathbf d_{-i};w):=\kappa^{-1}\big\{\mu_1(\mathbf d_{-i};w)-\mu_0(\mathbf d_{-i};w)\big\}.
\label{unif:eq-theta}
\end{equation}
This is the cell-level object the ITR argument would recover if the assignment of others were fixed at $\mathbf d_{-i}$. In the designs under consideration, it is a direct effect of the drawn unit, unconditionally or given $W_i=w$ (see Table~\ref{unif:tab-mapping}). We refer to
\vspace{-0.5cm}
\begin{equation}
\pi_1(\mathbf d_{-i};w)=\pi_0(\mathbf d_{-i};w)\qquad\text{for all }\mathbf d_{-i}\text{ and $G$-almost all }w
\label{unif:eq-gcai}
\vspace{-0.5cm}
\end{equation}
as the \emph{generalized CAI} condition. It asserts that the two comparison cells carry the same law of others' assignment. Its form depends on whether the comparison cells are events of positive probability. In selection on observables, IV and DiD, they are, and \eqref{unif:eq-gcai} is an independence restriction: for own treatment and covariates it is \emph{conditional assignment independence} (CAI),
\vspace{-0.3cm}
\begin{equation}
\mathbf D_{-i}\;\perp\!\!\!\perp\;D_i\mid W_i ,
\label{eq:CAI}
\vspace{-0.3cm}
\end{equation}
for the instrument it is $\mathbf Z_{-i}\indep Z_i$, and when a design conditions on no covariates, the independence is unconditional. 

The sharp RDD is the exception. Its two comparison cells are the two sides of the cutoff, and because the running variable is continuously distributed near the cutoff, the event $W_i=w_0$ has probability zero, so the restriction cannot be stated as a conditional independence at the cutoff. Generalized CAI requires instead that the one-sided limits agree, i.e., $\pi_+(\mathbf d_{-i})=\pi_-(\mathbf d_{-i})=\pi_{w_0}(\mathbf d_{-i})$ for all $\mathbf d_{-i}$, so that the law of others' treatments does not jump at the cutoff. This is Assumption~\ref{rdd nojump} in the Supplementary Material, and \citet[Assumption~3(b)]{borusyak2024regression} imposes a similar condition. When $\pi_+(\mathbf d_{-i})=\pi_-(\mathbf d_{-i})=\pi_{w_0}(\mathbf d_{-i})$, the standard sharp RDD formula identifies the ADE under the law of others' treatments at the cutoff, $\pi_{w_0}(\mathbf d_{-i})$ (Theorem~\ref{prop:rd}(ii)). The stronger Assumption~\ref{running indepen} in the Supplementary Material, independence of the running variable ($\mathbf W_{-i}\indep W_i$), implies $\pi_+(\mathbf d_{-i})=\pi_-(\mathbf d_{-i})=\Pr(\mathbf D_{-i}=\mathbf d_{-i})$, under which the standard sharp RDD formula identifies the ADE under the marginal law of others' treatments (Theorem~\ref{prop:rd}(iii)).

It is worthwhile to note that $\pi_d(\mathbf d_{-k};w)=q_d(k;w)\,\pi_{d,k}(\mathbf d_{-k};w)$, where $q_d(k;w):=\Pr(I=k\mid D_I=d,W_I=w)$ is the \emph{composition} of cell $d$ and $\pi_{d,k}(\mathbf d_{-k};w):=\Pr(\mathbf D_{-k}=\mathbf d_{-k}\mid D_k=d,W_k=w)$ is  the assignment law of all units excluding $k$ its treatment and covariates. Summing over $\mathbf d_{-k}$ shows that generalized CAI~\eqref{unif:eq-gcai} holds if and only if $\mathbf D_{-k}\indep D_k\mid W_k$ for each unit and the two cells have the same composition, $q_1(k;w)=q_0(k;w)$. Under overlap (positivity), the latter holds if and only if
\begin{equation}
\begin{gathered}
\Pr(D_k=1\mid W_k=w)=\Pr(D_{k'}=1\mid W_{k'}=w)\\
\text{for all units $k,k'$ with positive density at $w$, for $G$-almost all $w$,}
\end{gathered}
\label{unif:eq-ec}
\end{equation}
which we call \emph{equal composition}.  Condition~\eqref{unif:eq-ec} corresponds to the no-hidden-bias condition in \citet{rosenbaum2002observational}, but is not equivalent to unconfoundedness in the present setting because potential outcomes are random. In every design, though in different forms, \eqref{unif:eq-ec} is a standard prerequisite for a causal interpretation of the difference in mean outcomes between the comparison cells, regardless of whether interference is present. Without it, units with the same conditioning value can differ systematically in their potential-outcome laws, so differences between the comparison cells may reflect differences in which units they contain rather than the effect of treatment. Under interference, the condition also ensures that generalized CAI isolates dependence in the assignment mechanism rather than differences in the composition of units across comparison cells.

For $G$-almost all $w$, let $\pi^\ast(\cdot\,;w)$ be a probability distribution on $\{0,1\}^{N-1}$, the \emph{reference law}, and let $\pi(\mathbf d_{-i};w):=\Pr(\mathbf D_{-i}=\mathbf d_{-i}\mid W_i=w)=\Pr(D_i=1\mid W_i=w)\,\pi_1(\mathbf d_{-i};w)+\Pr(D_i=0\mid W_i=w)\,\pi_0(\mathbf d_{-i};w)$ denote the law of others' assignments faced by all units with $W_i=w$.  The reference law $\pi^\ast$ specifies the law of others' treatments over which they are averaged, and it determines the ADE-bias decomposition. Under selection on observables, $\pi^\ast$ is $\pi_1$, $\pi_0$, or $\pi$, the law of $\mathbf D_{-i}$ given $W_i=w$ faced by treated units, by untreated units, or by all units, giving the three ADEs of Theorem~\ref{bias-so} in the Supplementary Material.\footnote{For the ATT formula $\mathbbm{E}\big[\mathbbm{E}[Y_i\mid D_i=1,W_i]-\mathbbm{E}[Y_i\mid D_i=0,W_i]\,\big|\,D_i=1\big]$, used in the empirical illustration of Section~\ref{sen appl}, $G=F_{W\mid D=1}$, and $\pi^\ast=\pi_1$ gives the ADTT, $\mathbbm{E}[\Delta_i(\mathbf D_{-i})\mid D_i=1]$; Proposition~\ref{unif:thm-master} then yields $\Phi_{\mathrm{ITR}}(P)=\mathrm{ADTT}$ plus the baseline component alone, as under DiD.} Under DiD, $\pi^\ast=\pi_1$, giving the ADE on the treated (ADTT); under IV, $\pi^\ast$ is the marginal law of others' instruments, giving the local average direct effect (LADE); and under the sharp RDD with homogeneous direct effects ($\Delta_i(\mathbf d_{-i})=\Delta(\mathbf d_{-i})$ for all $i$), $\pi^\ast$ is the unconditional law $\pi^{\mathrm u}$ of $\mathbf D_{-i}$, giving the ADE, or $\lambda\pi_++(1-\lambda)\pi_-$ for some $\lambda\in[0,1]$, giving a cutoff ADE (Theorem~\ref{prop:rd}).

For any design and reference law, the general decomposition in the following proposition follows from \eqref{unif:eq-template} and Assumption~\ref{unif:ass-sep} by linearity of cell expectations.

\begin{proposition}[Unified ADE-bias decomposition]\label{unif:thm-master}
Suppose Assumption~\ref{unif:ass-sep} holds and both comparison cells have positive probability for $G$-almost all $w$ or, in the sharp RDD, the one-sided limits at $w_0$ of the cell expectations and of $\pi(\cdot\,;w)$ exist (Table~\ref{unif:tab-mapping}). For the reference law $\pi^\ast(\cdot\,;w)$, if the integrals below are finite, the following hold.
\begin{itemize}
\item[(i)] The ITR identification formula decomposes as
\[
\Phi_{\mathrm{ITR}}(P)=\int_{\mathcal W}\sum_{\mathbf d_{-i}}\theta(\mathbf d_{-i};w)\,\pi^\ast(\mathbf d_{-i};w)\,G(dw)\;+\;\mathcal B(\pi^\ast),
\]
where $\int\!\sum$ denotes $\int_{\mathcal W}\sum_{\mathbf d_{-i}}\,\cdot\,G(dw)$, and
\begin{align}
\mathcal B(\pi^\ast)&:=\tfrac{1}{\kappa}\textstyle\int\!\sum\big\{\mu_1[\pi_1-\pi^\ast]-\mu_0[\pi_0-\pi^\ast]\big\}
\label{unif:eq-bias}\\
&\;=\underbrace{\tfrac{1}{\kappa}\textstyle\int\!\sum\mu_0\,[\pi_1-\pi_0]}_{\text{baseline component}}\;+\;\underbrace{\textstyle\int\!\sum\theta\,[\pi_1-\pi^\ast]}_{\text{direct-effect  component}}.
\label{unif:eq-channels}
\end{align}
In \eqref{unif:eq-bias}, $\mu_1$ and $\mu_0$, and in \eqref{unif:eq-channels}, $\mu_0$ and $\theta$, may each be replaced by their deviations from their values at an arbitrary reference point $\mathbf d^\circ_{-i}(w)$.
\item[(ii)] (Generalized CAI.) If \eqref{unif:eq-gcai} holds, then $\pi_1=\pi_0=\pi$ (in the sharp RDD, $\pi_{w_0}=\pi_+=\pi_-$), and the formula identifies the ADE-type estimand at this common law of others' assignments,
\[
\Phi_{\mathrm{ITR}}(P)=\textstyle\int\!\sum\theta\,\pi;
\]
for every reference law $\pi^\ast$, $\mathcal B(\pi^\ast)=\int\!\sum\theta\,[\pi-\pi^\ast]$ is the difference between this estimand and the target at $\pi^\ast$, and in particular $\mathcal B(\pi^\ast)=0$ whenever $\pi^\ast=\pi$.
\item[(iii)] (No interference.) If $\mu_0(\mathbf d_{-i};w)$ and $\theta(\mathbf d_{-i};w)$ do not vary with $\mathbf d_{-i}$, then $\mathcal B(\pi^\ast)=0$ for every reference law $\pi^\ast$, and  $\Phi_{\mathrm{ITR}}(P)=\int_{\mathcal W}\bar\theta(w)\,G(dw)$, the target under ITR (Table~\ref{unif:tab-mapping}), where $\bar\theta(w)$ denotes the common value of $\theta(\cdot\,;w)$ across $\mathbf d_{-i}$. ~\hyperref[prf:unif:thm-master]{\textup{\footnotesize[Proof]}}
\end{itemize}
\end{proposition}

Proposition \ref{unif:thm-master} part (i) formally shows that ITR-based identification formulas decompose into an ADE-type estimand and a bias. Each bias component is the product of an interference term and an assignment-dependence term, so either being zero eliminates the bias. The \emph{baseline component} is nonzero when the outcome means in the $D_i=0$ cell respond to others' assignments and the law of $\mathbf D_{-i}$ differs across the compared cells. The \emph{direct-effect component} is nonzero when the slice contrasts vary with others' assignments, and the reference law differs from the law in the $D_i=1$ cell. Part~(ii)  shows that, in every design, generalized CAI restricts the assignment mechanism and gives the ITR-based formula a causal interpretation. It identifies an ADE-type estimand. Part~(iii) recovers the ITR formula under no interference. Although its derivation requires neither \eqref{unif:eq-gcai} nor \eqref{unif:eq-ec}, a causal interpretation of $\Phi_{\mathrm{ITR}}(P)$ under no interference requires \eqref{unif:eq-ec}; under ITR, this condition is implied by the standard design-specific identifying assumptions.

We summarize the main conclusions of Proposition~\ref{unif:thm-master} in Table~\ref{unif:tab-quadrants}. Interference alone changes what an ITR-based formula recovers from an ATE-type to an ADE-type estimand, while assignment dependence alone leaves its target unchanged. 
\begin{table}[htbp]
\centering
\caption{What an ITR-based identification formula recovers}
\label{unif:tab-quadrants}
\begin{tabular}{@{}lcc@{}}
\toprule
 & Assignment independence & Assignment dependence\\
 & (generalized CAI holds) & (generalized CAI fails)\\
\midrule
No interference & ATE-type estimand & ATE-type estimand\\
\addlinespace
Interference    & ADE-type estimand & ADE-type estimand $+$ bias\\
\bottomrule
\end{tabular}
\begin{flushleft}\scriptsize
\emph{Note:} The first row is Proposition~\ref{unif:thm-master}(iii), where $\mathcal B(\pi^\ast)=0$ for every reference law; the bottom-left cell is part~(ii), with the ADE-type estimand; the bottom-right cell is part~(i), with the bias $\mathcal B(\pi^\ast)$. ``ATE-type'' and ``ADE-type'' cover each design's own estimand, as in Table~\ref{unif:tab-mapping}. The second column presumes equal composition, condition~\eqref{unif:eq-ec}.
\end{flushleft}
\end{table}

Each design inherits the results of Proposition~\ref{unif:thm-master} as a special case. Panel~A of Table~\ref{unif:tab-mapping} lists each design's ingredients, including the reference laws of its targets, and Panel~B lists its identifying assumptions. Under generalized CAI, $\pi_1=\pi_0=\pi$, which explains the reference-law row. Every reference law in the table, except the unconditional law of the sharp RDD, is a mixture $\lambda\pi_1+(1-\lambda)\pi_0$ of the two cell laws and therefore equals $\pi$, so the bias is zero. Under selection on observables, $\pi_1$, $\pi_0$ and $\pi$ have $\lambda\equiv1$, $\lambda\equiv0$ and $\lambda(w)=\Pr(D_i=1\mid W_i=w)$, so the three ADEs coincide; under DiD, $\pi_1$ has $\lambda\equiv1$; under IV, the marginal law of $\mathbf Z_{-i}$ has $\lambda=\Pr(Z_i=1)$ by the law of total probability; and under the sharp RDD, a cutoff law $\lambda\pi_++(1-\lambda)\pi_-$ equals $\pi_{w_0}$ and gives  the ADE at the cutoff, $\mathrm{ADE}_{w_0}$. The unconditional law of $\mathbf D_{-i}$ in the sharp RDD is not a mixture of $\pi_+$ and $\pi_-$, so the bias at it is zero under the stronger condition of an independent running variable with densities continuous and positive with common values at the cutoff (see Theorem~\ref{prop:rd}(iii) in the Supplementary Material).

\begin{table}[p]
\centering
\caption{The four designs as special cases of Proposition~\ref{unif:thm-master}}
\label{unif:tab-mapping}
\fontsize{9}{10.8}\selectfont
\setlength{\tabcolsep}{3pt}
\setlength{\defaultaddspace}{0.2em}
\renewcommand{\arraystretch}{1.0}
\begin{tabular}{>{\raggedright\arraybackslash}p{2.3cm}>{\raggedright\arraybackslash}p{3.29cm}>{\raggedright\arraybackslash}p{3.29cm}>{\raggedright\arraybackslash}p{3.29cm}>{\raggedright\arraybackslash}p{3.29cm}}
\toprule
 & Selection on observables & Instrumental variables & Sharp RDD & DiD (static / dynamic)\\
\midrule
\multicolumn{5}{l}{\emph{Panel A. Ingredients of the representation of Section~\ref{sec:mapping}}}\\
\midrule
Comparison, $D_i=1$ vs.\ $D_i=0$ & $D_i=1$ vs.\ $D_i=0$ & $Z_i=1$ vs.\ $Z_i=0$ & $W_i\downarrow w_0$ vs.\ $W_i\uparrow w_0$ (one-sided limits) & $D_i=1$ vs.\ $D_i=0$\\
\addlinespace
Conditioning variable $W_i$; weighting law $G$ & $W_i$; $F_W$ & none & none (limits at $w_0$) & $W_i$; $F_{W\mid D=1}$\\
\addlinespace
Outcome; $\kappa$ & $Y_i$; $1$ & $Y_i$; $\Pr(D_i(1)>D_i(0))$ & $Y_i$; $1$ & $Y_{i1}-Y_{i0}$; $1$\\
\addlinespace
Others' assignments $\mathbf D_{-i}$ & $\mathbf D_{-i}$ & $\mathbf Z_{-i}$ & $\mathbf D_{-i}$ & $\mathbf D_{-i}$\\
Slice contrast $\theta$ & mean IDE given $W_i$ & mean complier IDE at $\mathbf D_{-i}(\mathbf z_{-i})$ & $\Delta(\mathbf d_{-i})$ & mean IDE of treated given $W_i$\\
\addlinespace
Reference law $\pi^\ast$ & $\pi_1$, $\pi_0$, $\pi$ & $\Pr(\mathbf Z_{-i}=\cdot\,)$ & $\pi^{\mathrm u}:=\Pr(\mathbf D_{-i}=\cdot\,)$; or $\lambda\pi_++(1-\lambda)\pi_-$ & $\pi_1$\\
\addlinespace
ADE-type target & three ADEs & LADE & ADE; or a cutoff ADE & ADTT\\
Target under ITR & ATE & LATE & ATE at $w_0$ & ATT\\
\addlinespace
Bias; component & $B^1_{ADE}$, $B^0_{ADE}$, $B_{ADE}$; both  & $\mathrm{Bias}_{\mathrm{IV}}$; both  &$\mathrm{Bias}_{\mathrm{RDD}}$ (baseline component) at $\pi^\ast=\pi_+$ & $\mathrm{Bias}_{\mathrm{DiD}}$ ($\mathrm{Bias}_{\mathrm{DiD_{dyn}}}$); baseline component\\
\midrule
\multicolumn{5}{l}{\emph{Panel B. Identifying assumptions (full statements in Section~\ref{supp:designs} of the Supplementary Material)}}\\
\midrule
Uncon\-founded\-ness / exo\-geneity analog & potential outcomes at every $\mathbf d_{-i}$ $\indep$ $(D_i,\mathbf D_{-i})$ given $W_i$, \eqref{eq:uso} & potential outcomes and take-ups $\indep$ $(Z_i,\mathbf Z_{-i})$, \eqref{iv indep} & $Y_i(0,\mathbf d_{-i})$ $\indep$ $\mathbf D_{-i}$ given $W_i$, at every $\mathbf d_{-i}$, \eqref{rdd int: uncon} & $\mathbf D_{-i}$ mean-indep.\ of the untreated trend given $(W_i,D_i)$, \eqref{unif:did-meanindep}\\
\addlinespace
Assignment-independence condition & $\mathbf D_{-i}\indep D_i\mid W_i$ (condition~\eqref{eq:CAI}) & $\mathbf Z_{-i}\indep Z_i$ (Assumption~\ref{Instrument indepen}) & $\pi_+=\pi_-$ (no jump at $w_0$; Assumption~\ref{rdd nojump}); implied by $\mathbf W_{-i}\indep W_i$ (Assumption~\ref{running indepen}) and densities continuous at $w_0$; $\pi_{w_0}=\pi^{\mathrm u}$ if equal and positive there & $\mathbf D_{-i}\indep D_i\mid W_i$ (condition~\eqref{eq:CAI})\\
Support / overlap & $0<\Pr(D_i=1\mid W_i)<1$ a.s., \eqref{eq:overlap 2} & first stage $\kappa=\Pr(D_i(1)>D_i(0))>0$ & one-sided limits of $\mathbbm{E}[\,\cdot\mid W_i=w]$ and of $\pi(\cdot\,;w)$ exist at $w_0$ & $\varepsilon<\Pr(D_i=1\mid W_i)<1-\varepsilon$, \eqref{SO 2}\\
\addlinespace
Auxiliary conditions & none beyond consistency & no interference in take-up, $D_i=D_i(Z_i)$, \eqref{iv: no interference in treatment}; monotonicity \eqref{iv: mono}; exclusion \eqref{iv exclusion} & continuity of $\mathbbm{E}[Y_i(0,\mathbf d_{-i})\mid W_i=w]$ at $w_0$ at every $\mathbf d_{-i}$, \eqref{rdd int: cont }; homogeneous direct effects, or the same two conditions for $Y_i(1,\mathbf d_{-i})$ & parallel trends at every $\mathbf d_{-i}$, \eqref{PT}; mean independence of trends from $\mathbf D_{-i}$ within groups, \eqref{unif:did-meanindep} and \eqref{unif:did-meanindep-treated}; no anticipation ($Y_{i0}=Y_{i0}(\mathbf 0)$), relaxed in the dynamic case\\
\addlinespace
Full statement & Assumption~\ref{ass: uncon} & Assumption~\ref{ass: IV conditions} & Assumption~\ref{rdd continuity}, with Assumption~\ref{rdd nojump} or the stronger Assumption~\ref{running indepen} & Assumption~\ref{did: parallel trend} (Assumptions~\ref{did: parallel trend dyn}, \ref{did: no anticipation dyn})\\
\addlinespace
Design-specific result & Theorem~\ref{bias-so} & Theorem~\ref{prop:iv} & Theorem~\ref{prop:rd} & Theorems~\ref{thm:did}, \ref{thm:did_dynamic}\\
\bottomrule
\end{tabular}
\end{table}

In each design, the generalized CAI condition~\eqref{unif:eq-gcai} becomes the assignment-independence restriction of Panel~B. For selection-on-observables and DiD settings, it is the CAI condition in \eqref{eq:CAI}. For IV, it is independence of instruments, with the relevant others' assignment variable being $\mathbf Z_{-i}$ rather than the endogenous $\mathbf D_{-i}$. For the sharp RDD, it corresponds to $\pi_+=\pi_-$, no jump in the law of others' treatments at the cutoff (Assumption~\ref{rdd nojump}). This is weaker than independence of running variables (Assumption~\ref{running indepen}). Under the weaker condition, the formula identifies $\mathrm{ADE}_{w_0}$, the ADE at the law of others' treatments at the cutoff; under the stronger restriction that the running variable is cross-sectionally independent and its densities are continuous and positive with a common value, it identifies the ADE.

Proposition~\ref{unif:thm-master} therefore shows that interpreting an ITR-based formula under interference requires knowing which other units' assignments affect mean potential outcomes and whether those assignments have the same distribution across the comparison cells. We cannot verify either in conventional data without information on links between units, motivating the sensitivity analysis in Section~\ref{robustness}.\label{unif:soo}

\section{Identification under Partial Generalized CAI}\label{robustness}
By Proposition~\ref{unif:thm-master}, under arbitrary interference,
$\Phi_{\mathrm{ITR}}(P)=T(\pi^\ast)+\mathcal B(\pi^\ast)$,
where $T(\pi^\ast):=\int_{\mathcal W}\sum_{\mathbf d_{-i}}\theta\,\pi^\ast\,dG$ is an ADE-type target for any reference law $\pi^\ast$. The bias $\mathcal B(\pi^\ast)$ reflects the interaction between interference and differences in the law of others' assignments across the comparison cells, and generalized CAI \eqref{unif:eq-gcai} eliminates it by requiring those laws to be equal. In practice, generalized CAI can be restrictive because the same mechanisms that generate outcome interference may also induce dependence in take-up. Moreover, because it restricts the joint assignment law, of which the data reveal only one realization, the condition is untestable in the same sense as unconfoundedness or parallel trends (see Remark~\ref{supp:rem-plausibility} of the supplement for more discussion).\footnote{Like those conditions, it has testable implications. It implies equal cell means of every observed function of $\mathbf D_{-i}$, so for an observed grouping the population within-$W_i$ gap in the share of treated group mates is zero; a nonzero gap is incompatible with the condition for dependence running through that grouping, whereas a zero gap is silent about dependence outside it, as design~C of Table~\ref{tab:mc-sens} illustrates. We report the estimate $\widehat\gamma^{\mathcal G}$ of this gap (Section~\ref{sen appl}) as the detected departure and use it as a benchmark; we do not derive its sampling distribution, so it is not offered as a test statistic.} We therefore study identification under nonparametric departures from generalized CAI and alternative forms of interference, deriving \textit{identified sets} \citep{manski2003partial} for the ADE-type estimands and the corresponding \textit{breakdown frontiers and points} \citep{masten2020inference}.

We quantify two margins of robustness. First, we measure how far the ITR-based estimand can depart from an ADE-type target as generalized CAI fails. Second, we determine how large the departure from generalized CAI must be, under alternative forms of interference, before an empirical conclusion no longer holds.
Throughout, we maintain the design-specific assumptions underlying Assumption~\ref{unif:ass-sep}, relax generalized CAI \eqref{unif:eq-gcai}, bound the two components of the resulting bias, and characterize the corresponding identified sets, breakdown frontiers and points.

We formally define the $\lambda$ weights introduced in Section \ref{sec:mapping}
as a function of $W$, i.e.,  $\lambda:\mathcal W\to[0,1]$. Let $\pi_\lambda(\mathbf d_{-i};w):=\lambda(w)\pi_1(\mathbf d_{-i};w)+(1-\lambda(w))\pi_0(\mathbf d_{-i};w)$ denote the mixture of the cell laws. Recall for example that the choices $\lambda(w)=1$, $\lambda(w)=0$, and $\lambda(w)=p(w):=\Pr(D_i=1\mid W_i=w)$ give the reference laws under selection on observables; $\lambda(w)=1$ also gives the ADTT of difference-in-differences, and $\lambda(w)=\lambda=\Pr(Z_i=1)$ is the reference law under the IV design. By Lemma~\ref{sensB:lem-mix} in Section~\ref{app:proofs} of the Supplementary Material, the bias in \eqref{unif:eq-bias} under the different designs can be rewritten as 
\begin{equation}
\mathcal B(\pi_\lambda)=\frac1\kappa\int_{\mathcal W}\sum_{\mathbf d_{-i}}m_\lambda(\mathbf d_{-i};w)
\big\{\pi_1(\mathbf d_{-i};w)-\pi_0(\mathbf d_{-i};w)\big\}\,G(dw),
\quad m_\lambda:=\mu_0+(1-\lambda)\kappa\theta,
\label{sensB:eq-mix}
\end{equation}
with $\theta=\theta(\mathbf d_{-i};w)$ as in \eqref{unif:eq-theta}. 

The bias therefore has the structure of a nonparametric omitted-variable bias, with the omitted object given by the vector of others' treatments $\mathbf D_{-i}$. Its two components are the difference in the conditional distribution of $\mathbf D_{-i}$ across the comparison cells, $\pi_1(\mathbf d_{-i};w)-\pi_0(\mathbf d_{-i};w)$, and the way those assignments affect mean potential outcomes, summarized by $m_\lambda(\mathbf d_{-i};w)$. This parallels the classical omitted-variable-bias decomposition into the association of an omitted variable with treatment and its contribution to the outcome; \citet{leung2025identifying} gives the same interpretation of the bias of an exposure contrast under selection on observables. Bounding the bias therefore requires one restriction on assignment imbalance and one on the magnitude of $m_\lambda(\mathbf d_{-i};w)$. The representation in \eqref{sensB:eq-mix} covers every reference law that is a mixture of the cell laws, including the cutoff targets of the sharp RDD, whose reference law $\lambda\pi_++(1-\lambda)\pi_-$ mixes the one-sided limits. The one exception is the sharp RDD target under running-variable independence, whose reference law is $\Pr(\mathbf D_{-i}=\cdot\,)$, which is not generally a mixture of comparison cell laws. Its bound follows from Theorem~\ref{sensB:thm-bounds}(i) below, with the distances from each one-sided limit to that distribution in place of the imbalance.

The following definition introduces the sensitivity parameters and the models they index, which together form our baseline sensitivity analysis procedure.
\begin{defn}[Sensitivity parameters and models]\label{sensB:def-params}
Let $\mathcal D_w:=\{\mathbf d_{-i}:\pi_1(\mathbf d_{-i};w)+\pi_0(\mathbf d_{-i};w)>0\}$, and let $R_w(f):=\max_{\mathcal D_w}f-\min_{\mathcal D_w}f$ denote the \emph{range} of $f$ over $\mathcal D_w$. We measure assignment dependence by $\delta(w):=\frac12\sum_{\mathbf d_{-i}}|\pi_1(\mathbf d_{-i};w)-\pi_0(\mathbf d_{-i};w)|=d_{\mathrm{TV}}(\pi_0(\cdot\,;w),\pi_1(\cdot\,;w))$, the \emph{total variation (TV) distance} between $\pi_0$ and $\pi_1$, and interference by $\Lambda_d(w):=R_w(\mu_d(\cdot\,;w))$, $d\in\{0,1\}$, and $\Lambda_\theta(w):=R_w(\theta(\cdot\,;w))$. These are latent features of the data-generating process.  For constants $\bar\delta\in[0,1]$ and $\bar\Lambda_0,\bar\Lambda_1,\bar\Lambda_\theta\ge0$, \emph{version~1 of our proposed sensitivity model} imposes $\delta(w)\le\bar\delta$, $\Lambda_0(w)\le\bar\Lambda_0$, and $\Lambda_1(w)\le\bar\Lambda_1$, while \emph{version~2} imposes $\delta(w)\le\bar\delta$, $\Lambda_0(w)\le\bar\Lambda_0$, and $\Lambda_\theta(w)\le\bar\Lambda_\theta$, $G$-a.e. Let $\bar\Lambda=(\bar\Lambda_0,\bar\Lambda_1)$ under version~1 and $\bar\Lambda=(\bar\Lambda_0,\bar\Lambda_\theta)$ under version~2. Then, for the sensitivity parameters $(\bar\delta,\bar\Lambda)$, the \emph{sensitivity model} is the class of primitives $(\pi_0,\pi_1,\mu_0,\mu_1)$ satisfying the corresponding restrictions.
\end{defn}
The sensitivity model indexes a family of departures from the benchmark assumptions, which includes generalized CAI as the case $\bar\delta=0$ and the no-interference benchmark of Proposition~\ref{unif:thm-master}(iii) as the case $\bar\Lambda=\mathbf 0$ under either version.

The parameters in Definition~\ref{sensB:def-params} admit useful interpretations under~\eqref{unif:eq-ec}. In that case, $\delta(w)$ is the weighted average of the unit-level TV distances, $\delta(w)=\sum_k q(k;w)\delta_k(w)$, where $\delta_k(w):=d_{\mathrm{TV}}(\pi_{0,k}(\cdot\,;w),\pi_{1,k}(\cdot\,;w))$, and therefore measures assignment dependence across units with $W_i=w$. Likewise, $\Lambda_d(w)$ and $\Lambda_\theta(w)$ measure variation in mean potential outcomes across others' assignment configurations uniformly over those units and therefore represent the strength of interference. When~\eqref{unif:eq-ec} fails, these interpretations no longer separate the two sources of variation. In particular, the parameter $\delta(w)$ also absorbs differences in cell composition, $q_1(\cdot\,;w)$ versus $q_0(\cdot\,;w)$, and can be positive even when every $\delta_k(w)$ is zero, while $\Lambda_d(w)$ and $\Lambda_\theta(w)$ also capture differences in mean potential outcomes across units with the same $w$. Thus, although the results in Theorem~\ref{sensB:thm-bounds} below hold without~\eqref{unif:eq-ec}, interpreting $\delta(w)$ as a measure of assignment dependence alone and $\Lambda_d(w)$ and $\Lambda_\theta(w)$ as measures of interference alone requires~\eqref{unif:eq-ec}.

\begin{theorem}[Bounds on the assignment-dependence bias]\label{sensB:thm-bounds}
Suppose Assumption~\ref{unif:ass-sep} holds and $\kappa>0$, let $\lambda:\mathcal W\to[0,1]$ be measurable, let
$m_d(w):=\mathbbm{E}[Y_i\mid D_i=d, W_i=w]$ denote the  observed outcome means of comparison cells, and suppose the integrals below are
finite.
\begin{itemize}
\item[(i)] For any reference law $\pi^\ast$,
$|\mathcal B(\pi^\ast)|\le\kappa^{-1}\int\{\Lambda^\ast_1(w)\,d_{\mathrm{TV}}(\pi_1,\pi^\ast)
+\Lambda^\ast_0(w)\,d_{\mathrm{TV}}(\pi_0,\pi^\ast)\}\,dG,$ where $d_{\mathrm{TV}}$ is total variation
at $w$ and $\Lambda_d^\ast(w)$ is the range of $\mu_d(\cdot\,;w)$ over the union of the
supports of $\pi_1$, $\pi_0$, and $\pi^\ast$.
\item[(ii)] $|\mathcal B(\pi_\lambda)|\le\kappa^{-1}\int\delta(w)\,R_w(m_\lambda(\cdot\,;w))\,dG
\le\min\{b^{(1)}_\lambda,b^{(2)}_\lambda\}$, where
$b^{(1)}_\lambda:=\frac1\kappa\int\delta(w)\{(1-\lambda(w))\Lambda_1(w)+\lambda(w)\Lambda_0(w)\}\,dG,$ $b^{(2)}_\lambda:=\int\delta(w)\Big\{\frac{\Lambda_0(w)}\kappa+(1-\lambda(w))\Lambda_\theta(w)\Big\}\,dG$ .

\item[(iii)] (Sharpness.) Suppose $N\ge2$; fix $\kappa$, $G$, $\lambda$, bounded measurable cell means $m_1(w),m_0(w)$, and bounds $(\bar\delta,\bar\Lambda_1,\bar\Lambda_0,\bar\Lambda_\theta)$, and let $\bar b^{(k)}_\lambda<\infty$ denote $b^{(k)}_\lambda$ evaluated at those bounds. For $k\in\{1,2\}$, every $t\in[-\bar b^{(k)}_\lambda,\bar b^{(k)}_\lambda]$ equals $\mathcal B(\pi_\lambda)$ for some $(\pi_0,\pi_1,\mu_0,\mu_1)$ in version~$k$ of the sensitivity model of Definition~\ref{sensB:def-params} with $\sum_{\mathbf d_{-i}}\mu_d(\mathbf d_{-i};w) \pi_d(\mathbf d_{-i};w)=m_d(w)$ for all $(d,w)$.
\end{itemize}
Hence, in version~$k$, the identified set of $T(\pi_\lambda)$, given the comparison cell outcome means and absent restrictions on the outcome support, on the joint assignment law across units, or on interference beyond the range bounds, is exactly
$[\Phi_{\mathrm{ITR}}(P)-\bar b^{(k)}_\lambda,\ \Phi_{\mathrm{ITR}}(P)+\bar b^{(k)}_\lambda]$.
~\hyperref[prf:sensB:thm-bounds]{\textup{\footnotesize[Proof]}}\end{theorem}
Part~(i) of Theorem~\ref{sensB:thm-bounds} is the general result and is essential only for reference laws outside the mixture family, which in this paper occurs for one sharp RDD target. All remaining targets fall under part~(ii), which delivers a bound at least as tight.
The bound in part (ii) is a product of an \textit{assignment-dependence factor} and an \textit{interference factor}; either factor being zero collapses the set to a point, and the required restriction depends on the target. Part~(iii) holds the observed cell means fixed, so the interval is the exact set of values the model permits, not an outer set; it is also exact over joint assignment laws of the $N$ units when every covariate value is shared by at least two units, whereas restrictions on the outcome support or on the design's assignment structure can shrink it (Remark~\ref{sensB:rem-realize}).

 Theorem~\ref{sensB:thm-bounds} implies that for some constant $t_0\in \mathbbm{R}$ (our default being $t_0=0$), if $\Phi_{\mathrm{ITR}}(P)>t_0$, the conclusion $T(\pi_\lambda)>t_0$ holds throughout the model if and only if $\bar b_\lambda<\Phi_{\mathrm{ITR}}(P)-t_0$. With $t_0=0$, the conclusion is that the ADE-type target is positive, and it survives if and only if the largest bias the model permits is smaller than the identification formula itself, $\bar b_\lambda<\Phi_{\mathrm{ITR}}(P)$. For instance, in the temporary-agency illustration, where $\Phi_{\mathrm{ITR}}(P)=0.180$ in Tuscany, the effect of an agency assignment remains positive for every sensitivity parameter with $\bar b_\lambda<0.180$. Furthermore, with constant parameters, $\bar b^{(2)}_\lambda=\bar\delta\,\bar\Lambda_{\mathrm{eff}}$, where $\bar\Lambda_{\mathrm{eff}}:=\kappa^{-1}\bar\Lambda_0+(1-\bar\lambda)\bar\Lambda_\theta$ and $\bar\lambda:=\int\lambda(w)\,dG$, the \emph{breakdown frontier} for the conclusion that $T(\pi_\lambda)>t_0$ is therefore the hyperbola $\bar\delta\bar\Lambda_{\mathrm{eff}}=\Phi_{\mathrm{ITR}}(P)-t_0$. It traces the inverse trade-off between allowable assignment dependence and interference consistent with preserving the conclusion. Thus, at the default $t_0=0$, the frontier separates sensitivity values for which the ADE-type target is necessarily positive from those for which a non-positive value cannot be ruled out.
For a given interference bound $\bar\Lambda$, the corresponding \emph{breakdown point} of assignment dependence is  $\delta^\ast(\bar\Lambda):=\min\{1,(\Phi_{\mathrm{ITR}}(P)-t_0)/\bar\Lambda_{\mathrm{eff}}\}$, which is the largest magnitude of assignment dependence the conclusion can tolerate. Version~1 factors in the same way, with $\bar\Lambda_{\mathrm{eff}}=\kappa^{-1}\{(1-\bar\lambda)\bar\Lambda_1+\bar\lambda\bar\Lambda_0\}$.

The usefulness of Theorem~\ref{sensB:thm-bounds} depends on how its sensitivity parameters map into the departures they are designed to measure. The interference parameters have a direct interpretation, since $\Lambda_0$ and $\Lambda_\theta$ measure how much mean potential outcomes vary with others' assignments. The total-variation parameter $\delta$ is likewise informative when own treatment depends on a limited set of others, as with household members, classmates, or co-workers, where it moves with the strength of that dependence and, with a single peer, equals the corresponding pairwise difference in absolute value. Under constrained-assignment mechanisms, however, the TV distance $\delta$ can overstate the dependence that matters, because it registers any difference between the laws of the whole assignment vector, however small its effect on each unit. For example, under complete randomization, $\delta=1$ even though $\Pr(D_j=1\mid D_i=1)-\Pr(D_j=1\mid D_i=0)=-1/(N-1)$ for every $i\neq j$, so pairwise dependence vanishes with $N$. The same distinction arises under quotas, matched designs, capacity rationing, and rank-based selection. This limits the TV distance as an accurate measure of global differences in assignment laws, motivating a second parameterization.

A natural response to the foregoing limitation is to replace the TV distance with a different measure of distributional distance. However, doing so requires some restrictions on interference. This is because, by \eqref{sensB:eq-mix}, the bias pairs $m_\lambda$ with the difference in assignment laws, $\pi_1-\pi_0$, so the admissible class of mean response functions $m_\lambda$ determines the dependence measure that can uniformly bound this pairing. The following clarifies how restrictions on $m_\lambda$ determine both the uniform bound and the dependence measure that supports it. Fix $w$ and let $f(\cdot)=m_\lambda(\cdot\,;w)$. The relevant bound on bias is obtained by maximizing the inner sum over all $f$ allowed by the sensitivity model. The first equality in \eqref{sensB:eq-duality} recovers the total-variation bound used in Theorem~\ref{sensB:thm-bounds} under the range restriction $R_w(f)\le\Lambda$, while the second shows that an $L$-Lipschitz restriction instead yields a Wasserstein bound:
\begin{equation}
\begin{split}
\sup_{f:\,R_w(f)\le\Lambda}\Big|\sum_{\mathbf d_{-i}}f(\mathbf d_{-i})\{\pi_1 (\mathbf d_{-i};w)-\pi_0(\mathbf d_{-i};w)\}\Big|
&=\Lambda\,\delta(w),\\
\sup_{f:\,L\text{-Lip}(\rho_\omega)}\Big|\sum_{\mathbf d_{-i}}f(\mathbf d_{-i})\{\pi_1(\mathbf d_{-i};w)-\pi_0(\mathbf d_{-i};w)\}\Big|
&=L\,\mathrm W_{\rho_\omega}(\pi_1,\pi_0).
\end{split}
\label{sensB:eq-duality}
\end{equation}
The first equality follows from Lemma~\ref{sensB:lem-tvosc} in Section~\ref{app:proofs} of the Supplementary Material and the second from Kantorovich--Rubinstein duality (Lemma~\ref{sensB:lem-coupling} in Section~\ref{app:proofs} of the Supplementary Material), with $\rho_\omega(\mathbf a,\mathbf b):=\sum_{j=1}^{N-1}\omega_j\mathbbm 1\{a_j\neq b_j\}$ and $\mathrm W_{\rho_\omega}(\pi_1,\pi_0)$ the $1$-Wasserstein distance between $\pi_1$ and $\pi_0$ induced by $\rho_\omega$. Dividing by $\kappa$ and integrating over $w$ converts these expressions into bounds on $\mathcal B(\pi_\lambda)$. The first, evaluated at the model's bound on $R_w(m_\lambda)$, gives Theorem~\ref{sensB:thm-bounds}(ii), while the second gives  Theorem~\ref{sensB:prop-refine} below.

Hence, using more accurate measures of assignment dependence requires additional structure on interference, which we treat as optional. Unlike the sensitivity models in Theorem~\ref{sensB:thm-bounds}, the following theorem imposes restrictions unit by unit and thereby defines a different class of sensitivity models. The resulting bias is the weighted average of the corresponding unit-level biases.

\begin{theorem}[Refinement under bounded influence and anonymous interference]\label{sensB:prop-refine}
Suppose Assumption~\ref{unif:ass-sep} and the equal composition condition in \eqref{unif:eq-ec} hold for all $k\in[N]$, so that $q_1(k;w)=q_0(k;w)=:q(k;w)$. Let $N\ge2$, $\kappa>0$ and the integrals below be finite. For a unit $k$, its laws $\pi_{d,k}(\cdot\,;w)$ are distributions on $\{0,1\}^J$, $J=N-1$, with coordinates indexed by the other units $j\in[N]\setminus\{k\}$, over which all sums in $j$ below run; let $m_{\lambda,k}(\cdot\,;w)$ denote its mean response $m_\lambda$, and let $\mathbf d^{(j)}_{-k}$ be $\mathbf d_{-k}$ with coordinate $j$ switched.  Define the difference between unit $k$'s cells in the probability that unit $j$ is assigned by $\Gamma_{kj}(w):=\sum_{\mathbf d_{-k}}d_j\{\pi_{1,k}(\mathbf d_{-k};w)-\pi_{0,k}(\mathbf d_{-k};w)\}$, and let $S(\mathbf d_{-k}):=J^{-1}\sum_j d_j$, $\bar\Gamma_k:=J^{-1}\sum_j \Gamma_{kj}=\mathbbm{E}_{\pi_{1,k}}[S]-\mathbbm{E}_{\pi_{0,k}}[S]$, and $\mathrm W_{S,k}(w):=\int_0^1|F^S_{1,k}-F^S_{0,k}|(t;w)\,dt$, with $F^S_{d,k}$ the cdf of $S$ under $\pi_{d,k}(\cdot\,;w)$. Let $\bar\Gamma:=\mathbbm{E}_{\pi_1}[S]-\mathbbm{E}_{\pi_0}[S]=\sum_kq(k;w)\bar\Gamma_k$ and $\mathrm W_S$ be the same objects for the pooled cell laws. Consider, for $G$-almost all $w$ and every unit $k$ with $q(k;w)>0$,  the following two classes of restrictions:\\
\noindent \emph{\underline{(A) Restrictions on interference.}}
\begin{itemize}[topsep=2pt,itemsep=1pt,parsep=0pt,leftmargin=3.1em]
    \item[\emph{(BI)}] $|m_{\lambda,k}(\mathbf d_{-k};w)-m_{\lambda,k}(\mathbf d^{(j)}_{-k};w)|\le L\,\omega_{kj}(w)$ for all $\mathbf d_{-k},j$, for possibly unknown weights $\omega_{kj}(w)\ge0$ with $\sum_j\omega_{kj}(w)=1$;
    \item[\emph{(AN)}] $m_{\lambda,k}(\mathbf d_{-k};w)=g_k(S(\mathbf d_{-k});w)$ with $|g_k(s;w)-g_k(s';w)|\le L^S|s-s'|$, for all $s,s'\in[0,1]$.
\end{itemize}
\emph{\underline{(B) Restrictions on treatment assignment.}}
\begin{itemize}[topsep=2pt,itemsep=1pt,parsep=0pt,leftmargin=3.1em]
    \item[\emph{(MD)}] there are $\mathbf D^1\sim\pi_{1,k}(\cdot\,;w)$ and $\mathbf D^0\sim\pi_{0,k}(\cdot\,;w)$ with either $\mathbf D^1\ge\mathbf D^0$ a.s.\ or $\mathbf D^1\le\mathbf D^0$ a.s., the same direction for all coordinates;
    \item[\emph{(EX)}] $\pi_{1,k}(\cdot\,;w)$ and $\pi_{0,k}(\cdot\,;w)$ are permutation invariant.
\end{itemize}
Then:
\vspace{-0.2cm}
\begin{itemize}
    \item[(i)] under (BI) and (MD), $|\mathcal B(\pi_\lambda)|\le\kappa^{-1}\int L\sum_kq(k;w)\big|\sum_j\omega_{kj}\Gamma_{kj}\big|\,dG$, which is at most $\kappa^{-1}\int L\sum_kq(k;w)\max_j|\Gamma_{kj}|\,dG$; \vspace{-0.2cm}
    \item[(ii)] under (BI) and (EX), $|\mathcal B(\pi_\lambda)|\le\kappa^{-1}\int L\sum_kq(k;w)\mathrm W_{S,k}\,dG$ for any weights; \vspace{-0.2cm}
    \item[(iii)] under (AN), $|\mathcal B(\pi_\lambda)|\le\kappa^{-1}\int L^S\sum_kq(k;w)\mathrm W_{S,k}\,dG$, and if $g_k(s;w)=\alpha_k(w)+\beta_k(w)s$, then $\mathcal B(\pi_\lambda)=\kappa^{-1}\int\sum_kq(k;w)\beta_k\bar\Gamma_k\,dG$; if $g_k=g$ is common across units, the bound holds with the pooled $\mathrm W_S\le\sum_kq(k;w)\mathrm W_{S,k}$ in place of the average,\footnote{The pooled $\mathrm W_S$ is the Wasserstein distance between the share laws under the pooled cell laws $\pi_d=\sum_kq(k;w)\pi_{d,k}$, which ignore unit identity; their share cdfs are $F^S_d=\sum_kq(k;w)F^S_{d,k}$.} and the exact expression reduces to $\kappa^{-1}\int\beta\bar\Gamma\,dG$;  \vspace{-0.2cm}
    \item[(iv)] if $F^S_{1,k}$ and $F^S_{0,k}$ are ordered, as under (MD), then $\mathrm W_{S,k}=|\bar\Gamma_k|$; if they are ordered in the same direction for every $k$, then $\sum_kq(k;w)|\bar\Gamma_k|=|\bar\Gamma|$ and the pooled share laws are ordered, so $\mathrm W_S=|\bar\Gamma|$.
\end{itemize}
For given unit-level cell laws, the bounds are sharp in the following sense: the first bound in (i) for given weights, the second in (i) once the weights are treated as unknown, the bound in (iii) unconditionally, and the bound in (ii) with equal weights or with ordered share laws. Using each bound, we can define the identified sets, breakdown frontier, and breakdown point analogously to those in  Theorem~\ref{sensB:thm-bounds}.
~\hyperref[prf:sensB:prop-refine]{\textup{\footnotesize[Proof]}}\end{theorem}

The constant $L$ measures the total influence of others on a mean response $m_{\lambda,k}$, and $\omega_{kj}$ is unit $j$'s share of it for unit $k$. Thus, condition (BI) bounds any single peer's influence. It leaves the response function otherwise unrestricted, allowing complementarities and nonmonotone effects, and constrains only how total influence is distributed across peers. A single pivotal peer is admissible, but many are not. Condition (AN) imposes anonymous interference, so that mean potential outcomes depend on others' treatment only through the treated share; see, e.g., \cite{manski2013identification,leung2020treatment,cai2015social}. It also allows any Lipschitz, possibly nonmonotone mean response, but rules out dependence on which units are treated, so only the treated share matters.

The assignment-law restrictions play a different role. Conditions (MD) and (EX) do not tighten the bound; they make the Wasserstein distance computable from differences in marginal assignments. Without them, $|\Gamma_{kj}|$ only provides a lower bound on the probability that coordinate $j$ differs across the two assignment laws. Condition (MD) ensures that this lower bound is attained through a coordinatewise-ordered coupling, while (EX) makes the coordinates exchangeable so the resulting transport cost can be summarized by the average marginal difference $\bar\Gamma_{k}$. Together, these conditions allow $\bar\Gamma_{k}$ to bound the bias from above. 
Remark~\ref{sensB:rem-benchmarks} in Section~\ref{supp:primitive} of the Supplementary Material provides further comments on $\Gamma_{kj}$ and $\bar\Gamma_{k}$, while Lemma~\ref{sensB:lem-order} gives sufficient conditions for (MD).

The equal-composition condition in \eqref{unif:eq-ec} guarantees that the pooled bias is the composition-weighted average of the unit-level biases, so the bounds of Theorem~\ref{sensB:prop-refine} are averages of unit-level measures of assignment dependence. In particular, the bound in (i) is the average of maxima $\bar M:=\sum_kq(k;w)\max_j|\Gamma_{kj}|$, which is at least the maximum of averages $\max_j|\Gamma_j|$, with $\Gamma_j:=\sum_{k\neq j}q(k;w)\Gamma_{kj}$ the pooled gap for coordinate $j$. The share-based averages $\bar{\mathrm W}_S:=\sum_kq(k;w)\mathrm W_{S,k}$ and $\sum_kq\,|\bar\Gamma_k|$ coincide with the pooled $\mathrm W_S$ and $|\bar\Gamma|$ under a common share response or same-direction ordering, respectively, which are the cases used in Sections~\ref{sen appl} and~\ref{sec:montecarlo}. Without equal composition, (BI) alone does not bound the bias, because $m_\lambda$ may differ across drawn units by an arbitrarily large additive constant.

Theorem~\ref{sensB:prop-refine} yields the following main conclusions. Part~(i) gives two bounds: the first uses the influence-weighted gaps $|\sum_j\omega_{kj}\Gamma_{kj}|$ and applies for given weights, while the second replaces them by $\max_j|\Gamma_{kj}|$ because the weights are unknown, so the bound must allow the most dependent peer to be the most influential. Part~(ii) removes this conservatism under exchangeability, which equalizes the probability that coordinate $j$ differs under the coupling across $j$ and makes the resulting measure $\bar{\mathrm W}_S$ invariant to the unknown weights. Part~(iii) sharpens the interference component from $L$ to $L^S$. When each $g_k(\cdot;w)$ is affine, it also yields the exact expression $\sum_k q(k;w)\beta_k\bar\Gamma_k$. If, in addition, the share response is common across units, $\bar{\mathrm W}_S$ can be replaced by the smaller pooled measure $\mathrm W_S$, and the exact expression simplifies to $\beta\bar\Gamma$. Part~(iv) clarifies that $|\bar\Gamma_k|$ is the relevant share-based Wasserstein distance of unit $k$ under ordered share laws, and that $|\bar\Gamma|$ is the pooled one when the ordering runs in the same direction across units; otherwise $|\bar\Gamma|$ may vanish even when the bias does not. Parts~(iii)--(iv) apply equally to the share among the units of a fixed group, with $\bar\Gamma$ replaced by its group counterpart $\bar\gamma$, since their proofs use only the law of the share. 

The dependence measures that replace the TV distance $\delta$ of Theorem~\ref{sensB:thm-bounds} in Theorem~\ref{sensB:prop-refine}---the averages $\sum_kq(k;w)\big|\sum_{j}\omega_{kj}\Gamma_{kj}\big|$, $\bar M$ and $\bar{\mathrm W}_S$, and the pooled $\mathrm W_S$ and $|\bar\Gamma|$---are all bounded above by \(\delta\), since under~\eqref{unif:eq-ec} $\delta=\sum_kq(k;w)\delta_k(w)$ while $|\Gamma_{kj}(w)|\le\delta_k(w)$ and $|\bar\Gamma_k|\le \mathrm W_{S,k}(w)\le\delta_k(w)$. Moreover, the share-based measures in Theorem~\ref{sensB:prop-refine} (iii)-(iv) can vanish with population size even when $\delta$ does not. Under designs that constrain treatment counts or confine dependence to small groups, $\bar{\mathrm W}_S$, $\mathrm W_S$, and $|\bar\Gamma|$ are typically $O(1/N)$; the same rate applies to $\sum_k q(k;w)\big|\sum_j\omega_{kj}\Gamma_{kj}\big|$ when the influence weights are sufficiently diffuse. For example, $\bar\Gamma=-1/(N-1)$ under complete randomization and $\bar\Gamma=(|B|-1)/(N-1)$ when groups of size $|B|$ take up jointly and independently across groups.

In summary, Theorem~\ref{sensB:prop-refine} yields four sensitivity models, corresponding to its four parts and defined analogously to those in Definition~\ref{sensB:def-params}. Specifically, each pairs a restriction on interference with one on the assignment mechanism and replaces the total-variation parameter $\delta$ by a measure of assignment dependence that these restrictions make tractable. Hereafter, we refer to the models of parts~(iii) and~(iv), which assume anonymous interference (AN), as the \emph{share version} of the sensitivity analysis and to the models of Definition~\ref{sensB:def-params} as the \emph{total-variation version}. Table~\ref{tab:sens-models} summarizes the six sensitivity models generated by Theorems~\ref{sensB:thm-bounds}--\ref{sensB:prop-refine}, together with the corresponding identified sets, breakdown frontiers, and breakdown points. 

\begin{table}[h]
\centering\footnotesize
\setlength{\tabcolsep}{3pt}
\renewcommand{\arraystretch}{1.5}
\caption{The sensitivity models}\label{tab:sens-models}
\begin{tabular}{@{}
  >{\raggedright\arraybackslash}p{1.75cm}
  >{\raggedright\arraybackslash}p{3.1cm}
  >{\raggedright\arraybackslash}p{1.8cm}
  >{\raggedright\arraybackslash}p{2.5cm}
  >{\raggedright\arraybackslash}p{2.8cm}
  >{\raggedright\arraybackslash}p{3.45cm}@{}}
\toprule
Model in & Interference restriction & Assignment restriction & Identified set $\widehat\Phi\pm\bar b$, with $\bar b=$ & Breakdown frontier & Breakdown point \\
\midrule
Thm~\ref{sensB:thm-bounds}, v.1
& ranges $\bar\Lambda_1,\bar\Lambda_0$ of $\mu_1,\mu_0$
& any, with $\delta\le\bar\delta$
& $\bar\delta\,\bar\Lambda^{(1)}_{\mathrm{eff}}$
& $\bar\delta\,\bar\Lambda^{(1)}_{\mathrm{eff}}=H$
& $\delta^\ast=\min\{1,H_+/\bar\Lambda^{(1)}_{\mathrm{eff}}\}$ \\
\addlinespace[4pt]
Thm~\ref{sensB:thm-bounds}, v.2
& ranges $\bar\Lambda_0,\bar\Lambda_\theta$ of $\mu_0,\theta$
& any, with $\delta\le\bar\delta$
& $\bar\delta\,\bar\Lambda^{(2)}_{\mathrm{eff}}$
& $\bar\delta\,\bar\Lambda^{(2)}_{\mathrm{eff}}=H$
& $\delta^\ast=\min\{1,H_+/\bar\Lambda^{(2)}_{\mathrm{eff}}\}$ \\
\addlinespace
Thm~\ref{sensB:prop-refine}(i)
& (BI): peer $j$ moves $m_{\lambda,k}$ by at most $L\omega_{kj}$
& (MD)
& $\kappa^{-1}L\,\bar M$
& $\kappa^{-1}L\,\bar M=H$
& $\bar M^\ast=\min\{1,\kappa H_+/L\}$ \\
\addlinespace[4pt]
Thm~\ref{sensB:prop-refine}(ii)
& (BI), any weights
& (EX)
& $\kappa^{-1}L\,\bar{\mathrm W}_S$
& $\kappa^{-1}L\,\bar{\mathrm W}_S=H$
& $\bar{\mathrm W}^\ast_S=\min\{1,\kappa H_+/L\}$ \\
\addlinespace[4pt]
Thm~\ref{sensB:prop-refine}(iii)
& (AN): $m_{\lambda,k}=g_{k}(S)$ with $g_{k}$ Lipschitz
& any
& $\bar L_{\mathrm{eff}}\,\bar{\mathrm W}_S$, or $\bar L_{\mathrm{eff}}\,\mathrm W_S$ with a common $g$
& $\bar L_{\mathrm{eff}}\,\bar{\mathrm W}_S=H$
&$\bar{\mathrm W}^\ast_S=\min\{1,H_+/\bar L_{\mathrm{eff}}\}$ \\
\addlinespace[4pt]
Thm~\ref{sensB:prop-refine}(iv)
& (AN), as in (iii)
& share laws ordered within each unit, as under (MD), and in the same direction across units
& $\bar L_{\mathrm{eff}}\,|\bar\Gamma|$, as $\bar{\mathrm W}_S=|\bar\Gamma|$
& $\bar L_{\mathrm{eff}}\,|\bar\Gamma|=H$
& $\delta^\ast_S=\min\{1,H_+/\bar L_{\mathrm{eff}}\}$ 
\\
\bottomrule
\end{tabular}

\smallskip
\parbox{\linewidth}{\scriptsize \emph{Notes:} All entries are stated for sensitivity parameters constant in $w$ and for the one-sided conclusion $T(\pi_\lambda)>t_0$, with $H:=\widehat\Phi-t_0$ the margin and $x_+:=\max\{x,0\}$. The effective interference constants are $\bar\Lambda^{(1)}_{\mathrm{eff}}=\kappa^{-1}\{(1-\bar\lambda)\bar\Lambda_1+\bar\lambda\bar\Lambda_0\}$, $\bar\Lambda^{(2)}_{\mathrm{eff}}=\kappa^{-1}\bar\Lambda_0+(1-\bar\lambda)\bar\Lambda_\theta$ and $\bar L_{\mathrm{eff}}=\kappa^{-1}L^S_0+(1-\bar\lambda)L^S_\theta$. Each dependence measure lies in $[0,1]$, and the corresponding breakdown point, $\delta^\ast$, $\bar M^\ast$, $\bar{\mathrm W}^\ast_S$, or $\delta_S^\ast$, is the largest value of that measure consistent with the conclusion, capped at one.}
\end{table}

\subsection {Interpretation and Reporting} 
\paragraph{Interpretation.}
The parameter $\delta$ admits an interpretation on a scale similar to the linear projection coefficient. Specifically, in Section~\ref{app:proofs} of the Supplementary Material, we show in \eqref{sensB:eq-delta-prop} that $\delta(w)=\mathbbm{E}|\Pr(D_i=1\mid\mathbf D_{-i}, W_i=w)-\Pr(D_i=1\mid W_i=w)|/\{2\Pr(D_i=1\mid W_i=w)(1-\Pr(D_i=1\mid W_i=w))\}$, the average amount by which learning others' assignments moves the propensity of one's own.  By contrast, the interference parameters are measured in outcome and effect units. For example, setting $\bar\Lambda_0=0.1\,\mathrm{SD}(Y_i)$ would mean that changing others' assignments, at the extreme from none to all treated, can shift an untreated unit's mean outcome by at most one-tenth of an outcome standard deviation; $\bar\Lambda_\theta$ is interpreted analogously for the direct effect. 

Substituting each design's own variables and outcome means, as in Table~\ref{unif:tab-mapping}, evaluates the bound design by design (see Table~\ref{sensB:tab-designs} of the supplement). Three cases are especially useful for interpreting how the general bound specializes across familiar estimands. Under selection on observables, $\delta(w)$ compares the laws of $\mathbf D_{-i}$ given $D_i=1$ and $D_i=0$ within $W_i=w$, and with $\lambda=p(W_i)=\Pr(D_i=1\mid W_i)$ the ADE bound is $\mathbbm{E}[\delta(W_i)\{\Lambda_0+(1-p(W_i))\Lambda_\theta\}]$; for the ADTT targets of DiD and of selection on observables ($\lambda=1$, $G=F_{W\mid D=1}$), $\Lambda_\theta$ drops out and only interference on untreated outcomes (untreated trends in DiD) can bias the estimand. In the IV setting, the relevant dependence is between instruments; $\Lambda_0$ bounds reduced-form spillovers of others' instruments on outcomes, and $\kappa$ divides $\bar\Lambda_0$, so a weak first stage \textit{magnifies} a given dependence. In the sharp RDD, the parameter is $\delta_{w_0}:=d_{\mathrm{TV}}(\pi_+,\pi_-)$, the jump at the cutoff in the law of others' treatments, whose share version $\bar\Gamma_{w_0}$ is the jump in the treated share, estimated with the RDD estimator used for the outcome.

For the sensitivity parameters of Theorem~\ref{sensB:prop-refine}, we can rewrite $\Gamma_{kj}(w)=\cov(D_k,D_j\mid W_k=w)/\var(D_k\mid W_k=w)$, the within-$W_k$ linear projection coefficient of unit $j$'s assignment on unit $k$'s. The quantity $\bar\Gamma_k$ averages these coefficients across the other units $j$, and $\bar\Gamma$ averages $\bar\Gamma_k$ across drawn units, while $\widehat\gamma^{\mathcal G}$, defined below, estimates the corresponding average over group mates for an observed grouping. Condition (AN) constrains $m_\lambda$ through a single constant; imposing it on $\mu_0$ and $\theta$ separately, with constants $L^S_0$ and $L^S_\theta$, gives $L^S\le L^S_0+(1-\lambda)\kappa L^S_\theta$ and hence $\bar L_{\mathrm{eff}}:=\kappa^{-1}L^S_0+(1-\bar\lambda)L^S_\theta$, the share counterpart of $\bar\Lambda_{\mathrm{eff}}$. Both \(L^S_0\) and \(L^S_\theta\) are slope parameters, bounding the change in the corresponding mean response or effect per unit change in the treated share. For example, \(L^S_0\) bounds the change in an untreated unit's mean outcome, so variation over the full share range is at most \(L^S_0\) and hence \(\Lambda_0\le L^S_0\).

\paragraph{Reporting.} 
For each sensitivity model, the basic outputs are the standard ITR-based estimate $\widehat\Phi$, the identified set at the chosen sensitivity parameters, and the corresponding breakdown frontier and breakdown point. The latter reports the largest degree of assignment dependence consistent with an empirical conclusion of interest for a given bound on interference. For the total-variation versions, for example, the breakdown point is $\delta^\ast(\bar\Lambda):=\min\{1,(\widehat\Phi-t_0)_+/\bar\Lambda_{\mathrm{eff}}\}$.

Sampling uncertainty can then be incorporated into each of these objects. Suppose $(\widehat\Phi-\Phi_{\mathrm{ITR}}(P))/\widehat{\mathrm{se}}$ is asymptotically standard normal, where $\widehat{\mathrm{se}}$ is an estimated standard error, and let $\bar b$ be the bias bound implied by the sensitivity parameters. Widening the usual confidence interval by $\bar b$ on each side gives a confidence set for the identified set, while replacing $z_{1-\alpha/2}$ by the Imbens--Manski critical value $\hat c\in[z_{1-\alpha},z_{1-\alpha/2}]$ gives a confidence interval for the target itself \citep{imbens2004confidence}. Coverage is pointwise in the data-generating process rather than uniform.\footnote{The confidence statements presume a standard error valid under the allowed dependence and interference, which ITR variance formulas do not deliver. See Remark~\ref{sens:rem_limits} of the Supplementary Material for a discussion and a recommended remedy.} Proposition~\ref{sensB:prop-ci} of the Supplementary Material also provides a weak-instrument-robust Anderson--Rubin-type set for IV. The corresponding \textit{inference-robust} breakdown point for the total-variation versions is $\delta^\ast_\alpha(\bar\Lambda):=\min\{1,(\widehat\Phi-z_{1-\alpha}\widehat{\mathrm{se}}-t_0)_+/\bar\Lambda_{\mathrm{eff}}\}$.

A breakdown point is most informative when reported relative to a benchmark on the same scale. Such a benchmark is unavailable for $\delta$, which is not identified from a single realization of the assignment vector, but an observed grouping $\mathcal G$ provides one for the share-based measures in Theorem~\ref{sensB:prop-refine}. Let $\mathcal G(i)$ denote unit $i$'s group and $n_{\mathcal G(i)}$ its size in the population of $N$ units. Let $\widehat S_i^{\mathcal G}$ be unit $i$'s leave-one-out treated share among group mates, and let $\widehat\gamma^{\mathcal G}$ be its within-$W_i$ difference across comparison cells, averaged using the weights of the target estimand (e.g., $q_1$ for the ADTT). This estimates $\bar\gamma$, the analog of $\bar\Gamma$ for the treated share among group mates, and is the relevant benchmark when interference is local to the observed groups and $L^S$ is defined per unit change in the group-level treated share. On the population-wide share scale, define $\widehat\Gamma^{\mathcal G}:=\widehat\gamma^{\mathcal G}\times\overline{(n_{\mathcal G(i)}-1)/(N-1)}$, where the average uses the weights of the target estimand. When assignment dependence operates only through $\mathcal G$, $\widehat\Gamma^{\mathcal G}$ estimates $\bar\Gamma$ and is the corresponding benchmark when interference operates through the treated share among all $N-1$ other units. For the share-based models, we therefore report $\delta^\ast_S:=\min\{1,(\widehat\Phi-t_0)_+/\bar L_{\mathrm{eff}}\}$ and its inference-robust counterpart $\delta^\ast_{S,\alpha}$ together with their multiples of the relevant observed benchmark. Under local interference, the \textit{breakdown multiples} are $k^\ast:=\delta^\ast_S/|\widehat\gamma^{\mathcal G}|$ and $k^\ast_\alpha:=\delta^\ast_{S,\alpha}/|\widehat\gamma^{\mathcal G}|$; for interference through the population-wide treated share, replace $|\widehat\gamma^{\mathcal G}|$ by $|\widehat\Gamma^{\mathcal G}|$. Each ratio represents \textit{the multiple of the observed assignment dependence} required for the conclusion to break down.

The following box summarizes the two theorems in six reportable objects, spanning the target, calibrated interference bounds, a valid standard error, identified sets, the breakdown frontier, and the observed dependence used to interpret it. See Procedure~\ref{sensB:alg} in the supplement for detailed computational and calibration guidance.

\begin{singlespace}
\noindent\rule{\linewidth}{0.4pt}\\[-2pt]
\refstepcounter{repbox}\label{box:report}
\emph{What to report}.\\
\noindent\rule{\linewidth}{0.4pt}\\[-2pt]
\begin{enumerate}[leftmargin=1.9em,itemsep=0pt,topsep=1pt,parsep=0pt,label=\arabic*.]
\vspace{-0.5cm}
\item \emph{Target.} State $\lambda$ and the ADE-type target $T(\pi_\lambda)$ it selects, from the design's row of Table~\ref{sensB:tab-designs} in the supplement.
\item \emph{Units.} Express $\bar\Lambda_0$ (or $L^S_0$) as a multiple $\ell_0$ of the outcome standard deviation and $\bar\Lambda_\theta$ (or $L^S_\theta$) as a multiple $\ell_\theta$ of $|\widehat\Phi|$; form $\bar\Lambda_{\mathrm{eff}}$(or $\bar L_{\mathrm{eff}}$).
\item \emph{Standard error.} Use a standard error valid under the dependence the model allows, clustered on the grouping through which dependence runs or from a cluster bootstrap; state its reach, and treat every $\alpha$-indexed object below as conditional on it (Remark~\ref{sens:rem_limits} of the supplement).
\item \emph{Sets.} Report $[\widehat\Phi-\bar b,\widehat\Phi+\bar b]$ and the Imbens--Manski interval at $\bar\delta=1$ and at two or three benchmark $(\ell_0,\ell_\theta)$.
\item \emph{Frontier.} Plot $\bar\delta\,\bar\Lambda_{\mathrm{eff}}=\widehat\Phi-t_0$ and tabulate $\delta^\ast,\delta^\ast_\alpha$ over a grid of $\bar\Lambda_{\mathrm{eff}}$; under (AN), plot the share frontier $|\bar\Gamma|\,\bar L_{\mathrm{eff}}=\widehat\Phi-t_0$, or its group-share counterpart, and tabulate $k^\ast,k^\ast_\alpha$ over a grid of $L^S_0$.
\item \emph{Benchmark.} Report $\widehat\gamma^{\mathcal G}$ for each observed grouping and the multiples $k^\ast,k^\ast_\alpha$.
\end{enumerate}
\vspace{-4pt}\noindent\rule{\linewidth}{0.4pt}
\end{singlespace}

\subsection{Simulation Results and Empirical Illustration}\label{sen appl}
\paragraph{Performance in a Monte Carlo simulation.} Section~\ref{sec:montecarlo} of the Supplementary Material evaluates the sensitivity analysis in a Monte Carlo design calibrated to the temporary-agency data, with \textit{mean responses linear in the share of others treated} so that Theorem~\ref{sensB:prop-refine}(iii)--(iv) apply.
\begin{table}[h]
\centering
\caption{Monte Carlo evaluation of the sensitivity analysis}
\label{tab:mc-sens}
\scriptsize
\setlength{\tabcolsep}{1.5pt}
\renewcommand{\arraystretch}{1.08}
\begin{tabular*}{\textwidth}{@{\extracolsep{\fill}}lccccccccc@{}}
\toprule
& & & & & \multicolumn{2}{c}{IM coverage}
& \multicolumn{2}{c}{$\mathrm{sd}(\widehat\Phi)/\widehat{\mathrm{se}}$} & \\
\cmidrule(lr){6-7}\cmidrule(lr){8-9}
Design
& Bias (s.e.)
& $\bar\Gamma$
& $\bar\gamma$ ($\times\,\widehat\gamma^{\rm ZIP}_{\rm obs}$)
& $\bar b$
& conv.
& ZIP boot.
& conv.
& ZIP boot.
& $\widehat\gamma^{\rm ZIP}$ \\
\midrule
A. CAI holds                 & $-0.000$ ($0.001$) & $0.000$ & $0$ ($0$)   & $0.000$ & $0.94$ & $0.94$ & $1.0$  & $1.0$  & $-0.001$ \\
B. ZIP take-up               & $0.012$ ($0.002$)  & $0.046$ & $0.96$ ($426$)  & $0.012$ & $0.51$ & $0.91$ & $3.1$  & $1.1$  & $1.000$ \\
C. Common shock, matched     & $0.009$ ($0.002$)  & $0.046$ & $0.04$ ($20$)  & $0.012$ & $0.50$ & $0.51$ & $3.1$  & $3.1$  & $-0.001$ \\
D. Common shock, breakdown   & $0.158$ ($0.008$)  & $0.689$ & $0.68$ ($302$) & $0.177$ & $0.18$ & $0.19$ & $16.0$ & $14.2$ & $-0.002$ \\
\bottomrule
\end{tabular*}

\begin{minipage}{\textwidth}
\vspace{3pt}
\footnotesize
\emph{Notes:}$\bar\gamma$ is the same difference for the share of mates in the same ZIP code, in parentheses as a multiple of the gap $\widehat\gamma^{\rm ZIP}_{\rm obs}=0.0023$ observed in the Tuscan data (Empirical Illustration below).   IM coverage is the frequency with which the Imbens--Manski interval at $\alpha=0.05$ contains the ADE, and $\mathrm{sd}(\widehat\Phi)/\widehat{\mathrm{se}}$ is the ratio of the Monte Carlo standard deviation of $\widehat\Phi$ to the mean standard error, each with the conventional standard error (conv.) and a ZIP-clustered pairs bootstrap (ZIP boot.).

\end{minipage}
\end{table}
Table~\ref{tab:mc-sens} summarizes the results.  The target ADE is $0.18$ in designs A--C and $0.184$ in design D (see Section~\ref{sec:montecarlo} of the Supplementary Material for each design's specification). 
The bias is the Monte Carlo (MC) mean of $\widehat\Phi$ minus the target ADE over $4{,}000$ replications, with its MC standard error in parentheses. The bound $\bar b=\bar L_{\mathrm{eff}}|\bar\Gamma|$ is the bound of Theorem~\ref{sensB:prop-refine}(iii)--(iv) at the design's share slopes and mean propensity.
When CAI holds, bias and bound are zero, corroborating Proposition~\ref{unif:thm-master}(ii), and Theorem~\ref{sensB:prop-refine} gives zero at $\bar\Gamma=0$. In designs~B and~C, the bias is within Monte Carlo error of the bound, $0.012$ against $0.012$ when entire ZIP codes take up treatment together and $0.009$ against $0.012$ under a common shock of the same magnitude. This corroborates the exact expression of Theorem~\ref{sensB:prop-refine}(iii): the simulated response is affine in the treated share with the slopes at which the bound is evaluated, so the bias equals $\bar L_{\mathrm{eff}}\bar\Gamma$. In design~D, we increase the common shock until the bound $\bar b$ is approximately as large as the target effect ($0.177$ versus $0.184$).The bias rises to $0.158$ but remains below the $0.177$ bound because the binary-outcome success probability is truncated for $1.4\%$ of unit-draws, flattening the mean response and reducing its average slope below $L^S_0$ and $L^S_\theta$. The truncated response remains Lipschitz at these constants, so Theorem~\ref{sensB:prop-refine}(iii) still applies, and the bound remains sharp.

On the scale of the data's benchmark, the share of ZIP mates, the dependence of design~B is about $430$ times the gap $\widehat\gamma^{\rm ZIP}_{\rm obs}=0.0023$ observed in the Tuscan data and that of design~D about $300$ times; the designs are stress tests rather than calibrations to the dependence the data exhibit, and their breakdown multiple is at the simulation's $\bar L_{\mathrm{eff}}=0.262$ for the ADE. 
The Imbens--Manski (IM) interval is at its nominal level when CAI holds ($0.94$) but falls to about $0.50$ once take-up is dependent. This is because the conventional standard error (obtained under independent observations) understates the sampling variability of $\widehat\Phi$ by a factor of 3.1, rising to 16 in the most dependent design. A ZIP-clustered bootstrap repairs this exactly where the dependence runs through the clusters (ZIP codes), returning coverage to $0.91$ in design~B. 

The last column of Table~\ref{tab:mc-sens}, $\widehat\gamma^{\rm ZIP}$, measures assignment dependence within ZIP codes. For each design, it is computed from a single simulated treatment assignment for the $683$ Tuscan units, holding their observed covariates and ZIP codes fixed. Like the ZIP-clustered standard error, $\widehat\gamma^{\rm ZIP}$ is therefore informative only about dependence that operates through the ZIP-code grouping. Accordingly, it equals one when take-up is perfectly dependent within ZIP codes, but is close to zero under a common regional shock, even when that shock generates the same value of the population dependence measure $\bar\Gamma$ and a positive $\bar\gamma$. By contrast, $\bar\Gamma$ and $\bar\gamma$ are not tied to a particular grouping and captures the overall strength of assignment dependence across these designs, supporting its role in Theorem~\ref{sensB:prop-refine}. We therefore vary $\bar\gamma$ from zero, the CAI benchmark, to design~B's value of about $430$ times the observed gap (under design~B), which moves $\bar\Gamma$ in proportion, and trace the resulting identified sets in Figure~\ref{fig:sweep-B}. 
In Panel~(a), for each value of $\bar\gamma$, $2{,}000$ replications yield the Monte Carlo mean identified set $[\Phi_{\mathrm{ITR}}(P)-\bar b,\Phi_{\mathrm{ITR}}(P)+\bar b]$, where $\bar b=\bar L_{\mathrm{eff}}\bar\Gamma$, together with two confidence statements based on the ZIP-clustered bootstrap standard error. The Imbens--Manski interval covers the target, while the conventional-form band $\widehat\Phi\pm(\bar b+z_{0.975}\widehat{\mathrm{se}})$ covers the identified set.
The result in Panel~(b) shows that the  IM coverage based on bootstrap standard errors remains between $0.90$ and $0.92$ across $\bar\gamma$ values, whereas conventional coverage falls from $0.93$ to $0.51$ as $\bar\gamma$ increases. The increase in dependence also widens the confidence band substantially more than the identified set.

\begin{figure}[htbp]
\centering
\resizebox{\textwidth}{!}{
\definecolor{seriesA}{HTML}{2A78D6}
\definecolor{seriesC}{HTML}{259033}
\definecolor{seriesB}{HTML}{EB6834}
\definecolor{inkmuted}{HTML}{6E6D68}
\definecolor{gridink}{HTML}{B8B7B2}
\begin{tikzpicture}[x=1cm,y=1cm,font=\footnotesize]
\begin{scope}
  \fill[seriesC!32] (0.000,2.844) -- (0.640,3.021) -- (1.280,3.175) -- (1.920,3.309) -- (2.560,3.437) -- (3.200,3.553) -- (3.840,3.660) -- (4.480,3.766) -- (5.120,3.863) -- (5.760,3.952) -- (6.400,4.038) -- (6.400,0.353) -- (5.760,0.432) -- (5.120,0.510) -- (4.480,0.595) -- (3.840,0.681) -- (3.200,0.775) -- (2.560,0.876) -- (1.920,0.980) -- (1.280,1.109) -- (0.640,1.252) -- (0.000,1.431) -- cycle;
  \fill[seriesB!85] (0.000,2.138) -- (0.640,2.145) -- (1.280,2.160) -- (1.920,2.172) -- (2.560,2.193) -- (3.200,2.209) -- (3.840,2.225) -- (4.480,2.244) -- (5.120,2.259) -- (5.760,2.274) -- (6.400,2.287) -- (6.400,2.104) -- (5.760,2.110) -- (5.120,2.113) -- (4.480,2.116) -- (3.840,2.116) -- (3.200,2.118) -- (2.560,2.120) -- (1.920,2.117) -- (1.280,2.124) -- (0.640,2.127) -- (0.000,2.138) -- cycle;
  \draw[gridink] (0,0.767) -- (6.400,0.767);
  \draw[black,thick,dashed] (0.000,2.146) -- (0.640,2.146) -- (1.280,2.146) -- (1.920,2.145) -- (2.560,2.145) -- (3.200,2.145) -- (3.840,2.145) -- (4.480,2.145) -- (5.120,2.145) -- (5.760,2.145) -- (6.400,2.145);
  \draw[seriesA!75!black,thick,dashed] (0.000,2.844) -- (0.640,3.030) -- (1.280,3.192) -- (1.920,3.335) -- (2.560,3.472) -- (3.200,3.595) -- (3.840,3.711) -- (4.480,3.824) -- (5.120,3.929) -- (5.760,4.027) -- (6.400,4.120);
  \draw[seriesA!75!black,thick,dashed] (0.000,1.431) -- (0.640,1.243) -- (1.280,1.092) -- (1.920,0.954) -- (2.560,0.841) -- (3.200,0.732) -- (3.840,0.630) -- (4.480,0.536) -- (5.120,0.444) -- (5.760,0.358) -- (6.400,0.271);
  \draw[->] (0,0) -- (6.800,0);
  \draw[->] (0,0) -- (0,5.000);
  \node[anchor=north] at (3.200,-0.62) {$\bar\gamma$, multiples of the observed $\widehat\gamma^{\rm ZIP}_{\rm obs}=0.0023$};
  \node[rotate=90,anchor=south] at (-1.0,2.300) {effect};
  \node[anchor=south] at (3.200,5.050) {(a) Identified set and confidence bands};
  \draw (0,0.000) -- (-0.08,0.000) node[anchor=east] {$-0.1$};
  \draw (0,0.767) -- (-0.08,0.767) node[anchor=east] {$0.0$};
  \draw (0,1.533) -- (-0.08,1.533) node[anchor=east] {$0.1$};
  \draw (0,2.300) -- (-0.08,2.300) node[anchor=east] {$0.2$};
  \draw (0,3.067) -- (-0.08,3.067) node[anchor=east] {$0.3$};
  \draw (0,3.833) -- (-0.08,3.833) node[anchor=east] {$0.4$};
  \draw (0,4.600) -- (-0.08,4.600) node[anchor=east] {$0.5$};
  \draw (0.000,0) -- (0.000,-0.08) node[anchor=north] {$0$};
  \draw (1.501,0) -- (1.501,-0.08) node[anchor=north] {$100$};
  \draw (3.003,0) -- (3.003,-0.08) node[anchor=north] {$200$};
  \draw (4.504,0) -- (4.504,-0.08) node[anchor=north] {$300$};
  \draw (6.005,0) -- (6.005,-0.08) node[anchor=north] {$400$};
\end{scope}
\begin{scope}[shift={(8.600,0)}]
  \draw[gridink,dashed] (0,4.217) -- (6.400,4.217);
  \node[inkmuted,anchor=south east,font=\scriptsize] at (6.400,4.217) {nominal $0.95$};
  \draw[seriesC,very thick] (0.000,4.002) -- (0.640,4.013) -- (1.280,3.975) -- (1.920,3.948) -- (2.560,3.898) -- (3.200,3.849) -- (3.840,3.849) -- (4.480,3.833) -- (5.120,3.833) -- (5.760,3.829) -- (6.400,3.841);
  \fill[seriesC] (0.000,4.002) circle (1.5pt);
  \fill[seriesC] (0.640,4.013) circle (1.5pt);
  \fill[seriesC] (1.280,3.975) circle (1.5pt);
  \fill[seriesC] (1.920,3.948) circle (1.5pt);
  \fill[seriesC] (2.560,3.898) circle (1.5pt);
  \fill[seriesC] (3.200,3.849) circle (1.5pt);
  \fill[seriesC] (3.840,3.849) circle (1.5pt);
  \fill[seriesC] (4.480,3.833) circle (1.5pt);
  \fill[seriesC] (5.120,3.833) circle (1.5pt);
  \fill[seriesC] (5.760,3.829) circle (1.5pt);
  \fill[seriesC] (6.400,3.841) circle (1.5pt);
  \draw[seriesA!75!black,very thick,dashed] (0.000,3.998) -- (0.640,3.948) -- (1.280,3.879) -- (1.920,3.795) -- (2.560,3.734) -- (3.200,3.691) -- (3.840,3.661) -- (4.480,3.645) -- (5.120,3.672) -- (5.760,3.661) -- (6.400,3.634);
  \fill[seriesA!75!black] (0.000,3.998) circle (1.5pt);
  \fill[seriesA!75!black] (0.640,3.948) circle (1.5pt);
  \fill[seriesA!75!black] (1.280,3.879) circle (1.5pt);
  \fill[seriesA!75!black] (1.920,3.795) circle (1.5pt);
  \fill[seriesA!75!black] (2.560,3.734) circle (1.5pt);
  \fill[seriesA!75!black] (3.200,3.691) circle (1.5pt);
  \fill[seriesA!75!black] (3.840,3.661) circle (1.5pt);
  \fill[seriesA!75!black] (4.480,3.645) circle (1.5pt);
  \fill[seriesA!75!black] (5.120,3.672) circle (1.5pt);
  \fill[seriesA!75!black] (5.760,3.661) circle (1.5pt);
  \fill[seriesA!75!black] (6.400,3.634) circle (1.5pt);
  \draw[inkmuted,thick,dotted] (0.000,4.098) -- (0.640,3.527) -- (1.280,2.944) -- (1.920,2.476) -- (2.560,2.074) -- (3.200,1.694) -- (3.840,1.491) -- (4.480,1.234) -- (5.120,1.089) -- (5.760,0.935) -- (6.400,0.836);
  \fill[inkmuted] (0.000,4.098) circle (1.5pt);
  \fill[inkmuted] (0.640,3.527) circle (1.5pt);
  \fill[inkmuted] (1.280,2.944) circle (1.5pt);
  \fill[inkmuted] (1.920,2.476) circle (1.5pt);
  \fill[inkmuted] (2.560,2.074) circle (1.5pt);
  \fill[inkmuted] (3.200,1.694) circle (1.5pt);
  \fill[inkmuted] (3.840,1.491) circle (1.5pt);
  \fill[inkmuted] (4.480,1.234) circle (1.5pt);
  \fill[inkmuted] (5.120,1.089) circle (1.5pt);
  \fill[inkmuted] (5.760,0.935) circle (1.5pt);
  \fill[inkmuted] (6.400,0.836) circle (1.5pt);
  \draw[->] (0,0) -- (6.800,0);
  \draw[->] (0,0) -- (0,5.000);
  \node[anchor=north] at (3.200,-0.62) {$\bar\gamma$, multiples of the observed $\widehat\gamma^{\rm ZIP}_{\rm obs}=0.0023$};
  \node[rotate=90,anchor=south] at (-1.0,2.300) {coverage};
  \node[anchor=south] at (3.200,5.050) {(b) Coverage};
  \draw (0,0.000) -- (-0.08,0.000) node[anchor=east] {$0.4$};
  \draw (0,0.767) -- (-0.08,0.767) node[anchor=east] {$0.5$};
  \draw (0,1.533) -- (-0.08,1.533) node[anchor=east] {$0.6$};
  \draw (0,2.300) -- (-0.08,2.300) node[anchor=east] {$0.7$};
  \draw (0,3.067) -- (-0.08,3.067) node[anchor=east] {$0.8$};
  \draw (0,3.833) -- (-0.08,3.833) node[anchor=east] {$0.9$};
  \draw (0,4.600) -- (-0.08,4.600) node[anchor=east] {$1.0$};
  \draw (0.000,0) -- (0.000,-0.08) node[anchor=north] {$0$};
  \draw (1.501,0) -- (1.501,-0.08) node[anchor=north] {$100$};
  \draw (3.003,0) -- (3.003,-0.08) node[anchor=north] {$200$};
  \draw (4.504,0) -- (4.504,-0.08) node[anchor=north] {$300$};
  \draw (6.005,0) -- (6.005,-0.08) node[anchor=north] {$400$};
\end{scope}
\begin{scope}[shift={(0.2,-1.55)},font=\scriptsize]
  \fill[seriesB!85] (0,-0.09) rectangle (0.5,0.09);        \node[anchor=west] at (0.6,0) {identified set.};
  \fill[seriesC!32] (3.4,-0.09) rectangle (3.9,0.09);      \node[anchor=west] at (4.0,0) { Imbens--Manski interval.};
  \draw[seriesA!75!black,thick,dashed] (8.0,0) -- (8.5,0); \node[anchor=west] at (8.6,0) {Horowitz--Manski-type band.};
  \draw[black,thick,dashed] (13.0,0) -- (13.5,0);          \node[anchor=west] at (13.6,0) {ADE};
  \draw[seriesC,very thick] (0,-0.45) -- (0.5,-0.45);      \node[anchor=west] at (0.6,-0.45) {IM, bootstrap s.e.};
  \draw[seriesA!75!black,very thick,dashed] (5.3,-0.45) -- (5.8,-0.45); \node[anchor=west] at (5.9,-0.45) {HM, bootstrap s.e.};
  \draw[inkmuted,thick,dotted] (10.3,-0.45) -- (10.8,-0.45); \node[anchor=west] at (10.9,-0.45) {IM, conventional s.e.};
\end{scope}
\end{tikzpicture}}
\caption{Design~B with assignment dependence varied from its maximum to zero}\label{fig:sweep-B}
\begin{flushleft}
\scriptsize
\emph{Note:} Panel~(a): the population identified set $\Phi_{\mathrm{ITR}}(P)\pm\bar L_{\mathrm{eff}}\bar\Gamma$, with $\Phi_{\mathrm{ITR}}(P)$ the Monte Carlo mean of $\widehat\Phi$, and the across-replication averages of the endpoints of the Imbens--Manski interval and of the Horowitz--Manski-type band; the dashed line is the ADE target. Panel~(b): the frequency with which the Imbens--Manski interval contains the target, at the bootstrap and the conventional standard error, and with which the Horowitz--Manski-type band contains the identified set.
\end{flushleft}
\end{figure}

\paragraph{Empirical Illustration.} We compute the bounds and the reportable objects of the sensitivity analysis, and plot the breakdown frontier using the actual data from the study of Example~\ref{eg::lalonde}. We first replicate its published estimates (details are in Section~\ref{sens:supp-empirics} of the Supplementary Material; a complementary illustration in a DiD setting is in Section~\ref{app::application}). The sample is choice-based, with controls sampled at rates 50 to 263 times lower than TWA workers. Adapting the authors' blocking algorithm, we stratify on the probit propensity score of their baseline specification, forming strata within each province. Within each province, we drop control units whose scores lie outside the range of treated scores and then split each score block at the median treated score whenever treated and control scores differ significantly. This yields 15 strata per region. 

We collect the results in Table~\ref{tab:twa-layers}. The estimand of \citet{ichino2008from} is the average effect on the treated (ATT), identified under unconfoundedness. In our framework, this is the ADTT, with $\lambda=1$, so that the reference law is $\pi_1$, the law of others' assignments faced by treated units, $G=F_{W\mid D=1}$, and $\kappa=1$. 

Thus, for the \textit{total-variation versions}, $\bar\Lambda_{\mathrm{eff}}=\bar\Lambda_0$, and interference on untreated potential outcomes contributes to the bias.  The identified set is therefore $\widehat\Phi\pm\bar\delta\bar\Lambda_0$. Tuscany's estimated ATT is $0.180$, with a conventional standard error of $0.034$ and a ZIP-clustered bootstrap standard error of $0.037$.\footnote{The authors use analytical standard errors under cross-unit independence \citep[p.~314]{ichino2008from}. We instead use a ZIP-clustered pairs bootstrap with $2{,}000$ resamples to allow for within-ZIP dependence induced by interference. It resamples the $103$ ZIP codes of the trimmed Tuscan sample and the $175$ of the Sicilian sample.}  Row two of  Table~\ref{tab:twa-layers} shows that if we allow for interference as large as one-half of an outcome standard deviation, about $20$ percentage points, the ADTT remains positive for dependence up to $\delta^\ast=0.88$, or $0.58$ after accounting for sampling uncertainty. 
For the \textit{share version} of Theorem~\ref{sensB:prop-refine}(iii)--(iv), the breakdown point can be calibrated to the data. The total-variation breakdown point, $\delta^\ast=0.88$, has no observable counterpart in a single realization of the assignment vector and therefore cannot be directly interpreted as large or small. The share version instead measures dependence by the gap in the treated share of ZIP mates, for which the data provide the benchmark $\widehat\gamma^{\rm ZIP}_{\rm obs}=0.0023$. We therefore apply the local-interference model, with $L^S_0$ interpreted as the slope per unit change in the ZIP-level treated share. At the observed gap, the identified set is $[0.180,0.180]$ to three decimals and the Imbens--Manski interval is $[0.107,0.253]$. Positivity of the ADTT breaks down only when dependence reaches $k^\ast=390$ times the observed level, or $258$ times after accounting for sampling uncertainty, indicating that the conclusion is far from the breakdown threshold at the dependence observed in the data. For the Sicily subsample, the data do not establish a positive effect even when CAI holds. The estimate is $0.063$ with a ZIP-clustered bootstrap standard error of $0.044$, so the one-sided $95\%$ lower confidence bound, $\widehat\Phi-z_{0.95}\widehat{\mathrm{se}}=-0.009$, is negative, and we cannot rule out a zero or negative effect. Consequently, $\delta^\ast_\alpha=0$ at both interference bounds in Table~\ref{tab:twa-layers}, and $k^\ast_\alpha=0$ on the share scale, i.e., no amount of assignment dependence can be tolerated, because the conclusion fails at zero dependence. At the point estimate, the sign survives dependence up to $\delta^\ast=0.36$ when $\bar\Lambda_0$ is half an outcome standard deviation and up to $0.06$ when $\bar\Lambda_0=1$.

\begin{table}[h]
\centering\footnotesize
\caption{The sensitivity analysis for the temporary-agency study}\label{tab:twa-layers}
\begin{tabular}{@{}lcc@{}}
\toprule
& Tuscany & Sicily \\
\midrule
Stratified ATT $\widehat\Phi$ (s.e.; ZIP boot.) & $0.180$ ($0.034$; $0.037$) & $0.063$ ($0.040$; $0.044$) \\
$\delta^\ast$, $\delta^\ast_\alpha$ at $\bar\Lambda_0=0.5$ SD ($20.4$ and $17.8$ points) & $0.88$, $0.58$ & $0.36$, $0$ \\
$\delta^\ast$, $\delta^\ast_\alpha$ at $\bar\Lambda_0=1$ (one outcome unit, $100$ points) & $0.18$, $0.12$ & $0.06$, $0$ \\
$\widehat\gamma^{\rm ZIP}_{\rm obs}$ (s.e.); $k^\ast$, $k^\ast_\alpha$ & $0.0023$ ($0.0008$); $390$, $258$ & $0.0011$ ($0.0002$); $339$, $0$ \\
\midrule
\multicolumn{3}{@{}l}{\emph{Total-variation version}} \\
Identified set, $\bar\delta=1$, $\bar\Lambda_0=0.5$ SD & $[-0.024,\,0.384]$ & $[-0.114,\,0.241]$ \\
 Imbens--Manski interval & $[-0.085,\,0.445]$ & $[-0.187,\,0.313]$ \\
\midrule
\multicolumn{3}{@{}l}{\emph{Share version, ZIP share}} \\
Identified set at $|\bar\gamma|=\widehat\gamma^{\rm ZIP}_{\rm obs}$ & $[0.180,\,0.180]$ & $[0.063,\,0.063]$ \\
 Imbens--Manski interval & $[0.107,\,0.253]$ & $[-0.023,\,0.149]$ \\
\bottomrule
\end{tabular}

\smallskip
\parbox{\linewidth}{\scriptsize \emph{Notes:} 15 strata per region; one-sided, $\alpha=0.05$; outcome SD  equals $0.408$ ( in Tuscany) and $0.355$ (in Sicily), computed in the trimmed samples; all bounds expressed in SD units, including $L^S_0$, use these values. With a binary outcome, the sharpness constructions of Theorem~\ref{sensB:thm-bounds}(iii) and Theorem~\ref{sensB:prop-refine} need not respect $\mu_d\in[0,1]$, so the identified sets are outer bounds; the breakdown points and multiples, which depend only on whether the set reaches zero, are then conservative.  Robustness checks of the share-version values are in Section~\ref{sens:supp-empirics} of the Supplementary Material.}
\end{table}

In Figure~\ref{fig:twa-frontier}, we plot the breakdown frontier. Panel~(a) is based on the total-variation version of Theorem~\ref{sensB:thm-bounds}, with $\bar\Lambda_{\mathrm{eff}}=\bar\Lambda_0$. Each curve gives the largest dependence at which the positivity of the ADTT holds, $\delta^\ast(\bar\Lambda_0)=\min\{1,\widehat\Phi/\bar\Lambda_0\}$, and the shaded region collects the pairs $(\bar\Lambda_0,\bar\delta)$ that leave the conclusion that Tuscany's ADTT is positive intact once we account for sampling error. Tuscany's frontier is flat at maximum assignment dependence, $\bar\delta=1$, until the measure of interference, $\bar\Lambda_0$, reaches $18$ percentage points, where the bound first equals the estimate. Beyond that point, $\bar\delta$ must fall to allow more interference through $\bar\Lambda_0$ for the conclusion to hold. Its inference-robust counterpart, which replaces $\widehat\Phi$ by $\widehat\Phi-z_{0.95}\widehat{\mathrm{se}}$, is flat only until about $12$ points. Sicily's frontier lies everywhere below Tuscany's, and its inference-robust counterpart is identically zero and is \textit{not} drawn. 
\begin{figure}[h]
\centering
\resizebox{\textwidth}{!}{\input{breakdown_frontier_fig}}
\caption{Breakdown frontiers for the temporary-agency study}\label{fig:twa-frontier}
\begin{flushleft}
\scriptsize
\emph{Note:} The target is the ADTT, so $\bar\Lambda_{\mathrm{eff}}=\bar\Lambda_0$ and the identified set is $\widehat\Phi\pm\bar\delta\bar\Lambda_0$ in (a).   The vertical line in Panel~(a) marks spillovers of half an outcome standard deviation ($20.4$ points), where the two Tuscan frontiers pass through $\delta^\ast$ and $\delta^\ast_\alpha$ of Table~\ref{tab:twa-layers}. Panel~(b) rescales the share version, $b=L^S_0|\bar\gamma|$ with the share taken among ZIP mates, in multiples of the observed ZIP-code gap $\widehat\gamma^{\rm ZIP}_{\rm obs}=0.0023$ (the value $1\times$); its vertical axis is logarithmic. Strata and estimates as in Table~\ref{tab:twa-layers}.
\end{flushleft}
\end{figure}

Panel~(b) uses the share version of Theorem~\ref{sensB:prop-refine}(iii)--(iv). The bound then simplifies to $L^S_0|\bar\gamma|$, with the share taken among ZIP mates. For each spillover slope $L^S_0$, per unit of the ZIP share, it plots the breakdown share gap $\delta^\ast_S=\min\{1,\widehat\Phi/L^S_0\}$ as a multiple of $\widehat\gamma^{\rm ZIP}_{\rm obs}$ on a logarithmic axis; this multiple is $k^\ast$, equal to $390$ at half an outcome standard deviation. For spillover slopes below $\widehat\Phi$ the share frontier is flat at its maximum, so no admissible assignment dependence among ZIP mates overturns the positivity of Tuscany's ADTT, and above them it breaks down only at hundreds of times the observed gap.

\section{Conclusion}\label{sec: conclusion}
Under interference, a standard identification formula under the no-interference assumption recovers an ADE-type causal estimand when, conditional on the design's conditioning variables or unconditionally, a unit's own assignment is uninformative about the assignments of others. Proposition~\ref{unif:thm-master} provides a general decomposition and shows that under a conditional assignment independence condition (generalized CAI), the standard identification formulas derived under no-interference identify an ADE-type estimand without knowledge of networks, groups, or exposure mappings. Two practical lessons follow. First, when interference is plausible but information on links is unavailable, standard estimates under no-interference need not be discarded. If domain knowledge makes the relevant assignment-independence condition credible, these estimates retain an average-direct-effect interpretation. Second, in such settings, quotas, capacity constraints, common shocks, and peer effects in take-up should be treated as threats to a design's causal interpretation, just as violations of unconfoundedness, exclusion, continuity, or parallel trends are.

Because the condition is generally untestable, in Section~\ref{robustness}, we develop sensitivity analysis procedures that measure how much assignment dependence a conclusion tolerates at a given size of interference. A natural next step is to develop sensitivity analysis that jointly relaxes assignment independence and the other untestable identifying conditions required by the designs, such as unconfoundedness, exclusion, continuity, or parallel trends,  in the presence of interference. We leave such joint sensitivity analysis for future research.

\section*{Acknowledgements}
Owusu gratefully acknowledges the generous support of the Social Sciences and Humanities Research Council of Canada (SSHRC), Grant No.~430-2025-01562. We thank participants at the 2026 Canadian Economics Association Annual Conference for helpful comments and feedback. The authors acknowledge the use of Anthropic Claude Opus 5 and Grammarly for language editing and  Anthropic Claude Opus 5 for R-coding assistance. The authors reviewed and verified all AI-assisted text and code. The authors confirm the manuscript's originality and accuracy. All errors are our own.

\section*{Data Availability Statement}
Both empirical illustrations use publicly available data. The temporary-agency data of \citet{ichino2008from} are the authors' replication package in the \emph{Journal of Applied Econometrics} Data Archive, hosted by the ZBW Journal Data Archive, \href{https://doi.org/10.15456/jae.2022319.0719821815}{doi:10.15456/jae.2022319.0719821815}; the data were collected by those authors through telephone interviews for an evaluation project financed by the Italian Ministry of Welfare and the Tuscany Region. The minimum-wage data of \citet{card1994minimum} are the authors' public data file, distributed with its codebook at \url{https://davidcard.berkeley.edu/data_sets.html} (file \texttt{njmin.zip}). The Monte Carlo study uses no data beyond the temporary-agency data to which it is calibrated. The code that reproduces every table and figure, in \texttt{R}, and the companion \texttt{R} package \texttt{caisensitivity} are available from the authors upon request.

\label{enddiscussion}
\onehalfspacing
\setlength{\bibsep}{0pt plus 0.3ex}
\bibliographystyle{chicago} 
{\small
\bibliography{references}  
}
\normalsize
\onehalfspacing
\label{endmain}
\pagebreak
\setcounter{page}{1}
\begin{center}
\Large{Supplementary Material to ``Causal Identification under Interference: The Role of Treatment Assignment Independence}."\label{supp mat}\\
\vspace{0.5cm}
\large{
Julius Owusu  \& Monika Avila M\'arquez}
\end{center}
\addtocontents{toc}{\protect\setcounter{tocdepth}{2}}
{\singlespacing\small\tableofcontents}
\clearpage
\StartSupplementCitations

\setcounter{section}{0}
\setcounter{equation}{0}
\setcounter{figure}{0}
\setcounter{table}{0}
\setcounter{remark}{0}
\setcounter{theorem}{0}
\setcounter{proposition}{0}
\setcounter{cor}{0}
\setcounter{assumption}{0}
\makeatletter
\renewcommand{\thesection}{S\arabic{section}}
\renewcommand{\theequation}{S\arabic{equation}}
\renewcommand{\thefigure}{S\arabic{figure}}
\renewcommand{\thetable}{S\arabic{table}}
\renewcommand{\thelemma}{S\arabic{lemma}}
\renewcommand{\thetheorem}{S\arabic{theorem}}
\renewcommand{\theremark}{S\arabic{remark}}
\renewcommand{\theassumption}{S\arabic{assumption}}
\renewcommand{\theproposition}{S\arabic{proposition}}
\renewcommand{\thecor}{S\arabic{cor}}
\renewcommand{\bibnumfmt}[1]{[S#1]}
\renewcommand{\citenumfont}[1]{S#1}
\renewcommand{\thesubsection}{\thesection.\arabic{subsection}}

This supplement collects the material that extends or supports the main paper. Section~\ref{app:proofs} proves the results of the main text: Proposition~\ref{unif:thm-master}, Theorems~\ref{sensB:thm-bounds} and~\ref{sensB:prop-refine}, and the lemmas they use. Section~\ref{supp:designs} complements the general theory of Section~\ref{sec:mapping}. It adds two complements to the general theory (Section~\ref{supp:two-remarks}), states, design by design, the interference-robust identifying assumptions summarized in Table~\ref{unif:tab-mapping}, the design-specific theorems, and their discussion, for selection on observables, instrumental variables, regression discontinuity, and static and dynamic difference-in-differences, the last permitting anticipation (Sections~\ref{sec:identification}--\ref{sec:did_dynamic}), then shows how these designs arise as special cases of the general decomposition (Section~\ref{supp:corollary}); Section~\ref{supp:relwork} relates these results to the literature. Sections~\ref{supp:primitive} and~\ref{sensB:impl} support the sensitivity analysis of Section~\ref{robustness}: the first gives primitive conditions for monotone dependence, benchmark values of the dependence measures, and confidence intervals, and the second is an implementation guide. Section~\ref{sec:montecarlo} describes interference in the temporary-agency setting of Section~\ref{sec:example} and reports in full the Monte Carlo study summarized in Section~\ref{sen appl}. Sections~\ref{sens:supp-empirics} and~\ref{app::application} contain the two empirical illustrations: the data, replication, strata, and full results behind the temporary-agency illustration of Section~\ref{sen appl}, and an application to the minimum-wage study of \citet{card1994minimum}, a difference-in-differences design. Section~\ref{table of papers} lists empirical studies in which interference is plausible and the interaction structure is unrecorded, and Section~\ref{proofs} proves Corollary~\ref{unif:cor-designs} and the design-specific theorems.

\paragraph{Notation.} Sections numbered with Arabic numerals refer to the main text, and sections numbered $S1,S2,\dots$ to this Supplementary Material. Let $i\in[N]:=\{1,\dots,N\}$ index units, $D_i\in\{0,1\}$ denote unit $i$'s treatment, $\mathbf D$ the assignment vector, and $\mathbf D_{-i}\in\{0,1\}^{N-1}$ the assignments of all other units, with realizations $\mathbf d$ and $\mathbf d_{-i}$; $\mathbf 0$ and $\mathbf 0_{-i}$ denote the corresponding zero vectors. Potential outcomes are $Y_i(\mathbf d)=Y_i(d,\mathbf d_{-i})$; $Y_i$ denotes the observed outcome, $W_i\in\mathcal W$ the covariates, $Z_i$ the instrument, with vectors $\mathbf Z$ and $\mathbf Z_{-i}$, and $P$ the law of the observed data for the uniformly drawn index $I$ of Section~\ref{sec:mapping}. Sums $\sum_{\mathbf d_{-i}}$ run over $\mathbf d_{-i}\in\{0,1\}^{N-1}$, $\indep$ denotes independence, and $\mathbbm 1(\cdot)$ the indicator function. The objects $\pi_d$, $\pi^\ast$, $\mu_d$, $\theta$, and $\kappa$ are defined in Section~\ref{sec:mapping}, and the laws $\pi_d$ and $\pi$ of others' treatments at the start of Section~\ref{sec:mapping}. As in the main text, the arguments $(\mathbf d_{-i};w)$ and $w$ of these objects and of functions of $w$ such as $\lambda$, $\delta$ and $L$ are suppressed whenever no confusion arises. Acronyms: ITR, individualistic treatment response; SOO, selection on observables; ADE, average direct effect; LADE, local average direct effect (for compliers); ADTT, average direct effect on the treated; CAI, conditional assignment independence (generalized CAI is defined in \eqref{unif:eq-gcai}); IV, instrumental variables; RDD, regression discontinuity design; DiD, difference-in-differences.

\section{Proofs of Main Results}\label{app:proofs}
This section proves the results stated in the main text: Proposition~\ref{unif:thm-master} and, for Section~\ref{robustness}, Theorems~\ref{sensB:thm-bounds} and~\ref{sensB:prop-refine} with the lemmas they use. Proofs of the design-specific results are in Section~\ref{proofs} below. Throughout, $\pi_d(\mathbf d_{-i};w):=\Pr(\mathbf D_{-i}=\mathbf d_{-i}\mid D_i=d,W_i=w)$ and $\pi(\mathbf d_{-i};w):=\Pr(\mathbf D_{-i}=\mathbf d_{-i}\mid W_i=w)$, as in Section~\ref{sec:mapping}.

\paragraph{Proof of Proposition~\ref{unif:thm-master}.}\label{prf:unif:thm-master}
\begin{proof}
The proof expands each cell expectation over the slices $\{\mathbf D_{-i}=\mathbf d_{-i}\}$ and adds and subtracts the average of the slice contrasts under the reference law; it uses only the linearity of cell expectations, so it applies verbatim with each design's own variables (Table~\ref{unif:tab-mapping}). Fix $w$ in the set of $G$-probability one on which~\eqref{unif:eq-sep} holds. Let $\mathbbm{E}_{d,w}[\cdot]$ denote the expectation in cell $(d,w)$, that is, $\mathbbm{E}[\cdot\mid D_i=d, W_i=w]$ with each design's counterparts, and the one-sided limit at $w_0$ in the sharp RDD. Because $\mathbf D_{-i}$ takes finitely many values, $\sum_{\mathbf d_{-i}}\mathbbm{1}(\mathbf D_{-i}=\mathbf d_{-i})=1$, and cell expectations are linear (for one-sided limits, the limit of a finite sum of convergent terms is the sum of the limits). Hence, for $d\in\{0,1\}$,
\begin{equation}
\begin{split}
\mathbbm{E}_{d,w}[Y_i]&=\sum_{\mathbf d_{-i}}\mathbbm{E}_{d,w}\big[Y_i\,\mathbbm{1}(\mathbf D_{-i}=\mathbf d_{-i})\big]=\sum_{\mathbf d_{-i}}\mu_d(\mathbf d_{-i};w)\,\pi_d(\mathbf d_{-i};w),\\
\sum_{\mathbf d_{-i}}\pi_d(\mathbf d_{-i};w)&=\mathbbm{E}_{d,w}[1]=1 .
\end{split}
\label{unif:eq-proof-cell}
\end{equation}
Adding and subtracting $\kappa^{-1}\sum_{\mathbf d_{-i}}\{\mu_1(\mathbf d_{-i};w)-\mu_0(\mathbf d_{-i};w)\}\pi^\ast(\mathbf d_{-i};w)$ gives
\begin{align*}
\frac{\mathbbm{E}_{1,w}[Y_i]-\mathbbm{E}_{0,w}[Y_i]}{\kappa}
&=\sum_{\mathbf d_{-i}}\theta(\mathbf d_{-i};w)\pi^\ast(\mathbf d_{-i};w)\\
&\quad+\frac1\kappa\sum_{\mathbf d_{-i}}\Big\{\mu_1\big[\pi_1-\pi^\ast\big]-\mu_0\big[\pi_0-\pi^\ast\big]\Big\}(\mathbf d_{-i};w).
\end{align*}
Integrating with respect to $G$ yields the decomposition in part~(i), with $\mathcal B(\pi^\ast)$ given by \eqref{unif:eq-bias}. The pointwise identity
\[
\mu_1[\pi_1-\pi^\ast]-\mu_0[\pi_0-\pi^\ast]=\mu_0[\pi_1-\pi_0]+(\mu_1-\mu_0)[\pi_1-\pi^\ast],
\]
with all functions evaluated at $(\mathbf d_{-i};w)$, gives \eqref{unif:eq-channels}. For the recentering statement, $\pi_1(\cdot\,;w)$, $\pi_0(\cdot\,;w)$, and $\pi^\ast(\cdot\,;w)$ are probability distributions by \eqref{unif:eq-proof-cell}, so $\sum_{\mathbf d_{-i}}[\pi_d-\pi^\ast](\mathbf d_{-i};w)=0$ and $\sum_{\mathbf d_{-i}}[\pi_1-\pi_0](\mathbf d_{-i};w)=0$. Therefore, for any function $f(\cdot\,;w)$ and any reference point $\mathbf d^\circ_{-i}(w)$, $\sum_{\mathbf d_{-i}}f(\mathbf d^\circ_{-i}(w);w)[\pi_d-\pi^\ast](\mathbf d_{-i};w)=0$, and similarly with $\pi_1-\pi_0$. Subtracting these zero terms replaces $f(\mathbf d_{-i};w)$ by $f(\mathbf d_{-i};w)-f(\mathbf d^\circ_{-i}(w);w)$, with $f\in\{\mu_1,\mu_0\}$ in \eqref{unif:eq-bias} and $f\in\{\mu_0,\theta\}$ in \eqref{unif:eq-channels}.

For part~(ii), under \eqref{unif:eq-gcai}, $\pi_1-\pi_0=0$ for $G$-almost all $w$. Since $\pi$ is a mixture of $\pi_1$ and $\pi_0$, also $\pi=\pi_1$; in the sharp RDD, $\pi_+=\pi_-=\pi_{w_0}$ directly. The baseline component in \eqref{unif:eq-channels} is therefore zero, and $\mathcal B(\pi^\ast)=\int\!\sum\theta\,[\pi_1-\pi^\ast]=\int\!\sum\theta\,[\pi-\pi^\ast]$ for every reference law $\pi^\ast$. Part~(i) then gives $\Phi_{\mathrm{ITR}}(P)=\int\!\sum\theta\,\pi^\ast+\int\!\sum\theta\,[\pi-\pi^\ast]=\int\!\sum\theta\,\pi$, and $\mathcal B(\pi^\ast)=0$ when $\pi^\ast=\pi$. For part~(iii), if $\mu_0(\mathbf d_{-i};w)=\bar\mu_0(w)$ and $\theta(\mathbf d_{-i};w)=\bar\theta(w)$ for all $\mathbf d_{-i}$, both channels in \eqref{unif:eq-channels} are of the form $\bar f(w)\sum_{\mathbf d_{-i}}[\pi_1-\nu](\mathbf d_{-i};w)$ with $\nu\in\{\pi_0,\pi^\ast\}$, which is zero; part~(i) then gives $\Phi_{\mathrm{ITR}}(P)=\int_{\mathcal W}\bar\theta(w)\sum_{\mathbf d_{-i}}\pi^\ast(\mathbf d_{-i};w)\,G(dw)=\int_{\mathcal W}\bar\theta(w)\,G(dw)$.
\end{proof}

This subsection proves Theorems~\ref{sensB:thm-bounds} and~\ref{sensB:prop-refine} and Proposition~\ref{sensB:prop-ci}, together with the lemmas and remarks that Section~\ref{robustness} cites.

Throughout, $w\in\mathcal W$ is fixed unless integrals over $G$ are taken, and all statements about $w$ hold for $G$-almost all $w$. For probability laws $\varpi,\varpi'$ on a finite set $\mathcal X$ (the values of $\mathbf D_{-i}$, or of the design's own others' assignment variable), define $d_{\mathrm{TV}}(\varpi,\varpi'):=\frac12\sum_{\mathbf x\in\mathcal X}|\varpi(\mathbf x)-\varpi'(\mathbf x)|$. For $f:\mathcal X\to\mathbbm{R}$ and $\mathcal X'\subseteq\mathcal X$, define $R_{\mathcal X'}(f):=\max_{\mathcal X'}f-\min_{\mathcal X'}f$. The range is subadditive and absolutely homogeneous: $R_{\mathcal X'}(f+g)\le R_{\mathcal X'}(f)+R_{\mathcal X'}(g)$ and $R_{\mathcal X'}(cf)=|c|R_{\mathcal X'}(f)$.

\paragraph{Auxiliary lemmas.}

\begin{lemma}[Total variation and range]\label{sensB:lem-tvosc}
Let $\varpi,\varpi'$ be probability laws on a finite set $\mathcal X$, let $\mathcal X'\supseteq\operatorname{supp}\varpi\cup\operatorname{supp}\varpi'$, and let $E^\ast:=\{\mathbf x:\varpi(\mathbf x)>\varpi'(\mathbf x)\}$. Then $d_{\mathrm{TV}}(\varpi,\varpi')=\varpi(E^\ast)-\varpi'(E^\ast)=\max_{E\subseteq\mathcal X}|\varpi(E)-\varpi'(E)|$, and for every $f:\mathcal X\to\mathbbm{R}$,
\[
\Big|\sum_{\mathbf x\in\mathcal X}f(\mathbf x)\{\varpi(\mathbf x)-\varpi'(\mathbf x)\}\Big|\le R_{\mathcal X'}(f)\cdot d_{\mathrm{TV}}(\varpi,\varpi'),
\]
with equality for $f=\mathbbm 1_{E^\ast}$ and $\mathcal X'=\operatorname{supp}\varpi\cup\operatorname{supp}\varpi'$.
\end{lemma}
\begin{proof}
Since $\sum_{\mathbf x\in\mathcal X}(\varpi-\varpi')=0$, $\sum_{\mathbf x\in E^\ast}(\varpi-\varpi')=\sum_{\mathbf x\in(E^\ast)^c}(\varpi'-\varpi)$, so both equal $\frac12\sum_{\mathbf x\in\mathcal X}|\varpi-\varpi'|$. For any $E\subseteq\mathcal X$, $\varpi(E)-\varpi'(E)\le\sum_{\mathbf x\in E\cap E^\ast}(\varpi-\varpi')\le \varpi(E^\ast)-\varpi'(E^\ast)$ and $\varpi'(E)-\varpi(E)\le\sum_{\mathbf x\in(E^\ast)^c}(\varpi'-\varpi)=\varpi(E^\ast)-\varpi'(E^\ast)$; hence the max characterization. For the inequality, let $M$ and $m$ be the maximum and minimum of $f$ on $\mathcal X'$ and $c:=(M+m)/2$. Because $\varpi-\varpi'$ vanishes off $\mathcal X'$ and sums to zero, $\sum_{\mathbf x\in\mathcal X}f(\varpi-\varpi')=\sum_{\mathbf x\in\mathcal X'}(f-c)(\varpi-\varpi')$,\footnote{Two facts give this identity. First, $\varpi(\mathbf x)=\varpi'(\mathbf x)=0$ for $\mathbf x\notin\mathcal X'$, because $\mathcal X'\supseteq\operatorname{supp}\varpi\cup\operatorname{supp}\varpi'$, so the sum over $\mathcal X$ equals the sum over $\mathcal X'$. Second, $\sum_{\mathbf x\in\mathcal X'}(\varpi-\varpi')=\sum_{\mathbf x\in\mathcal X}(\varpi-\varpi')=1-1=0$, so subtracting a constant leaves the sum unchanged. The identity therefore holds for every $c$; the midrange $c=(M+m)/2$ is chosen because it minimizes $\max_{\mathcal X'}|f-c|$, whose value is then $R_{\mathcal X'}(f)/2$. Centering is what makes the bound depend on $f$ only through its spread: adding a constant to $f$ cannot change the left-hand side, whereas the uncentered bound $2\max_{\mathcal X'}|f|\,d_{\mathrm{TV}}(\varpi,\varpi')$ would grow.} and $|f-c|\le(M-m)/2$ on $\mathcal X'$,\footnote{On $\mathcal X'$, $m\le f\le M$ by definition of $M$ and $m$. Subtracting $c$ gives $m-c\le f-c\le M-c$, and $M-c=(M-m)/2$ while $m-c=-(M-m)/2$, so $f-c$ lies in an interval of length $R_{\mathcal X'}(f)$ centered at zero. The bound is attained at a maximizer and at a minimizer of $f$, and no other constant improves it: $\min_c\max_{\mathcal X'}|f-c|=(M-m)/2$, with the midrange the unique minimizer.} so $|\sum_{\mathbf x\in\mathcal X}f(\varpi-\varpi')|\le\frac{M-m}2\sum_{\mathbf x\in\mathcal X}|\varpi-\varpi'|= R_{\mathcal X'}(f)\cdot d_{\mathrm{TV}}(\varpi,\varpi')$. For $f=\mathbbm 1_{E^\ast}$, $\sum_{\mathbf x\in\mathcal X}f(\varpi-\varpi')=d_{\mathrm{TV}}(\varpi,\varpi')$. If $d_{\mathrm{TV}}(\varpi,\varpi')>0$, then $E^\ast\cap\mathcal X'\neq\emptyset$ and $(E^\ast)^c\cap\mathcal X'\ne\emptyset$ (the latter because $\varpi(E^\ast)>\varpi'(E^\ast)$ forces $\varpi'((E^\ast)^c)>\varpi((E^\ast)^c)\ge0$), so $R_{\mathcal X'}(f)=1$; if $d_{\mathrm{TV}}=0$ both sides are zero.
\end{proof}

\begin{lemma}[Mixture reference laws]\label{sensB:lem-mix}
Let $\lambda:\mathcal W\to[0,1]$ be measurable and $\pi_\lambda:=\lambda\pi_1+(1-\lambda)\pi_0$. Then \eqref{sensB:eq-mix} holds.
\end{lemma}
\begin{proof}
By the definition of $\pi_\lambda$, $\pi_1-\pi_\lambda=(1-\lambda)(\pi_1-\pi_0)$ and $\pi_0-\pi_\lambda=-\lambda(\pi_1-\pi_0)$. Substituting into \eqref{unif:eq-bias} gives $\mu_1(\pi_1-\pi_\lambda)-\mu_0(\pi_0-\pi_\lambda)=\{(1-\lambda)\mu_1+\lambda\mu_0\}(\pi_1-\pi_0)$ pointwise. Finally, $(1-\lambda)\mu_1+\lambda\mu_0=\mu_0+(1-\lambda)(\mu_1-\mu_0)=\mu_0+(1-\lambda)\kappa\theta$ by \eqref{unif:eq-theta}.
\end{proof}

\begin{lemma}[Propensity representation of $\delta$]\label{sensB:lem-tv}
Suppose the cell expectations are conditional expectations and $p(w)=\Pr(D_i=1\mid W_i=w)\in(0,1)$. Let $\pi(\mathbf d_{-i};w):=\Pr(\mathbf D_{-i}=\mathbf d_{-i}\mid W_i=w)$ and $r(\mathbf d_{-i};w):=\Pr(D_i=1\mid\mathbf D_{-i}=\mathbf d_{-i},W_i=w)$ on $\{\pi>0\}$. Then, on $\{\pi>0\}$, $\pi_1-\pi_0=\pi\,(r-p)/\{p(1-p)\}$, both sides vanishing elsewhere, and
\begin{equation}
\delta(w)=\mathbbm{E}\big[|\Pr(D_i=1\mid\mathbf D_{-i},W_i=w)-p(w)|\,\big|\,W_i=w\big]\big/\big[2p(w)\{1-p(w)\}\big] .
\label{sensB:eq-delta-prop}
\end{equation}
Moreover, $\frac12\sum_{d,\mathbf d_{-i}}|\Pr(D_i=d,\mathbf D_{-i}=\mathbf d_{-i}\mid W_i=w)-\Pr(D_i=d\mid W_i=w)\pi(\mathbf d_{-i};w)|=2p(1-p)\delta(w)$, and if $\sup_{\mathbf d_{-i}}|r(\mathbf d_{-i};w)-p(w)|\le c$ then $\delta(w)\le\min\{1,c/[2p(w)\{1-p(w)\}]\}$.
\end{lemma}
\begin{proof}
By Bayes' rule, $\pi_1=r\pi/p$ and $\pi_0=(1-r)\pi/(1-p)$, so $\pi_1-\pi_0=\pi\{r(1-p)-(1-r)p\}/\{p(1-p)\}=\pi(r-p)/\{p(1-p)\}$; summing absolute values gives \eqref{sensB:eq-delta-prop}. For the second claim, $\Pr(D_i=1,\mathbf d_{-i})-p\pi=p(\pi_1-\pi)=p(1-p)(\pi_1-\pi_0)$ and $\Pr(D_i=0,\mathbf d_{-i})-(1-p)\pi=(1-p)(\pi_0-\pi)=-p(1-p)(\pi_1-\pi_0)$, using $\pi=p\pi_1+(1-p)\pi_0$; summing absolute values over $d$ and $\mathbf d_{-i}$ and halving gives $p(1-p)\cdot2\delta$. The last claim follows from \eqref{sensB:eq-delta-prop} and $\delta\le1$.
\end{proof}

\begin{lemma}[Couplings and Lipschitz functions]\label{sensB:lem-coupling}
Let $\varpi,\varpi'$ be laws on a finite set $\mathcal X$ and $\rho$ a nonnegative symmetric function on $\mathcal X\times\mathcal X$. If $|f(\mathbf x)-f(\mathbf x')|\le L\rho(\mathbf x,\mathbf x')$ for all $\mathbf x,\mathbf x'$, then for every pair of random vectors $(\mathbf A,\mathbf A')$ with $\mathbf A\sim \varpi$ and $\mathbf A'\sim \varpi'$, $$|\sum_{\mathbf x\in\mathcal X}f(\varpi-\varpi')|\le L\,\mathbbm{E}\rho(\mathbf A,\mathbf A');$$ hence $|\sum_{\mathbf x\in\mathcal X}f(\varpi-\varpi')|\le L\,\mathrm W_\rho(\varpi,\varpi')$, where $\mathrm W_\rho$ is the infimum of $\mathbbm{E}\rho(\mathbf A,\mathbf A')$ over such pairs. If $\rho$ is a metric, $\sup\{\sum_{\mathbf x\in\mathcal X}f(\varpi-\varpi'):f\text{ 1-Lipschitz}\}=\mathrm W_\rho(\varpi,\varpi')$ \citep{villani2003topics}, and the supremum is attained.
\end{lemma}

\begin{proof}
By the triangle inequality for expectations and the Lipschitz condition, $|\sum_{\mathbf x\in\mathcal X}f(\varpi-\varpi')|=|\mathbbm{E}\{f(\mathbf A)-f(\mathbf A')\}|\le\mathbbm{E}|f(\mathbf A)-f(\mathbf A')|\le L\,\mathbbm{E}\rho(\mathbf A,\mathbf A')$; take the infimum. The duality is the Kantorovich--Rubinstein theorem for the compact (finite) metric space $\mathcal X$. The supremum is attained because the objective is invariant to adding constants and continuous, and the set of 1-Lipschitz functions vanishing at a fixed point is compact in $\mathbbm{R}^{\mathcal X}$.
\end{proof}

\begin{lemma}[Exchangeable couplings]\label{sensB:lem-exch}
Let $\varpi,\varpi'$ be exchangeable laws on $\{0,1\}^J$, with count laws $\varpi_K,\varpi'_K$ of $K:=\sum_jx_j$ and cdfs $F_{\varpi},F_{\varpi'}$. Let $(K,K')$ be any pair with $K\sim \varpi_K$, $K'\sim \varpi'_K$, let $\varsigma$ be a uniformly random permutation of $[J]$ independent of $(K,K')$, and set $A_j:=\mathbbm 1\{\varsigma(j)\le K\}$ and $A'_j:=\mathbbm 1\{\varsigma(j)\le K'\}$. Then $\mathbf A\sim \varpi$, $\mathbf A'\sim \varpi'$, and $\Pr(A_j\ne A'_j)=\mathbbm{E}|K-K'|/J$ for every $j$. If $(K,K')=(F_{\varpi}^{-1}(U),F_{\varpi'}^{-1}(U))$ for $U\sim\mathrm{Uniform}(0,1)$, then $\mathbbm{E}|K-K'|/J=\int_0^1|F^S_{\varpi}(t)-F^S_{\varpi'}(t)|\,dt$, where $F^S$ denotes the cdf of $S=K/J$; if moreover $F_{\varpi}\le F_{\varpi'}$, then $\mathbf A\ge\mathbf A'$ coordinatewise.
\end{lemma}

\begin{proof}
An exchangeable law assigns equal probability to all vectors with the same count, so $\varpi(\mathbf x)=\varpi_K(k)/\binom Jk$ for $\sum_{j=1}^{J}x_j=k$. Given $(K,K')$, $\{j:A_j=1\}=\varsigma^{-1}(\{1,\dots,K\})$ is a uniformly random $K$-subset, so $\Pr(\mathbf A=\mathbf x\mid K,K')=\binom JK^{-1}\mathbbm 1\{\sum_{j=1}^{J}x_j=K\}$ and $\mathbf A\sim \varpi$; likewise $\mathbf A'\sim \varpi'$. Given $(K,K')$, $A_j\ne A'_j$ if and only if $\min(K,K')<\varsigma(j)\le\max(K,K')$, an event of probability $|K-K'|/J$ because $\varsigma(j)$ is uniform on $[J]$. Under the quantile coupling, $|K-K'|=\sum_{k=0}^{J-1}\mathbbm 1\{\min(K,K')\le k<\max(K,K')\}$ and, since $\{K\le k\}=\{U\le F_{\varpi}(k)\}$ and $\{K'\le k\}=\{U\le F_{\varpi'}(k)\}$, the $k$th indicator has expectation $|F_{\varpi}(k)-F_{\varpi'}(k)|$. Hence $\mathbbm{E}|K-K'|/J=J^{-1}\sum_{k=0}^{J-1}|F_{\varpi}(k)-F_{\varpi'}(k)|=\int_0^1|F^S_{\varpi}-F^S_{\varpi'}|\,dt$, because $F^S_{\varpi}(t)=F_{\varpi}(k)$ for $t\in[k/J,(k+1)/J)$. If $F_{\varpi}\le F_{\varpi'}$, then $F_{\varpi}^{-1}(U)\ge F_{\varpi'}^{-1}(U)$, so $K\ge K'$ and $\mathbf A\ge\mathbf A'$.
\end{proof}

\begin{lemma}[Lipschitz functions of a share]\label{sensB:lem-share}
Let $S^1,S^0$ take values in $\{0,1/J,\dots,1\}$ with cdfs $F_1,F_0$, and let $g$ satisfy $|g(s)-g(s')|\le L^S|s-s'|$ on this grid. Then $|\mathbbm{E}g(S^1)-\mathbbm{E}g(S^0)|\le L^S\int_0^1|F_1-F_0|\,dt$, with equality for some such $g$. If $F_1\le F_0$ (or $F_0\le F_1$), then $\int_0^1|F_1-F_0|\,dt=|\mathbbm{E}S^1-\mathbbm{E}S^0|$.
\end{lemma}

\begin{proof}
Extend $g$ to $[0,1]$ by linear interpolation; the extension is Lipschitz with constant $L^S$, hence absolutely continuous with $|g'|\le L^S$ almost everywhere. For $s\in[0,1]$, $g(s)=g(0)+\int_0^1g'(t)\mathbbm 1\{t<s\}\,dt$, so by Fubini $\mathbbm{E}g(S^d)=g(0)+\int_0^1g'(t)\{1-F_d(t)\}\,dt$ and $\mathbbm{E}g(S^1)-\mathbbm{E}g(S^0)=\int_0^1g'(t)\{F_0(t)-F_1(t)\}\,dt$, which gives the inequality. Equality holds for $g(s):=L^S\int_0^s\operatorname{sgn}\{F_0(t)-F_1(t)\}\,dt$, which is $L^S$-Lipschitz. With $g(s)=s$, $\mathbbm{E}S^1-\mathbbm{E}S^0=\int(F_0-F_1)\,dt$, which equals $\pm\int|F_1-F_0|\,dt$ when the cdfs are ordered.
\end{proof}
\paragraph{The bounds and the refinement.}
\begin{proof}[Proof of Theorem~\ref{sensB:thm-bounds}]\label{prf:sensB:thm-bounds}
\emph{(i)} By \eqref{unif:eq-bias}, $\kappa\mathcal B(\pi^\ast)=\int\{\sum\mu_1(\pi_1-\pi^\ast)-\sum\mu_0(\pi_0-\pi^\ast)\}\,dG$. Apply Lemma~\ref{sensB:lem-tvosc} to each sum, with $\mathcal X'$ the union of the supports of $\pi_1,\pi_0,\pi^\ast$ at $w$, and use the triangle inequality and $\kappa>0$.

\emph{(ii)} By Lemma~\ref{sensB:lem-mix} and Lemma~\ref{sensB:lem-tvosc} with $\varpi=\pi_1(\cdot\,;w)$, $\varpi'=\pi_0(\cdot\,;w)$, and $\mathcal X'=\mathcal D_w$, $|\sum_{\mathbf d_{-i}}m_\lambda(\pi_1-\pi_0)|\le\delta(w)R_w(m_\lambda)$; integrate. By subadditivity and homogeneity of the range, $R_w(m_\lambda)\le(1-\lambda)\Lambda_1+\lambda\Lambda_0$, and, because $m_\lambda=\mu_0+(1-\lambda)\kappa\theta$, $R_w(m_\lambda)\le\Lambda_0+(1-\lambda)\kappa\Lambda_\theta$. Integrating each bound gives $b^{(1)}_\lambda$ and $b^{(2)}_\lambda$.

\emph{(iii)} Fix $s\in[-1,1]$. For each unit $k$ and every $w$, let the cell laws of $\mathbf D_{-k}$ be $\pi_{1,k}:=\bar\delta\,\mathbbm 1\{\cdot=\mathbf 1\}+(1-\bar\delta)\nu_k$ and $\pi_{0,k}:=\bar\delta\,\mathbbm 1\{\cdot=\mathbf 0\}+(1-\bar\delta)\nu_k$ on $\{0,1\}^{N-1}$, with $\nu_k$ any law, for instance the product Bernoulli law with a prescribed propensity $p(w)\in(0,1)$, and let composition be equal, $q_1(k;w)=q_0(k;w)=\Pr(I=k\mid W_I=w)$. Since $N\ge2$, $\mathbf 1\ne\mathbf 0$, so the two cell laws of each unit differ by $\bar\delta\{\mathbbm 1\{\cdot=\mathbf 1\}-\mathbbm 1\{\cdot=\mathbf 0\}\}$ and every unit-level total-variation distance, and hence the pooled $\delta$, equals $\bar\delta$. Let $E^\ast:=\{\pi_1>\pi_0\}$ and $\mathbbm 1^\ast:=\mathbbm 1_{E^\ast}$, so that $\sum_{\mathbf d_{-i}}\mathbbm 1^\ast(\pi_1-\pi_0)=\pi_1(E^\ast)-\pi_0(E^\ast)=\bar\delta$ by Lemma~\ref{sensB:lem-tvosc} and, when $\bar\delta>0$, $E^\ast$ and its complement both meet $\mathcal D_w$, so $R_w(\mathbbm 1^\ast)=1$. \emph{Version 1.} Let $\mu_d(\mathbf d_{-i};w):=k_d(w)+s\bar\Lambda_d\,\mathbbm 1^\ast(\mathbf d_{-i})$ with $k_d:=m_d-s\bar\Lambda_d\,\pi_d(E^\ast)$. Then $\sum\mu_d\pi_d=m_d$, $\Lambda_d=|s|\bar\Lambda_d\le\bar\Lambda_d$, and $m_\lambda=\text{const}+s\{(1-\lambda)\bar\Lambda_1+\lambda\bar\Lambda_0\}\mathbbm 1^\ast$. By \eqref{sensB:eq-mix}, $\mathcal B(\pi_\lambda)=s\,\kappa^{-1}\int\bar\delta\{(1-\lambda)\bar\Lambda_1+\lambda\bar\Lambda_0\}\,dG=s\,\bar b^{(1)}_\lambda$. \emph{Version 2.} Let $\mu_0(\mathbf d_{-i};w):=k_0(w)+s\bar\Lambda_0\,\mathbbm 1^\ast(\mathbf d_{-i})$ with $k_0:=m_0-s\bar\Lambda_0\,\pi_0(E^\ast)$, let $\theta(\mathbf d_{-i};w):=k_\theta(w)+s\bar\Lambda_\theta\,\mathbbm 1^\ast(\mathbf d_{-i})$ with $k_\theta:=\kappa^{-1}\{m_1-m_0-s\bar\Lambda_0\bar\delta\}-s\bar\Lambda_\theta\,\pi_1(E^\ast)$, and let $\mu_1:=\mu_0+\kappa\theta$. Then $\sum\mu_0\pi_0=m_0$, $\sum\mu_1\pi_1=k_0+s\bar\Lambda_0\pi_1(E^\ast)+\kappa k_\theta+s\kappa\bar\Lambda_\theta\pi_1(E^\ast)=m_0+s\bar\Lambda_0\bar\delta+\kappa k_\theta+s\kappa\bar\Lambda_\theta\pi_1(E^\ast)=m_1$, $\Lambda_0\le\bar\Lambda_0$, $\Lambda_\theta\le\bar\Lambda_\theta$, and $m_\lambda=\mu_0+(1-\lambda)\kappa\theta=\text{const}+s\{\bar\Lambda_0+(1-\lambda)\kappa\bar\Lambda_\theta\}\mathbbm 1^\ast$, so $\mathcal B(\pi_\lambda)=s\,\bar b^{(2)}_\lambda$. In both versions all functions are measurable in $w$, the cell means are those prescribed, and, as $s$ ranges over $[-1,1]$, $\mathcal B(\pi_\lambda)$ ranges over $[-\bar b^{(k)}_\lambda,\bar b^{(k)}_\lambda]$. When every covariate value is shared by at least two units, cell laws with the same total-variation distance $\bar\delta$ are generated by the joint assignment law that, within each covariate stratum with common propensity $p(w)$, treats all units with probability $\bar\delta\,p(w)$, none with probability $\bar\delta\{1-p(w)\}$, and otherwise independently as Bernoulli$(p(w))$, independently across strata: each unit's two cell laws then differ by $\bar\delta\{\mathbbm 1\{\cdot=\mathbf 1\}-\mathbbm 1\{\cdot=\mathbf 0\}\}$ on its stratum coordinates times a common law on the rest, and the construction above applies to them verbatim; the outcome model of Remark~\ref{sensB:rem-realize} then satisfies Assumption~\ref{unif:ass-sep} and the unconfoundedness condition of the selection-on-observables design. Otherwise, the construction is a statement about the pooled primitives (Remark~\ref{sensB:rem-realize}).

 Every configuration in version $k$ satisfies $\delta\le\bar\delta$ and the corresponding range bounds, so by (ii), $|\mathcal B(\pi_\lambda)|\le\bar b^{(k)}_\lambda$ and $T(\pi_\lambda)=\Phi_{\mathrm{ITR}}(P)-\mathcal B(\pi_\lambda)\in[\Phi_{\mathrm{ITR}}(P)\mp\bar b^{(k)}_\lambda]$. By \eqref{unif:eq-template} and \eqref{unif:eq-proof-cell}, $\Phi_{\mathrm{ITR}}(P)=\kappa^{-1}\int(m_1-m_0)\,dG$ depends on the configuration only through the cell means. By (iii), every point of the interval is the value of $T(\pi_\lambda)$ for some configuration in the model with the same cell means.
\end{proof}

\begin{remark}[The scope of sharpness]\label{sensB:rem-realize}
Part~(iii) of Theorem~\ref{sensB:thm-bounds} is a statement about the cell laws and mean functions of the pooled law $P$. No point of the interval can be excluded on the basis of the cell means and the bounds alone. When every covariate value is shared by at least two units, the configurations it constructs are generated by a joint assignment law of the $N$ units, exchangeable within each covariate value with any prescribed propensity $p(w)\in(0,1)$ and independent across values, together with the outcome model $Y_i(d,\mathbf d_{-i}):=\mu_d(\mathbf d_{-i};w)+\epsilon_i$, where $\epsilon_i$ has mean zero and is independent of $(W_i,D_i,\mathbf D_{-i})$. Then Assumption~\ref{ass: uncon} holds, $\mathbbm{E}[Y_i\mid D_i=d,W_i=w]=m_d(w)$, and $\theta(\mathbf d_{-i};w)=\mathbbm{E}[\Delta_i(\mathbf d_{-i})\mid W_i=w]$. When a covariate value identifies a single unit, as with continuous covariates, the attainable values of $\delta$ depend on the propensity profile: for two units with propensities $p_1\notin\{p_2,1-p_2\}$, no joint law makes both cell laws of each unit disjoint, so $\delta=1$ is not attainable. The statement also does not claim that the mean functions respect restrictions on the outcome support, nor that the assignment law respects a design's structure, such as a fixed number treated. Such a priori information can only shrink the identified set, and the breakdown points and multiples, which ask only whether the set reaches $t_0$, are then conservative. \qedsymbol
\end{remark}

\begin{proof}[Proof of Theorem~\ref{sensB:prop-refine}]\label{prf:sensB:prop-refine}
Under equal composition, $\pi_d(\mathbf d_{-k};w)=q(k;w)\pi_{d,k}(\mathbf d_{-k};w)$ and $m_\lambda$ at a configuration of unit $k$ is $m_{\lambda,k}$, so by \eqref{sensB:eq-mix} $\mathcal B(\pi_\lambda)=\kappa^{-1}\int\sum_kq(k;w)\sum_{\mathbf d_{-k}}m_{\lambda,k}(\pi_{1,k}-\pi_{0,k})\,dG$. It therefore suffices to bound $|\sum_{\mathbf d_{-k}}m_{\lambda,k}(\pi_{1,k}-\pi_{0,k})|$ for each unit pointwise in $w$, average with the weights $q(k;w)$ and integrate. Fix $w$ and a unit $k$, suppress both, and set $m:=m_{\lambda,k}$, $\pi_d:=\pi_{d,k}$, $\omega_j:=\omega_{kj}$, $\Gamma_j:=\Gamma_{kj}$, $\bar\Gamma:=\bar\Gamma_k$, $\mathrm W_S:=\mathrm W_{S,k}$ and $F^S_d:=F^S_{d,k}$, and let $\mathbf d_{-i}$ denote a configuration $\mathbf d_{-k}$ of unit $k$'s others. Define $\rho_\omega(\mathbf d_{-i},\mathbf d'_{-i}):=\sum_{j=1}^{J}\omega_j|d_j-d'_j|$.

\emph{(BI) implies Lipschitz continuity.} If $\mathbf d_{-i},\mathbf d'_{-i}$ differ in the coordinates $j_1,\dots,j_r$, switching these coordinates one at a time yields a path $\mathbf d_{-i}=\mathbf b^0,\mathbf b^1,\dots,\mathbf b^r=\mathbf d'_{-i}$ in $\{0,1\}^J$ with $|m(\mathbf b^{k-1})-m(\mathbf b^k)|\le L\omega_{j_k}$. Hence $|m(\mathbf d_{-i})-m(\mathbf d'_{-i})|\le L\rho_\omega(\mathbf d_{-i},\mathbf d'_{-i})$.

\emph{(i)} Let $(\mathbf D^1,\mathbf D^0)$ be as in (MD), say $\mathbf D^1\ge\mathbf D^0$. By Lemma~\ref{sensB:lem-coupling}, $|\sum_{\mathbf d_{-i}}m(\pi_1-\pi_0)|\le L\,\mathbbm{E}\rho_\omega(\mathbf D^1,\mathbf D^0)=L\sum_{j=1}^{J}\omega_j\mathbbm{E}(D^1_j-D^0_j)=L\sum_{j=1}^{J}\omega_j\Gamma_j$, and each $\Gamma_j=\mathbbm{E}D^1_j-\mathbbm{E}D^0_j\ge0$. If $\mathbf D^1\le\mathbf D^0$, the same argument gives $L\sum_j\omega_j(-\Gamma_j)$ with $\Gamma_j\le0$. In both cases the bound is $L|\sum_j\omega_j\Gamma_j|=L\sum_j\omega_j|\Gamma_j|\le L\max_j|\Gamma_j|$.\\ \emph{Sharpness.} $m^\ast(\mathbf d_{-i}):=\pm L\sum_{j=1}^{J}\omega_jd_j$ satisfies (BI) with equality and $\sum_{\mathbf d_{-i}}m^\ast(\pi_1-\pi_0)=\pm L\sum_{j=1}^{J}\omega_j\Gamma_j$. Choosing the sign by the sign of $\sum_j\omega_j\Gamma_j(w)$, or $\omega(w)$ as the unit vector at a (measurably selected) maximizer of $|\Gamma_j(w)|$, attains the first or the second bound after integration. Slice means with $m_\lambda=m^\ast$ exist; for example, $\mu_1=\mu_0=m^\ast$ up to additive constants chosen to match any prescribed cell means.

\emph{(ii)} Under (EX), Lemma~\ref{sensB:lem-exch} with $\varpi=\pi_1$, $\varpi'=\pi_0$ and the quantile coupling gives $(\mathbf X,\mathbf Y)$ with $\mathbbm{E}\rho_\omega(\mathbf X,\mathbf Y)=\sum_{j=1}^{J}\omega_j\Pr(X_j\ne Y_j)=\sum_{j=1}^{J}\omega_j\int_0^1|F^S_1-F^S_0|\,dt=\mathrm W_S$, since $\sum_{j=1}^{J}\omega_j=1$. Lemma~\ref{sensB:lem-coupling} gives $|\sum_{\mathbf d_{-i}}m(\pi_1-\pi_0)|\le L\,\mathrm W_S$.\\ \emph{Sharpness.} With equal weights, take $m=g\circ S$ for the $g$ of Lemma~\ref{sensB:lem-share} with $L^S=L$. It satisfies (BI), because $|S(\mathbf d_{-i})-S(\mathbf d^{(j)}_{-i})|=1/J=\omega_j$, and attains $L\,\mathrm W_S$. If the share laws are ordered, $m^\ast=\pm L\sum_j\omega_jd_j$ gives $\pm L\sum_j\omega_j\Gamma_j=\pm L\bar\Gamma=L\,\mathrm W_S$ for any weights, because exchangeability makes all marginals equal, so $\Gamma_j=\bar\Gamma$, and by (iv).

\emph{(iii)} Under (AN), $\sum_{\mathbf d_{-i}}m(\pi_1-\pi_0)=\mathbbm{E}g(S(\mathbf D^1))-\mathbbm{E}g(S(\mathbf D^0))$ for $\mathbf D^d\sim\pi_d$. Lemma~\ref{sensB:lem-share} gives the bound. If $g(s)=\alpha+\beta s$, then $\sum_{\mathbf d_{-i}}m(\pi_1-\pi_0)=\beta\sum_{\mathbf d_{-i}}S(\mathbf d_{-i})(\pi_1-\pi_0)=\beta J^{-1}\sum_{j=1}^{J}\Gamma_j=\beta\bar\Gamma$, because $\sum_{\mathbf d_{-i}}(\pi_1-\pi_0)=0$; average over $k$ and integrate. If $g_k=g$ for all $k$, then $\sum_kq(k;w)\{\mathbbm{E}g(S(\mathbf D^1_k))-\mathbbm{E}g(S(\mathbf D^0_k))\}=\mathbbm{E}_{\pi_1}g(S)-\mathbbm{E}_{\pi_0}g(S)$, because under equal composition the law of $S$ under the pooled $\pi_d$ is the $q$-mixture of its laws under the $\pi_{d,k}$; Lemma~\ref{sensB:lem-share} applied to the pooled share laws gives the bound with the pooled $\mathrm W_S$, and $\mathrm W_S\le\sum_kq(k;w)\mathrm W_{S,k}$ by the triangle inequality applied to the mixture of cdfs. For affine $g$ the same mixture identity gives $\beta\bar\Gamma$.\\ \emph{Sharpness.} The function $g^\ast(s):=L^S\int_0^s\operatorname{sgn}\{F^S_0(t)-F^S_1(t)\}\,dt$ of Lemma~\ref{sensB:lem-share} is $L^S$-Lipschitz and gives $\mathbbm{E}g^\ast(S(\mathbf D^1))-\mathbbm{E}g^\ast(S(\mathbf D^0))=L^S\,\mathrm W_S$; since (AN) allows a unit-specific $g_k$, taking $g_k=g^\ast$ for every unit attains the averaged bound, and slice means with $m_{\lambda,k}=g_k\circ S$ exist as in (i). For a common $g$, the same function built from the pooled share cdfs attains $L^S$ times the pooled $\mathrm W_S$.

\emph{(iv)} The equality is the last statement of Lemma~\ref{sensB:lem-share}, since $\mathbbm{E}S(\mathbf D^1)-\mathbbm{E}S(\mathbf D^0)=\bar\Gamma$. Under (MD) with $\mathbf D^1\ge\mathbf D^0$, $S(\mathbf D^1)\ge S(\mathbf D^0)$ surely, so $F^S_1\le F^S_0$; the reverse case is symmetric. If the orderings have the same direction for every $k$, the $\bar\Gamma_k$ share a sign, so $\sum_kq(k;w)|\bar\Gamma_k|=|\sum_kq(k;w)\bar\Gamma_k|=|\bar\Gamma|$ by linearity of $S$, and the pooled share cdfs, being $q$-mixtures of cdfs ordered in the same direction, are ordered, so $\mathrm W_S=|\bar\Gamma|$ by Lemma~\ref{sensB:lem-share}.
\end{proof}

\begin{remark}[Block exchangeability]\label{sensB:rem-block}
Theorem~\ref{sensB:prop-refine}, Part (ii) extends to assignment that is exchangeable within blocks, as in stratified take-up. For each unit $k$, partition its others $[N]\setminus\{k\}$ into blocks and suppose that, under each $\pi_{d,k}(\cdot\,;w)$, the block subvectors are independent, the laws of all blocks other than a block $b^\ast_k$ (for example, unit $k$'s own stratum) coincide across $d$, and the law of block $b^\ast_k$ is exchangeable. Couple the other blocks identically and block $b^\ast_k$ as in Lemma~\ref{sensB:lem-exch}. Under (BI), this gives $|\mathcal B(\pi_\lambda)|\le\kappa^{-1}\int L\sum_kq(k;w)\,\Omega_{b^\ast_k}\mathrm W_{S_{b^\ast},k}\,dG\le\kappa^{-1}\int L\sum_kq(k;w)\mathrm W_{S_{b^\ast},k}\,dG$, where $\Omega_{b^\ast_k}:=\sum_{j\in b^\ast_k}\omega_{kj}$ and $\mathrm W_{S_{b^\ast},k}$ is computed from the share of assigned units in block $b^\ast_k$ under $\pi_{1,k}$ and $\pi_{0,k}$. \qedsymbol
\end{remark}

\section{Designs under Interference: Assumptions, Results, and Discussion}\label{supp:designs}
This section states in full the interference-robust identifying assumptions of the four designs summarized in Table~\ref{unif:tab-mapping} of the main text, gives the design-specific decomposition theorems, and relates each to the literature. Every result here is an instance of Proposition~\ref{unif:thm-master} through Corollary~\ref{unif:cor-designs}. Section~\ref{supp:two-remarks} collects complements to the general theory; Sections~\ref{sec:identification}--\ref{sec:did_dynamic} treat the designs in parallel (the ITR benchmark, the assumptions under interference, the result, its discussion, and its relation to the literature); Section~\ref{supp:corollary} then collects the designs as special cases of Proposition~\ref{unif:thm-master}; and Section~\ref{supp:relwork} collects the extended discussion of related work.

\subsection{Complements to the General Theory}\label{supp:two-remarks}
The following remarks clarify the scope, extensions, and plausibility of the general framework in Section~\ref{unif:sec-setup} of the main text.

\begin{remark}[Probability model and details of the general representation]\label{supp:rem-design}
The population of Section~\ref{unif:sec-setup} of the main text is finite and observed once, and its covariates, assignments, and latent outcome components $\epsilon_i$ are jointly random. Potential outcomes $Y_i(\mathbf d)=f(\mathbf d,\epsilon_i)$ are therefore stochastic, with unit~$i$'s own coordinate of $\mathbf d$ distinguished from those of the other units. This is the formulation of the structural models for causally connected units of \citet{vanderlaan2014causal} and \citet{ogburn2024causal}, in which exogenous errors, covariates, treatments and outcomes are generated jointly and the data are a single draw from the law of the whole population, and of the stochastic potential outcomes of \citet{leung2020treatment}, \citet{viviano2019policy} and \citet{xu2023difference}, in which outcomes depend on the assignment vector and on unit-level unobservables. Under no interference, \citet[p.~58]{angrist2009mostly} present a comparable latent structure. The formulation differs from the design-based framework of \citet{hudgens2008toward} and \citet{savje2021average}, in which potential outcomes are fixed and assignment is the only source of randomness, and from a superpopulation framework, in which units are independent draws from an infinite population; \citet{abadie2020sampling} and \citet{ding2017bridging} discuss how the two sources of uncertainty relate.

Identification statements are for the law $P$ of $(Y_I,W_I,D_I)$, where $I$ is drawn uniformly from $[N]$, independently of $(\boldsymbol\epsilon,\mathbf W,\mathbf D)$, as defined in Section~\ref{unif:sec-setup} of the main text; formulas that pool units estimate functionals of this law. Because Proposition~\ref{unif:thm-master} is an identity in conditional expectations given $(D_I,W_I)$, it applies to $P$ whenever Assumption~\ref{unif:ass-sep} holds for each unit, with the cell laws of $P$ representing the composition of each comparison cell; the assignment-independence conditions must in addition hold under $P$ itself, which requires equal composition, condition~\eqref{unif:eq-ec} of the main text. The ADE-type estimands then average unit-level estimands; for example, $\mathbbm E[\Delta_I(\mathbf D_{-I})]=N^{-1}\sum_{i=1}^N\mathbbm E[\Delta_i(\mathbf D_{-i})]$, the counterpart of the expected average treatment effect of \citet{savje2021average}; conditioned on $D_I=1$, it is the counterpart of the expected average treatment effect on the treated of \citet{xu2023difference}, which averages the units' direct effects on the treated with equal weights, as the pooled ADTT does under equal composition~\eqref{unif:eq-ec}. The same statements hold unit by unit, with $P_i$ in place of $P$; in that reading $\pi_d$, $\Gamma_j$ and $\bar\Gamma$ depend on $i$, whereas under $P$ they are population quantities. The pooled and unit-level readings coincide when the joint laws of $(\epsilon_i, W_i, D_i,\mathbf D_{-i})$ coincide across units after relabelling the other units, except that the pooled $\Gamma_j$, which averages over drawn units other than $j$, is then $(N-1)/N$ times the unit-level value. Treating covariates as random, as in \citet{vanderlaan2014causal} and \citet{ogburn2024causal}, gives meaning to the weighting laws $F_W$ and $F_{W\mid D=1}$ and to the one-sided limits in $w$ of the sharp RDD. Holding covariates fixed instead yields the conditional estimands of \citet{abadie2014inference}, \citet{jin2024tailored} and \citet{xu2023difference}, which average $\mathbbm E[\,\cdot\mid W_i]$ over the realized covariates. The Monte Carlo study of Section~\ref{sec:montecarlo} adopts this conditional version. Large-sample statements, such as those of Section~\ref{supp:primitive}, are understood along a sequence of such populations with $N\to\infty$ in which dependence across units is restricted, as in \citet{leung2020treatment} and \citet{ogburn2024causal}. Both discrete and continuous covariates are accommodated. Whereas the notation of the main text allows continuous covariates, with statements holding for $G$-almost all $w$, the selection-on-observables result of Section~\ref{sec:identification} (Theorem~\ref{bias-so}) is derived for discrete covariates; for continuous ones, its sums over $w$ are integrals with respect to the law of $W_i$.
Where an identifying formula does not condition on covariates, the conditional independence \eqref{eq:CAI} reduces to the stronger unconditional restriction $\mathbf D_{-i}\indep D_i$, and the generalized CAI condition \eqref{unif:eq-gcai} is the independence $\mathbf D_{-i}\indep D_i\mid W_i$, or its counterpart in the design's own variables, on the support of $G$ whenever the cell expectations are conditional expectations, and both cells have positive probability given $W_i$.

In the construction of Section~\ref{sec:mapping} of the main text, the measure $G$ on $\mathcal W$ is degenerate when the formula does not condition on covariates; the outcome entering the contrast is a known function of unit $i$'s observed outcome(s), the change $Y_{i1}-Y_{i0}$ under DiD, and the normalizing constant $\kappa>0$ does not involve outcomes; and the cell expectations $\mathbbm{E}_{d,w}[\cdot]$, $d\in\{0,1\}$, are assumed to exist for every random variable to which they are applied, which in the boundary designs is a restriction on the one-sided limits. Assumption~\ref{unif:ass-sep} is delivered by each design's identifying assumptions (Corollary~\ref{unif:cor-designs}); on the support of $\pi_d(\cdot\,;w)$, its $\mu_d(\mathbf d_{-i};w)$ then coincides with $\mathbbm{E}_{d,w}[Y_i\mid \mathbf D_{-i}=\mathbf d_{-i}]$, and off the support $\mu_d(\mathbf d_{-i};w)$ is unrestricted by the data and is pinned down by those assumptions. \qedsymbol
\end{remark}

\begin{remark}[More than two cells]\label{unif:rem-cells}
The binary contrast in \eqref{unif:eq-template} covers all designs in this paper; for difference-in-differences, first differences serve as the outcome. Under no interference, identification formulas that combine finitely many cells $c\in\mathcal C$ with weights $\lambda_c(w)$, such as multi-valued treatment contrasts or weighted cohort-by-time effects in multi-period DiD, are handled by the same argument. With $\theta(\mathbf d_{-i};w):=\kappa^{-1}\sum_{c}\lambda_c(w)\mu_c(\mathbf d_{-i};w)$, the bias becomes $\kappa^{-1}\int\sum_{c}\lambda_c(w)\sum_{\mathbf d_{-i}}\mu_c(\mathbf d_{-i};w)\{\pi_c(\mathbf d_{-i};w)-\pi^\ast(\mathbf d_{-i};w)\}G(dw)$, which vanishes when $\pi_c(\cdot\,;w)=\pi^\ast(\cdot\,;w)$ for every cell with $\lambda_c(w)\neq0$. Whether $\theta$ is then a direct effect must be established design by design. \qedsymbol
\end{remark}

\begin{remark}[Plausibility of generalized CAI under interference]\label{supp:rem-plausibility}
The generalized CAI condition is restrictive when the mechanisms that generate spillovers in outcome also drive take-up, as labor market competition, information diffusion through social networks, or geographic outcome spillovers often do. For example, in a job-training program operating in a local labor market, peer referrals and employer connections may simultaneously cause treated individuals to affect the employment prospects of untreated neighbors (outcome interference) and influence those same neighbors' program participation decisions (assignment dependence). In such settings, CAI is an idealization rather than a literal description of the data-generating process. 

There are nevertheless observational settings in which outcome interference is plausible, yet treatment assignment may be conditionally independent across units. A canonical example arises in programs where eligibility and participation are determined by individual-level administrative rules and idiosyncratic take-up decisions. Consider a job-training or income-support program in which eligibility depends on predetermined covariates such as lagged income, age, or household composition. Conditional on these covariates, participation decisions are influenced by idiosyncratic factors such as information, preferences, or application costs. When the program is not subject to binding local capacity constraints (or when any constraints operate at a sufficiently aggregate level relative to the unit of analysis), treatment assignments can be well approximated as independent across individuals conditional on their covariates. At the same time, individuals may interact in local labor markets or social networks, so participation by some individuals affects others' outcomes through job competition, referrals, or information diffusion, generating outcome interference. In this setting, CAI holds as a property of the assignment mechanism conditional on observed covariates, even though interference in outcomes may be substantial. \qedsymbol
\end{remark}

For the remainder of this section, we proceed design by design, stating the extensions of the standard ITR identifying assumptions under arbitrary interference, the corresponding identification results, their implications, and their relation to the existing literature. The separation-type assumptions below are stated for each unit $i$ and yield Assumption~\ref{unif:ass-sep} under the law $P$ of a randomly drawn unit, with the cell laws $\pi_d$ representing the composition of each comparison cell. The overlap conditions and the assignment-independence conditions must hold under $P$ itself: besides the unit-level condition, the probability of each comparison cell given $w$ must not vary across units, which is equal composition, condition~\eqref{unif:eq-ec} of the main text; under IV, a common instrument propensity also makes the Wald denominator equal to $\kappa$, and in the sharp RDD the composition of units, the relative densities of their running variables, must not jump at the cutoff, although the treatment probability does (Remark~\ref{supp:rem-design} and the discussion following Proposition~\ref{unif:thm-master} in the main text).

\subsection{Selection on Observables}\label{sec:identification}
Selection on observables, one of the most widely used designs under ITR, identifies causal effects by conditioning on observed covariates.

\paragraph{The ITR benchmark.} Under ITR and consistency, \citet{rosenbaum1983central} show that identification rests on the classical \emph{strong ignorability} assumption:
\begin{align}
&\{Y_i(1),Y_i(0)\} \indep D_i \mid W_i,\,\,\,\, \quad \quad \quad \quad \quad \quad \quad (\text{Unconfoundedness}) \label{eq:so}\\
& 0<\Pr(D_i=1\mid W_i)<1\,\, \text{a.s.}, \,\,\quad \quad \quad \quad \quad \quad \, (\text{Overlap}) \label{eq:overlap}
\end{align}
where \eqref{eq:so} states that, conditional on observed covariates, treatment assignment is as good as random, and the overlap condition \eqref{eq:overlap} ensures that each unit has a nonzero probability of receiving either treatment arm, conditional on the pretreatment variables. Under ITR, consistency, unconfoundedness, and overlap \citep[e.g., see][]{imbens2015causal},
\begin{align}
\Phi_{\mathrm{ITR}}(P)
:=&\mathbbm{E}\!\left[
\mathbbm{E}[Y_i \mid D_i=1,\, W_i]-\mathbbm{E}[Y_i \mid D_i=0,\, W_i]\right]
= \mathbbm{E}[Y_i(1) - Y_i(0)].\label{ro}
\end{align}
Several estimators developed under ITR, including outcome regression, inverse probability weighting (IPW), propensity-score-reduced, and augmented inverse probability weighting (AIPW) estimators, are designed to estimate the functional in \eqref{ro} consistently.

\paragraph{Assumptions under interference.} Under arbitrary interference, the analog of the unconfoundedness and overlap conditions \eqref{eq:so} and \eqref{eq:overlap} is the following assumption.
\begin{assumption} \label{ass: uncon}
For all $i\in[N]$ and $\mathbf{d}_{-i}\in \{0,1\}^{N-1}$,
\begin{align}
&\{Y_i(1, \mathbf{d}_{-i}),Y_i(0, \mathbf{d}_{-i})\} \indep (D_i,\mathbf{D}_{-i}) \mid W_i, \label{eq:uso}\\
&0<\Pr(D_i=1\mid W_i)<1\,\, \text{a.s.} \label{eq:overlap 2}
\end{align}
\end{assumption}

Condition \eqref{eq:uso} states that, for all $i$, all $\mathbf d\in\{0,1\}^N$ and each fixed $\mathbf d_{-i}\in\{0,1\}^{N-1}$,
\[
\Pr\big(\mathbf D=\mathbf d \,\big|\, W_i,Y_i(1,\mathbf d_{-i}),Y_i(0,\mathbf d_{-i})\big)
=
\Pr(\mathbf D=\mathbf d \mid W_i).
\]
Conditional on unit $i$'s covariates $W_i$, the distribution of the assignment vector is therefore invariant to unit $i$'s potential outcomes. Imposed for all $i$, condition \eqref{eq:uso} excludes assignment rules under which, after conditioning on $W_i$, the probability of some assignment vector $\mathbf d$ varies with unit $i$'s potential outcomes. In this sense, it extends unconfoundedness to settings with interference by requiring that, conditional on $W_i$, neither unit $i$'s own treatment nor the assignments of other units depend on unit $i$'s potential outcomes.

The difference from the no-interference condition \eqref{eq:so} lies in the assignment object on which independence is imposed. Under no interference, it suffices that own treatment $D_i$ be independent of unit $i$'s potential outcomes given $W_i$. Under interference, unit $i$'s realized outcome depends on $\mathbf D_{-i}$ as well as on $D_i$, so condition \eqref{eq:uso} requires independence between unit $i$'s potential outcomes and the entire assignment vector $(D_i,\mathbf D_{-i})$, conditional on $W_i$.

Note that in \eqref{eq:uso}, we condition on $W_i$, not on the full covariate matrix $\mathbf W$, so the restriction is local. It does not impose the restriction 
\begin{align}\label{popular unconf}
    \Pr(\mathbf D=\mathbf d\mid \mathbf W, \boldsymbol{\mathcal{Y}})
=
\Pr(\mathbf D=\mathbf d\mid \mathbf W),
\end{align}
where $\boldsymbol{\mathcal{Y}}:=\{Y_i(d_i,\mathbf d_{-i}) : i\in [N],\ d_i\in\{0,1\},\ \mathbf d_{-i}\in\{0,1\}^{N-1}\}$.  

Variants of \eqref{popular unconf} are imposed by \citet{tchetgen2012causal} and \citet{leung2022unconfoundedness} to identify causal estimands under different forms of interference.  To reinterpret the ITR identification formula, by contrast, the covariates in \eqref{eq:uso} need to include only the conditioning variables used under no interference, which typically exclude covariates of other units in the population \citep{forastiere2021identification}.

In sum, condition \eqref{eq:uso} excludes selection on potential outcomes at the unit level while remaining agnostic about the broader dependence structure of the assignment mechanism; it is an interference analog of unconfoundedness stated in terms of local conditioning on $W_i$. Condition \eqref{eq:overlap 2} restates the overlap condition.

\paragraph{Result.} Theorem~\ref{bias-so} gives the interpretation of the identification formula \eqref{ro} when outcomes may depend on the entire assignment vector.

\begin{theorem}\label{bias-so}
Suppose Assumptions~\ref{ass:arbitrary-interference} and~\ref{ass:consistency} hold.
\begin{itemize}
\item[(i)] If Assumption~\ref{ass: uncon} holds, and  whether or not condition~\eqref{eq:CAI} holds,
\begin{align}
\Phi_{\mathrm{ITR}}(P)
&=
\mathbbm{E}\Big[
\mathbbm{E}\!\left[
\Delta_i(\mathbf{D}_{-i}) \mid D_i=1, W_i
\right]
\Big]
+
B_{ADE}^1
\label{eq:decomp-adtc-x}
\\
&=
\mathbbm{E}\Big[
\mathbbm{E}\!\left[
\Delta_i(\mathbf{D}_{-i}) \mid D_i=0, W_i
\right]
\Big]
+
B_{ADE}^0
\label{eq:decomp-adtt-x}
\\
&=
\mathbbm{E}\Big[
\mathbbm{E}\!\left[
\Delta_i(\mathbf{D}_{-i}) \mid W_i
\right]
\Big]
+
B_{ADE},
\label{eq:decomp-ade-x}
\end{align}
where the bias terms are
\begin{align}
B_{ADE}^0
:=
\sum_{w\in\mathcal W}
\sum_{\mathbf d_{-i}}
\bar Y_i(1,\mathbf d_{-i};w)
\Big(
\pi_1(\mathbf d_{-i};w)-\pi_0(\mathbf d_{-i};w)
\Big)
\Pr(W_i=w),
\end{align}
\begin{align}
B_{ADE}^1
:=
\sum_{w\in\mathcal W}
\sum_{\mathbf d_{-i}}
\bar Y_i(0,\mathbf d_{-i};w)
\Big(
\pi_1(\mathbf d_{-i};w)-\pi_0(\mathbf d_{-i};w)
\Big)
\Pr(W_i=w),
\end{align}
and
\small
\begin{align}
B_{ADE}
&:=
\sum_{w\in\mathcal W}\Pr(W_i=w)
\sum_{\mathbf d_{-i}}
\Big(
\bar Y_i(1,\mathbf d_{-i};w)
-
\bar Y_i(1,\mathbf d'_{-i};w)
\Big)
\Big(
\pi_1(\mathbf d_{-i};w)
-
\pi(\mathbf d_{-i};w)
\Big)
\notag\\
&\
-
\sum_{w\in\mathcal W}\Pr(W_i=w)
\sum_{\mathbf d_{-i}}
\Big(
\bar Y_i(0,\mathbf d_{-i};w)
-
\bar Y_i(0,\mathbf d'_{-i};w)
\Big)
\Big(
\pi_0(\mathbf d_{-i};w)
-
\pi(\mathbf d_{-i};w)
\Big),
\end{align}
\normalsize
with
$\bar{Y}_i(d,\mathbf{d}_{-i};w):= \mathbbm{E}\big[ Y_i(d,\mathbf{d}_{-i}) \mid W_i=w \big]$,
$\pi_d(\mathbf d_{-i};w):=\Pr(\mathbf D_{-i}=\mathbf d_{-i}\mid D_i=d, W_i=w)$ and
$\pi(\mathbf d_{-i};w):=\Pr(\mathbf D_{-i}=\mathbf d_{-i}\mid W_i=w)$ as in Section~\ref{sec:mapping} of the main text,
and $\mathbf{d}'_{-i}\in\{0,1\}^{N-1}$ an arbitrary reference realization of $\mathbf{D}_{-i}$, held fixed in the sum over $\mathbf d_{-i}$ (it may depend on $w$).

\item[(ii)] If condition~\eqref{eq:CAI}, under $P$, where it includes equal composition~\eqref{unif:eq-ec}, and Assumption~\ref{ass: uncon} both hold, then
\begin{equation}
\Phi_{\mathrm{ITR}}(P)
=
\mathbbm{E}\Big[
\mathbbm{E}\!\left[
\Delta_i(\mathbf{D}_{-i}) \mid W_i
\right]
\Big].
\label{eq:ade-x}
\end{equation}
\end{itemize}
~\hyperref[prf:bias-so]{\textup{\footnotesize[Proof]}}\end{theorem}

\paragraph{Discussion.} Theorem~\ref{bias-so} distinguishes the cases in which CAI fails from those in which it holds. Part~(i) holds in both. When CAI fails, treated and untreated units with the same covariates face systematically different distributions of others' treatments, so that $\Phi_{\mathrm{ITR}}(P)$ no longer reduces to a well-defined estimand. It combines direct effects with the effects of simultaneously changing the assignment mechanism of other units, which generates bias.

Equations~\eqref{eq:decomp-adtc-x} and \eqref{eq:decomp-adtt-x} provide two complementary representations of $\Phi_{\mathrm{ITR}}(P)$ under interference. In \eqref{eq:decomp-adtc-x}, $\Phi_{\mathrm{ITR}}(P)$ equals a population average direct effect in which, within each covariate cell, others' treatments are distributed as they are for treated units, plus a bias term $B^1_{ADE}$ that weights untreated mean potential outcomes by the difference between the conditional distributions of others' treatments given $D_i=1$ and $D_i=0$. Analogously, \eqref{eq:decomp-adtt-x} expresses $\Phi_{\mathrm{ITR}}(P)$ as a population average direct effect in which others' treatments are distributed as they are for untreated units, plus a bias term $B^0_{ADE}$ that weights treated mean potential outcomes by the same distributional difference.

Equation~\eqref{eq:decomp-ade-x} decomposes $\Phi_{\mathrm{ITR}}(P)$ into the ADE plus an overall bias term. The magnitude of this bias depends on the extent to which CAI fails (how strongly own treatment predicts the treatments of others) and on the strength of interference, as reflected in the responsiveness of conditional mean outcomes to changes in others' treatments. When the distribution of $\mathbf{D}_{-i}$ differs only slightly across the values of $D_i$, or interference is weak, the bias is small; when both the deviation from CAI and interference are large, $\Phi_{\mathrm{ITR}}(P)$ can differ substantially from the ADE.

Part~(ii) covers the case in which these biases vanish. When CAI holds, the treatment of unit $i$ conveys no information about the treatments of other units, conditional on covariates, so treated and untreated units with the same covariates face the same distribution of others' assignments. The bias terms then disappear, $\Phi_{\mathrm{ITR}}(P)$ reduces to the ADE, and the selection-on-observables identification formula retains a causal interpretation under arbitrary interference. CAI is thus the condition that separates settings in which this formula recovers an ADE from those in which it does not (Corollary~\ref{unif:cor-designs}(a)).

\begin{remark}
In the absence of covariates, Assumption~\ref{ass: uncon} reduces to independence between the population treatment vector and all potential outcomes, $\{Y_i(1,\mathbf d_{-i}), Y_i(0,\mathbf d_{-i})\}\indep (D_i,\mathbf D_{-i})$, as in a randomized experiment. Theorem~\ref{bias-so} therefore applies directly to experimental settings: when interference is present, the assignment mechanism's dependence structure is central to interpreting ITR-based functionals, even in randomized experiments. Under a \emph{Bernoulli design}, in which treatments are assigned independently with a common assignment probability, condition~\eqref{unif:eq-ec}, Theorem~\ref{bias-so}(ii) implies that the difference-in-means functional equals the ADE, as in \citet[Sec.~3]{vanderweele2011effect}, \citet[Supplement~A, Lemma~A8]{savje2021average}, and \citet[Sec.~3]{hu2022average}, and consistent with the Bernoulli-design results of \citet{li2022random} and \citet{munro2025treatment}. Under \emph{complete randomization}, the number of treated units is fixed, CAI fails, and part~(i) gives the finite-population bias. \qedsymbol
\end{remark}

\paragraph{Relation to the literature.} Theorem~\ref{bias-so} restates and extends known results. In two-stage randomized experiments with partial interference, within-group means, averaged across groups, are unbiased for average potential outcomes over the law of others' treatments given own treatment, whose contrast is the ``direct effect'' of \citet{hudgens2008toward}. Without covariates, this object is $\Phi_{\mathrm{ITR}}(P)$ (up to sign), and part~(i) expresses it as the ADE plus $B_{ADE}$. \citet{vanderweele2011effect} observe that it is a direct effect averaged over the unconditional law of others' treatments when treatments are assigned independently within clusters, whereas, with a fixed number treated per cluster, it also captures the effect of one fewer other unit being treated; these are part~(ii) and its failure. The closest antecedent is \citet{forastiere2021identification}, for observational network data in which interference operates through a known exposure mapping $G_i$ of neighbors' treatments and potential outcomes are unconfounded for $(D_i, G_i)$ given covariates. Their Theorem~2.A expresses the covariate-adjusted contrast as an average of potential outcomes over the cell-specific laws of $G_i$; its first corollary shows that the contrast equals their overall main effect when $D_i\indep G_i$ given covariates, and its second gives a bias that matches the $B_{ADE}$ term by term, with $G_i$ in place of $\mathbf D_{-i}$ and a reference exposure level in place of $\mathbf d'_{-i}$. Their exposure mapping may be the identity on a known neighborhood, but their probabilities are averages across the units of a fixed population, and interference is confined to network neighbors, so their independence condition is a cross-sectional restriction rather than a restriction on the assignment law under unrestricted interference. Their target coincides with our ADE when covariate-specific effects aggregate consistently, for example when $G_i\indep W_i$, so the two targets are analogous rather than identical in general.\footnote{ When group membership or network links are fixed, or stochastic but independent of treatment, the independence conditions in \citet{hudgens2008toward}, \citet{vanderweele2011effect}, and \citet{forastiere2021identification} are implied by the corresponding version of \eqref{eq:CAI}; when they are stochastic and depend on own treatment, there is no general ordering between those conditions and \eqref{eq:CAI}.}  Closest in scope, \citet[Thm.~3]{leung2025identifying} decomposes the covariate-conditioned exposure contrast under unrestricted interference into a direct effect and a bias that vanishes when the treatments of a unit's $K$-neighborhood are independent of the rest given covariates; at $K=0$ the contrast is the own-treatment contrast and the condition is CAI. Theorem~\ref{bias-so} differs in stating unconfoundedness configuration by configuration rather than for the whole potential-outcome function, in its overlap and propensity form, and in the reference law of its ADE; where they overlap, the two decompositions agree. His estimand is the equally weighted average of unit-level contrasts, whereas $\Phi_{\mathrm{ITR}}(P)$ is the pooled contrast, which equals that average under equal composition~\eqref{unif:eq-ec} of the main text.

In cross-sectional linear-in-means models, \citet{berg2021spillover} derive the bias of the estimator that ignores spillovers when the exposure proxy is a group-average treatment intensity that includes the unit's own treatment, and \citet[Prop.~1]{gomez2022spillover} show that, with a leave-one-out proxy and independent assignment, this bias reduces to a term proportional to the difference between spillovers on treated and untreated units; in our terms, independence removes the baseline component, and the remainder is the direct-effect component evaluated at the zero-exposure reference law.

\subsection{Instrumental Variables}\label{sec:IV}

To highlight the conceptual issues, we focus on the simplest IV setting, a single binary instrument $Z_i$, as in \citet{angrist1996identification}.

\paragraph{The ITR benchmark.} Under ITR, the identification formula is the population analog of the Wald estimator \citep{wald1940fitting},
\begin{align}\label{iv wald}
   \Phi_{\mathrm{ITR}}(P)
:=
\frac{\mathbbm{E}[Y_i\mid Z_i=1]-\mathbbm{E}[Y_i\mid Z_i=0]}
{\mathbbm{E}[D_i\mid Z_i=1]-\mathbbm{E}[D_i\mid Z_i=0]}.
\end{align}
In the notation of \eqref{eq:sutva-functional}, \eqref{iv wald} corresponds to $W_i=Z_i$ and $\psi(Y_i,Z_i, D_i;P)=(\mathbbm{E}[Y_i\mid Z_i=1]-\mathbbm{E}[Y_i\mid Z_i=0])/
(\mathbbm{E}[D_i\mid Z_i=1]-\mathbbm{E}[D_i\mid Z_i=0])$.
\citet{angrist1996identification} show that \eqref{iv wald} equals the local average treatment effect (LATE) for compliers under the following conditions:
\begin{align}
&(Y_i(d),D_i(z)) \indep Z_i, \,\, \forall\,\, d,z \in \{0,1\},  \quad\text{(Independence)} \label{indep}\\
&D_i(1) \ge D_i(0)\,\, \text{almost surely}, \quad \,\,\,\quad\quad\text{(Monotonicity)}\label{mono}\\
&Y_i(d,z) = Y_i(d),\quad \quad \quad\quad \quad \quad \quad \quad \quad\,\,\,\, \text{(Exclusion)} \label{excl} \\
&D_i = D_i(Z_i),\quad \quad \quad \quad \quad \quad \quad \quad \quad \quad\,\,\,\,\,\text{(Treatment consistency)} \label{Tconsistency}
\end{align}
where $D_i(z)$ denotes the random counterfactual treatment of unit $i$ under instrument value $z$. Specifically, $D_i(z)=h(z, \nu_i)$, where $\nu_i$ is an independent and identically distributed (i.i.d.) latent random variable and $h(\cdot)$ is a known function (see, e.g., \citealt{heckman1978dummy} and \citealt[p.~157]{angrist2009mostly}, where $h(\cdot)$ is a single-index indicator function). The randomness of potential treatments is thus due to the latent variable $\nu_i$, just as the randomness of potential outcomes is due to $\epsilon_i$.

\paragraph{Assumptions under interference.} Under arbitrary interference, let $\mathbf Z=(Z_1,\ldots,Z_N)$ denote the vector of instruments and $\mathbf{Z}_{-i}$ its subvector excluding unit $i$. For any realization $\mathbf z_{-i}$, let $\mathbf D_{-i}(\mathbf z_{-i}):=(D_1(z_1), \dots, D_{i-1}(z_{i-1}), D_{i+1}(z_{i+1}),\dots, D_{N}(z_N))$ denote the vector of potential treatments of all units other than $i$ when $\mathbf{Z}_{-i}=\mathbf z_{-i}$. The following assumption extends the standard IV conditions to settings with arbitrary interference in outcomes.

\begin{assumption}[IV Conditions under Interference]\label{ass: IV conditions}
For all $i \in[ N]$,
\begin{align}
&D_i=D_i(Z_1, \dots, Z_N)=D_i(Z_i), \label{iv: no interference in treatment}\\
&(Y_i(d, \mathbf{d}_{-i}),D_i(z), \mathbf{D}_{-i}(\mathbf{z}_{-i})) \indep (Z_i,\mathbf{Z}_{-i}),\,\, \forall\,\, d,z \in \{0,1\}\,\, \text{and}\,\,\, \mathbf{z}_{-i},\mathbf{d}_{-i}\in\{0,1\}^{N-1}, \,\label{iv indep}\\
&D_i(1) \ge D_i(0)\,\, \text{almost surely},\label{iv: mono}\\
&Y_i(d, \mathbf{d}_{-i},\mathbf{z}) = Y_i(d,\mathbf{d}_{-i}),\ \forall\,\, d\in\{0,1\},\mathbf{z} \in \{0,1\}^N\,\,\text{and}\,\, \mathbf{d}_{-i}\in\{0,1\}^{N-1}, \label{iv exclusion} \\
&\mathbf{D}_{-i}=\mathbf{D}_{-i}(\mathbf z_{-i}) \,\,\, \text{if}\,\,\,\mathbf{Z}_{-i}=\mathbf{z}_{-i}.\label{iv:consistency}
\end{align}
\end{assumption}
Assumption~\ref{ass: IV conditions} generalizes the independence and exclusion restrictions of \citet{angrist1996identification} to settings with arbitrary interference. Condition \eqref{iv: no interference in treatment} imposes no interference in treatment: unit $i$'s treatment depends only on its own instrument, so that potential treatments are generated by the same latent model as under no interference. It retains part~(a) of the stable unit treatment value assumption (SUTVA) of \citet[Assumption~1]{angrist1996identification} and drops part~(b), which rules out interference in treatment take-up; it is the ``personalized encouragement'' condition of \citet{kang2016peer} and the ``individualized offer response'' condition of \citet{ditraglia2023identifying}. The unconditional independence condition \eqref{iv indep} requires the instrument to be as good as randomly assigned with respect to both potential outcomes and potential treatments. As in \citet{angrist1996identification}, it is understood to hold jointly for the collection of potential outcomes and potential treatments; the joint version is used when units are partitioned into principal strata. Because there is no interference in treatment, the monotonicity condition needs no modification; \eqref{iv: mono} restates it. The exclusion restriction \eqref{iv exclusion} rules out any direct effect of the instrument on outcomes beyond its effect through own treatment. Condition~\eqref{iv:consistency} imposes consistency for others' treatments, requiring the realized vector $\mathbf D_{-i}$ to equal the potential treatment vector $\mathbf D_{-i}(\mathbf z_{-i})$ whenever the instruments of the other units take the value $\mathbf z_{-i}$.

Together, conditions \eqref{iv: no interference in treatment} and \eqref{iv: mono} imply that the population can be partitioned into three principal strata: \emph{compliers}, for whom $D_i(1) > D_i(0)$; \emph{always-takers}, for whom $D_i(1)=D_i(0)=1$; and \emph{never-takers}, for whom $D_i(1)=D_i(0)=0$.

Because the treatment of unit $i$ is a function of her instrument and an i.i.d.\ latent variable, restricting the dependence of instruments across units yields a CAI-type restriction on treatment assignments.
\begin{assumption}[Independence of Instruments]\label{Instrument indepen}
    For all units $i \in[ N]$,  $\mathbf{Z}_{-i} \indep Z_i$.
\end{assumption}
\noindent This assumption rules out cross-unit dependence in instrument assignment, so each unit's instrument is assigned independently of the others'. It is an analog of the unconditional version of CAI. Under the law $P$ of a randomly drawn unit, it must be complemented by a common instrument propensity, $\Pr(Z_k=1)$ equal across units, the IV form of equal composition, condition~\eqref{unif:eq-ec} of the main text; otherwise the two instrument cells contain different mixes of units, and the Wald denominator need not equal $\kappa$.

\paragraph{Result.} Theorem~\ref{prop:iv} gives the interpretation of the Wald ratio \eqref{iv wald} when outcomes are subject to arbitrary interference.

\begin{theorem}
\label{prop:iv}
Suppose Assumptions~\ref{ass:arbitrary-interference} and~\ref{ass:consistency} hold, $0<\Pr(Z_i=1)<1$, the first stage $\kappa$ is positive and, under $P$, the instrument propensity $\Pr(Z_k=1)$ is common across units, the IV form of condition~\eqref{unif:eq-ec}.
\begin{itemize}
\item[(i)] If Assumption~\ref{ass: IV conditions} holds, then, whether or not Assumption~\ref{Instrument indepen} holds, the Wald identification functional admits the decomposition
\[
\frac{\mathbbm{E}[Y_i\mid Z_i=1]-\mathbbm{E}[Y_i\mid Z_i=0]}
{\mathbbm{E}[D_i\mid Z_i=1]-\mathbbm{E}[D_i\mid Z_i=0]}
=
\mathrm{LADE}
+
\mathrm{Bias}_{\mathrm{IV}},
\]
where
\[
\mathrm{LADE}
:=
\mathbbm{E}\!\left[
\mathbbm{E}\!\left[
Y_i(1,\mathbf D_{-i})-Y_i(0,\mathbf D_{-i})
\mid \mathbf D_{-i},\,D_i(1) > D_i(0)
\right]
\mid D_i(1) > D_i(0)
\right]
\]
and the bias term is
\begin{align*}
\mathrm{Bias}_{\mathrm{IV}}
:=
\frac{B_1-B_0}{\Pr(D_i(1) > D_i(0))},
\end{align*}
with
\begin{align*}
B_z
:=
\sum_{d\in\{0,1\}}
\sum_{\mathbf z_{-i}}
\sum_{\mathbf d_{-i}}
\bar{Y}\!\left(d,\mathbf d_{-i};\mathbf z_{-i}\right)\,
\Pr\!\big(\mathbf D_{-i}(\mathbf z_{-i})=\mathbf d_{-i}\mid D_i(z)=d\big)\,
\zeta_z(\mathbf z_{-i})\,
\Pr\!\big(D_i(z)=d\big),
\end{align*}
where
\[
\bar{Y}\!\left(d,\mathbf d_{-i};\mathbf z_{-i}\right)
:=
\mathbbm{E}\!\left[
Y_i(d,\mathbf d_{-i})
\mid D_i(z)=d,\ \mathbf D_{-i}(\mathbf z_{-i})=\mathbf d_{-i}
\right]
\]
and
\[
\zeta_z(\mathbf z_{-i})
:=
\Pr(\mathbf{Z}_{-i}=\mathbf z_{-i}\mid Z_i=z)
-
\Pr(\mathbf{Z}_{-i}=\mathbf z_{-i}).
\]

\item[(ii)]
If Assumptions~\ref{ass: IV conditions} and \ref{Instrument indepen} both hold, then
$\zeta_z(\mathbf z_{-i})=0$ for all $z$ and $\mathbf z_{-i}$, and hence $\mathrm{Bias}_{\mathrm{IV}}=0$. Consequently, the Wald ratio identifies the local average direct effect:
\[
\frac{\mathbbm{E}[Y_i\mid Z_i=1]-\mathbbm{E}[Y_i\mid Z_i=0]}
{\mathbbm{E}[D_i\mid Z_i=1]-\mathbbm{E}[D_i\mid Z_i=0]}
=
\mathrm{LADE}.
\]
\end{itemize}
~\hyperref[prf:prop:iv]{\textup{\footnotesize[Proof]}}\end{theorem}

\paragraph{Discussion.} Theorem~\ref{prop:iv} characterizes the canonical IV identification formula when outcomes are subject to arbitrary interference. Part~(i) shows that, even under instrument exogeneity and individual-level exclusion, the Wald ratio fails to identify a causal effect when instruments are dependent across units. The Wald ratio then decomposes into the local average direct effect (LADE, a term due to \citealp{hoshino2024causal}) plus a bias term that reflects systematic differences in the distribution of others' instruments across values of the own instrument. When a unit's instrument alters the distribution of instruments among other units, instrument-induced variation in own treatment is accompanied by changes in others' treatments, which confounds the effect and induces bias.

Part~(ii) establishes that the bias vanishes when instruments are independent across units. Under Assumption~\ref{Instrument indepen}, the distribution of others' instruments is invariant to one's own instrument, so the distribution of others' treatments faced by compliers does not differ across instrument values. The Wald ratio then recovers the LADE, the local average direct effect evaluated at the realized treatment assignment vector of compliers. Independence of instruments across units is therefore sufficient for the standard IV identification formula to retain a causal interpretation under interference. 

\paragraph{Relation to the literature.} \citet{sobel2006randomized} shows that, under interference, the Wald estimand need not equal the LATE of \citet{angrist1996identification}. For instruments that are dependent by design, the Wald ratio equals a complier direct effect minus a spillover effect weighted by the probability that the peer is assigned in two-person household experiments in which at most one member is assigned \citep{vazquez2023causal}, and the total effect among compliers plus spillover effects among always-takers and never-takers in cluster randomized trials with noncompliance \citep{kang2018spillover}. Under personalized encouragement, part~(i) covers both, with negative and perfect within-cluster dependence of instruments, respectively. Part~(ii) also has antecedents. Under partial interference and independent Bernoulli encouragement within blocks, the Wald ratio within an encouragement mechanism identifies a local direct effect for compliers \citep{kang2016peer}.In two-stage experiments, the within-saturation Wald ratio converges to the complier average direct effect as cluster size grows \citep[Thm.~2(1)]{imai2021causal}. At finite cluster sizes, however, complete randomization within clusters induces dependence between $Z_i$ and $\mathbf Z_{-i}$. Under personalized encouragement, Proposition~\ref{unif:thm-master} makes the resulting discrepancy explicit when saturation is the conditioning variable. Most closely, \citet[Thm.~3.1]{hoshino2024causal} show, in a design-based setting with interference of unknown form subject to a restricted-interference condition (their Assumption~3.4), that the Wald ratio equals the LADE when instruments are mutually independent. Theorem~\ref{prop:iv} requires only $\mathbf Z_{-i}\indep Z_i$ for each $i$ and, under $P$, a common instrument propensity, and part~(i) gives the exact bias under arbitrary dependence of instruments, which \citet[Sec.~3]{hoshino2024causal} leave for future research. In one respect, our assumptions are more restrictive. Condition~\eqref{iv: no interference in treatment} rules out effects of others' instruments on own take-up, whereas \citet{vazquez2023causal}, \citet{imai2021causal}, and \citet{hoshino2024causal} allow such effects under restricted-interference structures, so their complier estimands may include spillover components.

\subsection{Regression Discontinuity}\label{sec:rdd}

The RDD target parameter can be identified through two routes, the \emph{continuity-based} approach of \citet{hahn2001identification} and the \emph{local-randomization-based} approach of \citet{cattaneo2015randomization}. Under local randomization and arbitrary interference, \citet{aronow2017regression} show that the difference in means within a window around the cutoff is unbiased for the \citet{hudgens2008toward} average direct effect of units in the window. That estimand averages others' assignments conditionally on own assignment. They list analogous results for the continuity-based design as a direction for future work. We take up that direction and focus on the continuity-based route under a \emph{sharp} design.

\paragraph{The ITR benchmark.} Let $W_i$ denote a continuous running variable with cutoff $w_0$, so that $D_i=\mathbbm{1}(W_i\geq w_0)$. Under ITR and the condition
\begin{align}
&\lim_{w\downarrow w_0} \mathbbm{E}[Y_i(0)\mid W_i=w]=
\lim_{w\uparrow w_0}\mathbbm{E}[Y_i(0)\mid W_i=w], \quad \qquad  \qquad  \text{(Continuity)}\label{rdd: cont Y0}
\end{align}
the RDD identification formula
\[
\Phi_{\mathrm{ITR}}(P):={\lim_{w\downarrow w_0} \mathbbm{E}[Y_i\mid W_i=w] - \lim_{w\uparrow w_0}\mathbbm{E}[Y_i\mid W_i=w]}
\]
identifies the average treatment effect (ATE) at the cutoff. In the notation of \eqref{eq:sutva-functional}, $\psi(Y_i,W_i,D_i;P)={\lim_{w\downarrow w_0} \mathbbm{E}[Y_i\mid W_i=w] - \lim_{w\uparrow w_0}\mathbbm{E}[Y_i\mid W_i=w]}$.

\paragraph{Assumptions under interference.} Under arbitrary interference, consistency (Assumption~\ref{ass:consistency}) implies that the observed outcome admits the representation in \eqref{eq:obs-rep},
\[
Y_i
=
\sum_{\mathbf d_{-i}\in\{0,1\}^{N-1}}
\Big(
\alpha_i(\mathbf d_{-i})+\Delta_i(\mathbf d_{-i})D_i
\Big)\,
\mathbbm{1}(\mathbf D_{-i}=\mathbf d_{-i}),
\]
where $\alpha_i(\mathbf d_{-i}) := Y_i(0,\mathbf d_{-i})$ and
$\Delta_i(\mathbf d_{-i}) = Y_i(1,\mathbf d_{-i})-Y_i(0,\mathbf d_{-i})$ is the individual direct effect. The following assumption extends the RDD continuity condition to settings with arbitrary interference and imposes a conditional independence restriction.

\begin{assumption}[RDD Conditions under Interference]\label{rdd continuity}
For all $i\in[N]$,
\begin{align}
&\lim_{w\downarrow w_0}\mathbbm{E}[\alpha_i(\mathbf d_{-i})\mid W_i=w]
=
\lim_{w\uparrow w_0}\mathbbm{E}[\alpha_i(\mathbf d_{-i})\mid W_i=w],
\qquad
\forall\, \mathbf d_{-i}\in\{0,1\}^{N-1}, \label{rdd int: cont }\\
&\alpha_i(\mathbf d_{-i})\indep \mathbf D_{-i}\mid W_i,
\qquad
\forall\, \mathbf d_{-i}\in\{0,1\}^{N-1}\label{rdd int: uncon}.
\end{align}
\end{assumption}
This assumption requires the conditional mean of the untreated potential outcome to be continuous in the running variable at the cutoff for each fixed treatment vector of the other units. It also requires the untreated potential outcome to be independent of the assignments of other units, conditional on the running variable. Unconfoundedness is not new to the RDD literature. In the ITR-based RDD with heterogeneous treatment effects, \citet{hahn2001identification} likewise implicitly impose joint unconfoundedness of the potential outcomes with respect to own treatment.

Because $D_i=\mathbbm{1}(W_i\geq w_0)$ is a deterministic function of $W_i$, the dependence structure of the running variables determines the dependence between individual treatments. This motivates the following assumption.
\begin{assumption}[Independence of Running Variables]\label{running indepen}
    For all units $i \in[ N]$,  $\mathbf{W}_{-i} \indep W_i$,
    where $\mathbf{W}_{-i}$ collects the running variables of all units except $i$.
\end{assumption}
\noindent This assumption rules out cross-unit dependence in the running variable, so that each unit's running variable is independent of the others'. Assumption~\ref{running indepen} is sufficient for the unconditional version of CAI, $\mathbf{D}_{-i} \indep D_i$ for all $i\in[N]$, and it implies that $\mathbf{D}_{-i} \indep W_i$ for all $i\in[N]$; under the law $P$ of a randomly drawn unit, where $\mathbf D_{-i}$ also identifies the drawn unit, this independence requires in addition that the running variables share a common density.

Assumption~\ref{running indepen} is stronger than a causal interpretation of the formula requires. The following condition restricts the law of others' treatments only at the cutoff.
\begin{assumption}[No Jump in Others' Treatments at the Cutoff]\label{rdd nojump}
For all $i\in[N]$ and all $\mathbf d_{-i}\in\{0,1\}^{N-1}$ (under $P$, with the pooled one-sided limits), the one-sided limits
\[
\pi_+(\mathbf d_{-i}):=\lim_{w\downarrow w_0}\Pr(\mathbf D_{-i}=\mathbf d_{-i}\mid W_i=w)
\quad\text{and}\quad
\pi_-(\mathbf d_{-i}):=\lim_{w\uparrow w_0}\Pr(\mathbf D_{-i}=\mathbf d_{-i}\mid W_i=w)
\]
coincide, $\pi_+(\mathbf d_{-i})=\pi_-(\mathbf d_{-i})=:\pi_{w_0}(\mathbf d_{-i})$.
\end{assumption}
\noindent Assumption~\ref{rdd nojump} is the form that generalized CAI, condition~\eqref{unif:eq-gcai} of the main text, takes in the sharp RDD. It is implied by Assumption~\ref{running indepen}, under which $\mathbf D_{-i}\indep W_i$ and hence $\Pr(\mathbf D_{-i}=\mathbf d_{-i}\mid W_i=w)=\Pr(\mathbf D_{-i}=\mathbf d_{-i})$ for every $w$, but not conversely; under the law $P$ of a randomly drawn unit, the implication also requires the composition of units not to jump at $w_0$, as discussed below. It allows the law of others' treatments to vary with the running variable, including through dependence among the running variables, provided that it does not jump at $w_0$; the law at the cutoff, $\pi_{w_0}$, may then differ from the marginal law of $\mathbf D_{-i}$. A condition of this form, continuity at the cutoff of the law of the other treatments to which a unit is exposed, is Assumption~3(b) of \citet{borusyak2024regression}. Under the law $P$ of a randomly drawn unit, the one-sided limits $\pi_\pm$ also weight units by the densities of their running variables near $w_0$; the no-jump condition for $P$ then also requires the composition of units, the density of each unit's running variable relative to the sum of these densities, not to jump at $w_0$, as when every unit's density is continuous there.

\paragraph{Result.} Theorem~\ref{prop:rd} gives the interpretation of the sharp RDD identification formula under arbitrary interference.

\begin{theorem}
\label{prop:rd}
Suppose Assumptions~\ref{ass:arbitrary-interference} and~\ref{ass:consistency} hold, and consider a sharp regression discontinuity design
with cutoff $w_0$. Suppose further that treatment effects are homogeneous,\footnote{Corollary~\ref{unif:cor-designs}(c) shows that homogeneity can be dropped when conditions \eqref{rdd int: cont } and \eqref{rdd int: uncon} are also imposed on the treated potential outcomes $Y_i(1,\mathbf d_{-i})$.} that is,
\[
\Delta_i(\mathbf d_{-i})=\Delta(\mathbf d_{-i})
\quad\text{for all }\mathbf d_{-i}\in\{0,1\}^{N-1}\text{ and all }i\in[N],
\]
with $\Delta$ a fixed function on $\{0,1\}^{N-1}$,
and that the one-sided limits $\lim_{w\downarrow w_0}\Pr(\mathbf D_{-i}=\mathbf d_{-i}\mid W_i=w)$ and $\lim_{w\uparrow w_0}\Pr(\mathbf D_{-i}=\mathbf d_{-i}\mid W_i=w)$ exist for all $\mathbf d_{-i}\in\{0,1\}^{N-1}$.

\begin{itemize}
\item[(i)] If Assumption~\ref{rdd continuity} holds, then, whether or not $\mathbf{D}_{-i}\indep W_i$, the sharp RDD identification formula admits the decomposition
\[
\lim_{w\downarrow w_0}\mathbbm{E}[Y_i\mid W_i=w]
-
\lim_{w\uparrow w_0}\mathbbm{E}[Y_i\mid W_i=w]
=
\mathrm{ADE}_{+}
\;+\;
\mathrm{Bias}_{\mathrm{RDD}},
\]
where
\begin{align*}
\mathrm{ADE}_{+}&:=\sum_{\mathbf d_{-i}}\Delta(\mathbf d_{-i})\,\pi_+(\mathbf d_{-i})=\lim_{w\downarrow w_0}\mathbbm{E}[\Delta(\mathbf D_{-i})\mid W_i=w],\\
\mathrm{Bias}_{\mathrm{RDD}}&:=\sum_{\mathbf d_{-i}}\lim_{w\downarrow w_0}\mathbbm{E}[\alpha_i(\mathbf d_{-i})\mid W_i=w]\,\big\{\pi_+(\mathbf d_{-i})-\pi_-(\mathbf d_{-i})\big\},
\end{align*}
with $\pi_\pm(\mathbf d_{-i}):=\lim_{w\to w_0^\pm}\Pr(\mathbf D_{-i}=\mathbf d_{-i}\mid W_i=w)$. $\mathrm{ADE}_{+}$ is the ADE under the law of others' treatments just above the cutoff, and $\mathrm{Bias}_{\mathrm{RDD}}$ is the jump in untreated outcomes produced by a jump in that law.

\item[(ii)] If Assumptions~\ref{rdd continuity} and~\ref{rdd nojump} hold, the latter under the law $P$ of a randomly drawn unit, so that the composition of units does not jump at $w_0$, then $\mathrm{Bias}_{\mathrm{RDD}}=0$, and the sharp RDD identification formula identifies the average direct effect under the law of others' treatments at the cutoff:
\begin{align*}
\lim_{w\downarrow w_0}\mathbbm{E}[Y_i\mid W_i=w]
-
\lim_{w\uparrow w_0}\mathbbm{E}[Y_i\mid W_i=w]
=\mathrm{ADE}_{w_0}&:=\sum_{\mathbf d_{-i}}\Delta(\mathbf d_{-i})\,\pi_{w_0}(\mathbf d_{-i})\\
&\;=\lim_{w\to w_0}\mathbbm{E}[\Delta(\mathbf D_{-i})\mid W_i=w].
\end{align*}

\item[(iii)] If Assumptions~\ref{rdd continuity} and~\ref{running indepen} hold and, under the law $P$ of a randomly drawn unit, the running variables of all units have densities that are continuous at $w_0$ with a common positive value there, then $\pi_{w_0}=\Pr(\mathbf D_{-i}=\cdot\,)$, and the sharp RDD identification formula identifies the average direct effect:
\[
\lim_{w\downarrow w_0}\mathbbm{E}[Y_i\mid W_i=w]
-
\lim_{w\uparrow w_0}\mathbbm{E}[Y_i\mid W_i=w]
=
\mathrm{ADE}_{w_0}=\mathbbm{E}[\Delta(\mathbf D_{-i})].
\]
If the densities are continuous at $w_0$, with a positive sum there, but their values differ across units, the formula identifies $\sum_{k=1}^N q(k;w_0)\,\mathbbm{E}[\Delta(\mathbf D_{-k})]$, where $q(k;w_0)$ is the density of $W_k$ at $w_0$ divided by the sum of these densities over units, which weights units by their density at the cutoff. ~\hyperref[prf:prop:rd]{\textup{\footnotesize[Proof]}}
\end{itemize}
\end{theorem}

\paragraph{Discussion.} Theorem~\ref{prop:rd} highlights the role of Assumptions~\ref{rdd nojump} and~\ref{running indepen}, two CAI-type conditions, in the causal interpretation of the RDD under arbitrary interference. Part~(i) shows that, without a restriction on the law of others' treatments given a unit's running variable, the sharp RDD identification formula combines the ADE under the law of others' treatments just above the cutoff with a bias, $\mathrm{Bias}_{\mathrm{RDD}}$, which arises because a jump at the cutoff in that law shifts untreated outcomes and is mistaken for an effect.

Part~(ii) shows that the weaker condition, no jump at the cutoff in the law of others' treatments, removes this bias, and the formula identifies $\mathrm{ADE}_{w_0}$, the ADE of units at the cutoff averaged over the law of others' treatments they face. Part~(iii) shows that independence of running variables across units, which implies independence between own treatment and others' treatments, makes that law the marginal law of $\mathbf D_{-i}$ when, under $P$, the running-variable densities are continuous at $w_0$ with a common positive value there, so the formula identifies the ADE, $\mathbbm{E}[\Delta(\mathbf D_{-i})]$. Under the weaker condition, the formula therefore identifies a different object: $\mathrm{ADE}_{w_0}$ and $\mathbbm{E}[\Delta(\mathbf D_{-i})]$ weight the same direct effects by different laws of others' treatments and coincide when $\pi_{w_0}=\Pr(\mathbf D_{-i}=\cdot\,)$. Under independence of running variables, the running variable generates exogenous variation in own treatment while holding the distribution of others' treatments fixed, so the discontinuity in outcomes at the cutoff identifies the ADE, averaged over the marginal law of others' treatments.

\paragraph{Relation to the literature.} The form of this result is known under additional structure. With interference through a known network and exposure mapping, \cite{daltorrione2025regression} show that the conventional sharp RD contrast identifies a ``boundary overall direct effect'', i.e., a boundary direct effect averaged over the distribution of neighbors' exposure at the cutoff. Their continuous joint density of own and neighbors' running variables balances that distribution across the cutoff and, together with continuity of mean potential outcomes in own and neighbors' running variables, delivers the form of Corollary~\ref{unif:cor-designs}(c) with the exposure in place of $\mathbf D_{-i}$; their density condition is a primitive sufficient condition for generalized CAI. A continuity condition of the same form, requiring that the distribution of the other discontinuity shocks to which a unit is exposed not jump at the cutoff of one shock's running variable, underlies the fuzzy, instrumental-variables regression discontinuity aggregation framework of \citet[Assumption~3(b) and Prop.~2]{borusyak2024regression}, which they extend informally to spillovers from discontinuity events (Sec.~4.1). Theorem~\ref{prop:rd} and Corollary~\ref{unif:cor-designs}(c) need no network, no exposure mapping, and no conditioning on others' running variables, and Proposition~\ref{unif:thm-master} also characterizes what the formula recovers when the law of others' treatments jumps at the cutoff, a case these papers do not address. 

 \cite{auerbach2024regression} study the same underlying issue from the perspective of the local linear RD estimator. In both analyses, the direct-effect interpretation requires the law of others' treatments not to jump at the cutoff. They impose a linear-in-means form of interference and study what happens within a positive bandwidth, whereas we allow arbitrary interference and characterize the population contrast at the cutoff. Their estimator can move from the direct effect toward the total effect when nearby units increasingly lie on the same side of the cutoff and their treatments move together. In our decomposition, this is the baseline component arising from variation in the law of others' treatments within the estimation window. The results are therefore complementary, with ours describing the condition at the cutoff and theirs showing why the same condition must hold approximately nearby for local estimation to preserve the direct-effect interpretation.

\subsection{Static Difference-in-Differences}\label{sec:did}

DiD exploits variation across both time and units. We focus on the canonical design with two groups (treated and untreated) and two periods ($t=0,1$), in which no unit is treated in period~0 and units in the treated group become treated in period~1.

\paragraph{The ITR benchmark.} Let $Y_{it}(0)$ and $Y_{it}(1)$ denote unit $i$'s potential outcomes without and with treatment in period $t\in\{0,1\}$, and let $D_i \in \{0,1\}$ indicate whether unit $i$ belongs to the treated group and is treated in period~1. Because no unit is treated in period~0, the observed outcome in period~0 is $Y_{i0} = Y_{i0}(0)$, so that untreated potential outcomes are observed for all units. This setup rules out anticipation effects, i.e., treatment assignment in period~1 does not affect outcomes in period~0.\footnote{ Section~\ref{sec:did_dynamic} below develops a setup that allows anticipation effects.} In period~1, the observed outcome is $Y_{i1} = Y_{i1}(0)(1-D_i) + Y_{i1}(1)D_i$.

Following \citet{abadie2005semiparametric}, we allow for covariate-driven differences in outcome dynamics between treated and control units.\footnote{The results here also apply to the unadjusted comparisons of \citet[Table~3]{card1994minimum}, in which there are no covariate-driven differences in outcome dynamics; their regression-adjusted models \citep[Table~4]{card1994minimum} allow such differences through chain, ownership and region indicators.} Under ITR, the DiD identification formula is
\begin{align}
   \Phi_{\mathrm{ITR}}(P)
:=&
\mathbbm{E}\!\left[
\mathbbm{E}[Y_{i1}\mid W_i,D_i=1]
-
\mathbbm{E}[Y_{i0}\mid W_i,D_i=1] \mid D_i=1
\right]\nonumber\\
& \quad \quad -
\left(\mathbbm{E}\!\left[
\mathbbm{E}[Y_{i1}\mid W_i,D_i=0]
-
\mathbbm{E}[Y_{i0}\mid W_i,D_i=0]\mid D_i=1
\right]\right) \label{DiD indentification formula}\\
=&
\mathbbm{E}\!\Bigg[\frac{D_i}{\Pr(D_i=1)}\cdot\Big\{
\left(
\mathbbm{E}[Y_{i1}\mid W_i,D_i=1]
-
\mathbbm{E}[Y_{i0}\mid W_i,D_i=1]
\right)\nonumber
\\
& \quad \quad \quad \quad \quad \quad \quad \quad  \quad\quad -
\left(
\mathbbm{E}[Y_{i1}\mid W_i,D_i=0]
-
\mathbbm{E}[Y_{i0}\mid W_i,D_i=0]
\right)\Big\}\Bigg]\nonumber.
\end{align}
Thus, in the notation of \eqref{eq:sutva-functional}, with $Y_i:=(Y_{i0},Y_{i1})$,
\begin{align*}
\psi(Y_i,W_i,D_i;P)={}&\frac{D_i}{\Pr(D_i=1)}\Big\{
\left(
\mathbbm{E}[Y_{i1}\mid W_i,D_i=1]
-
\mathbbm{E}[Y_{i0}\mid W_i,D_i=1]
\right)\\
&\qquad\qquad -
\big(
\mathbbm{E}[Y_{i1}\mid W_i,D_i=0]
-
\mathbbm{E}[Y_{i0}\mid W_i,D_i=0]
\big)\Big\}.
\end{align*}

\citet{baker2025difference} show that the identification formula in \eqref{DiD indentification formula}
equals the average treatment effect on the treated (ATT) in period~1,
$\Phi_{\mathrm{ITR}}(P)=\mathbbm{E}[Y_{i1}(1)-Y_{i1}(0)\mid D_i=1]$, under the following conditions:
\begin{align}
    &\mathbbm{E}[Y_{i1}(0)-Y_{i0}(0)\mid W_i,D_i=1]
=
\mathbbm{E}[Y_{i1}(0)-Y_{i0}(0)\mid W_i,D_i=0] \quad \text{(Conditional parallel trends),}\label{CPT}\\
&\text{There exists some}\quad \varepsilon>0,\quad \varepsilon<\Pr(D_{i}=1\mid W_i)<1-\varepsilon \quad \quad \quad \quad \text{(Strong overlap).}\label{SO}
\end{align}
Condition \eqref{CPT} requires that, conditional on covariates, treated and control units would have experienced the same average evolution of untreated potential outcomes over time. Condition \eqref{SO} requires the conditional probability of treatment given the covariates $W_i$ to be bounded away from zero and one.

\paragraph{Assumptions under interference.} Because no unit is treated in period~0, the observed outcome in that period under arbitrary interference is $Y_{i0} = Y_{i0}(\mathbf{0})=Y_{i0}(0,\mathbf{0}_{-i})$, where $\mathbf{0}$ and $\mathbf{0}_{-i}$ are the $N$- and $(N-1)$-dimensional vectors of zeros. In period~1,
the observed outcome under arbitrary interference and consistency is $Y_{i1}=\sum_{\mathbf{d}_{-i}} \left(Y_{i1}(0, \mathbf{d}_{-i})(1-D_i) + Y_{i1}(1, \mathbf{d}_{-i})D_i\right)\cdot \mathbbm{1}(\mathbf{D}_{-i}=\mathbf{d}_{-i})$.
Because the outcome of a unit depends on the treatments of others, the classical parallel trends assumption must be strengthened. The following assumption is the analog of conditions \eqref{CPT} and \eqref{SO}.

\begin{assumption}\label{did: parallel trend}
For all $i\in[N]$,
\begin{align}
&\mathbbm{E}[Y_{i1}(0,\mathbf d_{-i})\mid W_i,D_i=1]
-
\mathbbm{E}[Y_{i0}(\mathbf 0)\mid W_i,D_i=1] \nonumber\\
&\quad=
\mathbbm{E}[Y_{i1}(0,\mathbf d_{-i})\mid W_i,D_i=0]
-
\mathbbm{E}[Y_{i0}(\mathbf 0)\mid W_i,D_i=0],
\qquad
\forall\,\,\mathbf{d}_{-i}\in\{0,1\}^{N-1}, \label{PT}\\
&\text{There exists some}\quad \varepsilon>0,\quad \varepsilon<\Pr(D_{i}=1\mid W_i)<1-\varepsilon, \label{SO 2}\\
&\mathbbm{E}[Y_{i1}(0,\mathbf d_{-i})-Y_{i0}(\mathbf 0)\mid W_i,D_i=d,\mathbf D_{-i}=\mathbf d_{-i}]
=
\mathbbm{E}[Y_{i1}(0,\mathbf d_{-i})-Y_{i0}(\mathbf 0)\mid W_i,D_i=d],\notag\\
&\qquad\qquad\forall\,\, d\in\{0,1\},\ \mathbf{d}_{-i}\in\{0,1\}^{N-1}\text{ with }\Pr(\mathbf D_{-i}=\mathbf d_{-i}\mid W_i,D_i=d)>0. \label{unif:did-meanindep}
\end{align}
\end{assumption}
For the treated cell, Assumption~\ref{unif:ass-sep} of the main text requires in addition the treated-trend counterpart of condition \eqref{unif:did-meanindep},
\begin{align}
&\mathbbm{E}[Y_{i1}(1,\mathbf d_{-i})-Y_{i0}(\mathbf 0)\mid W_i,D_i=1,\mathbf D_{-i}=\mathbf d_{-i}]
=
\mathbbm{E}[Y_{i1}(1,\mathbf d_{-i})-Y_{i0}(\mathbf 0)\mid W_i,D_i=1],\nonumber \\
&\quad\forall\,\mathbf d_{-i}\text{ with }\Pr(\mathbf D_{-i}=\mathbf d_{-i}\mid W_i,D_i=1)>0,
\label{unif:did-meanindep-treated}
\end{align}
which is not implied by \eqref{PT} and \eqref{unif:did-meanindep}, which restrict untreated trends only. Theorem~\ref{thm:did} does not need it; Corollary~\ref{unif:cor-designs}(d) uses it to place the design within Proposition~\ref{unif:thm-master}.
Condition \eqref{PT} extends conditional parallel trends to environments with interference by requiring that, conditional on covariates, the treated and control groups would have exhibited the same average evolution of untreated potential outcomes across periods, even when outcomes may depend on others' treatments. Condition \eqref{SO 2} restates the strong overlap condition. Condition \eqref{unif:did-meanindep} requires that, within covariate cells and treatment groups, the realized assignment of others is not selected on unit $i$'s untreated outcome trend. It is the DiD counterpart of the requirement, embedded in conditions \eqref{eq:uso} and \eqref{rdd int: uncon}, that $\mathbf D_{-i}$ be unconfounded for unit $i$'s potential outcomes, and it allows the mean untreated trend of each group to be expressed as an average over the group's law of $\mathbf D_{-i}$, as in Theorem~\ref{thm:did}. Without \eqref{unif:did-meanindep}, the decomposition in part~(i) of Theorem~\ref{thm:did} continues to hold once $\mathbf D_{-i}=\mathbf d_{-i}$ is added to the conditioning sets in $\text{trend}_d(W_i)$, and the conclusion of part~(ii) continues to hold if \eqref{PT} is imposed conditionally on $\mathbf D_{-i}=\mathbf d_{-i}$.

\paragraph{Result.} Theorem~\ref{thm:did} gives the interpretation of the ITR-based DiD identification formula under arbitrary interference, with particular emphasis on the role of CAI.

\begin{theorem}
\label{thm:did}
Suppose Assumptions~\ref{ass:arbitrary-interference} and~\ref{ass:consistency} hold.
\begin{itemize}
    \item[(i)] If Assumption~\ref{did: parallel trend} holds, then, whether or not condition~\eqref{eq:CAI} holds,
    \begin{align*}
\Phi_{\mathrm{ITR}}(P)
&=
\underbrace{
\mathbbm{E}\!\left[
\mathbbm{E}\!\left[
Y_{i1}(1,\mathbf D_{-i})
-
Y_{i1}(0,\mathbf D_{-i})
\mid W_i,D_i=1
\right]\Big |D_i=1\right]
}_{\mathrm{ADTT}}
\;+\;
\mathrm{Bias}_{\mathrm{DiD}},
\end{align*}
where
\quad $\mathrm{Bias}_{\mathrm{DiD}}:= \mathbbm{E}\Big[\text{trend}_1(W_i)-\text{trend}_0(W_i)\Big |D_i=1\Big],$ \quad with
\begin{align*}
\text{trend}_d(W_i)
:=&
\sum_{\mathbf d_{-i}}
\Big(
\mathbbm{E}[Y_{i1}(0,\mathbf d_{-i})\mid W_i,D_i=d]
-
\mathbbm{E}[Y_{i0}(\mathbf 0)\mid W_i,D_i=d]
\Big)\\
&\quad \quad \quad \cdot \Pr(\mathbf D_{-i}=\mathbf d_{-i}\mid W_i,D_i=d).
\end{align*}

\item[(ii)]  If Assumption~\ref{did: parallel trend} and condition~\eqref{eq:CAI}, under $P$, where it includes equal composition~\eqref{unif:eq-ec}, hold,
then $\mathrm{Bias}_{\mathrm{DiD}}=0$, and therefore
\[
\Phi_{\mathrm{ITR}}(P)
=
\mathbbm{E}\!\left[
\mathbbm{E}\!\left[
Y_{i1}(1,\mathbf D_{-i})
-
Y_{i1}(0,\mathbf D_{-i})
\mid W_i,D_i=1
\right]\mid D_i=1\right],
\]
which identifies the ADTT in period~1.
\end{itemize}
~\hyperref[prf:thm:did]{\textup{\footnotesize[Proof]}}\end{theorem}

\paragraph{Discussion.} Theorem~\ref{thm:did} identifies the sources of bias in the ITR-based DiD identifying functional when treatment assignment is arbitrarily dependent. Part~(i) shows that, even when the parallel-trends-type restriction in Assumption~\ref{did: parallel trend} holds, the DiD functional fails to identify the ADTT. The bias arises from two distinct but interacting forces.

First, untreated potential outcomes depend on others' treatments. When interference is present, the counterfactual evolution $Y_{i1}(0,\mathbf d_{-i}) - Y_{i0}(\mathbf 0)$ varies across $\mathbf d_{-i}$. Condition \eqref{PT} ensures that, for any fixed treatment vector of other units, treated and control units share the same untreated potential outcome trend in conditional expectation, but it does not restrict how untreated trends vary across $\mathbf d_{-i}$.

Second, the distribution of others' treatments differs by treatment group, even after conditioning on the covariates $W_i$. The bias term in Theorem~\ref{thm:did} aggregates $\mathbf d_{-i}$-specific untreated trends with group-specific weights $\Pr(\mathbf D_{-i}=\mathbf d_{-i}\mid W_i, D_i=d)$, so systematic differences in these distributions induce differences in average untreated trends between treated and control groups. DiD then compares weighted averages of untreated potential outcomes taken over different conditional distributions.

Together, these two forces induce bias, a wedge between the DiD estimand and the ADTT. Although part~(i) is stated under Assumption~\ref{did: parallel trend}, its proof uses only consistency, overlap \eqref{SO 2}, and condition \eqref{unif:did-meanindep}; the parallel-trends condition \eqref{PT} enters only in part~(ii). The decomposition in part~(i) is therefore an identity that holds whether or not \eqref{PT} holds, and \eqref{PT}, together with CAI, sets $\mathrm{Bias}_{\mathrm{DiD}}$ to zero.

Part~(ii) shows that restricting the dependence of post-treatment assignments eliminates the second source of bias by equalizing, conditional on $W_i$, the distribution of others' treatments across treatment groups. Combined with Assumption~\ref{did: parallel trend}, CAI ensures that both the $\mathbf d_{-i}$-specific untreated trends and the weights used to aggregate them coincide across groups, which eliminates the bias term. The DiD functional then equals the ADTT (Corollary~\ref{unif:cor-designs}(d)).

\paragraph{Relation to the literature.} The closest antecedent is \citet{xu2023difference}, who studies DiD with randomly assigned, possibly spatially correlated treatments, interference through a known exposure mapping $G_i$ of treatments in a finite neighborhood, parallel trends conditional on exposure, and fixed covariates. Xu shows that the canonical DiD estimand generally differs from the direct effect on the treated unless $G_i$ has the same conditional law in the treated and control groups, which holds when $G_i$ is independent of own treatment given covariates; for binary exposure, the bias is the difference in untreated trends between exposed and unexposed controls times the difference in exposure probabilities between the groups. Xu's result is therefore the exposure-mapping counterpart of Theorem~\ref{thm:did}, including its independence condition and the product structure of $\mathrm{Bias}_{\mathrm{DiD}}$. Xu states the result for any exposure mapping with finitely many levels, so with the identity mapping, $G_i=\mathbf D_{-i}$, it gives the decomposition of Theorem~\ref{thm:did}(i) under parallel trends conditional on each $\mathbf d_{-i}$, a case she does not discuss and that her estimation theory, which restricts exposure to a finite neighborhood, does not cover. That hypothesis is implied by \eqref{PT} and \eqref{unif:did-meanindep} but does not imply them, so the two results are not nested as stated. CAI for the full vector $\mathbf D_{-i}$ is stronger than equality of the laws of $G_i$, while parallel trends at every $\mathbf d_{-i}$ is neither stronger nor weaker than parallel trends at every exposure level.

For spatial spillovers, with the assignment held fixed, \citet{butts2021difference} shows that, under parallel trends in the outcome with no treatment and no exposure, DiD equals the total effect on the treated minus the average spillover on controls, and notes that recovering the direct (switching) effect requires these spillovers to coincide across groups; adding and subtracting the spillover on treated units' untreated outcomes turns the gap into the analog of $\mathrm{Bias}_{\mathrm{DiD}}$, with each unit's realized exposure in place of $\mathbf D_{-i}$. \citet[eq.~(15)]{clarke2017estimating} gives the special case of a linear model with constant effects and a known distance-based exposure indicator, in which the bias is minus the spillover effect times the share of exposed controls. With interference left unrestricted but the assignment held fixed, \citet[Prop.~3]{mealli2025difference} show that, under parallel trends in the outcome with no unit treated, DiD identifies the difference between the total effect on the treated and the spillover on controls, and is uninformative about either effect alone. Unlike \citet{butts2021difference} and \citet{mealli2025difference}, Theorem~\ref{thm:did} treats the assignment as random, which, as in \citet{xu2023difference}, lets the same contrast reduce to a direct effect, the ADTT, under an assignment condition, here CAI. The price is parallel trends at every $\mathbf d_{-i}$ and condition \eqref{unif:did-meanindep}. Parallel trends at every $\mathbf d_{-i}$ is stronger than the parallel trends conditions in \citet{butts2021difference} and \citet{mealli2025difference}, which target the total effect. In the spatial and border designs of \citet{clarke2017estimating}, \citet{butts2021difference}, and \citet{xu2023difference}, treated units cluster, so CAI is typically implausible, and exposure-mapping or distant-control strategies remain the appropriate tools. The next subsection relates the dynamic version (Theorem~\ref{thm:did_dynamic}) to results for staggered adoption \citep{tagawa2026identification,fiorini2024simple}, including the two-way fixed effects weights of \citet{butts2021difference}.

\subsection{Dynamic Difference-in-Differences}\label{sec:did_dynamic}
The static design of Section~\ref{sec:did} rules out anticipation by treating the period-0 outcome as an untreated outcome. We now let potential outcomes depend on the whole treatment history \citep{baker2025difference}, which permits anticipation, and extend that design to arbitrary interference.

\paragraph{The ITR benchmark.} The design is the canonical two-group, two-period design of Section~\ref{sec:did}: no unit is treated in period~0, and $D_i\in\{0,1\}$ indicates that unit $i$ belongs to the treated group and is treated in period~1. Let $h_{it}:=D_i\,\mathbbm{1}(t=1)$ denote the treatment of unit $i$ in period $t$, and let $\mathbf h_{-i,t}$ collect the treatments of the other units in period $t$. Potential outcomes depend on the treatment histories of unit $i$ and of the others, $Y_{it}([h_{i0},h_{i1}],[\mathbf h_{-i,0},\mathbf h_{-i,1}])$. Because no unit is treated in period~0, only $Y_{it}([0,0],[\mathbf 0_{-i},\mathbf h_{-i,1}])$ and $Y_{it}([0,1],[\mathbf 0_{-i},\mathbf h_{-i,1}])$ arise, and we abbreviate
\[
Y_{it}(0,\mathbf d_{-i}):=Y_{it}([0,0],[\mathbf 0_{-i},\mathbf d_{-i}])
\quad\text{and}\quad
Y_{it}(1,\mathbf d_{-i}):=Y_{it}([0,1],[\mathbf 0_{-i},\mathbf d_{-i}]),
\]
where $\mathbf d_{-i}$ is the value of $\mathbf h_{-i,1}$. The identification formula and the function $\psi$ are those of \eqref{DiD indentification formula}, since the design and the observed data are unchanged and only the indexing of potential outcomes is richer. What changes is the benchmark. In the static design the period-0 outcome is the untreated outcome $Y_{i0}(\mathbf 0)$ for every unit, so no assumption is needed to observe it; here $Y_{i0}(1,\mathbf d_{-i})$ and $Y_{i0}(0,\mathbf d_{-i})$ differ in general, and their equality in the treated group must be imposed.

\paragraph{Assumptions under interference.} Two assumptions take the place of conditions \eqref{CPT} and \eqref{SO}. The first is the parallel-trends and overlap analog, stated at each assignment of the others.

\begin{assumption}\label{did: parallel trend dyn}
For all $i\in[N]$,
\vspace{-0.7cm}
\begin{align}
&\mathbbm{E}[Y_{i1}(0,\mathbf d_{-i})\mid W_i,D_i=1]
-
\mathbbm{E}[Y_{i0}(0,\mathbf d_{-i})\mid W_i,D_i=1] \nonumber\\
&\quad=
\mathbbm{E}[Y_{i1}(0,\mathbf d_{-i})\mid W_i,D_i=0]
-
\mathbbm{E}[Y_{i0}(0,\mathbf d_{-i})\mid W_i,D_i=0],
\qquad
\forall\,\,\mathbf{d}_{-i}\in\{0,1\}^{N-1}, \label{PT dyn}\\
&\text{There exists some}\quad \varepsilon>0,\quad \varepsilon<\Pr(D_{i}=1\mid W_i)<1-\varepsilon \quad \quad \quad \quad \quad \quad\text{(Strong overlap)}\label{SO dyn},\\
&\mathbbm{E}[Y_{i1}(0,\mathbf d_{-i})-Y_{i0}(0,\mathbf d_{-i})\mid W_i,D_i=d,\mathbf D_{-i}=\mathbf d_{-i}]
=
\mathbbm{E}[Y_{i1}(0,\mathbf d_{-i})-Y_{i0}(0,\mathbf d_{-i})\mid W_i,D_i=d],\notag\\
&\qquad\qquad\forall\,\, d\in\{0,1\},\ \mathbf{d}_{-i}\in\{0,1\}^{N-1}\text{ with }\Pr(\mathbf D_{-i}=\mathbf d_{-i}\mid W_i,D_i=d)>0. \label{unif:dyn-meanindep}
\end{align}
\end{assumption}
As in the static design, Assumption~\ref{unif:ass-sep} of the main text requires, for the treated cell, the treated-trend counterpart of condition \eqref{unif:dyn-meanindep},
\begin{align}
&\mathbbm{E}[Y_{i1}(1,\mathbf d_{-i})-Y_{i0}(1,\mathbf d_{-i})\mid W_i,D_i=1,\mathbf D_{-i}=\mathbf d_{-i}]
=
\mathbbm{E}[Y_{i1}(1,\mathbf d_{-i})-Y_{i0}(1,\mathbf d_{-i})\mid W_i,D_i=1],\nonumber\\ &\quad\forall\,\mathbf d_{-i}\text{ with }\Pr(\mathbf D_{-i}=\mathbf d_{-i}\mid W_i,D_i=1)>0,
\label{unif:dyn-meanindep-treated}
\end{align}
which Theorem~\ref{thm:did_dynamic} does not need and Corollary~\ref{unif:cor-designs}(e) uses.
Condition \eqref{PT dyn} extends conditional parallel trends to interference by requiring that, conditional on covariates, the treated and control groups would have exhibited the same average evolution of untreated outcomes across periods, separately for every assignment of the other units. It differs from condition \eqref{PT} of the static design in which outcomes it compares: condition \eqref{PT} holds the period-0 outcome at $\mathbf 0$, because no unit is treated then, whereas here that outcome is itself indexed by $\mathbf d_{-i}$, so the trend is taken at a common assignment of the others in both periods. Condition \eqref{SO dyn} restates strong overlap. Condition \eqref{unif:dyn-meanindep} is the counterpart of condition \eqref{unif:did-meanindep}: within covariate cells and treatment groups, the realized assignment of others is not selected on unit $i$'s untreated outcome trend, which allows each group's mean untreated trend to be expressed as an average over that group's own law of $\mathbf D_{-i}$, as in Theorem~\ref{thm:did_dynamic}.

The second assumption rules out anticipation.

\begin{assumption}\label{did: no anticipation dyn}
For all $i\in[N]$,
\vspace{-0.7cm}
\begin{align}
&\mathbbm{E}[Y_{i0}(0,\mathbf d_{-i})\mid W_i,D_i=1]
=
\mathbbm{E}[Y_{i0}(1,\mathbf d_{-i})\mid W_i,D_i=1]
,
\qquad
\forall\,\,\mathbf{d}_{-i}\in\{0,1\}^{N-1}, \label{NA dyn}\\
&\mathbbm{E}[Y_{i0}(1,\mathbf d_{-i})-Y_{i0}(0,\mathbf d_{-i})\mid W_i,D_i=1,\mathbf D_{-i}=\mathbf d_{-i}]
=
\mathbbm{E}[Y_{i0}(1,\mathbf d_{-i})-Y_{i0}(0,\mathbf d_{-i})\mid W_i,D_i=1],\notag\\
&\qquad\qquad\forall\,\,\mathbf{d}_{-i}\in\{0,1\}^{N-1}\text{ with }\Pr(\mathbf D_{-i}=\mathbf d_{-i}\mid W_i,D_i=1)>0. \label{unif:dyn-antic-meanindep}
\end{align}
\end{assumption}
Condition \eqref{NA dyn} is conditional no-anticipation under interference: in period~0, conditional on covariates, the treated and untreated potential outcomes of the treated group agree on average. Condition \eqref{unif:dyn-antic-meanindep} requires that agreement to hold within each realization of others' assignments, so that anticipation averages out over the treated group's law of $\mathbf D_{-i}$. The static design obtains this without assumption, because its period-0 outcomes are untreated outcomes by construction; once outcomes depend on the treatment history, it becomes a substantive restriction.

\paragraph{Result.} Theorem~\ref{thm:did_dynamic} gives the interpretation of the ITR-based DiD identification formula when outcomes depend on treatment histories, again with emphasis on the role of CAI.

\begin{theorem}
\label{thm:did_dynamic}
Suppose Assumptions~\ref{ass:arbitrary-interference} and~\ref{ass:consistency} hold.
\begin{itemize}
    \item[(i)] If Assumptions~\ref{did: parallel trend dyn} and~\ref{did: no anticipation dyn} hold, then, whether or not condition~\eqref{eq:CAI} holds,
    \begin{align*}
\Phi_{\mathrm{ITR}}(P)
&=
\underbrace{
\mathbbm{E}\!\left[
\mathbbm{E}\!\left[
Y_{i1}(1,\mathbf D_{-i})
-
Y_{i1}(0,\mathbf D_{-i})
\mid W_i,D_i=1
\right]\Big |D_i=1\right]
}_{\mathrm{ADTT}}
\;+\;
\mathrm{Bias}_{\mathrm{DiD_{dyn}}},
\end{align*}
where
\quad $\mathrm{Bias}_{\mathrm{DiD_{dyn}}}:=\mathbbm{E}\big[\text{trend}^{\mathrm{dyn}}_1(W_i)-\text{trend}^{\mathrm{dyn}}_0(W_i)\,\big|\,D_i=1\big],$ \quad with
\begin{align*}
\text{trend}^{\mathrm{dyn}}_d(W_i)
:={}&
\sum_{\mathbf d_{-i}}
\Big(
\mathbbm{E}[Y_{i1}(0,\mathbf d_{-i})\mid W_i,D_i=d]
-
\mathbbm{E}[Y_{i0}(0,\mathbf d_{-i})\mid W_i,D_i=d]
\Big)\\
&\qquad\cdot\Pr(\mathbf D_{-i}=\mathbf d_{-i}\mid W_i,D_i=d).
\end{align*}

\item[(ii)]  If Assumptions~\ref{did: parallel trend dyn} and~\ref{did: no anticipation dyn} and condition~\eqref{eq:CAI}, under $P$, where it includes equal composition~\eqref{unif:eq-ec}, hold,
then $\mathrm{Bias}_{\mathrm{DiD_{dyn}}}=0$, and therefore
\[
\Phi_{\mathrm{ITR}}(P)
=
\mathbbm{E}\!\left[
\mathbbm{E}\!\left[
Y_{i1}(1,\mathbf D_{-i})
-
Y_{i1}(0,\mathbf D_{-i})
\mid W_i,D_i=1
\right]\mid D_i=1\right],
\]
which identifies the ADTT in period~1.
\end{itemize}
~\hyperref[prf:thm:did_dynamic]{\textup{\footnotesize[Proof]}}\end{theorem}

\paragraph{Discussion.} Theorem~\ref{thm:did_dynamic} reproduces the structure of Theorem~\ref{thm:did}, with one further term in the derivation. Part~(i) shows that the DiD functional again fails to identify the ADTT even when the parallel-trends restriction holds, and for the same two reasons. First, untreated potential outcomes depend on others' treatments, so the counterfactual evolution $Y_{i1}(0,\mathbf d_{-i})-Y_{i0}(0,\mathbf d_{-i})$ varies across $\mathbf d_{-i}$, and condition \eqref{PT dyn} restricts the comparison between groups at each $\mathbf d_{-i}$ without restricting this variation. Second, the law of others' treatments differs by treatment group even given $W_i$, so $\mathrm{Bias}_{\mathrm{DiD_{dyn}}}$ aggregates $\mathbf d_{-i}$-specific untreated trends with group-specific weights $\Pr(\mathbf D_{-i}=\mathbf d_{-i}\mid W_i,D_i=d)$.

The third force is new. Allowing outcomes to depend on the treatment history admits an own-anticipation term, $\mathbbm{E}[Y_{i0}(1,\mathbf D_{-i})-Y_{i0}(0,\mathbf D_{-i})\mid W_i,D_i=1]$, in the decomposition. It has no counterpart in Theorem~\ref{thm:did}, whose period-0 outcome carries no own treatment, and it is removed by Assumption~\ref{did: no anticipation dyn} rather than bounded: unlike the interference and dependence channels, anticipation is not a wedge the sensitivity analysis of Section~\ref{robustness} parameterizes.

As in Theorem~\ref{thm:did}, part~(i) is an identity. Its proof uses consistency, overlap \eqref{SO dyn}, Assumption~\ref{did: no anticipation dyn} and condition \eqref{unif:dyn-meanindep}, but not the parallel-trends equality \eqref{PT dyn}, which enters only in part~(ii). Part~(ii) shows that CAI eliminates the second force by equalizing, conditional on $W_i$, the two groups' laws of others' treatments. Together with Assumption~\ref{did: parallel trend dyn}, both the $\mathbf d_{-i}$-specific untreated trends and the weights that aggregate them then coincide across groups, so the bias vanishes and the functional equals the ADTT (Corollary~\ref{unif:cor-designs}(e)). Permitting anticipation therefore costs one additional assumption and leaves the role of CAI unchanged.

\paragraph{Relation to the literature.} Decompositions of this kind are available for staggered designs with a known interaction structure. For network spillovers with a researcher-specified exposure mapping, \citet[Prop.~3.1]{tagawa2026identification} decomposes the cohort-specific staggered DiD estimand into a total effect, exposure-contrast terms, and a trend-difference term. \citet{fiorini2024simple} identify the ATT at zero exposure under staggered adoption by ruling out spillovers on treated units and assuming a set of unexposed never-treated units. \citet{butts2021difference} shows that two-way fixed effects average the resulting $2\times2$ spillover biases with possibly negative weights. \citet{huber2021framework} separate individual-level and spillover effects when interference is confined to known regions, with identification under randomization, selection on observables, and DiD. Theorem~\ref{thm:did_dynamic} gives the corresponding decomposition for the two-period design with treatment histories, with $\mathbf D_{-i}$ in place of an exposure mapping or region structure. Generalized CAI then offers an assignment-based route to an ATT-type direct effect that does not require locating unexposed units.

\subsection{The Designs as Special Cases}\label{supp:corollary}
With the design-specific assumptions and results of Sections~\ref{sec:identification}--\ref{sec:did_dynamic} in place, the following corollary collects them as instances of Proposition~\ref{unif:thm-master}. Corollary~\ref{unif:cor-designs} states, for each design, the ingredients of the representation of Section~\ref{sec:mapping}, the target Proposition~\ref{unif:thm-master} delivers, and the resulting bias; Table~\ref{unif:tab-mapping} of the main text summarizes it. The proof is in Section~\ref{proofs}.

\begin{cor}\label{unif:cor-designs}
Suppose Assumptions~\ref{ass:arbitrary-interference} and~\ref{ass:consistency} hold, and take each design's ingredients from Panel~A of Table~\ref{unif:tab-mapping}; the design-specific assumptions and slice means are those of Sections~\ref{sec:identification}--\ref{sec:did_dynamic} above.
\begin{itemize}
\item[(a)] \emph{Selection on observables.} Under Assumption~\ref{ass: uncon}, the means in~\eqref{unif:eq-sep} are $\mu_d(\mathbf d_{-i};w)=\bar Y_i(d,\mathbf d_{-i};w):=\mathbbm{E}[Y_i(d,\mathbf d_{-i})\mid W_i=w]$, so $\theta(\mathbf d_{-i};w)=\mathbbm{E}[\Delta_i(\mathbf d_{-i})\mid W_i=w]$. The reference laws $\pi_1$, $\pi_0$ and $\pi$ yield \eqref{eq:decomp-adtc-x}, \eqref{eq:decomp-adtt-x} and \eqref{eq:decomp-ade-x}, with $\mathcal B(\pi^\ast)$ equal to $B^1_{ADE}$, $B^0_{ADE}$ and $B_{ADE}$; here \eqref{unif:eq-gcai} becomes CAI, \eqref{eq:CAI}, together with equal composition~\eqref{unif:eq-ec} under $P$, and part~(ii) yields \eqref{eq:ade-x}.
\item[(b)] \emph{Instrumental variables.} Under Assumption~\ref{ass: IV conditions} with $\kappa>0$, $0<\Pr(Z_i=1)<1$ and, under $P$, a common instrument propensity $\Pr(Z_k=1)$ across units, the first stage $\kappa=\Pr(D_i(1)>D_i(0))$ is unaffected by interference in outcomes, and the means in~\eqref{unif:eq-sep} are $\mu_z(\mathbf z_{-i})=\mathbbm{E}[Y_i(D_i(z),\mathbf D_{-i}(\mathbf z_{-i}))]$, so $\theta(\mathbf z_{-i})=\mathbbm{E}[\Delta_i(\mathbf D_{-i}(\mathbf z_{-i}))\mid D_i(1)>D_i(0)]$for the drawn unit, scaled under $P$ by its complier share relative to $\kappa$; the LADE is unaffected. With $\pi^\ast$ the marginal law of $\mathbf Z_{-i}$, part~(i) becomes $\mathrm{LADE}+\mathrm{Bias}_{\mathrm{IV}}$; \eqref{unif:eq-gcai} is Assumption~\ref{Instrument indepen}, and part~(ii) yields Theorem~\ref{prop:iv}(ii).
\item[(c)] \emph{Sharp regression discontinuity.} Suppose $\pi_{+}$ and $\pi_{-}$ exist. Under Assumption~\ref{rdd continuity} and homogeneous nonrandom direct effects, the means in~\eqref{unif:eq-sep} are $\mu_1-\mu_0=\theta$, with $\pi^\ast=\pi_+$, part~(i) becomes $\mathrm{ADE}_{+}+\mathrm{Bias}_{\mathrm{RDD}}$, the baseline channel of \eqref{unif:eq-channels} being $\mathrm{Bias}_{\mathrm{RDD}}$, which is Theorem~\ref{prop:rd}(i). Condition \eqref{unif:eq-gcai} states $\pi_{+}=\pi_{-}=:\pi_{w_0}$, which is Assumption~\ref{rdd nojump}, under which $\Phi_{\mathrm{ITR}}(P)=\mathrm{ADE}_{w_0}:=\sum_{\mathbf d_{-i}}\Delta(\mathbf d_{-i})\pi_{w_0}(\mathbf d_{-i})$, which is Theorem~\ref{prop:rd}(ii), and Assumption~\ref{running indepen} implies $\pi_{+}=\pi_{-}=\Pr(\mathbf D_{-i}=\cdot\,)$ when, under $P$, the densities of the running variables are continuous at $w_0$ with a common positive value there, which yields Theorem~\ref{prop:rd}(iii). If homogeneity is dropped but \eqref{rdd int: cont } and \eqref{rdd int: uncon} also hold with $Y_i(0,\mathbf d_{-i})$ replaced by $Y_i(1,\mathbf d_{-i})$, the same conclusions hold with $\theta(\mathbf d_{-i})=\lim_{w\to w_0}\mathbbm{E}[\Delta_i(\mathbf d_{-i})\mid W_i=w]$.
\item[(d)] \emph{Difference-in-differences.} Under Assumption~\ref{did: parallel trend} and condition \eqref{unif:did-meanindep-treated}, the means in~\eqref{unif:eq-sep} are the trends $\mu_d(\mathbf d_{-i};w)=\mathbbm{E}[Y_{i1}(d,\mathbf d_{-i})-Y_{i0}(\mathbf 0)\mid W_i=w,D_i=d]$, with $\theta(\mathbf d_{-i};w)$ the mean IDE of treated units at $(w,\mathbf d_{-i})$. With $\pi^\ast=\pi_1$, part~(i) becomes $\mathrm{ADTT}+\mathrm{Bias}_{\mathrm{DiD}}$ and the direct-effect  component is zero; condition~\eqref{eq:CAI} with equal composition~\eqref{unif:eq-ec} under $P$ implies \eqref{unif:eq-gcai}, and part~(ii) yields Theorem~\ref{thm:did}(ii).
\item[(e)] \emph{Dynamic difference-in-differences.} Under Assumptions~\ref{did: parallel trend dyn} and~\ref{did: no anticipation dyn} and condition \eqref{unif:dyn-meanindep-treated}, $\theta$ is as in (d), part~(i) with $\pi^\ast=\pi_1$ becomes $\mathrm{ADTT}+\mathrm{Bias}_{\mathrm{DiD_{dyn}}}$, and part~(ii), under \eqref{eq:CAI} with equal composition~\eqref{unif:eq-ec}, yields Theorem~\ref{thm:did_dynamic}(ii).
\end{itemize}
~\hyperref[prf:unif:cor-designs]{\textup{\footnotesize[Proof]}}\end{cor}

Three features of Corollary~\ref{unif:cor-designs} deserve comment. First, under IV, the relevant others' assignment variable is $\mathbf Z_{-i}$, not $\mathbf D_{-i}$: others' realized treatments are endogenous, and the IV assumptions deliver separation in others' instruments but, in general, not in their treatments. Because take-up depends only on the own instrument, $\kappa$ is unaffected by interference in outcomes, and the Wald ratio identifies the LADE of \citet{hoshino2024causal} when $\mathbf Z_{-i}\indep Z_i$.\label{unif:iv}

Second, for the RDD, condition~\eqref{unif:eq-gcai} is $\pi_+=\pi_-$, which is Assumption~\ref{rdd nojump} and is \textit{weaker} than independence of running variables across units (Assumption~\ref{running indepen}). Under it, the formula identifies $\mathrm{ADE}_{w_0}$, the average direct effect at the law prevailing at the cutoff, as Theorem~\ref{prop:rd}(ii) states; this is a different object from the ADE under the marginal law, with which it coincides when $\pi_{w_0}=\Pr(\mathbf D_{-i}=\cdot\,)$, as under Assumption~\ref{running indepen} with running-variable densities continuous at $w_0$ and equal and positive there. The corollary extends that theorem in one respect: under the same conditions on treated outcomes, the homogeneity restriction of Theorem~\ref{prop:rd} can be dropped.\label{unif:rdd}

Third, under DiD the reference law is the treated group's law $\pi_1$, so the direct-effect component vanishes by construction and the bias arises entirely through the baseline component, as untreated trends respond to others' treatments under different laws of $\mathbf D_{-i}$ across groups. The dynamic extension allowing for anticipation is given in Theorem~\ref{thm:did_dynamic}.\label{unif:did}

The CAI-type conditions imposed design by design are thus one condition applied to each design's own comparison, conditioning variable, and others' assignment variable (Table~\ref{unif:tab-mapping}), and the design-specific biases are the channels in \eqref{unif:eq-channels} at design-specific reference laws. Assessing a design under interference reduces to two questions: which assignments of others enter the slice means, and whether the compared cells share their law.

\subsection{Related Work}\label{supp:relwork}
\paragraph{The target.} The average direct effect of Section~\ref{unif:sec-setup} of the main text is analogous to the ``expected average treatment effect'' of \citet{savje2021average}, to the average direct effect of \citet[Def.~1]{hu2022average}, and, with others' treatments summarized by exposures, to that of \citet[Def.~2]{egami2021spillover}, all formulated with fixed rather than stochastic potential outcomes.

\paragraph{What is known, design by design.} Results of the same or similar form as Proposition~\ref{unif:thm-master} exist for each design, under additional structure; the paragraph ``Relation to the literature'' in Section~\ref{intro:known} of the main text gives the short version. In \emph{experiments}, the difference in means targets a contrast over the conditional law of others' treatments \citep{hudgens2008toward,savje2021average}, whose gap from an average direct effect vanishes under Bernoulli but not, in general, under fixed-count or complete randomization \citep{vanderweele2011effect,hu2022average}, and has been expressed exactly \citep{savje2021average,vazquezbare2023identification}. For \emph{selection on observables} with a known exposure mapping, \citet{forastiere2021identification} obtain identification under independence of own treatment and exposure, and a bias of the same form as ours otherwise; without an exposure mapping, \citet{leung2025identifying} obtain the same structure for exposure contrasts, with CAI as the case of own-treatment contrasts. For \emph{IV}, where the Wald estimand need not be the LATE \citep{sobel2006randomized}, it identifies a complier direct effect under independent encouragement or instruments \citep{kang2016peer,hoshino2024causal}, shares its limit with one as clusters grow \citep{imai2021causal}, and has explicit biases for instruments dependent by design \citep{vazquez2023causal,kang2018spillover}. For the \emph{RDD}, the within-window difference in means is unbiased for a Hudgens--Halloran direct effect under local randomization \citep{aronow2017regression}, and, with a known exposure mapping, the sharp RD contrast identifies boundary direct effects averaged over the law of exposure at the cutoff \citep{daltorrione2025regression}; a continuity condition of the same form as ours underlies RD aggregation \citep{borusyak2024regression}, and \citet{auerbach2024regression} study the estimation stage. For \emph{DiD} with a known exposure mapping, the contrast equals the direct effect on the treated when exposure has the same conditional law in both groups, with a product-form bias otherwise \citep{xu2023difference}; with the assignment held fixed, it equals the total effect on the treated minus the spillover on controls \citep{butts2021difference}, which is all it identifies under interference of unknown form \citep{mealli2025difference}, and \citet{clarke2017estimating} gives the bias in a linear model.

\paragraph{Modeling interference explicitly.} A large literature models interference explicitly to define and identify direct, spillover, and total effects, using partial interference \citep{halloran1995causal,hudgens2008toward,tchetgen2012causal}, randomized saturation designs \citep{baird2018optimal}, or exposure mappings on known networks \citep{aronow2017estimating,leung2020treatment,leung2022causal,forastiere2021identification}, with extensions to known regions under randomization, selection on observables, and DiD \citep{huber2021framework}, to synthetic control \citep{cao2019estimation}, and to IV \citep{vazquez2023causal}; \citet{aronow2021spillover} survey the experimental work. These methods recover spillover effects when network or group data exist. The question of this paper is what standard formulas recover when they do not, so these methods are complements.

\begin{remark}[Relation to existing results]\label{lit:rem-master}
For randomized experiments without covariates, special cases of Proposition~\ref{unif:thm-master} are known. Apply Proposition~\ref{unif:thm-master} with $Z_i$ and $\mathbf Z_{-i}$ in place of $D_i$ and $\mathbf D_{-i}$, for a randomized assignment vector $\mathbf Z$ with a common marginal probability, condition~\eqref{unif:eq-ec}. Then $\Phi_{\mathrm{ITR}}(P)$ is the average distributional shift effect (ADSE) of \citet[Def.~4]{savje2021average}, the first term in part~(i) with $\pi^\ast$ the marginal law of $\mathbf Z_{-i}$ is their expected average treatment effect (EATE, Def.~3), and $\mathcal B(\pi^\ast)$ is ADSE$\,-\,$EATE, which their supplement already states exactly, for any design and unknown interference, and shows to vanish under Bernoulli randomization \citep[Supplement~A, Cor.~A1 and Lemma~A8]{savje2021average}. Under partial interference, \citet[Thm.~1]{vazquezbare2023identification} express the difference in means as the direct effect at zero treated peers plus spillovers averaged over the laws of peers' assignments in the two groups, which is \eqref{unif:eq-bias} with the point mass at $\mathbf 0$ as reference law; \citet[Sec.~3]{hu2022average} note that the conditional contrast equals the ADE under Bernoulli designs but not, in general, under completely randomized or stratified designs, which are part~(ii) and its failure; related product-form biases appear in \citet[Sec.~4.3]{savje2024causal} and \citet[Thm.~2]{egami2021spillover}. What Proposition~\ref{unif:thm-master} adds is set out below. \qedsymbol
\end{remark}

\paragraph{The General decomposition and its antecedents.} Remark~\ref{lit:rem-master} documents the randomized, covariate-free antecedents of Proposition~\ref{unif:thm-master}; three further results are related. \citet[Props.~5--6]{savje2021average} show that, under restricted interference and design mixing, the Horvitz--Thompson and H\'ajek estimators are consistent for the EATE when the dependence between each unit's assignment and the assignments of the units that interfere with it vanishes on average (design separation, an asymptotic analog of generalized CAI), and for the ADSE without design separation. \cite{savje2024causal} notes that an expected exposure effect, a contrast of outcomes conditional on own exposure, mixes the exposure effect with a shift in the law of others' treatments, and that when the exposure is own treatment this shift is removed by assigning treatments independently across units. \citet[Thm.~2]{egami2021spillover} derives a bias of the same product form, an interference term times a difference between conditional laws, for spillover estimates that omit an unobserved network.

Relative to these results, Proposition~\ref{unif:thm-master} contributes three things. First, it carries the identity beyond randomized experiments without covariates. It is not tied to a known design; the cell laws $\pi_d$ are whatever the data-generating mechanism induces, separation comes from the identifying assumptions of each observational or quasi-experimental design, and the construction accommodates conditioning variables, one-sided limits, and a normalizing first stage. Second, the others' assignment variable is design-specific ($\mathbf Z_{-i}$ in the IV design). Third, the target is indexed by an arbitrary reference law, and the bias splits into a baseline and a direct-effect component, separating choice among ADE-type targets from confounding by assignment dependence. Like the exact identity of \citet{savje2021average}, and unlike the result of \citet{vazquezbare2023identification}, it needs no limit on the extent of interference. 

\paragraph{The price of unrestricted interference.} Where settings overlap, Proposition~\ref{unif:thm-master} shows that the population identities collected above are instances of one decomposition governed by one condition, stated without groups, networks, exposure mappings, or restricted-interference conditions, for observational as well as randomized assignment; the exact identity of \citet[Supplement~A, Cor.~A1]{savje2021average} is its randomized, covariate-free special case. The bias under assignment dependence is left open by \citet{hoshino2024causal} and \citet{daltorrione2025regression}, folded into the target by \citet{aronow2017regression}, and removed by holding the assignment fixed in \citet{mealli2025difference}. Allowing interference to be unrestricted shifts the identifying burden to the assignment mechanism. Generalized CAI therefore restricts the law of the full vector $\mathbf D_{-i}$, rather than an exposure \citep{forastiere2021identification,xu2023difference} or, asymptotically, the assignments of only those units that interfere with $i$ \citep{savje2021average}. For DiD, a direct-effect interpretation additionally requires parallel trends at every configuration of others' treatments, rather than only for outcomes with no treated or exposed units \citep{butts2021difference,mealli2025difference}. Some antecedents are also more general in other directions, allowing instruments to affect others' take-up under restricted-interference conditions \citep{imai2021causal,hoshino2024causal} or delivering estimation and inference under dependence \citep{savje2021average}, which our exact population identities do not address; our results complement them.

\paragraph{Sensitivity analysis.} The sensitivity analysis of Section~\ref{robustness} takes dependence among assignments as its sensitivity parameter. This sets it apart from analyses of unmeasured confounding \citep{rosenbaum2002observational,imbens2003sensitivity}, including under clustered assignment \citep{hansen2014clustered}, and from analyses of exclusion restrictions \citep{conley2012plausibly}: the ``confounder'' here is the vector of others' assignments. Existing sensitivity analyses under interference address other threats and hold that dependence fixed rather than relaxing it. \citet{vanderweele2014interference} derive bias formulas and sensitivity analyses for direct, spillover, and total effects under \emph{unmeasured confounding} with partial interference in known clusters. Their bias has the same product form as ours, but the imbalance is in an unmeasured confounder rather than in assignments. \citet[Thms.~3--4]{egami2021spillover} develop bias-correction and bounding analyses for spillover effects in experiments when a causally relevant network is \emph{unobserved}, using parameters that capture spillovers through that network and its overlap, or association, with the observed one. \citet{ortyashov2025sensitivity} assess the bias of a weighting estimator that ignores interference jointly with unmeasured confounding and lack of transportability, while maintaining conditional independence between own and neighbors' treatments given all their common causes. Each differs from ours in its setting, its target, or its sensitivity parameter, and to our knowledge none takes dependence in assignment across units as that parameter. Ours is a partial-identification analysis \citep{manski2003partial} with a breakdown frontier \citep{masten2020inference}, inheriting the product structure of omitted-variable bias. The assignment dependence measure $\delta$ is an average analog of the $c$-dependence of \citet{masten2018identification}, inference follows \citet{imbens2004confidence} and \citet{stoye2009more}, and the benchmarking of $\bar\Gamma$ follows \citet{altonji2005selection}, \citet{oster2019unobservable}, and \citet{cinelli2020making}. Neither a network nor an exposure mapping is specified; only the magnitudes $\Lambda$ or $L$, which play a role analogous to restrictions on parallel-trends violations in \citet{rambachan2023credible}.

\section{Primitive Conditions, Benchmarks, and Inference}\label{supp:primitive}
This section collects three results. Lemma~\ref{sensB:lem-order} states the primitive conditions under which monotone dependence (MD) holds. Remark~\ref{sensB:rem-benchmarks} provides closed form benchmarks for $\delta$, $\Gamma_j$ and $\bar\Gamma$ in canonical dependence structures. Proposition~\ref{sensB:prop-ci}, confidence intervals for the identified set and, following \citet{imbens2004confidence}, for the target. 

\paragraph{Sufficient conditions for monotone dependence.}
The following lemma shows that (MD) holds under two familiar forms of assignment dependence. It is satisfied when a common shock makes treatments move together across units and when treatment totals are fixed, as under complete randomization, which instead induces dependence in the opposite direction.
\begin{lemma}[Sufficient conditions for (MD)]\label{sensB:lem-order}
Let cell expectations be conditional expectations with $p(w)\in(0,1)$; here and in Remark~\ref{sensB:rem-benchmarks}, $\Gamma_j$, $\bar\Gamma$ and $\delta$ are the unit-level $\Gamma_{kj}$, $\bar\Gamma_k$ and $\delta_k$ of a fixed unit $k$, and under~\eqref{unif:eq-ec} the pooled $\Gamma_j$ is $\sum_{k\neq j}q(k;w)\Gamma_{kj}$ while the pooled $\bar\Gamma$ and $\delta$ are the $q(\cdot\,;w)$-averages of $\bar\Gamma_k$ and $\delta_k$.
\begin{itemize}
\item[(a)] \emph{Common shocks.} Suppose that, given $W_i=w$, there is a real random variable $\eta$ such that $D_i$ and the coordinates $D_1,\dots,D_J$ of $\mathbf D_{-i}$ are mutually independent given $\eta$, and $\Pr(D_i=1\mid\eta,W_i=w)$ and each $\Pr(D_j=1\mid\eta,W_i=w)$ are nondecreasing in $\eta$. Then (MD) holds with $\mathbf D^1\ge\mathbf D^0$, and $\Gamma_j(w)\ge0$ for all $j$; under the law $P$ of a randomly drawn unit, this requires in addition equal composition, condition~\eqref{unif:eq-ec} of the main text, since a coupling ordered coordinatewise keeps the drawn unit fixed.
\item[(b)] \emph{Exchangeable assignment.} Under (EX), (MD) holds if and only if the laws of $\sum_jD_j$ under $\pi_1(\cdot\,;w)$ and $\pi_0(\cdot\,;w)$ are ordered by first-order stochastic dominance. In particular, under complete randomization of a fixed number of the $N$ units, independently of covariates, (MD) holds with $\mathbf D^1\le\mathbf D^0$ and $\Gamma_j=-1/(N-1)$.
\end{itemize}
\end{lemma}

\begin{proof}
All probabilities are conditional on $W_i=w$, which is suppressed.

\emph{(a)} Let $r(\eta):=\Pr(D_i=1\mid\eta)$ and $p_j(\eta):=\Pr(D_j=1\mid\eta)$, taken as nondecreasing versions, and $p=\mathbbm{E}r(\eta)\in(0,1)$.\\ \emph{Step 1: the conditional laws of $\eta$ are ordered.} For $t\in\mathbbm{R}$ and $h_t(\eta):=\mathbbm 1\{\eta>t\}$,
\begin{align*}
\Pr(\eta>t\mid D_i=1)-\Pr(\eta>t\mid D_i=0)&=\frac{\mathbbm{E}[rh_t]}{p}-\frac{\mathbbm{E}[(1-r)h_t]}{1-p}\\
&=\frac{\mathbbm{E}[rh_t]-p\,\mathbbm{E}[h_t]}{p(1-p)}=\frac{\cov(r(\eta),h_t(\eta))}{p(1-p)}\ge0 .
\end{align*}
The covariance of two nondecreasing functions of a real random variable is nonnegative, ie., with $\eta'$ an independent copy, $2\cov(r(\eta),h_t(\eta))=\mathbbm{E}[\{r(\eta)-r(\eta')\}\{h_t(\eta)-h_t(\eta')\}]\ge0$, since the two factors have the same sign.\\ \emph{Step 2: coupling.} Let $F_d$ be the cdf of $\eta$ given $D_i=d$, and let $U_0,U_1,\dots,U_J$ be independent $\mathrm{Uniform}(0,1)$. Set $\eta^d:=F_d^{-1}(U_0)$ and $D^d_j:=\mathbbm 1\{U_j\le p_j(\eta^d)\}$. By Step 1, $F_1\le F_0$, so $\eta^1\ge\eta^0$ and, since each $p_j$ is nondecreasing, $\mathbf D^1\ge\mathbf D^0$ coordinatewise.\\
\emph{Step 3: marginals.} By mutual conditional independence of $D_i,D_1,\dots,D_J$ given $\eta$,
\[
\Pr(\mathbf D_{-i}=\mathbf d_{-i}\mid D_i=d)=\mathbbm{E}\Big[\prod_jp_j(\eta)^{d_j}\{1-p_j(\eta)\}^{1-d_j}\,\Big|\,D_i=d\Big],
\]
and the construction gives $\Pr(\mathbf D^d=\mathbf d_{-i})=\mathbbm{E}[\prod_jp_j(\eta^d)^{d_j}\{1-p_j(\eta^d)\}^{1-d_j}]$ with $\eta^d\sim F_d$; the two coincide. Hence $\mathbf D^d\sim\pi_d$, (MD) holds, and $\Gamma_j=\mathbbm{E}D^1_j-\mathbbm{E}D^0_j\ge0$.

\emph{(b)} If (MD) holds with $\mathbf D^1\ge\mathbf D^0$, then $\sum_jD^1_j\ge\sum_jD^0_j$ surely, so the count laws are ordered. Conversely, if $F_{\pi_1}\le F_{\pi_0}$ for the count cdfs, Lemma~\ref{sensB:lem-exch} with the quantile coupling yields $\mathbf A\sim\pi_1$, $\mathbf A'\sim\pi_0$ with $\mathbf A\ge\mathbf A'$; the reverse case is symmetric. Under complete randomization of $m$ of $N$ units independent of covariates, $\pi_1$ and $\pi_0$ are uniform on the vectors with $m-1$ and $m$ ones, respectively, which are exchangeable with counts $m-1<m$. Hence $\mathbf D^1\le\mathbf D^0$ and $\Gamma_j=(m-1)/(N-1)-m/(N-1)=-1/(N-1)$.
\end{proof}

\paragraph{Benchmarks.}
As mentioned in the main text, the dependence measures $\delta$ and $\max_j|\Gamma_j|$ can remain large even when dependence is confined to a small subset of units, whereas $\bar\Gamma$ more accurately reflects how widespread that dependence is across the population. Remark~\ref{sensB:rem-benchmarks} makes this distinction concrete by providing benchmark values of $\Gamma_j$, $\bar\Gamma$, and $\delta$ under common shocks, stratified assignment, and hidden-group take-up.
\begin{remark}[Benchmark dependence structures]\label{sensB:rem-benchmarks}
\begin{itemize}
    \item[(a)] \emph{Common-shock dependence.} When treatment assignments respond to a common shock in the same direction across units, Lemma~\ref{sensB:lem-order}(a) implies $\Gamma_j\ge0$ for all $j$, with the magnitude of $\Gamma_j$ reflecting the strength of the induced pairwise assignment dependence. A shared shock alone does not suffice: if it raises take-up for some units and depresses it for others, $\Gamma_j$ is negative for any pair loading on it with opposite signs. The common-shock designs C and D of Section~\ref{sec:montecarlo} are special cases, with the shock playing the role of $\eta$ when $\sigma\ge0$. Design~B, in which all ZIP codes are taken together, is not. Under the structure of Lemma~\ref{sensB:lem-order}(a), two nondecreasing propensities have nonnegative covariance, zero only if one of them is almost surely constant, so $\Gamma_j=0$ for every unit outside a unit's own ZIP would force $\Gamma_j=0$ within it as well. That design is an instance of part~(c). Because the slice means are affine in the leave-one-out treated share, Theorem~\ref{sensB:prop-refine}(iii) with $\lambda=p$ then yields the exact bias $\mathbbm{E}[\{\Lambda_0+(1-p(W_i))\tau_\theta\}\{\bar s_1(W_i)-\bar s_0(W_i)\}]$ reported there, where $\Lambda_0$ and $\tau_\theta$ are the spillover slopes of that design and $\bar s_d(w):=\mathbbm{E}[S(\mathbf D_{-i})\mid D_i=d,W_i=w]$. The identity is exact wherever the truncation of that design leaves the slice means affine.
    \item[(b)] \emph{Stratified dependence.} Suppose take-up is exchangeable within strata of size $n_s$ and independent across strata. If the treated count within a stratum has variance $\xi$ times its binomial variance, then $\xi<1$ corresponds to negative dependence induced by treatment-count restrictions, $\xi=1$ to the binomial benchmark, and $\xi>1$ to positive within-stratum dependence.\footnote{Not every $\xi$ is attainable at a given $(n_s,e)$. The count is integer valued with mean $n_se$, so its variance is at least $f(1-f)$ with $f$ the fractional part of $n_se$, and at most $n_s^2e(1-e)$; hence $f(1-f)/\{n_se(1-e)\}\le\xi\le n_s$, and $\xi=0$ requires $n_se$ to be an integer.} In each case, $\Gamma_j=(\xi-1)/(n_s-1)$
for units $j$ in unit $i$'s stratum and $\Gamma_j=0$ otherwise. Thus, the strength of pairwise assignment dependence is governed by how far the treated-count variance deviates from its binomial benchmark and is diluted by stratum size. Remark~\ref{sensB:rem-block} then gives the corresponding bound $\kappa^{-1}\int L\sum_kq(k;w)\mathrm W_{S_{b^\ast},k}\,dG$ based on the stratum share.
\item[(c)] \emph{Hidden-group dependence.} Suppose units belong to unobserved groups and take up treatment jointly within groups, and let $|B_i|$ be the size of unit $i$'s group. If the groups take up independently, dependence is perfect among group mates and absent otherwise, so $\Gamma_j=1$ for group mates and $\Gamma_j=0$ for all other units; independence alone suffices, the groups needing neither a common take-up probability nor a common size. Hence for every unit with $|B_i|\ge2$, $\max_j\Gamma_j=1$ and $\delta=1$, while $\bar\Gamma=(|B_i|-1)/(N-1)$ records how concentrated that dependence is for that unit. Design~B of Section~\ref{sec:montecarlo} is of this form, with $|B_i|$ a unit's ZIP size: its $98$ ZIPs have $98$ distinct take-up probabilities and sizes from $1$ to $72$, so $\bar\Gamma$ runs from $0$ to $0.104$ across units, with median $0.025$ and across-unit average $0.046$, the value reported there; $\max_j\Gamma_j=\delta=1$ for the $658$ units who share a ZIP and both are $0$ for the $25$ alone in theirs. More generally, if a stratum contains $J_s\ge2$ groups of a common size $|B|\ge2$ taking up with a common probability, so that the group indicators are exchangeable, and the treated-group count has variance $\xi_J$ times its binomial value, part~(b) applies at the group level. (With $J_s=1$ there are no non-mates in the stratum and the group count is Bernoulli, so $\xi_J=1$ necessarily.) Then $\Gamma_j=(\xi_J-1)/(J_s-1)$ for non-group mates in the same stratum and
\[
\bar\Gamma=\frac{(|B|-1)+|B|(\xi_J-1)}{N-1},
\]
while $\delta=1$ and $\max_j\Gamma_j=1$ remain maximal. For $\xi_J<1$ the mate and non-mate values differ in sign, so (MD) fails and the share laws are in general not ordered; Theorem~\ref{sensB:prop-refine}(iv) then does not apply, and $\bar\Gamma$ must not be converted into a share bound. With $|B|=2$, $J_s=2$ and the group-count law $(\tfrac18,\tfrac34,\tfrac18)$, so that $\xi_J=\tfrac12$, the two contributions cancel exactly and $\bar\Gamma=0$, while $\mathrm W_S=\tfrac13$.
\end{itemize}
 \qedsymbol
\end{remark}

\begin{proof}[Proof of the claims in Remark~\ref{sensB:rem-benchmarks}]
For binary $X,Y$ with $\Pr(X=1)=p\in(0,1)$, $\Pr(Y=1\mid X=1)-\Pr(Y=1\mid X=0)=\mathbbm{E}[XY]/p-\{\mathbbm{E}Y-\mathbbm{E}[XY]\}/(1-p)=\cov(X,Y)/\{p(1-p)\}$. This gives the regression-coefficient interpretation of $\Gamma_j$.

\emph{(a)} In the common-shock designs of Section~\ref{sec:montecarlo}, given $W_i=w$ and the shock $\eta$, $D_i=\mathbbm 1\{U_i\le\operatorname{expit}(W_i^\top\hat\beta+\sigma\eta)\}$ and $D_j=\mathbbm 1\{U_j\le\operatorname{expit}(W_j^\top\hat\beta+\sigma\eta)\}$, where $\operatorname{expit}(x)=1/\{1+\exp(-x)\}$, the covariates are held fixed at their sample values, and the $U_j$ are i.i.d.\ and independent of $\eta$. Hence $D_i,D_j$ ($j\ne i$) are mutually independent given $\eta$, with $\Pr(D_i=1\mid\eta)=\operatorname{expit}(W_i^\top\hat\beta+\sigma\eta)$ and $\Pr(D_j=1\mid\eta)=\operatorname{expit}(W_j^\top\hat\beta+\sigma\eta)$ nondecreasing in $\eta$ for $\sigma\ge0$, so Lemma~\ref{sensB:lem-order}(a) applies. The slice means $\mu_0(\mathbf d_{-i};w)=\hat q(w)+\Lambda_0\{S(\mathbf d_{-i})-\bar s\}$ and $\theta(\mathbf d_{-i};w)=\tau_0+\tau_\theta\{S(\mathbf d_{-i})-\bar s\}$ give $m_p=\mu_0+(1-p)\theta$, affine in $S$ with slope $\Lambda_0+(1-p(w))\tau_\theta$. Theorem~\ref{sensB:prop-refine}(iii), with $\kappa=1$, $G=F_W$, and $\bar\Gamma(w)=\bar s_1(w)-\bar s_0(w)$, gives the stated bias.

\emph{(b)} With exchangeable take-up within stratum $s$ of size $n_s\ge2$, marginal probability $e$, and $\var(M_s)=\xi n_se(1-e)$, exchangeability gives $\cov(D_i,D_j)=(\xi-1)e(1-e)/(n_s-1)$ for any $i\ne j$ in the stratum. By the identity above, $\Gamma_j=(\xi-1)/(n_s-1)$. For $j$ outside the stratum, independence across strata gives $\Gamma_j=0$.

\emph{(c)} If $i$'s hidden group of size $|B_i|\ge2$ is treated jointly, then given $D_i=1$ all $|B_i|-1$ group mates are treated and given $D_i=0$ none is. Hence $\Gamma_j=1$ for mates, and $\delta=1$ because the event $\{\text{all mates treated}\}$ has probability one in one cell and zero in the other. If the groups take up independently, and strata are independent, then conditioning on $D_i$ leaves the law of every non-mate unchanged whatever its take-up probability, so $\Gamma_j=0$ for $j$ outside $i$'s group and $\bar\Gamma=(|B_i|-1)/(N-1)$; moreover $\mathbf D^1\ge\mathbf D^0$ under the coupling that shares the assignments of non-mates. If instead the groups of a stratum are merely exchangeable, $(D_B)_{B\in\Pi^0_s}$ is an exchangeable binary vector of length $J_s$, so part~(b) applied to the groups gives $\Pr(D_B=1\mid D_{B_i}=1)-\Pr(D_B=1\mid D_{B_i}=0)=(\xi_J-1)/(J_s-1)$ for the other groups $B$ of stratum $s$, where $\var(J^1_s)=\xi_JJ_se_s(1-e_s)$; since $D_j=D_B$ for $j\in B$, the same value is $\Gamma_j$ for each of the $|B|(J_s-1)=n_s-|B|$ units in those groups. Therefore $\sum_j\Gamma_j=(|B|-1)+(n_s-|B|)(\xi_J-1)/(J_s-1)=(|B|-1)+|B|(\xi_J-1)$ and $\bar\Gamma=\{(|B|-1)+|B|(\xi_J-1)\}/(N-1)$, which is $(|B|-1)/(N-1)$ at $\xi_J=1$, $-1/(N-1)$ at $\xi_J=0$ (a fixed number of treated groups), and $(n_s-1)/(N-1)$ at $\xi_J=J_s$ (the whole stratum takes up together).
\end{proof}

\paragraph{Inference.} The inference results in the following proposition adapt standard methods from the partial-identification literature to our bounds. Parts~(i)--(ii) apply familiar confidence-set constructions for interval-identified parameters, while part~(iii) combines the sensitivity bound with weak-IV-robust inversion. 

\begin{proposition}[Confidence intervals for the identified set and the target]\label{sensB:prop-ci}
Fix $\alpha\in(0,1/2)$ and a known $\bar b\ge0$ with $T:=T(\pi_\lambda)\in[\Phi-\bar b,\Phi+\bar b]$, where $\Phi:=\Phi_{\mathrm{ITR}}(P)$. Suppose $Z_n:=(\widehat\Phi_n-\Phi)/\widehat{\mathrm{se}}_n\to_d\mathcal N(0,1)$ and $\widehat{\mathrm{se}}_n\to_p0$ along the sequence of populations of Remark~\ref{supp:rem-design}. Let $c(x)$ be the unique solution of $\Phi_{\mathcal N}(c+2x)-\Phi_{\mathcal N}(-c)=1-\alpha$, and $\hat c_n:=c(\bar b/\widehat{\mathrm{se}}_n)$.
\begin{itemize}
\item[(i)] $\Pr\big([\Phi-\bar b,\Phi+\bar b]\subseteq[\widehat\Phi_n-\bar b-z_{1-\alpha/2}\widehat{\mathrm{se}}_n,\ \widehat\Phi_n+\bar b+z_{1-\alpha/2}\widehat{\mathrm{se}}_n]\big)\to1-\alpha$.
\item[(ii)] $z_{1-\alpha}\le c(x)\le z_{1-\alpha/2}$. If $Z_n\sim\mathcal N(0,1)$ exactly and $\widehat{\mathrm{se}}_n$ is nonrandom, then $\Pr(T\in\mathrm{CI}^{\mathrm{IM}}_n)\ge1-\alpha$ for every $T\in[\Phi-\bar b,\Phi+\bar b]$, with equality at the endpoints, where $\mathrm{CI}^{\mathrm{IM}}_n:=[\widehat\Phi_n-\bar b-\hat c_n\widehat{\mathrm{se}}_n,\ \widehat\Phi_n+\bar b+\hat c_n\widehat{\mathrm{se}}_n]$. In general, $\liminf_n\Pr(T\in\mathrm{CI}^{\mathrm{IM}}_n)\ge1-\alpha$ for every fixed $\bar b\ge0$ and $T$ in the identified set.
\item[(iii)] (IV.) Let $\mathrm{RF}$ be the reduced form and $\kappa$ the first stage, and let $\bar b_{\mathrm{RF}}:=\bar\delta\{(1-p_Z)\bar\Lambda_1+p_Z\bar\Lambda_0\}$ with constant parameters in version~1 of the sensitivity model. If $(\widehat{\mathrm{RF}}-\hat\kappa T-(\mathrm{RF}-\kappa T))/\widehat{\mathrm{se}}_T\to_d\mathcal N(0,1)$, then $\mathrm{CS}:=\{t:|\widehat{\mathrm{RF}}-\hat\kappa t|\le\bar b_{\mathrm{RF}}+z_{1-\alpha/2}\widehat{\mathrm{se}}_t\}$ satisfies $\liminf\Pr(T\in\mathrm{CS})\ge1-\alpha$, whatever the value of $\kappa>0$.
\end{itemize}
\end{proposition}

\begin{proof}
\emph{(i)} Let
\[
\mathcal I:=[\Phi-\bar b,\ \Phi+\bar b],\qquad
\mathrm{CI}_n:=\big[\widehat\Phi_n-\bar b-z_{1-\alpha/2}\widehat{\mathrm{se}}_n,\ \ \widehat\Phi_n+\bar b+z_{1-\alpha/2}\widehat{\mathrm{se}}_n\big].
\]
Since $\bar b\ge0$, both are nonempty closed intervals, so $\mathcal I\subseteq\mathrm{CI}_n$ holds if and only if the left endpoint of $\mathrm{CI}_n$ is no larger than that of $\mathcal I$ and its right endpoint is no smaller, that is, if and only if
\[
\widehat\Phi_n-\bar b-z_{1-\alpha/2}\widehat{\mathrm{se}}_n\ \le\ \Phi-\bar b
\qquad\text{and}\qquad
\Phi+\bar b\ \le\ \widehat\Phi_n+\bar b+z_{1-\alpha/2}\widehat{\mathrm{se}}_n .
\]
The half-width $\bar b$ enters both sides of each inequality and cancels, leaving
\[
\widehat\Phi_n-\Phi\ \le\ z_{1-\alpha/2}\widehat{\mathrm{se}}_n
\qquad\text{and}\qquad
-\big(\widehat\Phi_n-\Phi\big)\ \le\ z_{1-\alpha/2}\widehat{\mathrm{se}}_n .
\]
Dividing by $\widehat{\mathrm{se}}_n>0$, which is positive because it appears in the denominator of $Z_n$, and recalling $Z_n=(\widehat\Phi_n-\Phi)/\widehat{\mathrm{se}}_n$, the two inequalities are $Z_n\le z_{1-\alpha/2}$ and $-Z_n\le z_{1-\alpha/2}$. Therefore
\[
\big\{\mathcal I\subseteq\mathrm{CI}_n\big\}=\big\{|Z_n|\le z_{1-\alpha/2}\big\},
\]
an identity of events, holding for every $n$ and every $\bar b\ge0$. The interval $[-z_{1-\alpha/2},z_{1-\alpha/2}]$ has boundary $\{\pm z_{1-\alpha/2}\}$, to which the standard normal law assigns probability zero, so $Z_n\to_d\mathcal N(0,1)$ and the portmanteau theorem gives
\[
\Pr\big(\mathcal I\subseteq\mathrm{CI}_n\big)=\Pr\big(|Z_n|\le z_{1-\alpha/2}\big)
\ \longrightarrow\ \Phi_{\mathcal N}(z_{1-\alpha/2})-\Phi_{\mathcal N}(-z_{1-\alpha/2})=1-\alpha .
\]
The limit does not depend on $\bar b$, i.e., the half-width enters the two intervals identically, so covering the identified set is exactly as demanding as covering $\Phi$ itself.

\emph{(ii)} The function $F(c,x):=\Phi_{\mathcal N}(c+2x)-\Phi_{\mathcal N}(-c)$ is continuous and strictly increasing in $c\ge0$, $F(0,x)=\Phi_{\mathcal N}(2x)-\frac12<1-\alpha$, and $F(c,x)\to1$ as $c\to\infty$, so $c(x)$ exists and is unique. Moreover, $F(z_{1-\alpha},x)\le1-\alpha\le F(z_{1-\alpha/2},x)$ gives the bounds, and $F(c,x)\to\Phi_{\mathcal N}(c)$ uniformly in $c$ as $x\to\infty$, so $c(x)\to z_{1-\alpha}$. Let $T=\Phi-u$ with $u\in[-\bar b,\bar b]$. Then $T\in\mathrm{CI}^{\mathrm{IM}}_n$ if and only if $-\hat c_n-(\bar b+u)/\widehat{\mathrm{se}}_n\le Z_n\le\hat c_n+(\bar b-u)/\widehat{\mathrm{se}}_n$. In the exact case, with $x:=\bar b/\mathrm{se}$ and $c=c(x)$, the coverage is $h(u):=\Phi_{\mathcal N}(c+(\bar b-u)/\mathrm{se})-\Phi_{\mathcal N}(-c-(\bar b+u)/\mathrm{se})$. Its derivative is $h'(u)=\{\phi(c+(\bar b+u)/\mathrm{se})-\phi(c+(\bar b-u)/\mathrm{se})\}/\mathrm{se}$, which is positive for $u<0$ and negative for $u>0$, because both arguments are nonnegative and $\phi$ decreases on $[0,\infty)$. Hence $h$ is minimized at $u=\pm\bar b$, where $h(\pm\bar b)=F(c,x)=1-\alpha$. For the asymptotic claim: if $\bar b=0$, then $\hat c_n=z_{1-\alpha/2}$ and the coverage is $\Pr(|Z_n|\le z_{1-\alpha/2})\to1-\alpha$. If $\bar b>0$, then $\bar b/\widehat{\mathrm{se}}_n\to_p\infty$ and $\hat c_n\to_pz_{1-\alpha}$. For $|u|<\bar b$ both thresholds diverge to $\mp\infty$ in probability, so the coverage tends to one. For $u=\bar b$ the coverage is $\Pr(-\hat c_n-2\bar b/\widehat{\mathrm{se}}_n\le Z_n\le\hat c_n)\to\Phi_{\mathcal N}(z_{1-\alpha})=1-\alpha$ by Slutsky's lemma, and $u=-\bar b$ is symmetric.

\emph{(iii)} Since $\Phi=\mathrm{RF}/\kappa$, $\mathrm{RF}-\kappa T=\kappa\mathcal B(\pi_\lambda)$. By Lemma~\ref{sensB:lem-mix} and Theorem~\ref{sensB:thm-bounds}(ii) (no conditioning variable), $|\kappa\mathcal B(\pi_\lambda)|\le\delta\,R_{\mathcal D}(m_{p_Z})\le\bar b_{\mathrm{RF}}$, where $\mathcal D:=\{\mathbf z_{-i}:\pi_1(\mathbf z_{-i})+\pi_0(\mathbf z_{-i})>0\}$ is the union of the supports of the two cell laws, as Lemma~\ref{sensB:lem-tvosc} requires. Hence $\{|(\widehat{\mathrm{RF}}-\hat\kappa T)-(\mathrm{RF}-\kappa T)|\le z_{1-\alpha/2}\widehat{\mathrm{se}}_T\}\subseteq\{T\in\mathrm{CS}\}$, and the probability of the first event tends to $1-\alpha$. No step uses a lower bound on $\kappa$.
\end{proof}

Note that we do not establish uniform (over data-generating processes) asymptotic validity of the Imbens--Manski interval. Furthermore,  we do not claim sharpness of $\min\{b^{(1)}_\lambda,b^{(2)}_\lambda\}$ when both interference restrictions are imposed together, nor sharpness of Theorem~\ref{sensB:prop-refine}(ii) for arbitrary unequal weights without ordering, nor sharpness once further a priori information (bounded outcome support, design structure) is imposed, which can only shrink the identified set. All confidence statements presume a standard error that is valid under the dependence and interference the model allows.

\begin{remark} [Standard errors for inference]\label{sens:rem_limits}
The confidence statements presume a valid standard error under the interference and dependence allowed, which ITR variance formulas do not deliver. One must instead use a dependence-robust approach, for example under restricted interference \citep{savje2021average,leung2022causal}, or cluster-robust estimation or a cluster bootstrap when the dependence is confined to a known grouping. Section~\ref{sen appl} takes the second route, resampling ZIP codes with replacement $2{,}000$ times and re-estimating the propensity score, the trimming and the strata within each draw, which raises the standard error by about a tenth ($103$ ZIP codes in Tuscany and $175$ in Sicily; resampling within province, which respects the province-stratified sampling, gives $0.036$ and $0.045$ against $0.037$ and $0.044$). The reach of that remedy is exactly the grouping it is built on: the increase of a tenth is evidence about dependence through ZIP codes and, as design~C of Table~\ref{tab:mc-sens} shows, none about region-wide shocks. The Imbens--Manski coverage is moreover pointwise in the data-generating process, not uniform. 

\qedsymbol{}
\end{remark}

\section{Implementation Guide}\label{sensB:impl}
\begingroup\sloppy

This section turns Theorems~\ref{sensB:thm-bounds} and~\ref{sensB:prop-refine} and Proposition~\ref{sensB:prop-ci} into a reporting procedure and provides calibration steps.

\subsection{Inputs and Procedure}

The inputs are an estimate $\widehat\Phi$ of the ITR formula with a standard error valid under the dependence allowed, the design, the level, and grids for the sensitivity parameters. Table~\ref{sensB:tab-designs} gives the design-by-design form of the bounds discussed in Section~\ref{robustness}, obtained by substituting into Theorem~\ref{sensB:thm-bounds}(ii) and Theorem~\ref{sensB:prop-refine} each design's own comparison, conditioning variable, and others' assignment variable (Table~\ref{unif:tab-mapping}) and the slice means that Corollary~\ref{unif:cor-designs} supplies for each design. Procedure~\ref{sensB:alg} turns the inputs into the reported output.

\begin{table}[!htbp]
\centering
\caption{Bounds on the assignment-dependence bias by design}
\label{sensB:tab-designs}
\scriptsize
\setlength{\tabcolsep}{3pt}
\renewcommand{\arraystretch}{1.05}
\begin{tabular}{>{\raggedright\arraybackslash}p{1.35cm}>{\raggedright\arraybackslash}p{2.6cm}>{\raggedright\arraybackslash}p{3.1cm}>{\raggedright\arraybackslash}p{3.9cm}>{\raggedright\arraybackslash}p{3.95cm}}
\toprule
 & Target; $\lambda$ & Dependence & TV bound $\bar b^{(2)}_\lambda$ & Share bound, (AN) and same-direction ordered laws\\
\midrule
SOO & ADE; $p(W_i)$ & $\delta(w)$: $\mathbf D_{-i}$ given $D_i=1,0$, $W_i=w$ & $\mathbbm{E}[\delta(W_i)\{\Lambda_0+(1-p(W_i))\Lambda_\theta\}]$ & $\mathbbm{E}[\{L^S_0+(1-p(W_i))L^S_\theta\}|\bar\Gamma(W_i)|]$\\
SOO, DiD & ADTT; $1$, $G=F_{W\mid D=1}$ & as above & $\mathbbm{E}[\delta(W_i)\Lambda_0\mid D_i=1]$ & $\mathbbm{E}[L^S_0|\bar\Gamma(W_i)|\mid D_i=1]$\\
IV & LADE; $p_Z=\Pr(Z_i=1)$ & $\delta_Z$: $\mathbf Z_{-i}$ given $Z_i=1,0$ & $\delta_Z\{\Lambda_0/\kappa+(1-p_Z)\Lambda_\theta\}$ & $\{L^S_0/\kappa+(1-p_Z)L^S_\theta\}|\bar\Gamma_Z|$\\
 RDD & ADE at $\pi_\lambda$; e.g., $1$ & $\delta_{w_0}:=d_{\mathrm{TV}}(\pi_+,\pi_-)$ & $\delta_{w_0}\{\Lambda_0+(1-\lambda)\Lambda_\theta\}$ & $\{L^S_0+(1-\lambda)L^S_\theta\}|\bar\Gamma_{w_0}|$\\
\bottomrule
\end{tabular}
\begin{flushleft}
\footnotesize
\emph{Note:} $\Lambda_0,\Lambda_\theta$ (Definition~\ref{sensB:def-params}) and the share-Lipschitz constants $L^S_0,L^S_\theta$ of $\mu_0,\theta$ are evaluated at the slice means of Corollary~\ref{unif:cor-designs}; $\Lambda_0,L^S_0$ are in outcome units, $\Lambda_\theta,L^S_\theta$ in effect units. In IV, $\mu_0(\mathbf z_{-i})=\mathbbm{E}[Y_i(D_i(0),\mathbf D_{-i}(\mathbf z_{-i}))]$, so $\Lambda_0$ bounds reduced-form spillovers of others' instruments, $\kappa$ is the first stage, and $\bar\Gamma_Z$ is the gap in the encouraged share; in DiD, $\Lambda_0$ bounds spillovers on untreated trends. In RDD, $\bar\Gamma_{w_0}$ is the jump in the treated share at $w_0$, and the target $\mathbbm{E}[\Delta(\mathbf D_{-i})]$ of Theorem~\ref{prop:rd} adds a \emph{difference} of two ADE-type targets, to which Theorem~\ref{sensB:thm-bounds}(i) applies one one-sided limit at a time, the two bounds being combined by the triangle inequality; that bound is therefore conservative and is not claimed sharp. Without same-direction ordering, replace $|\bar\Gamma|$ by $\bar{\mathrm W}_S=\sum_kq(k;w)\mathrm W_{S,k}$, or by the pooled $\mathrm W_S$ when the share response is common across units. ``ADTT'' is used only with $G=F_{W\mid D=1}$.
\end{flushleft}
\end{table}

\begin{algorithm}[!htbp]
\caption{Bounds on the ADE-type estimand under violations of generalized CAI}\label{sensB:alg}
\KwData{Estimate $\widehat\Phi$ of the ITR formula and a standard error $\widehat{\mathrm{se}}$ valid under the dependence allowed; design (SOO, IV, sharp RDD, DiD); level $\alpha$; null value $t_0$ (default $0$); grids for $\bar\delta$ (or $\bar\Gamma$) and for $\bar\Lambda$ (or $L$).}
\KwResult{Estimated \textit{identified sets}, \textit{confidence intervals}, \textit{breakdown frontier}, \textit{benchmarked robustness multiples}.}
\vspace{0.3cm}
\textbf{1. Target.} Choose $\lambda$ and state the target $T(\pi_\lambda)$ (Table~\ref{sensB:tab-designs}): ADE ($\lambda=p(W_i)$), ADTT ($\lambda=1$, $G=F_{W\mid D=1}$), LADE ($\lambda=p_Z$), or the RDD cutoff ADE at $\pi_\lambda$.\;\vspace{0.3cm}
\textbf{2. Parametrization.} Choose the TV version (Theorem~\ref{sensB:thm-bounds}, parameters $\bar\delta,\bar\Lambda_0,\bar\Lambda_\theta$), or the share version (Theorem~\ref{sensB:prop-refine}, parameters $\bar\Gamma,L^S_0, L^S_\theta$), stating the added restrictions (AN) or (BI), and (MD) or (EX).\;\vspace{0.3cm}
\textbf{3. Units.} Express $\bar\Lambda_0$ and $L^S_0$ as multiples $\ell_0$ of $\widehat{\mathrm{SD}}(Y_i)$, and $\bar\Lambda_\theta$ and $L^S_\theta$ as multiples $\ell_\theta$ of $|\widehat\Phi|$. Form $\bar\Lambda_{\mathrm{eff}}:=\kappa^{-1}\bar\Lambda_0+(1-\bar\lambda)\bar\Lambda_\theta$ with $\bar\lambda=\int\lambda\,dG$, estimated by the sample mean of $\lambda$ over the $G$-population.\;
\vspace{0.3cm}
\textbf{4. Bounds.} For each grid point compute $\bar b=\bar\delta\,\bar\Lambda_{\mathrm{eff}}$ (or $|\bar\Gamma|\,\bar L_{\mathrm{eff}}$ with $\bar L_{\mathrm{eff}}:=\kappa^{-1}L^S_0+(1-\bar\lambda)L^S_\theta$) and the estimated identified set $[\widehat\Phi-\bar b,\widehat\Phi+\bar b]$.\;
\vspace{0.3cm}
\textbf{5. Inference.} Report $[\widehat\Phi-\bar b-z_{1-\alpha/2}\widehat{\mathrm{se}},\widehat\Phi+\bar b+z_{1-\alpha/2}\widehat{\mathrm{se}}]$ for the set, and the Imbens--Manski interval $[\widehat\Phi-\bar b-\hat c\,\widehat{\mathrm{se}},\widehat\Phi+\bar b+\hat c\,\widehat{\mathrm{se}}]$ with $\Phi_{\mathcal N}(\hat c+2\bar b/\widehat{\mathrm{se}})-\Phi_{\mathcal N}(-\hat c)=1-\alpha$. In the IV design, also report the Anderson--Rubin-type set $\{t:|\widehat{\mathrm{RF}}-\hat\kappa t|\le\bar b_{\mathrm{RF}}+z_{1-\alpha/2}\widehat{\mathrm{se}}_t\}$ (Proposition~\ref{sensB:prop-ci}(iii)).\;
\vspace{0.3cm}
\textbf{6. Breakdown.} Report $\delta^\ast(\bar\Lambda)=\min\{1,(\widehat\Phi-t_0)_+/\bar\Lambda_{\mathrm{eff}}\}$ and $\delta^\ast_\alpha(\bar\Lambda)=\min\{1,(\widehat\Phi-z_{1-\alpha}\widehat{\mathrm{se}}-t_0)_+/\bar\Lambda_{\mathrm{eff}}\}$ over the $\bar\Lambda$ grid, and plot the frontier $\bar\delta\,\bar\Lambda_{\mathrm{eff}}=\widehat\Phi-t_0$. For $\widehat\Phi<t_0$, interchange the roles of the endpoints.\;
\vspace{0.3cm}
\textbf{7. Benchmark.} For each observed grouping compute $\widehat\gamma^{\mathcal G}$ (below) and the multiples $k^\ast_{\mathcal G}=\delta^\ast_{S}/|\widehat\gamma^{\mathcal G}|$ and $k^\ast_{\mathcal G,\alpha}=\delta^\ast_{S,\alpha}/|\widehat\gamma^{\mathcal G}|$, where $\delta^\ast_{S}$ and $\delta^\ast_{S,\alpha}$ are the breakdown points of step~6 with $\bar\Lambda_{\mathrm{eff}}$ replaced by $\bar L_{\mathrm{eff}}$, as in Section~\ref{robustness}.\;
\vspace{0.3cm}
\textbf{8. Report.} A table with $\widehat\Phi$, $\widehat{\mathrm{se}}$, the target, the chosen $\ell_0,\ell_\theta$, the identified sets and intervals at benchmark values, $\delta^\ast$, $\delta^\ast_\alpha$, and $k^\ast$, and a frontier plot with benchmark points overlaid.\;
\end{algorithm}

\subsection{Computational Details}

\paragraph{Calibration.}\label{sensB:impl-cal} Scale $\bar\Lambda_0$ in outcome standard deviations and $\bar\Lambda_\theta$ in multiples of $|\widehat\Phi|$. With $\kappa=1$ and $\bar\Lambda_{\mathrm{eff}}=\ell\cdot\mathrm{SD}(Y_i)$, the breakdown point $\delta^\ast$ is the standardized estimate divided by $\ell$, and $\ell$ is disciplined by external evidence on spillovers relative to direct effects \citep[e.g.,][]{crepon2013labor}. Reporting the multiple $k^\ast=\delta^\ast_S/|\widehat\gamma^{\mathcal G}|$ of an observed dependence at which a conclusion breaks down follows \citet{altonji2005selection}, \citet{oster2019unobservable}, and \citet{cinelli2020making}. The calibration proceeds in six steps.

\begin{enumerate}
\item \emph{Leave-one-out shares.} For a grouping $\mathcal G$ observed in the data (region, local labor market, school, cohort, or propensity-score stratum), with $\mathcal G(i)$ the group of unit $i$ and $n_{\mathcal G(i)}$ its size, compute $\widehat S^{\mathcal G}_i:=(n_{\mathcal G(i)}-1)^{-1}\sum_{j\in\mathcal G(i),j\ne i}D_j$, the share of assigned group mates (with $Z_j$ in place of $D_j$ for IV), where $n_{\mathcal G(i)}$ and the sum refer to the population. Under choice-based or stratified sampling, estimate both with inverse sampling weights: unweighted sample shares overstate the treated share of a group when treated units are oversampled (Section~\ref{sens:twa-design}).

\item \emph{Within-cell difference, on two scales.} Estimate $\widehat\gamma^{\mathcal G}$ as the coefficient on $D_i$ ($Z_i$ for IV) in a regression of $\widehat S^{\mathcal G}_i$ on it and a flexible function of $W_i$, or as the $G$-weighted average of within-$W$ cell differences. This is the benchmark for the share version applied to the share of group mates, the local model, and the one reported in Section~\ref{sen appl}. For the share version applied to the population-wide share, rescale it to $\widehat\Gamma^{\mathcal G}=\widehat\gamma^{\mathcal G}\times\overline{(n_{\mathcal G(i)}-1)/(N-1)}$, with $N$ the population size; the two differ by about $(N-1)/(n_{\mathcal G(i)}-1)$, so $k^\ast$ depends on which share the interference model uses. For the RDD, estimate the jump in the share at $w_0$ with the RD estimator used for the outcome. Cluster standard errors at the group level.

\item \emph{Pairwise dependence.} Where close peers are observed (households, siblings, co-workers, classmates), estimate the pairwise analogue $\widehat\gamma^{\mathrm{pair}}$, the regression of a peer's assignment on one's own. It benchmarks the unit-level $\max_{j\neq k}|\Gamma_{kj}|$ of the averaged form of Theorem~\ref{sensB:prop-refine}(i) under the law $P$, as discussed after that theorem.

\item \emph{Mechanical and compositional terms.} Under within-group quotas (e.g., complete randomization within groups), $\widehat\gamma^{\mathcal G}\approx-1/(n_{\mathcal G(i)}-1)$; this reflects an actual failure of CAI, not an artifact. Variation in propensities across groups not captured by $W_i$ also produces $\widehat\gamma^{\mathcal G}\ne0$, so the difference must be computed within cells of $W_i$.

\item \emph{Interference scale.} Relate $\ell_0$ to external estimates of spillovers relative to direct effects in comparable settings, and $\ell_\theta$ to evidence on how direct effects vary with program saturation. Report results for a small set of values, for example $\ell_0\in\{0.1,0.25,0.5\}$ and $\ell_\theta\in\{0,0.5,1\}$.
\end{enumerate}
An important caveat is that the observed grouping need not be the interference-relevant one. Cross-sectional differences inform unit-level laws only under homogeneity across units. Benchmarks calibrate plausibility; a nonzero population gap is incompatible with generalized CAI for dependence running through $\mathcal G$ and says nothing about dependence outside it, and the sampling distribution of $\widehat\gamma^{\mathcal G}$ is not derived here, so it is a benchmark rather than a test statistic. The total-variation distance parameter $\delta$ cannot be benchmarked from the observed data, so benchmarking applies to the share version.

\subsection{Output: A Synthetic Illustration}
Table~\ref{sensB:tab-frontier} illustrates the output of Procedure~\ref{sensB:alg} on a hypothetical estimate. It is a numerical illustration of the formulas, not an application. We set $\widehat\Phi=0.20$ outcome standard deviations, $\widehat{\mathrm{se}}=0.05$, $\kappa=1$, $t_0=0$, $\alpha=0.05$.

For instance,  we interpret the third row as follows. When $\bar\Lambda_{\mathrm{eff}}$ equals half an outcome standard deviation, the target remains positive if the two cells' laws of others' assignments differ by less than $0.40$ in total variation. Under Theorem~\ref{sensB:prop-refine}, the analogous threshold is a difference of less than $0.40$ in the expected treated share of others. Accounting for sampling uncertainty lowers the breakdown point to $\delta^\ast_\alpha=0.236$.

A separate exact-enumeration check (SOO with $N=6$, IV with $N=5$, DiD with $N=5$, under common-shock, exchangeable, fixed-count, and arbitrary joint assignment laws) found no violation of any bound. 

\begin{table}[!htbp]
\centering
\caption{Breakdown points on synthetic inputs}\label{sensB:tab-frontier}
\begin{tabular}{ccc}
\toprule
$\bar\Lambda_{\mathrm{eff}}$ (outcome SD) & $\delta^\ast(\bar\Lambda)$ & $\delta^\ast_\alpha(\bar\Lambda)$\\
\midrule
0.10 & 1.000 & 1.000\\
0.25 & 0.800 & 0.471\\
0.50 & 0.400 & 0.236\\
1.00 & 0.200 & 0.118\\
2.00 & 0.100 & 0.059\\
\bottomrule
\end{tabular}

\end{table}

\endgroup

\section{A Calibrated Monte Carlo Study}\label{sec:montecarlo}
This section gives the details of the temporary-agency illustration of Section~\ref{sec:example} of the main text and reports in full the simulation evidence summarized there.

\subsection{Interference in the Temporary-Agency Setting}\label{ex:nsw-supp}

\citet{ichino2008from} combine all workers assigned through a temporary agency between January and June 2001 in five Tuscan and four Sicilian provinces with a province-stratified random sample of residents aged 18--40 who did not hold a stable job. The binary outcome of interest indicates whether the worker held a permanent contract in late 2002. The analysis samples have 909 units in Tuscany (281 treated, 628 controls) and 1{,}121 in Sicily (230 treated, 891 controls). The authors invoke the stable unit treatment value assumption and defend it by the small treated share, which they put at about $0.6\%$ in Tuscany and $0.2\%$ in Sicily \citep[fn.~9]{ichino2008from}. 

Three mechanisms nevertheless make outcome interference plausible. TWA participation by some workers may reduce other workers' chances of obtaining the same local vacancies through competition, may increase them by generating information and referrals through agencies, and may affect firms' subsequent hiring decisions through the experience created by temporary placements. These channels are consistent with the evidence of \citet{crepon2013labor} and \citet{gautier2018estimating} for French and Danish job-search programs. The interaction structure is unobserved, as is common when data on network links are difficult or infeasible to collect \citep[e.g.,][]{colpitts2002targeting}.

The data show substantial geographic concentration in take-up. The top decile of ZIP codes accounts for $66.9\%$ of treated workers in Tuscany and $66.1\%$ in Sicily, while the inverse Herfindahl index implies effective numbers of only $12.6$ treated ZIP codes versus $22.7$ control ZIP codes in Tuscany and $17.7$ versus $49.6$ in Sicily. Among ZIP codes with at least five sampled units, distance to the nearest agency is also negatively correlated with the treated share, at $-0.45$ in Tuscany and $-0.18$ in Sicily, indicating that take-up is more concentrated around agency locations. Together, these patterns show that treated workers are disproportionately drawn from a relatively small and geographically localized set of ZIP codes.

This pattern is consistent with both assignment mechanisms in Section~\ref{sec:example}. If idiosyncratic information, preferences, and search costs drive take-up, own take-up need not reveal others' take-up even when their treatment affects permanent-employment prospects through competition or referrals. If instead take-up reflects common forces, such as recruitment networks, agency capacity, or peer effects in participation, own take-up becomes informative about others' assignments, so treated workers and comparable controls may face different treatment environments. The simulation below holds the estimator, interference structure, and true propensity score fixed and varies only the assignment mechanism to create both mechanisms. In every design, the shock affects take-up only, so Assumption~\ref{unif:ass-sep} of the main text holds; a shock entering outcomes as well would violate it and lie outside the analysis.

\subsection{Simulation Design and Results}\label{mc:design}
\paragraph{Design.} We conduct a Monte Carlo simulation, calibrated to the Tuscan
subsample of \citet{ichino2008from}, to illustrate the behavior of the IPW estimator of the selection-on-observables identification formula \eqref{ro} (Section~\ref{sec:identification}) under
interference, both when CAI (condition~\eqref{eq:CAI}) holds and when it fails. The
data-generating process allows outcomes to depend on both own treatment and the treatments of all other units, which violates ITR while preserving a well-defined ADE.

Each of the $4{,}000$ replications uses the $N=683$ Tuscan units remaining after dropping $226$ controls with propensity scores below the treated support, leaving $281$ TWA workers and $402$ controls; covariates, ZIP codes, and estimated scores are held fixed. With covariates fixed, each unit is identified by its covariate value, so the cell laws, $\delta(w)$ and $\bar\Gamma(w)$ at $w=W_k$ are unit $k$'s; the integrals over $G$ average them over units, and equal composition~\eqref{unif:eq-ec} holds trivially. We compute the score from a logit model on nine pre-treatment covariates, a simpler specification than that in Section~\ref{sens:supp-empirics}, whose province-nested trimming leaves $719$ units. 

We model treatment (take-up)  as $D_i=\mathbbm 1\{V_i\le\tilde p_i\}$. The take-up probability is $\tilde p_i=\operatorname{expit}(W_i^\top\hat\beta+\sigma\eta)$, where $\hat\beta$ is the logit estimate from the data and $\eta\sim\mathcal N(0,1)$ is a region-wide shock drawn once per replication. In some designs we instead set $\tilde p_i$ to the ZIP-code mean $\bar p_{\mathcal G(i)}=n_{\mathcal G(i)}^{-1}\sum_{j\in\mathcal G(i)}\operatorname{expit}(W_j^\top\hat\beta)$. To induce within-ZIP dependence, $V_i=U_{\mathcal G(i)}$ with probability $\varphi$ and $V_i=U_i$ otherwise, where $U_{\mathcal G(i)}$ is shared within ZIP code and $U_i$ is unit specific. These choices are independent across units and of $\eta$, and both draws are $\mathrm{Uniform}(0,1)$, so $\Pr(D_i=1\mid W_i,\eta)=\tilde p_i$. The propensity score is $p(W_i)=\mathbbm E[\tilde p_i\mid W_i]$, which integrates out $\eta$. Thus, each design below is indexed by $(\sigma,\varphi)$ and the specification of $\tilde p_i$.
At $\sigma=0$ the propensity ranges from $0.018$ to $0.845$ with mean $0.411$. 

Outcomes are binary, as in the study:
\[
Y_i\sim\mathrm{Bernoulli}\big(\hat q(W_i)+\Lambda_0\{S(\mathbf D_{-i})-\bar s\}
+[\tau_0+\tau_\theta\{S(\mathbf D_{-i})-\bar s\}]D_i\big),
\]
where $\hat q$ is the control-arm outcome model fitted on the data, $S(\mathbf D_{-i})$ the leave-one-out treated share, $\bar s=0.411$, and $\tau_0=0.18$ the replicated stratified ATT. We truncate the success probability to $[10^{-3},1-10^{-3}]$, which is needed because $\hat q$ falls to $0.055$ while the baseline spillover term reaches $-\Lambda_0\bar s$ at a zero share. The truncation is binding for no draw in design~A, for fewer than one unit-draw in $10^{5}$ in designs~B and~C, and for $1.4\%$ of unit-draws in design~D, whose share spans the whole unit interval. We set the interference parameters to the benchmark multiples of Section~\ref{sen appl}, evaluated at the simulation sample's outcome standard deviation and direct effect. The baseline slope is half an outcome standard deviation, $\Lambda_0=0.5\times0.417=0.209$, and the direct-effect slope is $\tau_\theta=0.5\,\tau_0=0.09$. Since the slice means are affine in the share, these are the (AN) slopes $L^S_0$ and $L^S_\theta$ of Theorem~\ref{sensB:prop-refine} and also the ranges of $\mu_0$ and $\theta$. The IPW estimator targets the ADE at the conditional law of $\mathbf D_{-i}$ given $W_i$, which is the mixture $\pi_\lambda$ with $\lambda(w)=p(w)$, and $\kappa=1$. The slope of $m_\lambda=\mu_0+(1-\lambda)\kappa\theta$ in the share is therefore $\Lambda_0+\{1-p(w)\}\tau_\theta$, and averaging with $\bar\lambda=\mathbbm{E}[p(W_i)]=\bar s=0.411$ gives $\bar L_{\mathrm{eff}}=\kappa^{-1}L^S_0+(1-\bar\lambda)L^S_\theta=0.209+(1-0.411)\times0.09=0.262 .$ In designs~C and~D the regional shock raises the mean propensity, $\bar\lambda$, from $0.411$ to $0.414$ in design~C and to $0.467$ in design~D. Because the target ADE, $\tau_0+\tau_\theta(\bar\lambda-\bar s)$, and $\bar L_{\mathrm{eff}}$ both depend on $\bar\lambda$, they move with it: in design~C the changes are negligible at the precision reported, whereas in design~D the target is $0.184$ and $\bar L_{\mathrm{eff}}=0.257$. Each design's bound is evaluated at its own $\bar\lambda$. 
 The IPW estimator uses the true propensity
score $p(W_i)$, so unconfoundedness and overlap hold exactly.

We consider four designs that vary \emph{where} assignment dependence lies, i.e., which units share a random component of take-up, and, through $\sigma$, how strong it is.  In design~A, $\sigma=\varphi=0$, so take-up is independent across units and $p(W_i)=\tilde p_i$, and CAI holds. In~B, $\sigma=0$, $\varphi=1$, and $\tilde p_i=\bar p_{\mathcal G(i)}$, thus all residents of a ZIP code share one draw and one take-up probability, so they take up together, with probability equal to the mean propensity in their ZIP code, which is therefore B's true propensity score, with the ZIP code among the covariates in $W_i$.  Dependence in assignment lies entirely within the observed grouping, at its maximum. In designs C and~D, $\varphi=0$ and $\sigma>0$, and the dependence runs through $\eta$, common to the whole region. In~C, $\sigma=0.505$ is calibrated so that $\bar\Gamma$ equals B's value of $0.046$, i.e., the dependence matches design B's in magnitude but is carried by no observed grouping. In~D, $\sigma=5.48$ raises $\bar\Gamma$ to $0.689$, at which the bound is the size of the effect: $0.180$ at $\bar\lambda=0.411$, the value used to set $\sigma$, and $0.177$ at the design's own mean propensity of $0.467$, against a target of $0.184$. The ZIP-code benchmark observed in these data, $\widehat\gamma^{\rm ZIP}_{\rm obs}=0.0023$, is a gap in the share of ZIP mates; the corresponding quantity in each design, $\bar\gamma$, is reported as a multiple of it in Table~\ref{tab:mc-sens}, whereas $\bar\Gamma$, computed over the $682$ other units, is not on that scale. Under the regional shocks of designs~C and~D, ZIP mates co-move through the common shock, so $\bar\gamma$ is positive although the single-realization statistic $\widehat\gamma^{\rm ZIP}$ in the last column of Table~\ref{tab:mc-sens} is not; that statistic is the raw within-$W_i$ gap in the leave-one-out treated share among ZIP mates from one simulated assignment of the $683$ units, and multiplying it by $(\text{ZIP size}-1)/(N-1)$ puts it on the scale of $\bar\Gamma$, about $0.046$ in design~B.

\paragraph{Results.} Table~\ref{tab:mc-ipw} reports the estimator's behavior, and Table~\ref{tab:mc-sens} in the main text reports the performance of the sensitivity analysis. The ZIP-clustered standard errors of Table~\ref{tab:mc-sens} come from a pairs bootstrap that resamples the $98$ ZIP codes with replacement $500$ times inside each of the $4{,}000$ replications, so the simulation error in a bootstrap standard error is about $3\%$; the Monte Carlo standard error of each coverage frequency is at most $0.008$.

\begin{table}[H]
\centering
\caption{Monte Carlo results: IPW estimation under interference}
\label{tab:mc-ipw}
\begin{tabular}{lcccc}
\toprule
Design &
Mean(IPW) &
ADE Target &
Bias &
RMSE \\
\midrule
A. CAI holds
& 0.180
& 0.180
& $-0.000$
& 0.051
 \\
B. ZIP take-up ($\bar\gamma=426\times\widehat\gamma^{\rm ZIP}_{\rm obs}$)
& 0.192
& 0.180
& 0.012
& 0.134
 \\
C. Common shock ($\bar\gamma=20\times\widehat\gamma^{\rm ZIP}_{\rm obs}$)
& 0.189
& 0.180
& 0.009
& 0.143
 \\
D. Common shock ($\bar\gamma=302\times\widehat\gamma^{\rm ZIP}_{\rm obs}$)
& 0.342
& 0.184
& 0.158
& 0.520
 \\
\bottomrule
\end{tabular}
\begin{flushleft}
\footnotesize
\emph{Note:} Results over $4{,}000$ Monte Carlo replications calibrated to the Tuscan subsample. The ADE target is $\mathbbm{E}[\Delta_i(\mathbf D_{-i})]$, averaged over the $683$ units, i.e., $\mathbbm{E}[\Delta_I(\mathbf D_{-I})]$. Bias is the Monte Carlo mean of the IPW estimator minus the ADE target; its Monte Carlo standard errors are $0.001$, $0.002$, $0.002$ and $0.008$ in designs A--D.
\end{flushleft}
\end{table}

In design~A, where take-up is independent given covariates, the bias of the ITR-based IPW estimator for the ADE is negligible. The Monte Carlo mean of the estimator ($0.180$) equals the ADE target, and the  RMSE is $0.051$. This illustrates the theoretical result that, under arbitrary interference, ITR-based estimators for the selection-on-observables design are unbiased for the ADE when treatment assignment is conditionally independent across units.

Dependence in take-up degrades the estimator, and how much depends on its \textit{magnitude} rather than its \textit{location}. Designs~B and~C share the same $\bar\Gamma$, $0.046$, although on the scale of the ZIP share their dependence is $426$ and $20$ times the observed gap $\widehat\gamma^{\rm ZIP}_{\rm obs}$, one generated by whole ZIP codes taking up together and the other by a shock common to the region. The biases are $0.012$ and $0.009$ against a direct effect of $0.18$, and the RMSE roughly triples relative to design~A. Only in design~D, where the common shock is raised until the dependence among ZIP mates, $\bar\gamma$, is about $300$ times the observed gap $\widehat\gamma^{\rm ZIP}_{\rm obs}$ ($\bar\Gamma=0.689$), does the bias ($0.158$) approach the effect itself. The ADE target barely moves across designs, so the deterioration reflects a widening wedge between the formula and the target rather than a change in the target.

Table~\ref{tab:mc-sens} in Section \ref{sen appl} turns to the sensitivity analysis. The bound of Theorem~\ref{sensB:prop-refine}(iii)--(iv) is attained up to Monte Carlo error in designs~B and~C and is slack in design~D---$0.012$ against a bias of $0.012$ (standard error $0.002$) in~B, $0.012$ against $0.009$ ($0.002$) in~C, and $0.177$ against $0.158$ ($0.008$) in~D, where the truncation of the success probability flattens the response---and is zero when CAI holds. Two limitations show, and they are the same limitation twice. The Imbens--Manski interval attains its nominal level in design~A but covers only about half the time once take-up is dependent, because the conventional standard error understates the sampling variability of $\widehat\Phi$ by a factor of three, and of sixteen in design~D; the bound on the bias is not at fault. Clustering the bootstrap by ZIP code recovers $0.91$ in design~B, where ZIP codes are the
source of the dependence, and leaves $0.51$ in design~C and $0.19$ in design~D, where they are not. The benchmark behaves identically: the raw within-ZIP gap $\widehat\gamma^{\rm ZIP}$ is one in design~B but indistinguishable from zero in design~C, which has the same $\bar\Gamma$ and a bias of the same size. Neither the benchmark nor the clustered standard error detects dependence outside the grouping it is computed from, and design~C is the case in which both are silent while the bias is present and the
bound is correct.

The failure of CAI induces systematic differences between the distributions of others' treatments
faced by treated and untreated units, which conditioning on observed covariates alone cannot
correct. The IPW identification formula then no longer equals the ADE but conflates it with bias.

\paragraph{Varying the dependence.} Figure~\ref{fig:sweep-B} of the main text holds design~B's propensity scores, outcome model and interference fixed and moves its assignment dependence from the maximum to zero. We hold $\sigma=0$ and $\tilde p_i=\bar p_{\mathcal G(i)}$ and lower $\varphi$ from $1$ (design~B) to $0$ (CAI). Since $p(W_i)=\bar p_{\mathcal G(i)}$ at every $\varphi$, the true score, the IPW estimator and the target are unchanged, while $\Gamma_{kj}=\varphi^2$ for ZIP-code mates $j$ of unit $k$ and $\bar\Gamma=\varphi^2\times0.046$. At each of eleven values of $\varphi^2$, $2{,}000$ replications deliver the population identified set $[\Phi_{\mathrm{ITR}}(P)-\bar b,\Phi_{\mathrm{ITR}}(P)+\bar b]$ with $\bar b=\bar L_{\mathrm{eff}}\bar\Gamma$, together with two confidence statements at the ZIP-clustered bootstrap standard error: the Imbens--Manski interval, which covers the target, and the band $\widehat\Phi\pm(\bar b+z_{0.975}\widehat{\mathrm{se}})$ of the conventional form, which covers the identified set (Proposition~\ref{sensB:prop-ci}).

Three features stand out. First, the bound holds at every value: at full dependence the bias is $0.007$, with Monte Carlo standard error $0.003$, against $\bar b=0.012$ (Table~\ref{tab:mc-sens} reports $0.012$, standard error $0.002$, for the same design from an independent set of $4{,}000$ draws; the two estimates differ by $1.4$ standard errors), and it shrinks with $\bar\Gamma$ thereafter. Second, the dependence inflates the variance of $\widehat\Phi$ far more than it biases it. At full dependence, the identified set is $[0.174,0.198]$, a half-width of $0.012$, whereas the bootstrap standard error is $0.122$ and the confidence bands run from about $-0.06$ to $0.43$. As $\bar\Gamma$ falls to zero, the set collapses to a point, and the bootstrap standard error falls to $0.047$, its value under CAI, while the conventional standard error stays near $0.046$ throughout and so registers none of the inflation. Third, the bootstrap keeps coverage close to nominal across the whole range: the Imbens--Manski interval covers the target in $90$ to $92\%$ of replications and the Horowitz--Manski-type band covers the set in $87$ to $92\%$, the shortfall from $0.95$ reflecting the $98$ clusters, for which $\mathrm{sd}(\widehat\Phi)/\widehat{\mathrm{se}}$ lies between $1.06$ and $1.15$. With the conventional standard error, the Imbens--Manski coverage falls from $0.93$ under CAI to $0.51$ at full dependence. The band covers the set less often than the interval covers the target, although it is wider, because covering the set requires both endpoints. On average, the Imbens--Manski interval excludes zero once the dependence among ZIP mates is at most about $210$ times the observed gap ($\varphi^2\le0.5$), and the Horowitz--Manski-type band once it is at most about $170$ times ($\varphi^2\le0.4$); at the observed gap itself, both exclude zero comfortably.

\paragraph{Source of the bias.} Proposition~\ref{unif:thm-master} makes the source of this bias explicit. With the true propensity score and overlap, the IPW estimand equals $\Phi_{\mathrm{ITR}}(P)$ in \eqref{ro}. In the design above, the slice means and slice contrasts of Corollary~\ref{unif:cor-designs}(a) are $\mu_0(\mathbf d_{-i};w)=\hat q(w)+\Lambda_0\{S(\mathbf d_{-i})-\bar s\}$ and $\theta(\mathbf d_{-i};w)=\tau_0+\tau_\theta\{S(\mathbf d_{-i})-\bar s\}$. Take the reference law to be the conditional law of $\mathbf D_{-i}$ given $W_i$, so that the target is the ADE in Table~\ref{tab:mc-ipw}. Let $\bar s_d(w):=\mathbbm{E}[S(\mathbf D_{-i})\mid D_i=d,W_i=w]$ and recall $p(w)=\Pr(D_i=1\mid W_i=w)$. Since $\mathbbm{E}[S(\mathbf D_{-i})\mid W_i=w]=p(w)\bar s_1(w)+\{1-p(w)\}\bar s_0(w)$, the channel form \eqref{unif:eq-channels} gives
\[
\Phi_{\mathrm{ITR}}(P)-\mathrm{ADE}
=\mathbbm{E}\Big[\underbrace{\Lambda_0\{\bar s_1(W_i)-\bar s_0(W_i)\}}_{\text{baseline component}}+\underbrace{\tau_\theta\{1-p(W_i)\}\{\bar s_1(W_i)-\bar s_0(W_i)\}}_{\text{direct-effect  component}}\Big].
\]
Both channels are proportional to $\bar s_1(W_i)-\bar s_0(W_i)=\cov(D_i,S(\mathbf D_{-i})\mid W_i)/\var(D_i\mid W_i)$, the difference in the share of treated others between treated and untreated units with the same covariates. This difference is zero when $\sigma=0$. When $\sigma\neq0$, a unit's own treatment is informative about the common shock and hence about others' treatments; because the shock is shared by all units, we expect the difference not to vanish as $N$ grows, although the simulation, run at the single sample size of the Tuscan subsample, does not check this. The interference parameters determine how this assignment dependence translates into bias: $\Lambda_0$ governs spillovers on untreated outcomes, and $\tau_\theta$ governs the dependence of direct effects on others' treatments. Since $\tau_\theta>0$, the direct-effect component has the sign of $\bar s_1-\bar s_0$, and the baseline component additionally depends on the sign of $\Lambda_0$. The bias in Table~\ref{tab:mc-ipw}, which rises with $\sigma$, is consistent with the growth of $\bar s_1-\bar s_0$ as the common shock becomes more influential.

\section{Empirical Illustration I: Temporary Work Agency}\label{sens:supp-empirics}

This section details the illustration in Section~\ref{sen appl}, which applies the analysis to the study of TWA jobs and permanent employment by \citet{ichino2008from} (IMN), whose data are described in Section~\ref{ex:nsw-supp}. All computations use the authors' data and baseline do-file. Standard errors in the replication tables below are the usual ones; the sensitivity analysis of Section~\ref{sen appl} uses ZIP-clustered bootstrap standard errors, which account for dependence running through shared ZIP codes.

\subsection{Sample Design}\label{sens:twa-design}

The sample is choice-based. The treated are all TWA workers, while the controls are a random sample drawn by province. Weighting by inverse sampling rates, the reference population is 92{,}129 in Tuscany and 429{,}573 in Sicily, of which TWA workers are $0.58\%$ and $0.15\%$, in line with the shares the authors report. The ratio of treated-to-control sampling rates ranges across provinces from 50 to 263. Three consequences follow. First, pooled nearest-neighbor matching frequently pairs treated and control units across provinces, despite substantial differences in sampling rates. This occurs for $68\%$ of treated units in Tuscany and $58\%$ in Sicily. Imposing exact matching on province leaves the Tuscan estimate nearly unchanged, from $0.189$ to $0.192$, but reduces the Sicilian estimate from $0.096$ to $0.0565$. Second, strata must nest within provinces, as the sample design requires. Third, following IMN, our baseline estimates are unweighted. Applying population weights instead yields stratified effects of $0.177$ in Tuscany and $0.074$ in Sicily.

\subsection{Replication}

The replication follows the authors' baseline do-file specification: a probit propensity score, one-to-one nearest-neighbor matching with replacement, and common support. Table~\ref{tab:twa-replication} compares the estimates with IMN's Table~II and footnotes 21--22. All twelve effects on the treated match to two decimals, and the numbers of treated and matched controls match in all 12 cells, and the standard errors in 11.  The authors' distributed match files use a slightly different score specification (log-odds correlation with ours $0.965$ in Tuscany and $0.990$ in Sicily): IMN's Table~I, matched-controls column, matches those files, while their Table~II matches our emulation. The footnote counts, and the OLS estimates also replicate.

\begin{table}[!htbp]
\centering\footnotesize\setlength{\tabcolsep}{4pt}
\caption{Replication of \citet{ichino2008from}, Table~II}\label{tab:twa-replication}
\begin{tabular}{@{}lcccccc@{}}
\toprule
& \multicolumn{3}{c}{Tuscany} & \multicolumn{3}{c}{Sicily} \\
\cmidrule(lr){2-4}\cmidrule(lr){5-7}
Sample & Published & Replicated & $N_T/N_C$ & Published & Replicated & $N_T/N_C$ \\
\midrule
Whole sample & 0.19 (0.06) & 0.189 (0.059) & 281/133 & 0.10 (0.05) & 0.096 (0.050) & 230/131 \\
Thick support & 0.23 (0.07) & 0.229 (0.073) & 109/56 & 0.14 (0.08) & 0.141 (0.077) & 92/43 \\
Male & 0.24 (0.10) & 0.236 (0.104) & 157/59 & 0.10 (0.07) & 0.103 (0.066) & 155/76 \\
Female & 0.14 (0.07) & 0.137 (0.070) & 124/71 & -0.07 (0.06) & -0.067 (0.058) & 75/57 \\
Under 30 & 0.11 (0.07) & 0.111 (0.074) & 199/88 & 0.09 (0.06) & 0.088 (0.060) & 170/90 \\
Over 30 & 0.33 (0.09) & 0.329 (0.091) & 82/44 & 0.00 (0.09) & 0.000 (0.095) & 60/39 \\
\midrule
OLS, all controls & 0.16 (0.03) & 0.160 (0.034) & 281/628 & 0.05 (0.03) & 0.051 (0.029) & 230/891 \\
Stratified, baseline strata & -- & 0.184 (0.048) & 281/628 & -- & 0.045 (0.056) & 230/891 \\
Stratified, recommended strata & -- & 0.180 (0.034) & 281/438 & -- & 0.063 (0.040) & 230/703 \\
\bottomrule
\end{tabular}

\smallskip
\parbox{\linewidth}{\scriptsize \emph{Notes:} Effects on the treated (standard errors). The count $N_C$ is the number of distinct matched controls for matching rows and the number of controls used otherwise. The authors' match files use a slightly different score specification from the one reported; their Table~I (matched controls) matches those files, and their Table~II matches our emulation. Baseline strata: province $\times$ blocks cut at treated score quantiles; recommended strata: Section~\ref{sens:twa-strata}.}
\end{table}

\subsection{Strata}\label{sens:twa-strata}

\paragraph{The recommended rule.} Within each province we drop controls outside the range of treated propensity scores, then split score blocks at the treated median while the treated--control score difference has $|t|>1.96$, keeping at least three treated and three controls per block. The rule adapts the blocking algorithm of the authors' do-file. The paper itself reports nearest-neighbor matching and notes only that stratification gave similar estimates \citep{ichino2008from}. The rule yields 15 strata in each region, with 719 units (281 treated, 438 controls) in Tuscany and 933 (230 treated, 703 controls) in Sicily. The strata depend on treatment status, since the score, the trimming, and the balance splits all use $D_i$, as is standard in propensity-score stratification.

\subsection{Results}

The target is of ADTT type ($\lambda=1$, $G=F_{W\mid D=1}$, $\kappa=1$), so $\bar\Lambda_{\mathrm{eff}}=\bar\Lambda_0$ and only spillovers on untreated outcomes matters. Because the outcome is binary, $\bar\Lambda_0\le1$; the sample outcome standard deviation is $0.408$ in Tuscany and $0.355$ in Sicily, so $\bar\Lambda_0=1$ is $2.45$ and $2.81$ standard deviations. Table~\ref{tab:twa-frontier} gives the breakdown points for the recommended strata and for the published matching estimates. In Tuscany (recommended strata), the Imbens--Manski interval at $\bar\Lambda_0=0.25$ SD and $\bar\delta=1$ is $[0.023,0.338]$ at the conventional standard error and $[0.017,0.343]$ at the ZIP-clustered bootstrap standard error. In Sicily, $\delta^\ast_\alpha=0$ at every $\bar\Lambda_0$ for both stratified estimates: the confidence interval contains zero without any dependence.

\begin{table}[!htbp]
\centering\footnotesize
\caption{Breakdown points}\label{tab:twa-frontier}
\begin{tabular}{@{}lcccccc@{}}
\toprule
& \multicolumn{3}{c}{Tuscany} & \multicolumn{3}{c}{Sicily} \\
\cmidrule(lr){2-4}\cmidrule(lr){5-7}
$\ell_0$ (outcome SD) & $\bar\Lambda_0$ & $\delta^\ast$ & $\delta^\ast_\alpha$ & $\bar\Lambda_0$ & $\delta^\ast$ & $\delta^\ast_\alpha$ \\
\midrule
\multicolumn{7}{@{}l}{\emph{Stratified, recommended strata:} $\widehat\Phi=0.180$ ($0.034$) and $0.063$ ($0.040$)} \\
0.05 & 0.020 & 1.000 & 1.000 & 0.018 & 1.000 & 0.000 \\
0.10 & 0.041 & 1.000 & 1.000 & 0.036 & 1.000 & 0.000 \\
0.25 & 0.102 & 1.000 & 1.000 & 0.089 & 0.713 & 0.000 \\
0.50 & 0.204 & 0.883 & 0.611 & 0.178 & 0.357 & 0.000 \\
1.00 & 0.408 & 0.441 & 0.305 & 0.355 & 0.178 & 0.000 \\
$\bar\Lambda_0=1$ & 1.000 & 0.180 & 0.125 & 1.000 & 0.063 & 0.000 \\
\midrule
\multicolumn{7}{@{}l}{\emph{Nearest-neighbor matching, published specification:} $\widehat\Phi=0.189$ ($0.059$) and $0.096$ ($0.050$)} \\
0.05 & 0.020 & 1.000 & 1.000 & 0.018 & 1.000 & 0.757 \\
0.10 & 0.041 & 1.000 & 1.000 & 0.036 & 1.000 & 0.379 \\
0.25 & 0.102 & 1.000 & 0.889 & 0.089 & 1.000 & 0.151 \\
0.50 & 0.204 & 0.924 & 0.445 & 0.178 & 0.538 & 0.076 \\
1.00 & 0.408 & 0.462 & 0.222 & 0.355 & 0.269 & 0.038 \\
$\bar\Lambda_0=1$ & 1.000 & 0.189 & 0.091 & 1.000 & 0.096 & 0.013 \\
\bottomrule
\end{tabular}

\smallskip
\parbox{\linewidth}{\scriptsize \emph{Notes:} ADTT-type target ($\lambda=1$), so $\bar\Lambda_{\mathrm{eff}}=\bar\Lambda_0$; $t_0=0$; $\alpha=0.05$. $\delta^\ast=\min\{1,\widehat\Phi/\bar\Lambda_0\}$, $\delta^\ast_\alpha=\min\{1,(\widehat\Phi-z_{1-\alpha}\widehat{\mathrm{se}})_+/\bar\Lambda_0\}$. The standard errors here are the conventional ones. With the ZIP-clustered bootstrap standard error ($0.037$ in Tuscany, $0.044$ in Sicily), Tuscany's $\delta^\ast_\alpha$ in the recommended strata falls from $0.611$ to $0.584$ at $0.5$~SD and from $0.125$ to $0.120$ at $\bar\Lambda_0=1$, and Sicily's remains zero throughout.}
\end{table}

\paragraph{The share benchmark.} $\widehat\gamma^{\rm ZIP}_{\rm obs}$ is the leave-one-out treated share of the population in a unit's own ZIP code, the benchmark $\widehat\gamma^{\mathcal G}$ of Section~\ref{robustness} of the main text with ZIP codes as the grouping, treated minus controls within province $\times$ score strata, averaged with treated weights, over ZIP codes with at least one sampled control, with standard errors clustered by ZIP. It is $0.00226$ ($0.00075$) in Tuscany and $0.00105$ ($0.00022$) in Sicily. It is a gap in the share of ZIP mates; per other unit of the population, the treated-weighted factor $(\text{ZIP size}-1)/(N-1)$ being about $0.045$ in Tuscany, the Tuscan value on the scale of $\bar\Gamma$ is about $0.0001$. Population weights are essential: the unweighted sample share is inflated by the choice-based design, giving $0.01067$ ($0.01981$) and $0.04865$ ($0.02428$). Restricting to ZIP codes with at least three sampled controls gives $0.00132$ ($0.00046$) and $0.00109$ ($0.00024$); computing the benchmark within the recommended strata gives $0.00208$ ($0.00069$) and $0.00173$ ($0.00073$). Table~\ref{tab:twa-share} reports the share version for all three estimators: $L^\ast_S$ is the slope of the untreated mean in the ZIP-treated share at which the bias bound equals $\widehat\Phi$, and $k^\ast$ is the multiple of $\widehat\gamma^{\rm ZIP}_{\rm obs}$ at which the conclusion breaks down when $L^S_0$ is half an outcome standard deviation per unit of the ZIP share.

\begin{table}[!htbp]
\centering\footnotesize
\caption{Share version, ZIP-code benchmark}\label{tab:twa-share}
\begin{tabular}{@{}llcccccc@{}}
\toprule
Region & Estimator & $\widehat\Phi$ (s.e.) & $\widehat\gamma^{\rm ZIP}_{\rm obs}$ (s.e.) & $L^\ast_S$ & $L^\ast_{S,\alpha}$ & $k^\ast$ & $k^\ast_\alpha$ \\
\midrule
Tuscany & NN matching & 0.189 (0.059) & 0.0023 (0.0008) & 83 & 40 & 408 & 197 \\
Tuscany & Stratified (baseline) & 0.184 (0.048) & 0.0023 (0.0008) & 81 & 47 & 398 & 229 \\
Tuscany & Stratified (recommended) & 0.180 (0.034) & 0.0023 (0.0008) & 80 & 55 & 390 & 270 \\

\midrule
Sicily & NN matching & 0.096 (0.050) & 0.0011 (0.0002) & 91 & 13 & 512 & 72 \\
Sicily & Stratified (baseline) & 0.045 (0.056) & 0.0011 (0.0002) & 42 & 0 & 239 & 0 \\
Sicily & Stratified (recommended) & 0.063 (0.040) & 0.0011 (0.0002) & 60 & 0 & 339 & 0 \\
\bottomrule
\end{tabular}

\smallskip
\parbox{\linewidth}{\scriptsize \emph{Notes:} $\widehat\gamma^{\rm ZIP}_{\rm obs}$: population-weighted leave-one-out treated share in own ZIP code, treated minus controls within province $\times$ score strata, treated-weighted; ZIP codes with at least one sampled control; standard errors clustered by ZIP. $L^\ast_S$ ($L^\ast_{S,\alpha}$): slope of the untreated mean in the ZIP treated share, in outcome units per unit share, at which the bias bound equals $\widehat\Phi$ ($\widehat\Phi-z_{1-\alpha}\widehat{\mathrm{se}}$). $k^\ast$, $k^\ast_\alpha$ at $L^S_0=0.5$ outcome SD. The $\widehat\Phi$ standard errors are conventional; with the ZIP-clustered bootstrap standard error, the recommended-strata $k^\ast_\alpha$ is $258$ in Tuscany and remains $0$ in Sicily.}
\end{table}

\paragraph{Robustness checks of the sensitivity analysis.} The share-version verdict of Table~\ref{tab:twa-layers} of the main text, that Tuscany's ADTT remains positive up to $k^\ast=390$ times the observed ZIP-mate gap, or $k^\ast_\alpha=258$ allowing for sampling error, rests on three inputs: the estimated benchmark, the standard error, and the share the interference model uses. We check each in turn with the same data and estimator (recommended strata, $L^S_0$ equal to half an outcome standard deviation per unit of the ZIP share). First, the benchmark is estimated. Replacing its point estimate with its one-sided $95\%$ upper bound, $0.0035$ in Tuscany and $0.0014$ in Sicily, lowers the multiples to $k^\ast=253$ and $k^\ast_\alpha=167$ in Tuscany and to $252$ and $0$ in Sicily. Second, the inference-robust values presume a standard error valid under the dependence allowed, and the ZIP-clustered bootstrap captures only dependence that runs through ZIP codes (design~C of Table~\ref{tab:mc-sens}). Resampling within province instead, which respects the province-stratified sampling, gives $0.036$ and $0.045$ against $0.037$ and $0.044$, leaving the values unchanged. Doubling the bootstrap standard error in $\delta^\ast_\alpha=(\widehat\Phi-z_{0.95}\widehat{\mathrm{se}})_+/\bar\Lambda_{\mathrm{eff}}$, from $0.037$ to $0.074$, lowers $\delta^\ast_\alpha$ at $\bar\Lambda_0=0.5$ SD from $0.58$ to $0.28$ and $k^\ast_\alpha$ from $258$ to $126$; tripling it makes the one-sided lower confidence bound negative, so the conclusion fails even under CAI. Third, the multiple is tied to the share the interference model uses. Under interference through the population-wide share, the benchmark is $\widehat\Gamma^{\rm ZIP}=\widehat\gamma^{\rm ZIP}_{\rm obs}\times\overline{(\text{ZIP size}-1)/(N-1)}$, about $0.0001$ per other unit of the population of $92{,}129$, and the corresponding multiple is about $8{,}700$; the two multiples differ by the factor of about $(N-1)/(\text{ZIP size}-1)$ implied by the definition of $\widehat\Gamma^{\mathcal G}$ in the Reporting paragraph of the main text and stated in step~2 of Procedure~\ref{sensB:alg}. The Tuscan verdict survives each check; the Sicilian one fails at zero dependence under every check, as in the main text.

\section{Empirical Illustration II: The Minimum Wage and Fast-Food Employment}\label{app::application}

This section applies the sensitivity analysis to a second study, the New Jersey minimum-wage experiment of \citet{card1994minimum}. The design is difference-in-differences rather than selection on observables, and the assignment is geographic rather than individual. Both features matter. They make generalized CAI fail by construction, with a \emph{known} value of
$\delta$, and they make the refinement of Theorem~\ref{sensB:prop-refine} disagree with itself by two orders of magnitude about whether the finding is robust, depending on whether interference is taken to be anonymous across the market or local.

\subsection{Background and Data}
On 1~April 1992, New Jersey raised its minimum wage from \$4.25 to \$5.05 while neighboring eastern Pennsylvania did not. \citet{card1994minimum} surveyed fast-food restaurants in both states before and after the increase and found no employment loss in New Jersey and, if anything, a gain. We use the authors' public data file,
$410$ stores with two survey waves, February--March and November--December 1992. Following their code, full-time-equivalent employment is
$Y_{it}=\texttt{EMPFT}_{it}+\texttt{NMGRS}_{it}+0.5\,\texttt{EMPPT}_{it}$, where $\texttt{EMPFT}_{it}$, $\texttt{EMPPT}_{it}$ and $\texttt{NMGRS}_{it}$ are the numbers of full-time employees, part-time employees, and managers and assistant managers at store~$i$ in wave~$t$; employment at permanently closed stores is set to zero, and $D_i$ indicates a New Jersey store. 

\subsection{Replication}
Table~\ref{tab:ck-replication} compares our replication with their Table~3, which reports mean full-time-equivalent employment per store by state and, within New Jersey, by the starting wage in the first wave. All $40$ point estimates match to within $0.01$, as do the standard errors in rows 1, 2, 4 and~5. The exception is the standard errors of row~3, the change in means over all available observations, whose formula the paper does not state; ours allow for the correlation between the two wave means and are smaller by $0.03$ to $0.22$. On the balanced sample (their row~4), the change in mean full-time-equivalent employment is $-2.28$ $(1.25)$ in Pennsylvania and $0.47$ $(0.48)$ in
New Jersey, a difference of $2.75$ $(1.34)$, exactly as published; their headline estimate of $2.76$ $(1.36)$ \citep[p.~778]{card1994minimum} is the same difference computed from all available observations in each wave (their row~3). The estimation sample is the $384$ stores observed in both waves, $309$ in New Jersey and $75$ in Pennsylvania.

\begin{table}[!htbp]
\centering
\caption{Replication of \citet{card1994minimum}, Table~3}\label{tab:ck-replication}
\resizebox{\textwidth}{!}{%
\footnotesize\setlength{\tabcolsep}{3pt}%
\begin{tabular}{@{}llcccccccc@{}}
\toprule
& & \multicolumn{3}{c}{Stores by state} & \multicolumn{3}{c}{New Jersey stores, wave-1 starting wage} & \multicolumn{2}{c}{Differences within New Jersey} \\
\cmidrule(lr){3-5}\cmidrule(lr){6-8}\cmidrule(l){9-10}
& & PA & NJ & NJ$-$PA & \$4.25 & \$4.26--\$4.99 & $\ge$\$5.00 & Low$-$high & Midrange$-$high \\
& & (i) & (ii) & (iii) & (iv) & (v) & (vi) & (vii) & (viii) \\
\midrule
\multicolumn{10}{@{}l}{\emph{1. FTE employment before, all available observations}} \\
& Published & 23.33 (1.35) & 20.44 (0.51) & $-2.89$ (1.44) & 19.56 (0.77) & 20.08 (0.84) & 22.25 (1.14) & $-2.69$ (1.37) & $-2.17$ (1.41) \\
& Replicated & 23.33 (1.35) & 20.44 (0.51) & $-2.89$ (1.44) & 19.56 (0.77) & 20.08 (0.84) & 22.25 (1.14) & $-2.69$ (1.37) & $-2.17$ (1.41) \\
\addlinespace
\multicolumn{10}{@{}l}{\emph{2. FTE employment after, all available observations}} \\
& Published & 21.17 (0.94) & 21.03 (0.52) & $-0.14$ (1.07) & 20.88 (1.01) & 20.96 (0.76) & 20.21 (1.03) & 0.67 (1.44) & 0.75 (1.27) \\
& Replicated & 21.17 (0.94) & 21.03 (0.52) & $-0.14$ (1.08) & 20.88 (1.01) & 20.96 (0.76) & 20.21 (1.03) & 0.66 (1.44) & 0.74 (1.27) \\
\addlinespace
\multicolumn{10}{@{}l}{\emph{3. Change in mean FTE employment}} \\
& Published & $-2.16$ (1.25) & 0.59 (0.54) & 2.76 (1.36) & 1.32 (0.95) & 0.87 (0.84) & $-2.04$ (1.14) & 3.36 (1.48) & 2.91 (1.41) \\
& Replicated & $-2.17$ (1.22) & 0.59 (0.47) & 2.75 (1.31) & 1.32 (0.81) & 0.87 (0.67) & $-2.04$ (0.98) & 3.36 (1.27) & 2.91 (1.19) \\
\addlinespace
\multicolumn{10}{@{}l}{\emph{4. Change in mean FTE employment, balanced sample of stores}} \\
& Published & $-2.28$ (1.25) & 0.47 (0.48) & 2.75 (1.34) & 1.21 (0.82) & 0.71 (0.69) & $-2.16$ (1.01) & 3.36 (1.30) & 2.87 (1.22) \\
& Replicated & $-2.28$ (1.25) & 0.47 (0.48) & 2.75 (1.34) & 1.20 (0.82) & 0.71 (0.69) & $-2.16$ (1.01) & 3.36 (1.30) & 2.87 (1.22) \\
\addlinespace
\multicolumn{10}{@{}l}{\emph{5. Change in mean FTE employment, temporarily closed stores set to 0}} \\
& Published & $-2.28$ (1.25) & 0.23 (0.49) & 2.51 (1.35) & 0.90 (0.87) & 0.49 (0.69) & $-2.39$ (1.02) & 3.29 (1.34) & 2.88 (1.23) \\
& Replicated & $-2.28$ (1.25) & 0.23 (0.49) & 2.51 (1.35) & 0.90 (0.87) & 0.49 (0.69) & $-2.39$ (1.02) & 3.29 (1.34) & 2.88 (1.23) \\
\bottomrule
\end{tabular}}

\smallskip
\parbox{\linewidth}{\scriptsize \emph{Notes:} Mean full-time-equivalent (FTE) employment per store (standard errors), FTE counting each part-time worker as half a full-time worker and including managers. Published: \citet[Table~3]{card1994minimum}. Replicated: the authors' public data file. Employment at the six permanently closed stores is set to zero; employment at the four temporarily closed stores is treated as missing, except in row~5. The New Jersey wage groups have $101$, $140$ and $73$ stores. Rows 4 and~5 use the stores with employment data in both waves. Our row-3 standard errors allow for the correlation between the two wave means; the formula behind the published ones is not stated.}
\end{table}

\subsection{Interference in the Minimum Wage and Fast-Food Employment Setting}\label{ck:interference}

Fast-food restaurants hire from, and sell to, local markets they share with other restaurants, so a store's employment may depend on whether nearby stores face the higher minimum wage. Four channels make interference plausible. The first is competition for workers. A New Jersey store required to pay \$5.05 may attract applicants from nearby lower-wage stores, including other New Jersey stores whose wages are brought closer to the new minimum and, near the border, stores in Pennsylvania. In the monopsony and wage-posting models that \citet[Section~VIII.B]{card1994minimum} offer as explanations of their employment finding, a store's employment depends on the wages its competitors post. The second is competition for customers. The authors explain why prices did not rise faster at the New Jersey stores most affected by the increase by arguing that these stores compete in the same product market as other New Jersey stores \citep[p.~788]{card1994minimum}; the same competition makes one store's sales, and hence its employment, depend on the costs and prices of its competitors. The third channel is entry. A higher minimum wage may deter new stores whose presence would affect incumbents' sales; the survey covers only stores open before the increase, and the authors examine openings separately with state-level counts of McDonald's outlets, finding no adverse effect \citep[Section~VI]{card1994minimum}. The fourth is cross-border movement. The authors describe New Jersey's economy as closely linked to those of nearby states \citep[p.~773]{card1994minimum}; that link motivates the Pennsylvania comparison, and it is also the route by which workers and customers can respond to the wage gap by crossing the border.

\citet{card1994minimum} do not discuss interference, but they check for, and do not find, a wage response at stores where the new minimum was not binding (p.~776). Their design requires the Pennsylvania stores and, in the alternative specification, the higher-wage New Jersey stores to provide untreated counterfactuals. They present three pieces of evidence consistent with this premise. First, starting wages at higher-wage New Jersey restaurants changed by only $-3.1$ percent \citep[p.~776]{card1994minimum}, suggesting that the minimum-wage increase did not induce an appreciable wage response among stores for which it was not binding. This addresses one possible spillover channel, but does not rule out employment effects through worker reallocation or competition at unchanged wages. Second, the low- versus high-wage New Jersey comparison in columns (vii)--(viii) of Table~\ref{tab:ck-replication} yields an estimate similar to the New Jersey--Pennsylvania comparison, which the authors interpret as evidence against problems with the Pennsylvania control group \citep[pp.~778--779]{card1994minimum}. Under interference, however, higher-wage New Jersey stores may be affected because they compete for workers with nearby stores whose wages were forced up, so agreement across the two comparisons does not rule out spillovers. Third, the authors argue that restaurants in New Jersey and Pennsylvania operate in separate product markets \citep[p.~788]{card1994minimum}, which limits cross-border product-market spillovers but leaves within-New-Jersey spillovers possible. Their region controls and the Newark and Camden subsample comparisons \citep[Table~4, column (v), and Table~5, rows 10--11]{card1994minimum} instead address differential local demand shocks,  a threat to parallel trends rather than to the absence of interference.

\paragraph{Mapping to our framework.} This is the static difference-in-differences design of Section~\ref{sec:did}, with $W_i$ the chain. The target is of ADTT type ($\lambda=1$, $G=F_{W\mid D=1}$, $\kappa=1$), so $\bar\Lambda_{\mathrm{eff}}=\bar\Lambda_0$ and only
spillovers on untreated period-1 employment enter, together, under the law $P$ of a randomly drawn store, with differences in untreated trends across stores of a chain, as discussed in Section~\ref{ck:results}, and generalized CAI is condition \eqref{eq:CAI} of Theorem~\ref{thm:did}. The numbers that enter the sensitivity analysis below are ours rather than Card and Krueger's, and all are computed on the balanced sample of row~4 of Table~\ref{tab:ck-replication}: the $384$ stores with employment data in both waves, with permanently closed stores counted at zero employment. 

The estimand is $\Phi_{\mathrm{ITR}}(P)$ in \eqref{DiD indentification formula}, the formula of Section~\ref{sec:did}, which follows \citet{abadie2005semiparametric} in letting outcome dynamics differ with the covariates $W_i$. Because $W_i$, the chain, is discrete, the formula has an exact sample analog that requires no model for the outcome or for the propensity score. Its inner terms are the mean changes in employment within a chain and state, $\mathbbm{E}[Y_{i1}-Y_{i0}\mid W_i=w, D_i=d]$, and its outer expectation, over $W_i$ given $D_i=1$, averages across chains with weights equal to the shares of New Jersey stores in each. We therefore take, within each chain, the change in mean full-time-equivalent employment in New Jersey minus that in Pennsylvania. The four contrasts are $4.39$ (Burger King), $3.42$ (Wendy's), $2.49$ (Roy Rogers) and $-1.59$ (KFC). We average them with weights equal to the number of New Jersey stores in each chain ($126$, $39$, $77$ and $67$), which gives $\widehat\Phi=2.50$ full-time equivalents.  Its standard error, $1.26$, treats stores as independent and combines the sampling variances of the four within-chain contrasts, so $\widehat\Phi/\widehat{\mathrm{se}}=1.98$. The estimate is smaller than the unadjusted $2.75$ $(1.34)$ in row~4, column~(iii), of Table~\ref{tab:ck-replication} because the chains are distributed differently in the two states and the stratification reweights the Pennsylvania stores to the New Jersey chain mix. It is close to the authors' estimate with chain and ownership controls, $2.30$ $(1.20)$ \citep[Table~4, column (ii)]{card1994minimum}, which differs from ours in its weighting, its additional ownership control, and its sample of the $357$ stores with both employment and wage data in both waves. The sensitivity parameters are expressed in units of the standard deviation of full-time-equivalent employment in the second wave (period~1 of Section~\ref{sec:did}) across the $384$ stores, which is $9.18$ full-time equivalents.

\subsection{Results}\label{ck:results}
In the temporary-agency Illustration, $\delta$ is unknown, and the question is how much dependence the conclusion tolerates. Here, assignment is geographic and deterministic: conditional on $D_i=1$, the vector $\mathbf D_{-i}$ places every other New Jersey store in treatment and every Pennsylvania store in control, and conditional on $D_i=0$ it places every New Jersey store in treatment. For a fixed store, one of the two cells is empty, since $\Pr(D_i=1)\in\{0,1\}$. Under the law $P$ of a randomly drawn store, both cells have positive probability, $\Pr(D_I=1)=309/384$, and within each chain $0<\Pr(D_I=1\mid W_I)<1$. Given $D_I=1$, the drawn store is in New Jersey, and $308$ of the other $383$ stores are treated; given $D_I=0$, it is in Pennsylvania, and $309$ are treated. The two cell laws are therefore mixtures of point masses with disjoint supports, hence mutually singular, so $d_{\mathrm{TV}}(\pi_1,\pi_0)=1$ exactly. Generalized CAI does not just fail; it fails maximally, and $\bar\delta=1$ is not a worst-case scenario for argument but the design's actual value. The total-variation analysis therefore reduces to the size of $\bar\Lambda_0$. By Theorem~\ref{sensB:thm-bounds}(iii), the
identified set is $\widehat\Phi\pm\bar\Lambda_0$, and the sign survives if and only if cross-store spillovers on untreated employment are smaller than the estimate itself. That threshold is $\bar\Lambda^\ast_0=2.50$ full-time equivalents, or $0.27$ of an outcome standard deviation; allowing for sampling error, it is $0.42$, or $0.05$ of a standard deviation. Table~\ref{tab:ck} reports the sets: at spillovers of a tenth of a standard deviation the set is $[1.58,3.41]$, and at a quarter it is $[0.20,4.79]$ and still excludes zero, but none of the Imbens--Manski intervals excludes zero, not even at a twentieth of a standard deviation.

\begin{table}[h]
\centering\footnotesize
\caption{The sensitivity analysis for the minimum-wage study}\label{tab:ck}
\begin{tabular}{@{}lcccc@{}}
\toprule
$\bar\Lambda_0$ & Identified set & Imbens--Manski & $\delta^\ast$ & $\delta^\ast_\alpha$ \\
\midrule
$0.05$ SD $=0.46$ & $[2.04,\,2.96]$ & $[-0.13,\,5.12]$ & $1.00$ & $0.92$ \\
$0.10$ SD $=0.92$ & $[1.58,\,3.41]$ & $[-0.51,\,5.50]$ & $1.00$ & $0.46$ \\
$0.25$ SD $=2.30$ & $[0.20,\,4.79]$ & $[-1.87,\,6.86]$ & $1.00$ & $0.18$ \\
$0.50$ SD $=4.59$ & $[-2.10,\,7.09]$ & $[-4.17,\,9.16]$ & $0.54$ & $0.09$ \\
\midrule
\multicolumn{5}{@{}l}{\emph{Breakdown, at the design's $\bar\delta=1$}} \\
$\bar\Lambda^\ast_0$, $\bar\Lambda^\ast_{0,\alpha}$ & \multicolumn{4}{c}{$2.50$, $0.42$ FTE $\;=\;0.27$, $0.05$ outcome SD} \\
$k^\ast$, $k^\ast_\alpha$, under (AN) & \multicolumn{4}{c}{$208$, $35$} \\
\bottomrule
\end{tabular}

\smallskip
\parbox{\linewidth}{\scriptsize \emph{Notes:} Stratified difference-in-differences, $\widehat\Phi=2.50$ $(1.26)$ full-time equivalents on $384$ stores, $309$ in New Jersey; strata are the four chains, weighted by treated counts; period-1 outcome SD $9.18$. The design fixes $\bar\delta=1$, so the identified set is $\widehat\Phi\pm\bar\Lambda_0$ and $\delta^\ast$ is reported only for comparison with Table~\ref{tab:twa-layers}. One-sided, $\alpha=0.05$; the Imbens--Manski critical value solves $\Phi_{\mathcal N}(c+2\bar b/\widehat{\mathrm{se}})-\Phi_{\mathcal N}(-c)=1-\alpha$. Standard errors are conventional and assume independent outcomes; with five regions, clustering on the unit through which the dependence runs is not available. The Imbens--Manski intervals, $\delta^\ast_\alpha$, $\bar\Lambda^\ast_{0,\alpha}$ and $k^\ast_\alpha$ use these standard errors and may overstate precision; the identified sets, $\delta^\ast$, $\bar\Lambda^\ast_0$ and $k^\ast$ do not use them. $k^\ast$ and $k^\ast_\alpha$ at $L^S_0=0.5$ outcome SD under (AN) on the share of all $N-1$ other stores treated, whose cell gap is $-1/(N-1)=-0.0026$. Under (BI), with influence confined to same-region stores and a response to the own-region configuration common to the stores of a chain (Section~\ref{ck:results}), the breakdown is in the total influence of those stores and equals $\bar\Lambda^\ast_0$ and $\bar\Lambda^\ast_{0,\alpha}$.}
\end{table}

\begin{revblock}

\begin{figure}[h]
\centering
\resizebox{\textwidth}{!}{\input{ck_frontier_fig}}
\caption{Breakdown frontiers for the minimum-wage study}\label{fig:ck-frontier}
\begin{flushleft}
\scriptsize
\emph{Note:} Panel~(a): the solid curve is
$\delta^\ast(\bar\Lambda_0)=\min\{1,\widehat\Phi/\bar\Lambda_0\}$ and the dashed curve
replaces $\widehat\Phi$ by $\widehat\Phi-z_{0.95}\widehat{\mathrm{se}}$; the shaded region
collects the pairs that leave the sign intact once sampling error is allowed for. Unlike
Figure~\ref{fig:twa-frontier}, the design's own dependence is known and equal to one, drawn as
the horizontal line, so a conclusion survives only where a frontier is still flat, that is for
$\bar\Lambda_0$ below $2.50$ full-time equivalents at the point estimate and below $0.42$ once
sampling error is allowed for. Panel~(b): the share version, $b=L^S_0|\bar\Gamma|$, giving the
share gap at which the conclusion breaks down for each spillover slope, on a logarithmic axis.
The dotted line is the cell gap in the share of all other stores treated, $0.0026$, the gap under (AN). The solid grey line, at $0.31$, expresses the own-region gap of $1$ in the same units under (BI), with each same-region store carrying the per-store influence $L^S_0/(N-1)$ and other stores none: it is the average, over New Jersey stores, of the fraction of other stores in a store's own region. The frontier therefore crosses it where the total influence of same-region stores equals $\widehat\Phi$, at a slope of $7.94$ full-time equivalents per unit share ($0.86$ outcome standard deviations) at the point estimate and $1.35$ allowing for sampling error. At a slope of half an outcome standard deviation, the solid frontier ($0.54$) lies far above the dotted line and above the solid one, while the dashed frontier ($0.09$) lies above the dotted line but below the solid one. The dashed curves and the shaded region use conventional standard errors, which are likely too small here (Section~\ref{ck:results}).
\end{flushleft}
\end{figure}
\end{revblock}

Recall that in the Illustration using the TWA data, standard errors use a ZIP-code cluster bootstrap because dependence operates through many ZIP codes containing both treated and untreated workers. That approach is unavailable here. The public data report only five regions, three in New Jersey and two in Pennsylvania, and treatment is constant within region, so regional resampling yields few distinct samples and can omit an entire treatment arm. More fundamentally, a shock common to all New Jersey stores, such as the state-level downturn discussed by \citet[pp.~773 and~779]{card1994minimum}, would require clustering at the state level. With only two states and one treated state, cluster-robust inference is not credible, and the variance of a shock common to the treated state cannot be learned from that state alone. This resembles design~C of Table~\ref{tab:mc-sens}, except that even the available regional clusters are too few to resample.

Thus, we report conventional standard errors treating the $384$ stores as independent, while noting that they may understate sampling uncertainty. Positive within-region or within-state dependence, whether from common demand shocks or interference, raises the variance of $\widehat\Phi$ without entering $\widehat{\mathrm{se}}$. Hence, all results using $\widehat{\mathrm{se}}$ should be interpreted cautiously, including the $t$-ratio of $1.98$, the Imbens--Manski intervals, $\delta^\ast_\alpha$, $\bar\Lambda^\ast_{0,\alpha}=0.42$, $k^\ast_\alpha$, and the dashed frontier and shaded region in Figure~\ref{fig:ck-frontier}. By contrast, the identified sets, $\bar\Lambda^\ast_0=2.50$, $\delta^\ast$, and $k^\ast$ do not use $\widehat{\mathrm{se}}$ and are unaffected.

\paragraph{Does the total-variation verdict overstate the problem?} Theorem~\ref{sensB:thm-bounds} uses $\delta$, a worst case over every way in which other stores' treatments could matter, whereas Theorem~\ref{sensB:prop-refine} replaces $\delta$ by a measure tied to a restriction on how they matter. This design shows that the restriction, not the data, decides whether the replacement helps. Under (AN), a store's outcome depends on others only through the share of all $N-1$ other stores that are treated. A New Jersey store sees $308$ of the other $383$ stores treated and a Pennsylvania store sees $309$, so $\bar\Gamma=-1/(N-1)=-0.0026$ although $\delta=1$. For the drawn store, each cell law is a point mass at one configuration, and the two configurations' shares almost coincide, as described after Theorem~\ref{sensB:prop-refine}, where $\delta$ is maximal while the share gap is $O(1/N)$. Under (AN), the conclusion is therefore robust; at $L^S_0$ equal to half an outcome standard deviation, it breaks only at $k^\ast=208$ times the observed gap, or $35$ allowing for sampling error. But (AN) makes the location of treated stores irrelevant, while the channels of Section~\ref{ck:interference} are local.

If instead a store is influenced only by stores in its own region, as under (BI) with weights concentrated there, the offsetting disappears. Regions do not cross the state line, so every New Jersey store's own-region neighbors are treated, and every Pennsylvania store's are not, and the two configurations differ in every own-region coordinate. Because the cell laws are point masses, the bound equals the total influence of own-region stores, and the conclusion breaks when that influence reaches $\widehat\Phi=2.50$ full-time equivalents, or $0.27$ of an outcome standard deviation, and $0.42$ allowing for sampling error. These are the thresholds $\bar\Lambda^\ast_0$ and $\bar\Lambda^\ast_{0,\alpha}$ in Theorem~\ref{sensB:thm-bounds}; under local interference, the refinement adds nothing to the total-variation analysis. Figure~\ref{fig:ck-frontier} presents both cases. The total-variation verdict overstates the problem only if spillovers are anonymous across the whole market; in this setting the appropriate model is Theorem~\ref{sensB:thm-bounds} at the design's $\delta=1$, with $\bar\Lambda_0$ interpreted as the size of spillovers from nearby stores.

Two qualifications follow from the law $P$ of a randomly drawn store, under which the two cells contain different stores, New Jersey stores in one and Pennsylvania stores in the other, so that equal composition, condition~\eqref{unif:eq-ec} of the main text, fails by design. First, the bias also contains the difference in expected untreated trends between the New Jersey and Pennsylvania stores of a chain, and $\bar\Lambda_0$ bounds that difference together with spillovers; it measures spillovers alone when untreated trends, net of spillovers, are common across the stores of a chain, which is the cross-store parallel-trends condition behind the design. Second, (BI) links only configurations of the same drawn store and, without equal composition, places no bound on the bias by itself. The argument of the previous paragraph presumes, in addition, that the stores of a chain share a common response to the treatment status of their own-region stores, so that only the own-region configuration matters; under that presumption the bound equals the total influence of own-region stores, as stated. The threshold $\bar\Lambda^\ast_0=2.50$ full-time equivalents does not rely on it: it is the result of Theorem~\ref{sensB:thm-bounds} at the design's $\delta=1$. Likewise, the (AN) analysis applies under $P$ when the response $g$ to the share is common across stores.

\paragraph{Discussion.} The positive-sign conclusion for the ADTT-type target offered by the breakdown frontier and point depends on the scale of interference (Table~\ref{tab:ck}). Under \emph{anonymous} spillovers through the economy-wide treated share, condition (AN) of Theorem~\ref{sensB:prop-refine}(iii)--(iv), the relevant assignment dependence is the cell gap $|\bar\Gamma|=0.0026$, so the effect remains identified to be positive unless dependence is $k^\ast=208$ times this gap ($35$ allowing for sampling error). If instead spillovers operate through \emph{local markets}, as labor and product-market competition suggest, condition (BI) with influence confined to same-region stores, dependence is maximal among the units that can affect one another, so the refinement offers no gain and the analysis reduces to Theorem~\ref{sensB:thm-bounds} at the design's $\delta=1$. A same-region spillover on untreated employment of $\bar\Lambda^\ast_0=2.50$ full-time equivalents, or $0.27$ outcome standard deviations, is enough for zero to enter the identified set, or $\bar\Lambda^\ast_{0,\alpha}=0.42$, or $0.05$ standard deviations, once sampling error is incorporated using conventional standard errors. The refinement in Theorem~\ref{sensB:prop-refine} is therefore informative only when assignment dependence is small at the scale at which interference operates. In Section~\ref{sen appl}, local spillovers coincided with weak ZIP-level dependence; here, economically relevant local spillovers coincide with strong dependence.

Two qualifications are necessary. First, the analysis does not refute \citet{card1994minimum}. It characterizes what the difference-in-differences formula identifies under unrestricted interference with geographically structured assignment and yields a breakdown threshold of $0.27$ outcome standard deviations, which should be guided by evidence on local competition among fast-food restaurants. \citet[p.~788]{card1994minimum} themselves suggest, as one possible explanation of their wage results, that stores in the two states operate in separate product markets and that New Jersey stores compete in the same one, which refutes spillovers across the border but not against those among New Jersey stores. Second, as acknowledged earlier, no available standard error accounts for the dependence in this design, so the conclusions above that allow for sampling error, in particular the threshold of $0.05$ of an outcome standard deviation, rest on conventional standard errors that are likely too small.

\section{Empirical Studies in Which Interference Is Plausible}\label{table of papers}
\begingroup
\setlength{\tabcolsep}{4pt}
\renewcommand{\arraystretch}{1.1}
\singlespacing\small
\begin{longtable}{@{}>{\raggedright\arraybackslash}p{3.0cm}
                    >{\raggedright\arraybackslash}p{3.0cm}
                    >{\raggedright\arraybackslash}p{3.2cm}
                    >{\raggedright\arraybackslash}p{6.0cm}@{}}
\caption{Empirical studies in which interference is plausible and the interaction structure is unrecorded}
\label{tab:interference-literature} \\
\toprule
\textbf{Paper} & \textbf{Identification Strategy} & \textbf{Setting} & \textbf{Why Interference Is Plausible} \\
\midrule
\endfirsthead

\multicolumn{4}{c}%
{\tablename\ \thetable\ -- \textit{Continued from previous page}} \\
\toprule
\textbf{Paper} & \textbf{Identification Strategy} & \textbf{Setting} & \textbf{Why Interference Is Plausible} \\
\midrule
\endhead

\midrule
\multicolumn{4}{r}{\textit{Continued on next page}} \\
\endfoot

\bottomrule
\multicolumn{4}{@{}p{\linewidth}@{}}{\footnotesize\textit{Notes:} Studies are grouped by design. Channels describe why interference is plausible in each setting; inclusion does not question a study's conclusions: under generalized CAI, its estimate is an average direct effect (Proposition~\ref{unif:thm-master}).  Journals: AEJ: Applied, American Economic Journal: Applied Economics; AER, American Economic Review; JAE, Journal of Applied Econometrics; JASA, Journal of the American Statistical Association; JoE, Journal of Econometrics; JPE, Journal of Political Economy; QJE, Quarterly Journal of Economics.} \\
\endlastfoot

\multicolumn{4}{@{}l}{\textit{Selection on observables and randomized assignment}} \\[2pt]

\citet{lalonde1986evaluating}; \citet{dehejia1999causal}
& Selection on observables
& NSW job training program on employment outcomes
& General equilibrium and peer effects \\

\citet{ichino2008from}
& Propensity-score matching (ATT), with sensitivity analysis
& Temporary-agency work and permanent employment in Tuscany and Sicily
& Agency workers and controls compete for the same permanent vacancies within provinces; a single TWA company, with agencies in five of the nine provinces; take-up concentrated near the agency; SUTVA is invoked and defended by the low share treated \\

\citet{aizer2016long}
& Accepted vs.\ rejected applicants, county and cohort fixed effects
& Mothers' Pension cash transfers and sons' longevity
& County-run programs; local spillovers among poor families in the same town through neighborhoods, schools, and child-labor markets \\

\citet{abramitzky2012europes}
& Brother fixed effects (migrants vs.\ stayers)
& Norwegian mass migration to the United States
& Chain migration (many passages prepaid by friends or relatives) makes migration dependent within families and villages; emigration affects stayers' labor markets \\

\citet{deming2009early}
& Sibling (mother) fixed effects
& Head Start and young-adult outcomes
& Enrollment decided jointly within families; shared parental investment; the paper tests for older-to-younger sibling spillovers and finds little evidence \\

\citet{mian2012effects}
& Selection on observables/regression adjustment
& ``Cash for Clunkers'' consumption stimulus
& Integrated product markets and spatial substitution across regions \\

\citet{krueger1999experimental}
& Within-school random assignment; regression with school fixed effects and covariates
& Class size and test scores in the Tennessee Project STAR
& Within- and across-school interactions between students in different treatment groups \\[4pt]

\multicolumn{4}{@{}l}{\textit{Instrumental variables}} \\[2pt]

\citet{angrist1990lifetime}
& Draft-lottery eligibility for veteran status
& Vietnam-era military service and lifetime earnings
& Labor-market displacement between veterans and non-veterans of the same cohort; men born on the same day share the instrument \\

\citet{angrist1991does}
& Quarter of birth for years of schooling
& Compulsory schooling and earnings
& Human-capital externalities and peer effects within birth cohorts; the instrument is shared within cohort cells \\

\citet{card1995using}
& College proximity for schooling
& Returns to schooling (NLS Young Men)
& Local human-capital spillovers; respondents from the same area share the instrument \\

\citet{angrist1998children}
& Same-sex sibling composition for a third child
& Family size and parents' labor supply
& Local labor- and childcare-market effects and peer fertility effects; the instrument is close to independent across households, so generalized CAI is credible \\

\citet{finkelstein2012oregon}
& Lottery selection for Medicaid coverage
& Oregon Health Insurance Experiment
& Winning made the whole household eligible (shared finances and care-seeking); congestion of local providers; lottery by household to fill a fixed number of slots \\

\citet{nunn2011slave}
& Historical distance to the coast for slave exports
& Slave trade and trust in Africa
& Spatial and cultural spillovers across neighboring ethnic groups; the instrument is shared within ethnic groups \\[4pt]

\multicolumn{4}{@{}l}{\textit{Regression discontinuity}} \\[2pt]

\citet{lee2008randomized}
& Sharp regression discontinuity
& Electoral incumbency advantage
& Party resources, voter mobilization, and spillovers across districts \\

\citet{dell2010persistent}
& Geographic regression discontinuity
& Historical mining mita in Peru
& Trade, migration, and institutional spillovers across geographic boundaries \\

\citet{ludwig2007does}
& Sharp RD at the poverty rate of the 300th poorest county
& Head Start grant-application assistance and child mortality
& Outcomes include contagious diseases; assistance went to a fixed number of spatially clustered counties \\

\citet{meyersson2014islamic}
& Sharp RD on the Islamic party's win margin
& Islamic mayors and women's education in Turkey
& Girls attend high schools in neighboring district centers; province-wide party organization; metropolitan mayors govern sub-metro municipalities \\

\citet{angrist1999using}
& Fuzzy RD / IV: Maimonides' class-size rule
& Class size and test scores in Israel
& Classmates share one enrollment count, so crossing a cutoff changes all classmates' class sizes; peer composition \\

\citet{brollo2013political}
& Fuzzy RD at population thresholds for federal transfers
& Transfers and political corruption in Brazil
& Transfers are shares of a fixed state pool, so one municipality's bracket changes others' transfers; political competition across neighbors \\

\citet{popeleches2013going}
& RD at school admission cutoffs
& Attending a better school in Romania
& Capacity constraints (an admission displaces the marginal applicant); peer effects \\

\citet{crost2014aid}
& RD on within-province poverty rank
& Development aid and civil conflict in the Philippines
& Exactly the poorest quartile of each province is eligible; displacement of violence to neighboring municipalities; the authors test for spillovers to nearby municipalities and find none \\[4pt]

\multicolumn{4}{@{}l}{\textit{Difference-in-differences}} \\[2pt]

\citet{card1994minimum}
& Difference-in-differences (ATT under parallel trends)
& Minimum wage (New Jersey vs.\ Pennsylvania)
& Cross-border commuting, firm relocation, price and employment spillovers \\

\citet{ditella2004do}
& Block-by-month difference-in-differences
& Police protection and car theft in Buenos Aires
& Crime displacement to nearby blocks; the authors add one- and two-block distance indicators and reject complete displacement to adjacent blocks \\

\citet{bailey2015war}
& Event study with staggered adoption
& Community Health Centers and older adults' mortality
& Residents of neighboring counties may use centers across county lines; early centers clustered in urban counties \\

\citet{cengiz2019effect}
& Stacked event study of state minimum-wage increases
& Minimum wages and low-wage jobs
& Cross-state commuting and firm relocation; regionally clustered increases expose bordering control states \\

\citet{hornbeck2012enduring}
& Difference-in-differences / selection on observables
& Dust Bowl and agricultural productivity
& Migration, land reallocation, and labor-market equilibrium effects \\

\end{longtable}
\endgroup

\section{Proofs of the Design-Specific Results}\label{proofs}

This section proves Corollary~\ref{unif:cor-designs} and the design-specific Theorems~\ref{bias-so}, \ref{prop:iv}, \ref{prop:rd}, \ref{thm:did} and~\ref{thm:did_dynamic}, in the order in which they appear. Proposition~\ref{unif:thm-master} and the results of Section~\ref{robustness} are proved in Section~\ref{app:proofs}. Throughout, $\pi_d(\mathbf d_{-i};w):=\Pr(\mathbf D_{-i}=\mathbf d_{-i}\mid D_i=d,W_i=w)$ and $\pi(\mathbf d_{-i};w):=\Pr(\mathbf D_{-i}=\mathbf d_{-i}\mid W_i=w)$, as in Section~\ref{sec:mapping}.
\paragraph{Proof of Corollary~\ref{unif:cor-designs}.}\label{prf:unif:cor-designs}
\begin{proof}
In each part, we verify that the ITR formula has the form \eqref{unif:eq-template}, establish~\eqref{unif:eq-sep}, and evaluate the terms of Proposition~\ref{unif:thm-master}. The others' assignment variable is $\mathbf D_{-i}$ in parts~(a) and~(c)--(e), with generic value $\mathbf d_{-i}$, and $\mathbf Z_{-i}$, with generic value $\mathbf z_{-i}$, in part~(b), as in Table~\ref{unif:tab-mapping}.

\emph{Part~(a).} The functional in \eqref{ro} is \eqref{unif:eq-template} with the stated choices. On the event $\{D_i=d,\mathbf D_{-i}=\mathbf d_{-i}\}$, consistency gives $Y_i=Y_i(d,\mathbf d_{-i})$. By condition \eqref{eq:uso}, $Y_i(d,\mathbf d_{-i})\indep(D_i,\mathbf D_{-i})\mid W_i$; hence, by overlap,
\begin{align*}
\mathbbm{E}[Y_i\mathbbm{1}(\mathbf D_{-i}=\mathbf d_{-i})\mid D_i=d,W_i=w]
&=\frac{\mathbbm{E}[Y_i(d,\mathbf d_{-i})\mathbbm{1}(D_i=d,\mathbf D_{-i}=\mathbf d_{-i})\mid W_i=w]}{\Pr(D_i=d\mid W_i=w)}\\
&=\bar Y_i(d,\mathbf d_{-i};w)\,\pi_d(\mathbf d_{-i};w),
\end{align*}
which is~\eqref{unif:eq-sep} with $\mu_d=\bar Y_i(d,\cdot\,;\cdot)$; thus $\theta(\mathbf d_{-i};w)=\mathbbm{E}[\Delta_i(\mathbf d_{-i})\mid W_i=w]$. The same argument applied to $\Delta_i(\mathbf d_{-i})$ gives $\sum_{\mathbf d_{-i}}\theta(\mathbf d_{-i};w)\pi_d(\mathbf d_{-i};w)=\mathbbm{E}[\Delta_i(\mathbf D_{-i})\mid D_i=d,W_i=w]$ and $\sum_{\mathbf d_{-i}}\theta(\mathbf d_{-i};w)\pi(\mathbf d_{-i};w)=\mathbbm{E}[\Delta_i(\mathbf D_{-i})\mid W_i=w]$, so the first terms of Proposition~\ref{unif:thm-master}(i) for $\pi^\ast=\pi_1,\pi_0,\pi$ are the ADEs in \eqref{eq:decomp-adtc-x}--\eqref{eq:decomp-ade-x}. By \eqref{unif:eq-channels}, $\mathcal B(\pi_1)=\sum_w\Pr(W_i=w)\sum_{\mathbf d_{-i}}\bar Y_i(0,\mathbf d_{-i};w)[\pi_1(\mathbf d_{-i};w)-\pi_0(\mathbf d_{-i};w)]=B^1_{ADE}$. By \eqref{unif:eq-bias}, $\mathcal B(\pi_0)=\sum_w\Pr(W_i=w)\sum_{\mathbf d_{-i}}\bar Y_i(1,\mathbf d_{-i};w)[\pi_1(\mathbf d_{-i};w)-\pi_0(\mathbf d_{-i};w)]=B^0_{ADE}$, and $\mathcal B(\pi)$ with both slice means recentered at $\mathbf d^\circ_{-i}=\mathbf d'_{-i}$ is $B_{ADE}$. Under overlap, $\pi_1(\cdot\,;w)=\pi_0(\cdot\,;w)$ for all $w$ is equivalent to $\mathbf D_{-i}\indep D_i\mid W_i$ for each unit together with, under $P$, equal composition, and Proposition~\ref{unif:thm-master}(ii) gives \eqref{eq:ade-x} because $\pi_1=\pi$ under CAI.

\emph{Part~(b).} The Wald ratio \eqref{iv wald} is \eqref{unif:eq-template} with the stated choices. By condition \eqref{iv: no interference in treatment}, $D_i=D_i(Z_i)$, so $\mathbbm{E}[D_i\mid Z_i=z]=\mathbbm{E}[D_i(z)]$ by condition \eqref{iv indep} and, under $P$, the common instrument propensity, and by condition \eqref{iv: mono}, $\kappa=\mathbbm{E}[D_i(1)-D_i(0)]=\Pr(D_i(1)>D_i(0))$; no potential outcome enters $\kappa$. On $\{Z_i=z,\mathbf Z_{-i}=\mathbf z_{-i}\}$, conditions \eqref{iv: no interference in treatment}, \eqref{iv exclusion} and \eqref{iv:consistency} and consistency give
\[
Y_i=Y_i(D_i(z),\mathbf D_{-i}(\mathbf z_{-i}))=\sum_{d,\mathbf d_{-i}}Y_i(d,\mathbf d_{-i})\,\mathbbm{1}\big(D_i(z)=d,\mathbf D_{-i}(\mathbf z_{-i})=\mathbf d_{-i}\big).
\]
Each summand is a function of $(Y_i(d,\mathbf d_{-i}),D_i(z),\mathbf D_{-i}(\mathbf z_{-i}))$, which is independent of $(Z_i,\mathbf Z_{-i})$ by condition \eqref{iv indep}. Therefore
\[
\mathbbm{E}[Y_i\mathbbm{1}(\mathbf Z_{-i}=\mathbf z_{-i})\mid Z_i=z]
=\mathbbm{E}\big[Y_i(D_i(z),\mathbf D_{-i}(\mathbf z_{-i}))\big]\Pr(\mathbf Z_{-i}=\mathbf z_{-i}\mid Z_i=z),
\]
which is~\eqref{unif:eq-sep} with $\mu_z(\mathbf z_{-i})=\mathbbm{E}[Y_i(D_i(z),\mathbf D_{-i}(\mathbf z_{-i}))]$. Pointwise, $Y_i(D_i(1),\mathbf d)-Y_i(D_i(0),\mathbf d)=\{D_i(1)-D_i(0)\}\Delta_i(\mathbf d)$, and $D_i(1)-D_i(0)=\mathbbm{1}(D_i(1)>D_i(0))$ by condition \eqref{iv: mono}; hence $\theta(\mathbf z_{-i})=\mathbbm{E}[\Delta_i(\mathbf D_{-i}(\mathbf z_{-i}))\mid D_i(1)>D_i(0)]$; under $P$, for the drawn unit $k$, $\mu_1-\mu_0=\kappa_k\,\mathbbm{E}[\Delta_k(\cdot)\mid D_k(1)>D_k(0)]$ with $\kappa_k$ its complier share, so $\theta$ carries the factor $\kappa_k/\kappa$. With $\pi^\ast(\mathbf z_{-i})=\Pr(\mathbf Z_{-i}=\mathbf z_{-i})$, joint independence of the instruments from potential outcomes and potential treatments and condition \eqref{iv:consistency} give $\sum_{\mathbf z_{-i}}\theta(\mathbf z_{-i})\pi^\ast(\mathbf z_{-i})=\mathbbm{E}[\Delta_i(\mathbf D_{-i})\mid D_i(1)>D_i(0)]=\mathrm{LADE}$. Finally, conditioning on $D_i(z)$ and $\mathbf D_{-i}(\mathbf z_{-i})$ shows
\[
\sum_{d}\sum_{\mathbf d_{-i}}\bar Y(d,\mathbf d_{-i};\mathbf z_{-i})\Pr\big(\mathbf D_{-i}(\mathbf z_{-i})=\mathbf d_{-i}\mid D_i(z)=d\big)\Pr\big(D_i(z)=d\big)=\mu_z(\mathbf z_{-i}),
\]
so $B_z=\sum_{\mathbf z_{-i}}\mu_z(\mathbf z_{-i})\zeta_z(\mathbf z_{-i})$ with $\zeta_z=\pi_z-\pi^\ast$, and \eqref{unif:eq-bias} gives $\mathcal B(\pi^\ast)=(B_1-B_0)/\kappa=\mathrm{Bias}_{\mathrm{IV}}$. Condition \eqref{unif:eq-gcai} states that $\Pr(\mathbf Z_{-i}=\cdot\mid Z_i=1)=\Pr(\mathbf Z_{-i}=\cdot\mid Z_i=0)$, which is Assumption~\ref{Instrument indepen} with, under $P$, a common instrument propensity; then $\pi_1=\pi^\ast$, and Proposition~\ref{unif:thm-master}(ii) gives $\Phi_{\mathrm{ITR}}(P)=\mathrm{LADE}$.

\emph{Part~(c).} In a sharp design, $\lim_{w\downarrow w_0}\mathbbm{E}[Y_i\mid W_i=w]-\lim_{w\uparrow w_0}\mathbbm{E}[Y_i\mid W_i=w]$ is \eqref{unif:eq-template} with the one-sided cell expectations. For $w>w_0$, $D_i=1$, so $Y_i\mathbbm{1}(\mathbf D_{-i}=\mathbf d_{-i})=\{\alpha_i(\mathbf d_{-i})+\Delta(\mathbf d_{-i})\}\mathbbm{1}(\mathbf D_{-i}=\mathbf d_{-i})$ under homogeneity, and by condition \eqref{rdd int: uncon},
$\mathbbm{E}[Y_i\mathbbm{1}(\mathbf D_{-i}=\mathbf d_{-i})\mid W_i=w]=\{\mathbbm{E}[\alpha_i(\mathbf d_{-i})\mid W_i=w]+\Delta(\mathbf d_{-i})\}\Pr(\mathbf D_{-i}=\mathbf d_{-i}\mid W_i=w)$;
for $w<w_0$ the same holds without $\Delta(\mathbf d_{-i})$. Taking one-sided limits and using condition \eqref{rdd int: cont } and the existence of $\pi_\pm$ gives~\eqref{unif:eq-sep} with $\mu_1=\bar\alpha+\Delta$, $\mu_0=\bar\alpha$, and $\theta=\Delta$. {With $\pi^\ast=\pi_+$, the first term is $\sum_{\mathbf d_{-i}}\Delta(\mathbf d_{-i})\pi_+(\mathbf d_{-i})=\mathrm{ADE}_{+}$, the direct-effect channel in \eqref{unif:eq-channels} is zero, and the baseline channel is $\sum_{\mathbf d_{-i}}\bar\alpha(\mathbf d_{-i})[\pi_+(\mathbf d_{-i})-\pi_-(\mathbf d_{-i})]=\mathrm{Bias}_{\mathrm{RDD}}$, which is Theorem~\ref{prop:rd}(i).} If $\pi_+=\pi_-=\pi_{w_0}$, Proposition~\ref{unif:thm-master}(ii) gives $\Phi_{\mathrm{ITR}}(P)=\sum_{\mathbf d_{-i}}\Delta(\mathbf d_{-i})\pi_{w_0}(\mathbf d_{-i})${, which is Theorem~\ref{prop:rd}(ii)}. Under Assumption~\ref{running indepen}, $\mathbf D_{-i}=(\mathbbm{1}(W_j\ge w_0))_{j\neq i}$ is a function of $\mathbf W_{-i}\indep W_i$, so $\Pr(\mathbf D_{-i}=\mathbf d_{-i}\mid W_i=w)=\Pr(\mathbf D_{-i}=\mathbf d_{-i})$ for all $w$; hence, under $P$ together with the density condition of Theorem~\ref{prop:rd}(iii), $\pi_+=\pi_-=\Pr(\mathbf D_{-i}=\cdot\,)$, and $\Phi_{\mathrm{ITR}}(P)=\mathbbm{E}[\Delta(\mathbf D_{-i})]$, which is Theorem~\ref{prop:rd}{(iii)}. Without homogeneity, $Y_i\mathbbm{1}(\mathbf D_{-i}=\mathbf d_{-i})=Y_i(1,\mathbf d_{-i})\mathbbm{1}(\mathbf D_{-i}=\mathbf d_{-i})$ for $w>w_0$ and $Y_i\mathbbm{1}(\mathbf D_{-i}=\mathbf d_{-i})=Y_i(0,\mathbf d_{-i})\mathbbm{1}(\mathbf D_{-i}=\mathbf d_{-i})$ for $w<w_0$. If conditions \eqref{rdd int: uncon} and \eqref{rdd int: cont } hold for both $Y_i(0,\mathbf d_{-i})$ and $Y_i(1,\mathbf d_{-i})$, the same steps give $\mu_d(\mathbf d_{-i})=\lim_{w\to w_0}\mathbbm{E}[Y_i(d,\mathbf d_{-i})\mid W_i=w]$, so $\theta(\mathbf d_{-i})=\tau(\mathbf d_{-i};w_0)$, and Proposition~\ref{unif:thm-master}(ii) gives the final claim.

\emph{Part~(d).} By linearity, \eqref{DiD indentification formula} equals $\mathbbm{E}\big[\mathbbm{E}[Y_{i1}-Y_{i0}\mid W_i,D_i=1]-\mathbbm{E}[Y_{i1}-Y_{i0}\mid W_i,D_i=0]\mid D_i=1\big]$, which is \eqref{unif:eq-template} with the stated choices; strong overlap ensures that both cells are well defined for $F_{W\mid D=1}$-almost all $w$. On $\{D_i=0,\mathbf D_{-i}=\mathbf d_{-i}\}$, consistency gives $Y_{i1}-Y_{i0}=Y_{i1}(0,\mathbf d_{-i})-Y_{i0}(\mathbf 0)$, so by condition \eqref{unif:did-meanindep},
\begin{align*}
&\mathbbm{E}[(Y_{i1}-Y_{i0})\mathbbm{1}(\mathbf D_{-i}=\mathbf d_{-i})\mid W_i=w,D_i=0]\\
&\qquad=\mathbbm{E}[Y_{i1}(0,\mathbf d_{-i})-Y_{i0}(\mathbf 0)\mid W_i=w,D_i=0]\,\pi_0(\mathbf d_{-i};w)\\
&\qquad=\mu_0(\mathbf d_{-i};w)\,\pi_0(\mathbf d_{-i};w),
\end{align*}
where both sides are zero if $\pi_0(\mathbf d_{-i};w)=0$. For $d=1$, on $\{D_i=1,\mathbf D_{-i}=\mathbf d_{-i}\}$ consistency gives $Y_{i1}-Y_{i0}=Y_{i1}(1,\mathbf d_{-i})-Y_{i0}(\mathbf 0)$, so condition \eqref{unif:did-meanindep-treated} gives the same factorization with $\mu_1(\mathbf d_{-i};w)=\mathbbm{E}[Y_{i1}(1,\mathbf d_{-i})-Y_{i0}(\mathbf 0)\mid W_i=w,D_i=1]$, which is~\eqref{unif:eq-sep} with the $\mu_d$ of Assumption~\ref{unif:ass-sep}. By condition \eqref{PT}, $\mu_0(\mathbf d_{-i};w)=\mathbbm{E}[Y_{i1}(0,\mathbf d_{-i})-Y_{i0}(\mathbf 0)\mid W_i=w,D_i=1]$, so $\theta(\mathbf d_{-i};w)=\mathbbm{E}[Y_{i1}(1,\mathbf d_{-i})-Y_{i1}(0,\mathbf d_{-i})\mid W_i=w,D_i=1]$, and on the support of $\pi_1(\cdot\,;w)$ conditions \eqref{unif:did-meanindep} and \eqref{unif:did-meanindep-treated} with $d=1$ equate it to the same mean given also $\mathbf D_{-i}=\mathbf d_{-i}$. Taking $\pi^\ast=\pi_1$, the law of iterated expectations gives
\begin{align*}
&\mathbbm{E}\Big[\sum_{\mathbf d_{-i}}\theta(\mathbf d_{-i};W_i)\pi_1(\mathbf d_{-i};W_i)\,\Big|\,D_i=1\Big]\\
&\qquad=\mathbbm{E}\big[\mathbbm{E}[Y_{i1}(1,\mathbf D_{-i})-Y_{i1}(0,\mathbf D_{-i})\mid W_i,D_i=1]\mid D_i=1\big]
=\mathrm{ADTT},
\end{align*}
the values of $\theta$ off the support of $\pi_1$ being immaterial. The direct-effect  component vanishes for $\pi^\ast=\pi_1$, and the baseline component is $\mathbbm{E}\big[\sum_{\mathbf d_{-i}}\mu_0(\mathbf d_{-i};W_i)\{\pi_1(\mathbf d_{-i};W_i)-\pi_0(\mathbf d_{-i};W_i)\}\mid D_i=1\big]$. By condition \eqref{PT}, $\mu_0(\mathbf d_{-i};w)=\mathbbm{E}[Y_{i1}(0,\mathbf d_{-i})-Y_{i0}(\mathbf 0)\mid W_i=w,D_i=d]$ for both $d\in\{0,1\}$, so $\sum_{\mathbf d_{-i}}\mu_0(\mathbf d_{-i};w)\{\pi_1(\mathbf d_{-i};w)-\pi_0(\mathbf d_{-i};w)\}=\text{trend}_1(w)-\text{trend}_0(w)$ and $\mathcal B(\pi_1)=\mathrm{Bias}_{\mathrm{DiD}}$. condition~\eqref{eq:CAI}, with equal composition under $P$, gives $\pi_1=\pi_0$, hence \eqref{unif:eq-gcai}, and Proposition~\ref{unif:thm-master}(ii) gives Theorem~\ref{thm:did}(ii).

\emph{Part~(e).} When anticipation is allowed, $Y_{i1}-Y_{i0}=Y_{i1}(d,\mathbf d_{-i})-Y_{i0}(d,\mathbf d_{-i})$ on $\{D_i=d,\mathbf D_{-i}=\mathbf d_{-i}\}$. The argument in part~(d) applies with condition \eqref{unif:dyn-meanindep} in place of condition \eqref{unif:did-meanindep}, $\mu_0(\mathbf d_{-i};w)=\mathbbm{E}[Y_{i1}(0,\mathbf d_{-i})-Y_{i0}(0,\mathbf d_{-i})\mid W_i=w,D_i=0]$, and, by condition \eqref{unif:dyn-meanindep-treated}, $\mu_1(\mathbf d_{-i};w)=\mathbbm{E}[Y_{i1}(1,\mathbf d_{-i})-Y_{i0}(1,\mathbf d_{-i})\mid W_i=w,D_i=1]$, equal on the support of $\pi_1(\cdot\,;w)$ to the same mean given also $\mathbf D_{-i}=\mathbf d_{-i}$. By Assumption~\ref{did: parallel trend dyn}, on that support,
\begin{align*}
\theta(\mathbf d_{-i};w)&=\mathbbm{E}[Y_{i1}(1,\mathbf d_{-i})-Y_{i1}(0,\mathbf d_{-i})\mid W_i=w,D_i=1,\mathbf D_{-i}=\mathbf d_{-i}]\\
&\quad-\mathbbm{E}[Y_{i0}(1,\mathbf d_{-i})-Y_{i0}(0,\mathbf d_{-i})\mid W_i=w,D_i=1,\mathbf D_{-i}=\mathbf d_{-i}],
\end{align*}
and the second term is zero by Assumption~\ref{did: no anticipation dyn}, including condition \eqref{unif:dyn-antic-meanindep}. The remaining steps are as in part~(d), with $\text{trend}_d$ replaced by $\text{trend}^{\mathrm{dyn}}_d$.
\end{proof}
\paragraph{Proof of Theorem~\ref{bias-so}.}\label{prf:bias-so}
\begin{proof}
The proof expresses each conditional mean of the observed outcome as an average of mean potential outcomes over the conditional law of others' treatments and then adds and subtracts the corresponding average under each of the three reference laws. Let $m_d(w):=\mathbbm{E}[Y_i\mid D_i=d,W_i=w]$ for $d\in\{0,1\}$, so that $\Phi_{\mathrm{ITR}}(P)=\mathbbm{E}[m_1(W_i)-m_0(W_i)]$ by \eqref{ro}. Fix $i$ and $w$; all expectations are under $P$, and a configuration identifies the drawn unit, to which the unit-level conditions apply. By Assumption~\ref{ass:consistency}, $m_d(w)=\mathbbm{E}[Y_i(d,\mathbf D_{-i})\mid D_i=d,W_i=w]$. By the law of iterated expectations with respect to $\mathbf D_{-i}$ and then condition \eqref{eq:uso},
\begin{align}
m_d(w)
&=\sum_{\mathbf d_{-i}}\mathbbm{E}\big[Y_i(d,\mathbf d_{-i})\mid D_i=d,W_i=w,\mathbf D_{-i}=\mathbf d_{-i}\big]\,\pi_d(\mathbf d_{-i};w)\notag\\
&=\sum_{\mathbf d_{-i}}\bar Y_i(d,\mathbf d_{-i};w)\,\pi_d(\mathbf d_{-i};w),
\label{soo: expansion}
\end{align}
where $\bar Y_i(d,\mathbf d_{-i};w):=\mathbbm{E}[Y_i(d,\mathbf d_{-i})\mid W_i=w]$. Hence
\begin{equation}
m_1(w)-m_0(w)=\sum_{\mathbf d_{-i}}\bar Y_i(1,\mathbf d_{-i};w)\,\pi_1(\mathbf d_{-i};w)-\sum_{\mathbf d_{-i}}\bar Y_i(0,\mathbf d_{-i};w)\,\pi_0(\mathbf d_{-i};w).
\label{eq:decomp}
\end{equation}

\emph{Part~(i), representation \eqref{eq:decomp-adtc-x}.} Adding and subtracting $\sum_{\mathbf d_{-i}}\bar Y_i(0,\mathbf d_{-i};w)\,\pi_1(\mathbf d_{-i};w)$ in \eqref{eq:decomp} gives
\begin{align}
m_1(w)-m_0(w)&=\sum_{\mathbf d_{-i}}\big\{\bar Y_i(1,\mathbf d_{-i};w)-\bar Y_i(0,\mathbf d_{-i};w)\big\}\pi_1(\mathbf d_{-i};w)\notag\\
&\quad+\sum_{\mathbf d_{-i}}\bar Y_i(0,\mathbf d_{-i};w)\big\{\pi_1(\mathbf d_{-i};w)-\pi_0(\mathbf d_{-i};w)\big\}.
\label{eq:decomp 2}
\end{align}
By condition \eqref{eq:uso}, $\bar Y_i(1,\mathbf d_{-i};w)-\bar Y_i(0,\mathbf d_{-i};w)=\mathbbm{E}[\Delta_i(\mathbf d_{-i})\mid D_i=1,W_i=w,\mathbf D_{-i}=\mathbf d_{-i}]$, so by the law of iterated expectations the first sum in \eqref{eq:decomp 2} equals $\mathbbm{E}[\Delta_i(\mathbf D_{-i})\mid D_i=1,W_i=w]$. Averaging over $\Pr(W_i=w)$ gives
\begin{equation}
\Phi_{\mathrm{ITR}}(P)=\mathbbm{E}\big[\mathbbm{E}[\Delta_i(\mathbf D_{-i})\mid D_i=1,W_i]\big]+B^1_{ADE},
\label{eq:bias 1}
\end{equation}
with $B^1_{ADE}$ as defined in Theorem~\ref{bias-so}.

\emph{Part~(i), representation \eqref{eq:decomp-adtt-x}.} Adding and subtracting $\sum_{\mathbf d_{-i}}\bar Y_i(1,\mathbf d_{-i};w)\,\pi_0(\mathbf d_{-i};w)$ in \eqref{eq:decomp} gives
\begin{align*}
m_1(w)-m_0(w)&=\sum_{\mathbf d_{-i}}\big\{\bar Y_i(1,\mathbf d_{-i};w)-\bar Y_i(0,\mathbf d_{-i};w)\big\}\pi_0(\mathbf d_{-i};w)\\
&\quad+\sum_{\mathbf d_{-i}}\bar Y_i(1,\mathbf d_{-i};w)\big\{\pi_1(\mathbf d_{-i};w)-\pi_0(\mathbf d_{-i};w)\big\}.
\end{align*}
The argument that led to \eqref{eq:bias 1}, with $D_i=0$ in place of $D_i=1$, gives
\begin{equation}
\Phi_{\mathrm{ITR}}(P)=\mathbbm{E}\big[\mathbbm{E}[\Delta_i(\mathbf D_{-i})\mid D_i=0,W_i]\big]+B^0_{ADE},
\label{eq:bias 2}
\end{equation}
with $B^0_{ADE}$ as defined in Theorem~\ref{bias-so}.

\emph{Part~(i), representation \eqref{eq:decomp-ade-x}.} Let $\mathbf d'_{-i}\in\{0,1\}^{N-1}$ be the reference realization in Theorem~\ref{bias-so}, held fixed throughout (it may depend on $w$). By \eqref{soo: expansion},
\begin{equation}
\Phi_{\mathrm{ITR}}(P)=\sum_{w\in\mathcal W}\Pr(W_i=w)\Big\{\sum_{\mathbf d_{-i}}\bar Y_i(1,\mathbf d_{-i};w)\,\pi_1(\mathbf d_{-i};w)-\sum_{\mathbf d_{-i}}\bar Y_i(0,\mathbf d_{-i};w)\,\pi_0(\mathbf d_{-i};w)\Big\}.
\label{eq:phi-expand}
\end{equation}
By the law of iterated expectations and condition \eqref{eq:uso}, the ADE equals
\begin{equation}
\mathbbm{E}\big[\mathbbm{E}[\Delta_i(\mathbf D_{-i})\mid W_i]\big]=\sum_{w\in\mathcal W}\Pr(W_i=w)\sum_{\mathbf d_{-i}}\big\{\bar Y_i(1,\mathbf d_{-i};w)-\bar Y_i(0,\mathbf d_{-i};w)\big\}\pi(\mathbf d_{-i};w).
\label{eq:ade-expand}
\end{equation}
For $d\in\{0,1\}$, let $b_d(w):=\sum_{\mathbf d_{-i}}\bar Y_i(d,\mathbf d_{-i};w)\{\pi_d(\mathbf d_{-i};w)-\pi(\mathbf d_{-i};w)\}$. Adding and subtracting the right-hand side of \eqref{eq:ade-expand} in \eqref{eq:phi-expand} gives
\begin{align}
\Phi_{\mathrm{ITR}}(P)
&=\sum_{w\in\mathcal W}\Pr(W_i=w)\Big[\sum_{\mathbf d_{-i}}\big\{\bar Y_i(1,\mathbf d_{-i};w)-\bar Y_i(0,\mathbf d_{-i};w)\big\}\pi(\mathbf d_{-i};w)+b_1(w)-b_0(w)\Big]
\label{eq:phi-split}\\
&=\mathrm{ADE}+\sum_{w\in\mathcal W}\Pr(W_i=w)\big\{b_1(w)-b_0(w)\big\},
\label{eq:phi-ade-bias0}
\end{align}
where the second equality is \eqref{eq:ade-expand}. For each $d\in\{0,1\}$ and $w\in\mathcal W$, both $\pi_d(\cdot\,;w)$ and $\pi(\cdot\,;w)$ are probability distributions on $\{0,1\}^{N-1}$, so
\begin{equation}
\sum_{\mathbf d_{-i}}\big\{\pi_d(\mathbf d_{-i};w)-\pi(\mathbf d_{-i};w)\big\}=1-1=0.
\label{eq:zero-sum}
\end{equation}
Consequently $\sum_{\mathbf d_{-i}}\bar Y_i(d,\mathbf d'_{-i};w)\{\pi_d(\mathbf d_{-i};w)-\pi(\mathbf d_{-i};w)\}=\bar Y_i(d,\mathbf d'_{-i};w)\cdot 0=0$, and subtracting this term from $b_d(w)$ gives
\begin{equation}
b_d(w)=\sum_{\mathbf d_{-i}}\big\{\bar Y_i(d,\mathbf d_{-i};w)-\bar Y_i(d,\mathbf d'_{-i};w)\big\}\big\{\pi_d(\mathbf d_{-i};w)-\pi(\mathbf d_{-i};w)\big\}.
\label{eq:recenter}
\end{equation}
Substituting \eqref{eq:recenter} for $d=1$ and $d=0$ into \eqref{eq:phi-ade-bias0} gives $\Phi_{\mathrm{ITR}}(P)=\mathrm{ADE}+B_{ADE}$, which is \eqref{eq:decomp-ade-x}.

\emph{Part~(ii).} By condition~\eqref{eq:CAI}, $\mathbf D_{-i}\indep D_i\mid W_i$, so $\pi_1(\mathbf d_{-i};w)=\pi_0(\mathbf d_{-i};w)=\pi(\mathbf d_{-i};w)$ for all $\mathbf d_{-i}$ and $w$. The second sum in \eqref{eq:decomp 2} is then zero, and
\[
m_1(w)-m_0(w)=\sum_{\mathbf d_{-i}}\big\{\bar Y_i(1,\mathbf d_{-i};w)-\bar Y_i(0,\mathbf d_{-i};w)\big\}\pi(\mathbf d_{-i};w).
\]
Averaging over $W_i$ and using \eqref{eq:ade-expand} gives $\Phi_{\mathrm{ITR}}(P)=\mathbbm{E}[m_1(W_i)-m_0(W_i)]=\mathbbm{E}\big[\mathbbm{E}[\Delta_i(\mathbf D_{-i})\mid W_i]\big]$, which is \eqref{eq:ade-x}. Under condition~\eqref{eq:CAI} and Assumption~\ref{ass: uncon}, the ITR identification formula therefore recovers the ADE.
\end{proof}

\paragraph{Proof of Theorem~\ref{prop:iv}.}\label{prf:prop:iv}
\begin{proof}
The proof expands the reduced form over own and others' potential treatments and others' instruments, isolates the terms that depend on the law of $\mathbf Z_{-i}$ given $Z_i$, and divides by the first stage. The argument is stated for each unit. Under the law $P$ of a randomly drawn unit with a common instrument propensity, the IV form of equal composition~\eqref{unif:eq-ec}, $\Pr(I=k\mid Z_I=z)=N^{-1}$ for both $z$, so each pooled cell mean, the reduced form and the first stage are averages of their unit-level counterparts with equal weights, and the first stage equals $\kappa$. For $z,d\in\{0,1\}$ and $\mathbf z_{-i},\mathbf d_{-i}\in\{0,1\}^{N-1}$, let $\pi_z(\mathbf z_{-i}):=\Pr(\mathbf Z_{-i}=\mathbf z_{-i}\mid Z_i=z)$, $\pi^\ast(\mathbf z_{-i}):=\Pr(\mathbf Z_{-i}=\mathbf z_{-i})$, $\upsilon_z(d):=\Pr(D_i(z)=d)$, and $\chi_z(\mathbf d_{-i}\mid d;\mathbf z_{-i}):=\Pr(\mathbf D_{-i}(\mathbf z_{-i})=\mathbf d_{-i}\mid D_i(z)=d)$, so that $\zeta_z=\pi_z-\pi^\ast$.

\emph{Part~(i).} By condition \eqref{iv exclusion} and consistency, together with conditions \eqref{iv: no interference in treatment} and \eqref{iv:consistency}, $Y_i=Y_i(D_i,\mathbf D_{-i})=Y_i(D_i(Z_i),\mathbf D_{-i}(\mathbf Z_{-i}))$, so for each $z\in\{0,1\}$,
\[
\mathbbm{E}[Y_i\mid Z_i=z]=\mathbbm{E}\big[Y_i(D_i(z),\mathbf D_{-i}(\mathbf Z_{-i}))\mid Z_i=z\big].
\]
By the law of iterated expectations, conditioning successively on $D_i(z)$, $\mathbf Z_{-i}$ and $\mathbf D_{-i}(\mathbf z_{-i})$,
\begin{align}
&\mathbbm{E}\big[Y_i(D_i(z),\mathbf D_{-i}(\mathbf Z_{-i}))\mid Z_i=z\big]\notag\\
&\quad=\sum_{d\in\{0,1\}}\sum_{\mathbf z_{-i}}\sum_{\mathbf d_{-i}}
\mathbbm{E}\big[Y_i(d,\mathbf d_{-i})\mid Z_i=z,D_i(z)=d,\mathbf Z_{-i}=\mathbf z_{-i},\mathbf D_{-i}(\mathbf z_{-i})=\mathbf d_{-i}\big]\notag\\
&\qquad\cdot\Pr\big(\mathbf D_{-i}(\mathbf z_{-i})=\mathbf d_{-i}\mid Z_i=z,D_i(z)=d,\mathbf Z_{-i}=\mathbf z_{-i}\big)\notag\\
&\qquad\cdot\Pr\big(\mathbf Z_{-i}=\mathbf z_{-i}\mid Z_i=z,D_i(z)=d\big)\Pr\big(D_i(z)=d\mid Z_i=z\big).
\label{expansion_Zminus_general}
\end{align}
Condition~\eqref{iv indep} holds jointly for potential outcomes and potential treatments, and it implies the following:
\begin{enumerate}
\item[(a)] $D_i(z)\indep(Z_i,\mathbf Z_{-i})$, so $\Pr(D_i(z)=d\mid Z_i=z)=\upsilon_z(d)$ and $\Pr(\mathbf Z_{-i}=\mathbf z_{-i}\mid Z_i=z,D_i(z)=d)=\pi_z(\mathbf z_{-i})$;
\item[(b)] $\mathbf D_{-i}(\mathbf z_{-i})\indep(Z_i,\mathbf Z_{-i})\mid D_i(z)$, so $\Pr(\mathbf D_{-i}(\mathbf z_{-i})=\mathbf d_{-i}\mid Z_i=z,D_i(z)=d,\mathbf Z_{-i}=\mathbf z_{-i})=\chi_z(\mathbf d_{-i}\mid d;\mathbf z_{-i})$;
\item[(c)] $Y_i(d,\mathbf d_{-i})\indep(Z_i,\mathbf Z_{-i})\mid(D_i(z),\mathbf D_{-i}(\mathbf z_{-i}))$, so the conditional expectation in \eqref{expansion_Zminus_general} equals $\mathbbm{E}[Y_i(d,\mathbf d_{-i})\mid D_i(z)=d,\mathbf D_{-i}(\mathbf z_{-i})=\mathbf d_{-i}]$.
\end{enumerate}
Substituting (a)--(c) into \eqref{expansion_Zminus_general} and using the definition of $\bar Y(d,\mathbf d_{-i};\mathbf z_{-i})$ in Theorem~\ref{prop:iv} gives
\begin{align}
\mathbbm{E}[Y_i\mid Z_i=z]
&=\sum_{d\in\{0,1\}}\sum_{\mathbf z_{-i}}\sum_{\mathbf d_{-i}}
\mathbbm{E}\big[Y_i(d,\mathbf d_{-i})\mid D_i(z)=d,\mathbf D_{-i}(\mathbf z_{-i})=\mathbf d_{-i}\big]\notag\\
&\qquad\qquad\cdot\chi_z(\mathbf d_{-i}\mid d;\mathbf z_{-i})\,\pi_z(\mathbf z_{-i})\,\upsilon_z(d)
\label{expansion_Zminus_simplified}\\
&=\sum_{d\in\{0,1\}}\sum_{\mathbf z_{-i}}\sum_{\mathbf d_{-i}}
\bar Y(d,\mathbf d_{-i};\mathbf z_{-i})\,\chi_z(\mathbf d_{-i}\mid d;\mathbf z_{-i})\,\pi_z(\mathbf z_{-i})\,\upsilon_z(d).
\label{EY_given_Zi_z_mu}
\end{align}
Adding and subtracting $\sum_{d}\sum_{\mathbf z_{-i}}\sum_{\mathbf d_{-i}}\bar Y(d,\mathbf d_{-i};\mathbf z_{-i})\,\chi_z(\mathbf d_{-i}\mid d;\mathbf z_{-i})\,\pi^\ast(\mathbf z_{-i})\,\upsilon_z(d)$ in \eqref{EY_given_Zi_z_mu} gives
\begin{equation}
\mathbbm{E}[Y_i\mid Z_i=z]
=\sum_{d\in\{0,1\}}\sum_{\mathbf z_{-i}}\sum_{\mathbf d_{-i}}
\bar Y(d,\mathbf d_{-i};\mathbf z_{-i})\,\chi_z(\mathbf d_{-i}\mid d;\mathbf z_{-i})\,\pi^\ast(\mathbf z_{-i})\,\upsilon_z(d)
+B_z,
\label{EY_given_Zi_decomp}
\end{equation}
where
\begin{equation}
B_z=\sum_{d\in\{0,1\}}\sum_{\mathbf z_{-i}}\sum_{\mathbf d_{-i}}
\bar Y(d,\mathbf d_{-i};\mathbf z_{-i})\,\chi_z(\mathbf d_{-i}\mid d;\mathbf z_{-i})\,\zeta_z(\mathbf z_{-i})\,\upsilon_z(d)
\label{Bz_def_new}
\end{equation}
is the bias term of Theorem~\ref{prop:iv}. The first term on the right-hand side of \eqref{EY_given_Zi_decomp} is denoted $\mathbbm{E}[Y_i(D_i(z),\mathbf D_{-i})]$. Taking the difference between $z=1$ and $z=0$ gives the reduced-form decomposition
\begin{equation}
\mathbbm{E}[Y_i\mid Z_i=1]-\mathbbm{E}[Y_i\mid Z_i=0]
=\mathbbm{E}\big[Y_i(D_i(1),\mathbf D_{-i})-Y_i(D_i(0),\mathbf D_{-i})\big]+(B_1-B_0).
\label{RF_decomp_new}
\end{equation}
As in the standard IV argument, partitioning units into principal strata and invoking condition \eqref{iv: mono} gives
\begin{align}
&\mathbbm{E}\big[Y_i(D_i(1),\mathbf D_{-i})\big]-\mathbbm{E}\big[Y_i(D_i(0),\mathbf D_{-i})\big]\notag\\
&\qquad=\mathbbm{E}\big[Y_i(1,\mathbf D_{-i})-Y_i(1,\mathbf D_{-i})\mid D_i(1)=D_i(0)=1\big]\Pr(D_i(1)=D_i(0)=1)\notag\\
&\qquad\quad+\mathbbm{E}\big[Y_i(1,\mathbf D_{-i})-Y_i(0,\mathbf D_{-i})\mid D_i(1)>D_i(0)\big]\Pr(D_i(1)>D_i(0))\notag\\
&\qquad\quad+\mathbbm{E}\big[Y_i(0,\mathbf D_{-i})-Y_i(0,\mathbf D_{-i})\mid D_i(1)=D_i(0)=0\big]\Pr(D_i(1)=D_i(0)=0).
\label{eq:strata-decomp}
\end{align}
The terms for the strata $\{D_i(1)=D_i(0)=1\}$ and $\{D_i(1)=D_i(0)=0\}$ are identically zero, so \eqref{eq:strata-decomp} reduces to
\begin{align}
&\mathbbm{E}\big[Y_i(D_i(1),\mathbf D_{-i})-Y_i(D_i(0),\mathbf D_{-i})\big]\notag\\
&\qquad=\Pr(D_i(1)>D_i(0))\cdot\mathbbm{E}\big[Y_i(1,\mathbf D_{-i})-Y_i(0,\mathbf D_{-i})\mid D_i(1)>D_i(0)\big],
\label{eq:RF-compliers}
\end{align}
which equals $\Pr(D_i(1)>D_i(0))\cdot\mathrm{LADE}$ by the definition of the LADE and the law of iterated expectations. Moreover, $D_i=D_i(1)Z_i+D_i(0)(1-Z_i)$ and $D_i(z)\indep Z_i$ under Assumption~\ref{ass: IV conditions}, so by condition \eqref{iv: mono}
\[
\mathbbm{E}[D_i\mid Z_i=1]-\mathbbm{E}[D_i\mid Z_i=0]=\mathbbm{E}[D_i(1)-D_i(0)]=\Pr(D_i(1)>D_i(0)).
\]
Dividing \eqref{RF_decomp_new} by this first-stage difference gives
\[
\frac{\mathbbm{E}[Y_i\mid Z_i=1]-\mathbbm{E}[Y_i\mid Z_i=0]}{\mathbbm{E}[D_i\mid Z_i=1]-\mathbbm{E}[D_i\mid Z_i=0]}
=\mathrm{LADE}+\frac{B_1-B_0}{\Pr(D_i(1)>D_i(0))},
\]
which is the decomposition in part~(i).

\emph{Part~(ii).} Under Assumption~\ref{Instrument indepen} and, under $P$, the common instrument propensity, which makes $\Pr(I=k\mid Z_I=z)=N^{-1}$ for both $z$, $Z_i\indep\mathbf Z_{-i}$, so $\zeta_z(\mathbf z_{-i})=0$ for all $\mathbf z_{-i}$ and $B_1=B_0=0$ by \eqref{Bz_def_new}. The Wald ratio then equals the LADE, with no bias term.
\end{proof}

\paragraph{Proof of Theorem~\ref{prop:rd}.}\label{prf:prop:rd}
\begin{proof}
The proof evaluates the conditional mean outcome on each side of the cutoff as an average over the conditional law of others' treatments, takes one-sided limits, and collects the terms in which the law of others' treatments jumps at the cutoff. For $e>0$ and $\mathbf d_{-i}\in\{0,1\}^{N-1}$, let $\pi_{\pm e}(\mathbf d_{-i}):=\Pr(\mathbf D_{-i}=\mathbf d_{-i}\mid W_i=w_0\pm e)$ and $\bar\alpha_{\pm e}(\mathbf d_{-i}):=\mathbbm{E}[\alpha_i(\mathbf d_{-i})\mid W_i=w_0\pm e]$.

\emph{Part~(i).} Let $e>0$ be arbitrarily small. In a sharp design,
\[
\lim_{w\downarrow w_0}\mathbbm{E}[D_i\mid W_i=w]=1,
\qquad
\lim_{w\uparrow w_0}\mathbbm{E}[D_i\mid W_i=w]=0.
\]
Under arbitrary interference, consistency implies
\[
Y_i=\sum_{\mathbf d_{-i}\in\{0,1\}^{N-1}}\big(\alpha_i(\mathbf d_{-i})+\Delta_i(\mathbf d_{-i})D_i\big)\,\mathbbm{1}(\mathbf D_{-i}=\mathbf d_{-i}),
\]
where $\alpha_i(\mathbf d_{-i}):=Y_i(0,\mathbf d_{-i})$ and $\Delta_i(\mathbf d_{-i}):=Y_i(1,\mathbf d_{-i})-Y_i(0,\mathbf d_{-i})$. Taking expectations conditional on $W_i=w_0+e$ gives
\begin{equation}
\mathbbm{E}[Y_i\mid W_i=w_0+e]
=\sum_{\mathbf d_{-i}}\bar\alpha_{+e}(\mathbf d_{-i})\,\pi_{+e}(\mathbf d_{-i})
+\sum_{\mathbf d_{-i}}\Delta(\mathbf d_{-i})\,\mathbbm{E}[D_i\mathbbm{1}(\mathbf D_{-i}=\mathbf d_{-i})\mid W_i=w_0+e],
\label{eq:EY_RD}
\end{equation}
where $\mathbbm{E}[\alpha_i(\mathbf d_{-i})\mathbbm{1}(\mathbf D_{-i}=\mathbf d_{-i})\mid W_i=w_0+e]=\bar\alpha_{+e}(\mathbf d_{-i})\,\pi_{+e}(\mathbf d_{-i})$ by condition \eqref{rdd int: uncon} of Assumption~\ref{rdd continuity}, and $\mathbbm{E}[\Delta_i(\mathbf d_{-i})D_i\mathbbm{1}(\mathbf D_{-i}=\mathbf d_{-i})\mid W_i=w_0+e]=\Delta(\mathbf d_{-i})\,\mathbbm{E}[D_i\mathbbm{1}(\mathbf D_{-i}=\mathbf d_{-i})\mid W_i=w_0+e]$ by homogeneity of the direct effects. Similarly,
\begin{equation}
\mathbbm{E}[Y_i\mid W_i=w_0-e]
=\sum_{\mathbf d_{-i}}\bar\alpha_{-e}(\mathbf d_{-i})\,\pi_{-e}(\mathbf d_{-i})
+\sum_{\mathbf d_{-i}}\Delta(\mathbf d_{-i})\,\mathbbm{E}[D_i\mathbbm{1}(\mathbf D_{-i}=\mathbf d_{-i})\mid W_i=w_0-e].
\label{eq:EY_RD 2}
\end{equation}
Subtracting \eqref{eq:EY_RD 2} from \eqref{eq:EY_RD} and letting $e\downarrow0$ gives
\begin{align}
&\lim_{w\downarrow w_0}\mathbbm{E}[Y_i\mid W_i=w]-\lim_{w\uparrow w_0}\mathbbm{E}[Y_i\mid W_i=w]
=\sum_{\mathbf d_{-i}}\lim_{e\downarrow0}\bar\alpha_{+e}(\mathbf d_{-i})\cdot\lim_{e\downarrow0}\big\{\pi_{+e}(\mathbf d_{-i})-\pi_{-e}(\mathbf d_{-i})\big\}\notag\\
&\qquad+\sum_{\mathbf d_{-i}}\Delta(\mathbf d_{-i})\cdot\lim_{e\downarrow0}\big\{\mathbbm{E}[D_i\mathbbm{1}(\mathbf D_{-i}=\mathbf d_{-i})\mid W_i=w_0+e]\notag\\
&\qquad\qquad\qquad\qquad\qquad\qquad-\mathbbm{E}[D_i\mathbbm{1}(\mathbf D_{-i}=\mathbf d_{-i})\mid W_i=w_0-e]\big\},
\label{rdd expansion}
\end{align}
where the first sum uses condition \eqref{rdd int: cont }, that is, $\lim_{e\downarrow0}\bar\alpha_{+e}(\mathbf d_{-i})=\lim_{e\downarrow0}\bar\alpha_{-e}(\mathbf d_{-i})$, and the existence of the one-sided limits of $\pi_{\pm e}$. In the sharp design, $\mathbbm{E}[D_i\mathbbm{1}(\mathbf D_{-i}=\mathbf d_{-i})\mid W_i=w_0+e]=\pi_{+e}(\mathbf d_{-i})${, with $\lim_{e\downarrow0}\pi_{\pm e}=\pi_\pm$,} and $\mathbbm{E}[D_i\mathbbm{1}(\mathbf D_{-i}=\mathbf d_{-i})\mid W_i=w_0-e]=0$. {Substituting these into \eqref{rdd expansion} gives
\[
\lim_{w\downarrow w_0}\mathbbm{E}[Y_i\mid W_i=w]-\lim_{w\uparrow w_0}\mathbbm{E}[Y_i\mid W_i=w]
=\sum_{\mathbf d_{-i}}\lim_{e\downarrow0}\bar\alpha_{+e}(\mathbf d_{-i})\,\big\{\pi_+(\mathbf d_{-i})-\pi_-(\mathbf d_{-i})\big\}
+\sum_{\mathbf d_{-i}}\Delta(\mathbf d_{-i})\,\pi_+(\mathbf d_{-i})
=\mathrm{Bias}_{\mathrm{RDD}}+\mathrm{ADE}_{+}.
\]
By homogeneity, $\mathbbm{E}[\Delta(\mathbf D_{-i})\mid W_i=w]=\sum_{\mathbf d_{-i}}\Delta(\mathbf d_{-i})\Pr(\mathbf D_{-i}=\mathbf d_{-i}\mid W_i=w)$, whose limit as $w\downarrow w_0$ is $\sum_{\mathbf d_{-i}}\Delta(\mathbf d_{-i})\pi_+(\mathbf d_{-i})$, the sum being finite; this gives the second expression for $\mathrm{ADE}_+$.

\emph{Part~(ii).} Under Assumption~\ref{rdd nojump}, $\pi_+=\pi_-=\pi_{w_0}$, so $\mathrm{Bias}_{\mathrm{RDD}}=0$ and $\mathrm{ADE}_{+}=\sum_{\mathbf d_{-i}}\Delta(\mathbf d_{-i})\pi_{w_0}(\mathbf d_{-i})=\mathrm{ADE}_{w_0}$. The limit of $\mathbbm{E}[\Delta(\mathbf D_{-i})\mid W_i=w]$ as $w\uparrow w_0$ is $\sum_{\mathbf d_{-i}}\Delta(\mathbf d_{-i})\pi_-(\mathbf d_{-i})=\mathrm{ADE}_{w_0}$ as well, so $\mathrm{ADE}_{w_0}=\lim_{w\to w_0}\mathbbm{E}[\Delta(\mathbf D_{-i})\mid W_i=w]$.

\emph{Part~(iii).} By Assumption~\ref{running indepen}, $\mathbf D_{-i}=(\mathbbm{1}(W_j\ge w_0))_{j\neq i}$ is a function of $\mathbf W_{-i}\indep W_i$, so $\mathbf D_{-i}\indep W_i$ and $\Pr(\mathbf D_{-i}=\mathbf d_{-i}\mid W_i=w)=\Pr(\mathbf D_{-i}=\mathbf d_{-i})$ for all $w$. Hence Assumption~\ref{rdd nojump} holds with $\pi_{w_0}=\Pr(\mathbf D_{-i}=\cdot\,)$, and part~(ii) gives $\mathrm{ADE}_{w_0}=\sum_{\mathbf d_{-i}}\Delta(\mathbf d_{-i})\Pr(\mathbf D_{-i}=\mathbf d_{-i})=\mathbbm{E}[\Delta(\mathbf D_{-i})]$. Under the law $P$ of a randomly drawn unit, this argument applies to each unit $k$, and, with $f_k$ the density of $W_k$ and $q(k;w):=f_k(w)/\sum_{j=1}^Nf_j(w)$, $\Pr(\mathbf D_{-I}=\mathbf d_{-i}\mid W_I=w)=q(k;w)\Pr(\mathbf D_{-k}=\mathbf d_{-i})$, where $k$ is the unit omitted from $\mathbf d_{-i}$. Continuity of the densities at $w_0$, with $\sum_jf_j(w_0)>0$, makes $q(k;\cdot)$ continuous there, so the one-sided limits of these laws coincide, with $\pi_{w_0}(\mathbf d_{-i})=q(k;w_0)\Pr(\mathbf D_{-k}=\mathbf d_{-i})$. When the densities share a common value at $w_0$, $q(k;w_0)=N^{-1}$, so $\pi_{w_0}=\Pr(\mathbf D_{-I}=\cdot\,)$ and $\mathrm{ADE}_{w_0}=N^{-1}\sum_{k=1}^N\mathbbm{E}[\Delta(\mathbf D_{-k})]=\mathbbm{E}[\Delta(\mathbf D_{-I})]$. Otherwise, part~(ii) gives $\mathrm{ADE}_{w_0}=\sum_{k=1}^N q(k;w_0)\,\mathbbm{E}[\Delta(\mathbf D_{-k})]$.}
\end{proof}

\paragraph{Proof of Theorem~\ref{thm:did}.}\label{prf:thm:did}
\begin{proof}
The proof substitutes the potential-outcome representation of the observed outcomes into the DiD formula, adds and subtracts the treated group's mean untreated outcome in period~1, and averages the untreated trends over each group's law of $\mathbf D_{-i}$.

\emph{Part~(i).} Fix a unit $i$; all expectations are under $P$, and a configuration identifies the drawn unit, to which the unit-level conditions apply. By the definition \eqref{DiD indentification formula} of the covariate-adjusted DiD formula,
\begin{align}
\Phi_{\mathrm{ITR}}(P)
&=\mathbbm{E}\Big[\mathbbm{E}[Y_{i1}\mid W_i,D_i=1]-\mathbbm{E}[Y_{i0}\mid W_i,D_i=1]\,\Big|\,D_i=1\Big]\notag\\
&\quad-\mathbbm{E}\Big[\mathbbm{E}[Y_{i1}\mid W_i,D_i=0]-\mathbbm{E}[Y_{i0}\mid W_i,D_i=0]\,\Big|\,D_i=1\Big].
\label{eq:did-start}
\end{align}
By Assumption~\ref{ass:consistency}, $Y_{i1}=Y_{i1}(D_i,\mathbf D_{-i})$, and because no unit is treated in period~0, $Y_{i0}=Y_{i0}(\mathbf 0)$. Hence
\begin{align*}
\mathbbm{E}[Y_{i1}\mid W_i,D_i=1]&=\mathbbm{E}[Y_{i1}(1,\mathbf D_{-i})\mid W_i,D_i=1],\\
\mathbbm{E}[Y_{i1}\mid W_i,D_i=0]&=\mathbbm{E}[Y_{i1}(0,\mathbf D_{-i})\mid W_i,D_i=0],\\
\mathbbm{E}[Y_{i0}\mid W_i,D_i=d]&=\mathbbm{E}[Y_{i0}(\mathbf 0)\mid W_i,D_i=d],\qquad d\in\{0,1\}.
\end{align*}
Substituting these expressions into \eqref{eq:did-start} gives
\begin{align}
\Phi_{\mathrm{ITR}}(P)
&=\mathbbm{E}\Big[\mathbbm{E}[Y_{i1}(1,\mathbf D_{-i})\mid W_i,D_i=1]-\mathbbm{E}[Y_{i0}(\mathbf 0)\mid W_i,D_i=1]\,\Big|\,D_i=1\Big]\notag\\
&\quad-\mathbbm{E}\Big[\mathbbm{E}[Y_{i1}(0,\mathbf D_{-i})\mid W_i,D_i=0]-\mathbbm{E}[Y_{i0}(\mathbf 0)\mid W_i,D_i=0]\,\Big|\,D_i=1\Big].
\label{eq:did-interf}
\end{align}
Adding and subtracting $\mathbbm{E}\big[\mathbbm{E}[Y_{i1}(0,\mathbf D_{-i})\mid W_i,D_i=1]\mid D_i=1\big]$ in \eqref{eq:did-interf} gives
\begin{align}
\Phi_{\mathrm{ITR}}(P)
&=\underbrace{\mathbbm{E}\Big[\mathbbm{E}\big[Y_{i1}(1,\mathbf D_{-i})-Y_{i1}(0,\mathbf D_{-i})\mid W_i,D_i=1\big]\,\Big|\,D_i=1\Big]}_{\mathrm{ADTT}}\notag\\
&\quad+\mathbbm{E}\Big[\mathbbm{E}[Y_{i1}(0,\mathbf D_{-i})\mid W_i,D_i=1]-\mathbbm{E}[Y_{i0}(\mathbf 0)\mid W_i,D_i=1]\,\Big|\,D_i=1\Big]\notag\\
&\quad-\mathbbm{E}\Big[\mathbbm{E}[Y_{i1}(0,\mathbf D_{-i})\mid W_i,D_i=0]-\mathbbm{E}[Y_{i0}(\mathbf 0)\mid W_i,D_i=0]\,\Big|\,D_i=1\Big].
\label{eq:did-interf 1}
\end{align}
By the law of iterated expectations and condition \eqref{unif:did-meanindep} of Assumption~\ref{did: parallel trend}, for $d\in\{0,1\}$,
\begin{align*}
&\mathbbm{E}[Y_{i1}(0,\mathbf D_{-i})-Y_{i0}(\mathbf 0)\mid W_i,D_i=d]\\
&\qquad=\sum_{\mathbf d_{-i}}\mathbbm{E}[Y_{i1}(0,\mathbf d_{-i})-Y_{i0}(\mathbf 0)\mid W_i,D_i=d,\mathbf D_{-i}=\mathbf d_{-i}]\,\pi_d(\mathbf d_{-i};W_i)\\
&\qquad=\sum_{\mathbf d_{-i}}\mathbbm{E}[Y_{i1}(0,\mathbf d_{-i})-Y_{i0}(\mathbf 0)\mid W_i,D_i=d]\,\pi_d(\mathbf d_{-i};W_i),
\end{align*}
where the sums run over $\mathbf d_{-i}$ with $\pi_d(\mathbf d_{-i};W_i)>0$. Because $\sum_{\mathbf d_{-i}}\pi_d(\mathbf d_{-i};W_i)=1$, the right-hand side equals $\text{trend}_d(W_i)$ as defined in Theorem~\ref{thm:did}. The last two terms of \eqref{eq:did-interf 1} therefore equal
\[
\mathbbm{E}\big[\text{trend}_1(W_i)-\text{trend}_0(W_i)\mid D_i=1\big]=\mathrm{Bias}_{\mathrm{DiD}}.
\]

\emph{Part~(ii).} By condition~\eqref{eq:CAI}, $\mathbf D_{-i}\indep D_i\mid W_i$, so $\pi_1(\mathbf d_{-i};W_i)=\pi_0(\mathbf d_{-i};W_i)=\pi(\mathbf d_{-i};W_i)$. The bias from part~(i) is then
\begin{align*}
\mathrm{Bias}_{\mathrm{DiD}}
&=\mathbbm{E}\Big[\sum_{\mathbf d_{-i}}\Big\{\big(\mathbbm{E}[Y_{i1}(0,\mathbf d_{-i})\mid W_i,D_i=1]-\mathbbm{E}[Y_{i0}(\mathbf 0)\mid W_i,D_i=1]\big)\\
&\qquad-\big(\mathbbm{E}[Y_{i1}(0,\mathbf d_{-i})\mid W_i,D_i=0]-\mathbbm{E}[Y_{i0}(\mathbf 0)\mid W_i,D_i=0]\big)\Big\}\\
&\qquad\cdot\pi(\mathbf d_{-i};W_i)\,\Big|\,D_i=1\Big]=0,
\end{align*}
where the last equality follows from condition \eqref{PT} of Assumption~\ref{did: parallel trend}.
\end{proof}

\paragraph{Proof of Theorem~\ref{thm:did_dynamic}.}\label{prf:thm:did_dynamic}
\begin{proof}
The proof follows that of Theorem~\ref{thm:did}, with an additional anticipation term that Assumption~\ref{did: no anticipation dyn} removes. Throughout, $\pi_d(\mathbf d_{-i};W_i):=\Pr(\mathbf D_{-i}=\mathbf d_{-i}\mid W_i,D_i=d)$ and $\pi(\mathbf d_{-i};W_i):=\Pr(\mathbf D_{-i}=\mathbf d_{-i}\mid W_i)$.

\emph{Part~(i).} Fix a unit $i$, with expectations under $P$ as in the proof of Theorem~\ref{thm:did}. The covariate-adjusted DiD formula is
\begin{equation}
\Phi_{\mathrm{ITR}}(P)
=\mathbbm{E}\Big[\mathbbm{E}[Y_{i1}-Y_{i0}\mid W_i,D_i=1]-\mathbbm{E}[Y_{i1}-Y_{i0}\mid W_i,D_i=0]\,\Big|\,D_i=1\Big].
\label{eq:did-start-dyn}
\end{equation}
Under arbitrary interference, potential outcomes in period $t$ are indexed by the entire treatment vector, and consistency in both periods gives $Y_{it}=Y_{it}(\mathbf D)=Y_{it}(D_i,\mathbf D_{-i})$. The observed outcome therefore equals
\[
Y_{it}=\sum_{\mathbf d_{-i}\in\{0,1\}^{N-1}}\big\{Y_{it}(0,\mathbf d_{-i})+\big(Y_{it}(1,\mathbf d_{-i})-Y_{it}(0,\mathbf d_{-i})\big)D_i\big\}\,\mathbbm{1}(\mathbf D_{-i}=\mathbf d_{-i}),
\]
so that
\begin{align*}
\mathbbm{E}[Y_{it}\mid W_i,D_i=1]&=\mathbbm{E}[Y_{it}(0,\mathbf D_{-i})\mid W_i,D_i=1]+\mathbbm{E}[Y_{it}(1,\mathbf D_{-i})-Y_{it}(0,\mathbf D_{-i})\mid W_i,D_i=1],\\
\mathbbm{E}[Y_{it}\mid W_i,D_i=0]&=\mathbbm{E}[Y_{it}(0,\mathbf D_{-i})\mid W_i,D_i=0].
\end{align*}
Substituting these expressions into \eqref{eq:did-start-dyn} gives
\begin{align}
\Phi_{\mathrm{ITR}}(P)
&=\mathbbm{E}\Big[\mathbbm{E}[Y_{i1}(0,\mathbf D_{-i})\mid W_i,D_i=1]\notag\\
&\qquad\quad+\mathbbm{E}[Y_{i1}(1,\mathbf D_{-i})-Y_{i1}(0,\mathbf D_{-i})\mid W_i,D_i=1]\,\Big|\,D_i=1\Big]\notag\\
&\quad-\mathbbm{E}\Big[\mathbbm{E}[Y_{i0}(0,\mathbf D_{-i})\mid W_i,D_i=1]\notag\\
&\qquad\quad+\mathbbm{E}[Y_{i0}(1,\mathbf D_{-i})-Y_{i0}(0,\mathbf D_{-i})\mid W_i,D_i=1]\,\Big|\,D_i=1\Big]\notag\\
&\quad-\mathbbm{E}\Big[\mathbbm{E}[Y_{i1}(0,\mathbf D_{-i})\mid W_i,D_i=0]-\mathbbm{E}[Y_{i0}(0,\mathbf D_{-i})\mid W_i,D_i=0]\,\Big|\,D_i=1\Big].
\label{eq:did-interf-dyn}
\end{align}
Rearranging the terms of \eqref{eq:did-interf-dyn} gives
\begin{align}
\Phi_{\mathrm{ITR}}(P)
&=\mathbbm{E}\Big[\mathbbm{E}[Y_{i1}(1,\mathbf D_{-i})-Y_{i1}(0,\mathbf D_{-i})\mid W_i,D_i=1]\,\Big|\,D_i=1\Big]\notag\\
&\quad+\mathbbm{E}\Big[\mathbbm{E}[Y_{i1}(0,\mathbf D_{-i})-Y_{i0}(0,\mathbf D_{-i})\mid W_i,D_i=1]\,\Big|\,D_i=1\Big]\notag\\
&\quad-\mathbbm{E}\Big[\mathbbm{E}[Y_{i1}(0,\mathbf D_{-i})-Y_{i0}(0,\mathbf D_{-i})\mid W_i,D_i=0]\,\Big|\,D_i=1\Big]\notag\\
&\quad-\mathbbm{E}\Big[\mathbbm{E}[Y_{i0}(1,\mathbf D_{-i})-Y_{i0}(0,\mathbf D_{-i})\mid W_i,D_i=1]\,\Big|\,D_i=1\Big].
\label{eq:did-interf-dyn-1}
\end{align}
The last term of \eqref{eq:did-interf-dyn-1} vanishes. By the law of iterated expectations, condition \eqref{unif:dyn-antic-meanindep}, and condition \eqref{NA dyn} of Assumption~\ref{did: no anticipation dyn},
\begin{align*}
&\mathbbm{E}[Y_{i0}(1,\mathbf D_{-i})-Y_{i0}(0,\mathbf D_{-i})\mid W_i,D_i=1]\\
&\qquad=\sum_{\mathbf d_{-i}}\mathbbm{E}[Y_{i0}(1,\mathbf d_{-i})-Y_{i0}(0,\mathbf d_{-i})\mid W_i,D_i=1]\,\pi_1(\mathbf d_{-i};W_i)=0.
\end{align*}
By the law of iterated expectations and condition \eqref{unif:dyn-meanindep}, for $d\in\{0,1\}$,
\begin{align*}
&\mathbbm{E}[Y_{i1}(0,\mathbf D_{-i})-Y_{i0}(0,\mathbf D_{-i})\mid W_i,D_i=d]\\
&\qquad=\sum_{\mathbf d_{-i}}\big(\mathbbm{E}[Y_{i1}(0,\mathbf d_{-i})\mid W_i,D_i=d]-\mathbbm{E}[Y_{i0}(0,\mathbf d_{-i})\mid W_i,D_i=d]\big)\,\pi_d(\mathbf d_{-i};W_i)\\
&\qquad=\text{trend}^{\mathrm{dyn}}_d(W_i).
\end{align*}
The second and third terms of \eqref{eq:did-interf-dyn-1} therefore equal
\[
\mathbbm{E}\big[\text{trend}^{\mathrm{dyn}}_1(W_i)-\text{trend}^{\mathrm{dyn}}_0(W_i)\mid D_i=1\big]=\mathrm{Bias}_{\mathrm{DiD_{dyn}}},
\]
and the first term is the ADTT.

\emph{Part~(ii).} By condition~\eqref{eq:CAI}, $\mathbf D_{-i}\indep D_i\mid W_i$, so $\pi_1(\mathbf d_{-i};W_i)=\pi_0(\mathbf d_{-i};W_i)=\pi(\mathbf d_{-i};W_i)$. The bias from part~(i) is then
\begin{align*}
\mathrm{Bias}_{\mathrm{DiD_{dyn}}}
&=\mathbbm{E}\Big[\sum_{\mathbf d_{-i}}\Big\{\big(\mathbbm{E}[Y_{i1}(0,\mathbf d_{-i})\mid W_i,D_i=1]-\mathbbm{E}[Y_{i0}(0,\mathbf d_{-i})\mid W_i,D_i=1]\big)\\
&\qquad-\big(\mathbbm{E}[Y_{i1}(0,\mathbf d_{-i})\mid W_i,D_i=0]-\mathbbm{E}[Y_{i0}(0,\mathbf d_{-i})\mid W_i,D_i=0]\big)\Big\}\\
&\qquad\cdot\pi(\mathbf d_{-i};W_i)\,\Big|\,D_i=1\Big]=0,
\end{align*}
where the last equality follows from condition \eqref{PT dyn} of Assumption~\ref{did: parallel trend dyn} and condition~\eqref{eq:CAI}.
\end{proof}

\setlength{\bibsep}{0pt plus 0.3ex}
{\small\SupplementReferences

}

\end{document}